\documentclass{article}
\pdfoutput=1

\usepackage[english]{babel}
\usepackage[letterpaper,top=2cm,bottom=2cm,left=2.75cm,right=2.75cm,marginparwidth=1.75cm]{geometry}

\usepackage{amsmath,amssymb,amsthm}
\usepackage{bbm}

\usepackage{pifont} 
\newcommand{\hollowstar}{\text{\ding{73}}}

\allowdisplaybreaks[4]

\usepackage{indentfirst}

\usepackage{authblk}

\usepackage{abstract}

\usepackage{geometry}
\usepackage[bookmarksnumbered]{hyperref}
\usepackage{appendix}
\usepackage{natbib}
\usepackage{xr}
\usepackage{pdfpages}

\usepackage{tikz}
\usepackage{subcaption}
\usetikzlibrary{positioning, fit, shapes.misc,arrows.meta, matrix}
\usetikzlibrary{calc} 

\usepackage{booktabs} 
\usepackage{array}    
\usepackage{multirow} 

\usepackage{float} 
\usepackage{graphicx}

\numberwithin{equation}{section}
\numberwithin{figure}{section}

\usepackage{cases} 
\usepackage{empheq} 

\usetikzlibrary{matrix}

\usepackage{thmtools,thm-restate}
\usepackage{cleveref}

\newtheorem{theorem}{Theorem}[section]
\newtheorem{lemma}[theorem]{Lemma}

\newtheorem{corollary}[theorem]{Corollary}
\newtheorem{definition}{Definition}[section]

\newtheorem{remark}{Remark}[section]

\title{Emergence of Homophily under Contextual Mechanisms\thanks{A more detailed title we used before is: ``Emergence of Homophily: Policy Requirement, Unrestricted Credit and Strategic Segmentation''. The model in this paper mainly follows the spirit of Schelling (1978).}} 

\author{Jiaxing Weng, Haijun Yang, Tongyu Wang\thanks{J. Weng: School of Economics, Ocean University of China, Qingdao, 266100, China; wengjiaxing@stu.ouc.edu.cn. H. Yang: School of Economics and Management, Beihang University, Beijing, 100191, China; navy@buaa.edu.cn.T. Wang: School of Economics, Ocean University of China, Qingdao, 266100, China; wty@ouc.edu.cn.}}

\begin{document}
\maketitle
\begin{abstract}
	\begin{sloppypar}
		This paper introduces a tractable model to study incentive-compatible homophily under both external environments—such as exogenous shocks or policy constraints—and internal micromotives based on interactive attributes. We propose a set of invariants that capture main features of homophily and the well-defined partition dynamics leading to perfect global homophily. The criteria for homophily formation are characterized via isomorphism. Within this framework, we demonstrate the emergence of macro-complementarity coupled with micro-substitution, where local individuals' utility function is nonlinear and submodular. We discuss two types of financial networks and their differences: hierarchical structure emerges from short-term liquidity transactions, whereas core–periphery structure is based on a stock-based perspective.
	\end{sloppypar}
\end{abstract}

\newpage
\section{Introduction}
\noindent\textbf{Motivation} ``Birds of a feather flock together'' suggests an emergent collective adaptation to contextual mechanisms, whereby the strategic interactions of agents possessing salient attributes endogenously shape the \textit{homophily} patterns of network under given constraints. 

For example, banks typically choose their strategic behaviors within a given macro- or micro-prudential regulatory framework. When liquidity environment deteriorates, banks, while complying with regulatory requirements (such as the statutory leverage ratio and reserve requirements), determine their optimal counterparties and transaction volumes based on their own parameters, such as leverage, cash holdings, and illiquid assets. Optimality in this context is always grounded in self-interest, and is commonly defined through concepts such as Nash equilibrium. In many special cases—such as when strategic complementarities or positive externalities exist within the trading network—the system tends to exhibit a unique form of homophily, with all banks clustering together. However, in settings akin to zero-sum games, this outcome may not hold (see Proposition \ref{Homogeneous and Homophily}); instead, the evolved stable state may consist of multiple distinct forms of homophily. Motivated by this, we propose a new analytical framework that allows for exogenous macro-level mechanisms while preserving the dynamic interactions among banks. This framework provides a consistent definition of optimality, characterizes the conditions for the emergence of new configurations of homophily, and captures both the internal structures of networks at their onset and the principles underlying transitions between different homophily patterns.

In this theory, we aim to address the following questions: What mechanisms drive the formation or differentiation of incentive-compatible homophily patterns? What are the inherent characteristics of a homophily partition? Under what internal structures can agents sustain a stable form of homophily? And how do we characterize the comparative statics of homophily with respect to external environments and individual attributes?

We study a model in which banks are subject to homogeneous macro-prudential policies and exogenous shocks, but face heterogeneous micro-prudential policy requirements and possess heterogeneous attributes. The liquidation demand function of bank $i$, denoted $s^i$, is submodular and nonlinear, and its fire-sale outcome $x^i$ depends on the policy context, market liquidity conditions, risk environment, and the choices of other banks. We allow $x^i \in \mathbb{R}$, where $x^i \leq 0$ is interpreted as credit creation and $x^i > 0$ as asset sales. This does not imply a zero-sum game; rather, it reflects the heterogeneity in the signs of banks' utilities.

We characterize homophily as different subspaces $\mathcal{P}_{s\,|\,t}$, each of them ensures the existence and uniqueness of a liquidation equilibrium. Within a given subspace, individual banks decide whether to belong to the current homophily according to an incentive compatibility condition. Perfection dynamics refers to the iterative application of joint operations within each subspace $\mathcal{P}_{s\,|\,t}$: banks that satisfy the relevant micromotives are retained within the homophily, while those whose incentives differ are guided to alternative homophily patterns. To capture this evolutionary process, we extract invariants within $\mathcal{P}_{s\,|\,t}$.

\vspace{0.15cm}

\noindent\textbf{Results} Our first main result characterizes the homophily transition under escalating exogenous shocks. We employ an equivalent form of the incentive compatibility condition, termed the crowding-out effect (henceforth, COE), to captures banks' interactive states and the alignment (or misalignment) of incentives to form homophily. When banks can be distinctly partitioned into those engaging in credit creation and those conducting asset sales, the two exhibit opposite incentives: the former tends to form a homophily, while the latter resists it (Lemma \ref{Lemma: Crowding-Out Effect Opposite}). However, zooming in to a local perspective, the credit-creating banks lack sufficient internal cohesion to sustain a stable cluster, whereas the latter's cluster configuration is micro-incentive-compatible (Lemma \ref{Lemma: Crowding-Out Effect Homophily}). Building on this, we generalize the COE to a broader range. 

In a two-bank setting, we examine how homophily-formation incentives evolve as systemic risk deteriorates. The comparative statics identify an ignition condition under risk-free environment, which determines the banks' initial liquidation state profiles—the starting point. This condition allows asymmetric behaviors: one bank to create credit and another liquidates assets, while guaranteeing the irreversibility of this configuration. The leverage-based threshold hierarchy rules out the possibility of role-swapping, yet the system still admits simultaneous credit creation under moderate risk (Proposition \ref{Propositions about Counterfactual Monopoly Space2}).

We then explore three types of homophily transition based on banks' relative attributes. The sequence of transitions is fully governed by the threshold hierarchy: highly leveraged banks, being more resilient to external shocks, adjust their behavior later than less leveraged ones (Proposition \ref{Propositions about Counterfactual Monopoly Space} and Proposition \ref{Propositions about Counterfactual Monopoly Space2}). Extending the analysis to general bank groups, we define the maximal bailout cluster $\mathcal{B}(\varepsilon)$ as the subset in which all banks maintain credit creation, and the maximal bail-in cluster $\mathcal{B}^{\mathcal{I}}(\varepsilon)$ as the subset in which all conduct asset sales. Theorem \ref{Decomposition and Compression Equivalence} (Compression Equivalence) describes the banking system's phase transition as financial environment deteriorate, providing the exogenous shock intervals corresponding to each unit reduction in the size of $\mathcal{B}(\varepsilon)$ and the associated COE states. The overall phase transition depends on local critical points and is equivalent to the state shifts of boundary banks of $\mathcal{B}(\varepsilon)$ and $\mathcal{B}^{\mathcal{I}}(\varepsilon)$. The proposed concept ``Chain'' further refines the connection between the threshold hierarchy and homophily transitions, distinguishing the partial order of liquidation volumes from the hierarchy—showing that an inferior bank may still exhibit a higher partial order in credit creation (see Figure \ref{Partial order and Hierarchy}).

We next turn to the homophily differentiation dynamics, which incorporate bank heterogeneity, risk environment, and the policy context. We establish the necessary and sufficient condition for the existence and uniqueness of the clearing system equilibrium, which requires sufficient market liquidity to accommodate exogenous shocks (Proposition \ref{Solutions in FCP}). At the same time, considering the inapplicability of such clearing process, we introduce a clearing structure that, while unrestricted, preserves the desirable properties of clearing equilibrium (Lemma \ref{Lemma ForDecomposition and Compression Equivalence}). We define this structure as homophily, and refer to the collection of subspaces evolving from a primitive generative space (i.e., the initial homophily) as the perfection generating space.

The evolutionary process is characterized by each bank's ``stay-or-leave'' decision following its assessment of individual or group COE states (termed the individual or partition perfection operation, respectively) and heterogeneous micro-prudential policy (termed the switch $p^s$). We demonstrate that when strict regulation is imposed on credit-creating banks ($p^s = 1$), the partition perfection operation on the subspace $\mathcal{P}_{s\,|\,t}$ is equivalent to performing individual perfection for bail-out banks and partition perfection for bail-in banks (Remark \ref{The equivalent operations p^s=1}). Formally, this shows that the partition perfection operation is well-defined, satisfying the incentive compatibility condition within each cluster (Lemma \ref{Lemma: Crowding-Out Effect Opposite} and Lemma \ref{Lemma: Crowding-Out Effect Homophily}).

Furthermore, by imposing the finite-risk mitigation condition, we ensure that the outcome of the joint operation is well-defined. This assumption implies that the generalized supply within the current subspace $\mathcal{P}_{s\,|\,t}$ cannot simultaneously satisfy the demands of all banks (Inequality \ref{finite-risk mitigation2}). A direct implication is that, after the joint operation, the subset $\mathcal{B}(\varepsilon)$ persists within the current subspace, while $\mathcal{B}^{\mathcal{I}}(\varepsilon)$ either migrates or generates a new homophily subspace (Lemma \ref{Lemma For Decomposition Chains in any Generating Space}). These lead to a well-defined perfection dynamics, consistent with both individual and cluster-level incentive compatibility conditions.

The above perfection dynamics rest on the fact that each subspace $\mathcal{P}_{s\,|\,t}$ can be weakly decomposed into three regions or strongly decomposed into four regions, two of which necessarily contain $\mathcal{B}(\varepsilon)$ and $\mathcal{B}^{\mathcal{I}}(\varepsilon)$, respectively (see Figure \ref{fig: Weak and Strong Decomposition}, Theorem \ref{Decomposition and Compression Equivalence} and Theorem \ref{Partition-induced transition}).  By adjusting the banks' relative attribute values, we can flexibly alter the number of banks located within each region of the weakly decomposed sets of $\mathcal{B}^{\mathcal{I}}(\varepsilon)\cup\mathcal{B}(\varepsilon)$ (Lemma \ref{Lemma ForPartitionInduced Equilibrium Transition Snd}). Under strong decomposition, the two intermediate regions form two pre-regular chains, any one of which can be further decomposed into a regular chain and another distinct pre-regular chain (Lemma \ref{Lemma for Chain}). 

Theorem \ref{Partition-induced transition} establishes that applying similar regulatory adjustments across banks allows precise control over the number of regular chains in the decomposition, while keeping the existing level of exogenous shocks unchanged. This can be viewed as the inverse process of Theorem \ref{Decomposition and Compression Equivalence}, through which we construct a bijective correspondence among individual bank assets $A_i$, leverage ratios $\theta^i$, and the exogenous shock $\varepsilon$, and verify the isomorphism structure among them (see Figure \ref{Maps}, Lemma \ref{Lemma ForPartitionInduced Equilibrium Transition} and Lemma \ref{Lemma For Commutative Diagram}). Therefore, the perfection dynamics can be also viewed as iterative weak decompositions of the primitive generative space $\mathcal{P}_{0}$ (see Figure \ref{fig: Weak and Strong Decomposition} and Figure \ref{PGS_p=1}). 

Moreover, Theorem \ref{Partition-induced transition} together with the isomorphic relation among variables determines both $(i)$ the strength of internal cohesion among banks forming a stable homophily—i.e., sufficient to constitute a pure $\mathcal{B}(\varepsilon)$ or $\mathcal{B}^{\mathcal{I}}(\varepsilon)$—and $(ii)$ the logic of transitions across distinct homophily structures, namely, how changes in individual relative attributes induce transitions to another homophily. In other words, banks with similar attributes values tend to cluster under given exogenous shocks and policy contexts, and the attribute range required to sustain a specific homophily can be precisely tuned. Proposition \ref{Homogeneous and Homophily} presents a comparative-statics application, showing how we can easily extend from a simple structural configuration to more complex ones through Theorem \ref{Partition-induced transition}.

Having characterized the perfection dynamics, we further describe its equilibrium properties. Each subspace $\mathcal{P}_{s\,|\,t}$ within the perfection generating space is mutually independent or evolves independently, and banks exhibit a ``no-turning-back'' property—they never return any subspace they have previously exited. We demonstrate that the perfection dynamics necessarily converge to an equilibrium state in finite time, as both the number of banks and weak decomposition steps are finite. Moreover, we provide the exact upper and lower time bounds for this iterative process, determined by the relative number and attributes of banks in the binary structure ${_\chi}\mathcal{P}_{s\,|\,t}\,\cup\,{^\chi}\mathcal{P}_{s\,|\,t}$ of the primitive generative space. The maximum time complexity is $\mathcal{O}(n^2)$; however, for the special tiered structure described in Proposition \ref{Homogeneous and Homophily}, equilibrium arises within three iterations (see Figure \ref{Homgeneity and hierarchy}).  Furthermore, once the perfection dynamics stabilize, each subspace $\mathcal{P}_{s\,|\,t}$ becomes purified.

In addition to the aforementioned properties, we further generalize several of Schelling's insights \citep{schelling1978micromotives}. We introduce two distinct types of individual incentives: one reflecting the expectation of others' behavior (termed Schelling point) and the other representing the optimal behavior under uncertainty (termed Knightian point). For banks in $\mathcal{B}^{\mathcal{I}}(\varepsilon)$, the Schelling point corresponds to the micro-prudential policy imposed on $\mathcal{B}(\varepsilon)$, i.e., $p^s = 1$, while the Knightian point corresponds to the no-turning-back property (Proposition \ref{Schelling Point and Knightian Point}). These two perspectives are dual: they jointly guarantee the maximum utility of $\mathcal{B}^{\mathcal{I}}(\varepsilon)$, whereas the minimum utility of $\mathcal{B}(\varepsilon)$. Moreover, for $\mathcal{B}^{\mathcal{I}}(\varepsilon)$, the entry of new banks that preserve the subspace's structural characteristics increases the utility of all existing banks within $\mathcal{B}^{\mathcal{I}}(\varepsilon)$ (termed emergence of complementarity, Proposition \ref{Emergence of Complementarity}); in contrast, for $\mathcal{B}(\varepsilon)$, the reverse holds (termed persistence of substitution, Proposition \ref{Persistence of Substitution}). We find this asymmetry persists when the utility function is supmodular (Proposition \ref{Emergence of Substitution} and Proposition \ref{Persistence of Complementarity}).

Lastly, we discuss several illustrative examples, with particular attention to financial networks.   Proposition \ref{Homogeneous and Homophily} shows that the hierarchical structure originates from short-term liquidity trading, which can be formulated as either a matching or an optimal transport problem. In contrast, the core–periphery structure is examined from a stock-based view. We examine this highly clustered and self-fulfilling network formation in the context of bank-runs.
\vspace{0.15cm}

\noindent\textbf{Related Literature} Our paper is related to the literature on sorting and matching, such as \citet{boerma2023composite}, \citet{levy2015preferences}, \citet{anderson2024comparative}, \citet{lindenlaub2017sorting}, \citet{moldovanu2007contests}, and \citet{staab2024formation}. Among them, \citet{boerma2023composite} employs optimal transport theory to characterize composite sorting, while we analyze the homophily formation dynamics of heterogeneous agents facing heterogeneous micro-prudential policy requirements, homogeneous macro-prudential policies, and exogenous shocks. Building on the incentive compatibility conditions discussed in \citet{levy2015preferences}, we develop a theory that fully characterizes the invariants and equilibrium properties of perfection dynamics. Although our methodology differs from these studies, it yields several common structural insights regarding group structure and their dynamics: $(i)$ within a group (homophily), agents with different attributes and utilities (i.e., regional hierarchy) can coexist, and this property holds independently other groups' internal configurations; $(ii)$ we characterize the rationale for transitions between different groups, clarifying the boundaries between groups; $(iii)$ we introduce a concept analogous to the ``layer'' notion in optimal problem proposed by \citet{boerma2023composite} and \citet{delon2012local}, which we refer to as homophily, allowing independent analysis across different homophilies; and $(iv)$ comparative statics are tractable. Unlike \citet{anderson2024comparative}, who focus on positive quadrant dependence, we employ an isomorphism among variables to characterize interdependence.\footnote{\citet{ham2021notionsanonymityfairnesssymmetry} analyzes the strategic structure of finite strategic-form games through the lens of game isomorphism.} This enables us to extend from a simple intra-group structure to a complex combination of bank settings without altering the external group's macro structure.

In a broader perspective, our theory can be also applied to the analysis of the Schelling model \citep{schelling1978micromotives}. A strand of literature investigates segregation dynamics. For instance, \citet{ortega2021schelling} develop a residential segregation system on a lattice and explore its dynamics driven by agents' certain preferences and tolerance levels, employing numerical simulations to analyze the resulting statistical properties of networks. Similarly, \citet{zakine2024socioeconomic} apply mean-field approximation to agent-based model to analyze segregation dynamics. Their simulation results reveal phase separation phenomena and identify the corresponding critical points and scaling exponents. Our group segregation results are derived analytically from the well-defined perfection dynamics, where the formation of group emerges endogenously from incentive-compatible interactions among agents. Moreover, our theory allows for explicit analysis of how exogenous variables influence group formation. Agents' micromotives are also constrained by locally heterogeneous prudential policies, ensuring the dynamics operator satisfies within-cluster cohesion. Furthermore, the proposed invariants characterize the process of equilibrium phase transitions, identify the boundaries of transitions both between and within groups as well as the equivalence relations among them. In addition, our framework can trace how macro-level behaviors emerge endogenously from micromotives.

Our analysis of the banking system also relates to the literature on endogenous formation of financial networks \citep{jackson2021systemic, bernard2022bail,elliott2021systemic}, and the identification of liquidity transmission channels \citep{chen2025trade, acharya2024liquidity, eisenschmidt2024monetary, hachem2021liquidity, kahn2021sources, denbee2021network, wang2025biggest}.

The remainder of the paper is organized as follows: Section \ref{Section: Model} introduces the model, Section \ref{Section: Properties of Equilibrium} provides the properties of clearing equilibrium and the equivalent incentive compatibility conditions, and characterizes the homophily transition under escalating exogenous shocks. Section \ref{Characterization of Invariants, Perfection and Homophily} proposes a set of invariants to develop the perfection dynamics and comparative statics. Section \ref{Section: Micromotives, Mechanism and Macrobehavior} generalizes Schelling's insights. Section \ref{Section: Discussion} discusses several illustrative examples. Section \ref{Conlusion} concludes our contribution.
\section{Model}\label{Section: Model}
\noindent In order to elucidate the model's core mechanisms and key insights with maximum transparency, we develop the analysis in a simplified interaction framework.
\vspace{0.15cm}

\noindent\textbf{Agents} We consider the financial market consisting of $n$ banks and $k$ industries (or investment portfolios, collectively referred to as industries), along with a sink node $Q$ (representing self-interested external entities). Individual banks are indexed by $\mathcal{I}_{\text{Bank}} = \{i, j, \ldots, n \mid n \in \mathbb{N}_{+}, n < +\infty\}$, and industries are indexed by $\mathcal{I}_{\text{Ind}} = \{1, \ldots, l, \ldots, k \mid k \in \mathbb{N}_{+}, k < +\infty\}$.

For each bank $i$, its initial assets, liabilities, and equity are denoted as $A_{i}, L_{i},$ and $E_{i}$, respectively. The balance sheet identity requires that $A_{i} = L_{i} + E_{i}$. Each bank allocates all of its assets across industries. Let $V_{ik}$ represent the amount of assets that bank $i$ invests in industry $k$. It follows that $A_{i} = \sum_{k} V_{ik}$. All liabilities $L_{i}$ and equity $E_{i}$ of the banks are ultimately absorbed by the sink node $Q$.
\vspace{0.15cm}

\noindent\textbf{Response to Crisis} We assume that the regulatory leverage ratio for banks is set at $\bar{\theta} \in (0,1)$, which implies $\frac{E_i}{A_i}\geq \bar{\theta}$. The economic system experiences an initial shock of magnitude $\varepsilon \in (0,1)$, which leads to a write-down of assets to $(1-\varepsilon) A_i$. In accordance with the accounting identity, the reduction in assets equals the reduction in equity, as short-term liabilities remain unchanged. The losses resulting from the shock are absorbed by arbitrageurs. These dynamics give rise to Lemma \ref{Drop in LR} and Corollary \ref{Rational sink node}, where Corollary \ref{Rational sink node} indicates that banks have an incentive to sell assets, that is, such sales are individually rational.\footnote{This does not suggest that the sink node abstains from injecting funds into the banking system. It will employ interventions in the form of equity when doing so yields a more favorable outcome (for sink node) than inaction. The incentives behind these decisions are discussed in detail in Online Appendix Section OA4.}

\begin{lemma} \label{Drop in LR}
	\textnormal{For any bank $i$, if $A_i \gg E_i$, then its leverage ratio declines after the exogenous shock. That is, if $\Delta A_i(A_i, \varepsilon) = \Delta E_i$, then $\theta_{t=0}^i > \theta_{t=1}^i$.}
\end{lemma}

\begin{corollary} \label{Rational sink node}
	\textnormal{Under the assumption of a rational sink node, bank $i$ can only comply with the leverage regulation by selling assets.}
\end{corollary}

The value of industry $k$ is given by $M_k = \sum_j V_{jk}$. Banks hold homogeneous disposal rights over industry assets, as formalized by the disposal matrix $\mathbf{\Pi} \in \mathbb{R}^{n \times k}$ (see Definition \ref{Matrix for disposal}). When a bank $i$ anticipates a fire sale of amount $s^i$ to meet the regulatory leverage requirement, it triggers a devaluation in the affected industries. The market impact on a industry is modeled by the devaluation factor $\kappa = e^{-\beta \sum_i x^i \pi_{ik}}$, where $\beta$ is the devaluation parameter and $x^i$ is the actual sales volume.\footnote{The framework for fire sales and the disposal matrix (corresponding to the ``relative liabilities matrix'' in their models) used here are adapted from the models established by \citet{bernard2022bail} and \citet{eisenberg2001systemic}.} Consequently, the post-sale value of the industry is reduced to $M_k \times \kappa$, resulting in a total devaluation loss of $M_k \times (1 - \kappa)$. 

We assume that industries possess no short-term capacity for expansion during a shock; therefore, the total fire sales volume in each industry $k$ must be non-negative, i.e., $\sum_i \beta x^i \pi_{ik} \geq 0$ for all $k$. In addition, although the contagion effect of fire sales propagates through the banking system whenever any bank sells an asset, the aggregate value of the industry remains unchanged. We will elaborate on this point in Online Appendix Section OA4 by examining the incentives of the sink node $Q$.

\begin{definition}\label{Matrix for disposal}
	\textnormal{(Disposal Matrix and Contagion Matrix)} 
	
	\textnormal{The \textit{disposal matrix} $\mathbf{\Pi} = (\pi_{ik})^{n \times k} \in \mathbb{R}^{n \times k}$ consists of each bank's relative asset share across industries, where $\pi_{ik}$ represents the proportion of bank $i$'s assets in industry $k$ relative to its total assets. The \textit{contagion matrix} $\mathbf{P} = (p_{ki})^{k \times n} \in \mathbb{R}^{k \times n}$ consists of each industry's relative value share across banks, where $p_{ki}$ represents the proportion of bank $i$'s assets in industry $k$ relative to the total value in that industry.}
	\begin{subequations}
		\begin{align}
			\pi_{ik} &= \begin{cases} \frac{V_{ik}}{\sum\limits_k V_{ik}} & \text{if} \ V_{ik} > 0\text{,} \\ 0 & \text{otherwise.} \end{cases} \\
			p_{ki} &= \begin{cases} \frac{V_{ik}}{\sum\limits_i V_{ik}} & \text{if} \ V_{ik} > 0\text{,} \\ 0 & \text{otherwise.} \end{cases}
		\end{align}
	\end{subequations}
\end{definition}
\vspace{0.15cm}

\noindent\textbf{Monopoly and Specialization} Bank $i$ has a \textit{monopoly} position in industry $k$ if $V_{ik}=M_k$ (i.e., $p_{ki}=1$). Bank  \textit{specializes} in industry $k$ if $V_{jk}>0$ and $V_{jl}=0,\,l\neq k$ (i.e., $\pi_{jk}=1$). Moreover, if a bank has more than one monopoly position, we add them together into one entry in order to get a full rank matrix. In Figure \ref{fig:network-adjacency, 2 to 2}, bank $A$ monopolizes industry $1$, while bank $B$ specializes in industry $2$.

We reformulate disposal matrix into its reduced-form $\mathring{\mathbf{\Pi}}$. The procedure is, we eliminate all columns corresponding to monopolized industries and the rows associated with banks who have monopolistic position from matrix  (i.e., for all bank $i$ who monopolizes industry $k$, we remove $\{V_{il},\forall l\in \mathcal{I}_{Ind}\}$ and $\{V_{jk},\forall j\in \mathcal{I}_{Bank}\}$ from matrix $\mathbf{\Pi}$). The same operation is applied to banks holding specialized positions and ultimately yielding the matrix $\mathring{\mathbf{\Pi}}$. 

Clearly, if not all banks monopolize or specialize an industry, then $\mathring{\mathbf{\Pi}} \neq \varnothing$. We assume the investment proportions of different banks across distinct industries in matrix $\mathring{\mathbf{\Pi}}^{i'\times k'}$ are linearly independent (i.e., there does not exist constant $\omega$ such that $(\pi'_{ik})^{1\times k'}=\omega \cdot (\pi'_{jk})^{1\times k'}$), which suggests a stochastic selection for these proportions\footnote{This assumption aligns with empirical evidence demonstrating inherent randomness in a subset of banking sector choices \citep{wang2025biggest}.}.
\begin{figure}[htbp]
	\centering
	
	\begin{subfigure}[b]{0.4\textwidth}
		\centering
		\begin{tikzpicture}[
			node distance=1.5cm and 3cm,  
			vertex/.style={circle, draw=black, thick, fill=gray!25, minimum size=10mm, font=\normalsize},  
			square/.style={rectangle, draw=black, thick, minimum size=10mm, font=\normalsize},  
			edge/.style={-{Stealth[length=3.5mm]}, thick, bend left=18},  
			labelstyle/.style={font=\large,fill=white, inner sep=1pt}  
			]
			
			\draw[red, dashed] (-1.2,-2.4) rectangle (4.5,2.5);
			
			\node[vertex] (A) at (0,1) {A};          
			\node[vertex, below=of A] (B) {B};       
			\node[square, right=2.4cm of A, yshift=0.5cm] (1) {1};  
			\node[square, below=of 1] (2) {2};       
			
			\draw[edge] (A) to node[midway, above left, labelstyle] {15} (1);  
			\draw[edge] (A) to node[midway, above, labelstyle] {10} (2);       
			\draw[edge] (B) to node[midway, right, labelstyle] {20} (2);       
			
		\end{tikzpicture}
		\caption{Network Structure}
	\end{subfigure}
	\hfill
	\begin{subfigure}[b]{0.1\textwidth}
		\centering
		\vspace{2.5cm} 
		{\Huge \boldmath$\Longleftrightarrow$}
		\vspace{2.5cm} 
	\end{subfigure}
	\hfill
	\begin{subfigure}[b]{0.4\textwidth}
		\centering
		\begin{tikzpicture}[
			matrixcell/.style={draw, minimum width=1.8cm, minimum height=1.8cm,fill=gray!7, anchor=center, font=\large},
			boundingbox/.style={draw=red, dashed, minimum width=\textwidth, minimum height=5cm}
			]
			
			\draw[red, dashed] (-3,-2.4) rectangle (2.7,2.5);
			
			\matrix (mat) [matrix of nodes,
			nodes={matrixcell},
			row sep=-\pgflinewidth,
			column sep=-\pgflinewidth,
			ampersand replacement=\&] {
				15 \& 10 \\
				0  \& 20 \\
			};
			
			\node[left=3pt of mat-1-1] {\textbf{A}};
			\node[left=3pt of mat-2-1] {\textbf{B}};
			\node[above=3pt of mat-1-1] {\textbf{1}};
			\node[above=3pt of mat-1-2] {\textbf{2}};
		\end{tikzpicture}
		\caption{Corresponding Adjacency Matrix}
	\end{subfigure}
	\caption{Network structure and its corresponding adjacency matrix representation}
	\label{fig:network-adjacency, 2 to 2}
\end{figure}
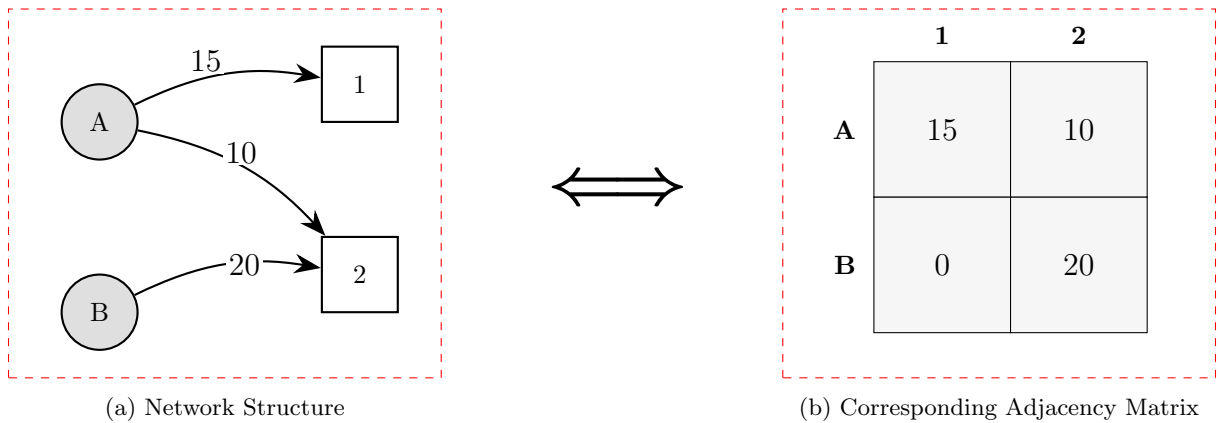
\subsection{Clearing Process}
\noindent During a crisis, banks with severely impaired leverage ratios need to decide and liquidate a certain amount of assets. The magnitude of this liquidation depends on (i) the individual bank's objectives, such as merely meeting regulatory requirements to resolve immediate distress, or preserving a higher leverage ratio to buffer against future uncertain shocks; and (ii) regulatory constraints on credit, for instance, requiring that the amount of assets sold by a bank be non-negative, i.e., $x^i \in \mathbb{R}_+$, which we term \textit{restricted credit}. The alternative, \textit{unrestricted credit}, permits $x^i \in \mathbb{R}$.

We adopt the unrestricted credit coupled with a \textit{basic clearing process} as our benchmark to determine magnitude of ``liquidation''. In this setup, $x^i \geq 0$ denotes \textit{asset sales}, while $x^i < 0$ corresponds to \textit{credit creation}. The logic of the basic clearing process, as formalized in Definition \ref{Clearing Process for Regulational Requirement}, is to liquidate exact amount to meet the regulatory requirement. We demonstrate that this benchmark is superior to a restricted credit regime (where $x^i \in \mathbb{R}_+$ is imposed) in terms of both individual rationality and social efficiency. Social efficiency lies in offsetting interbank devaluation spillovers, while individual rationality is reflected in benefiting from industrial booms and minimizing excessive asset depreciation. A comprehensive analysis of alternative clearing processes, including the restricted credit case, is provided in Online Appendix Section OA1.

\begin{definition}\label{Clearing Process for Regulational Requirement}
	Basic Clearing Process (BCP).
	
	\textnormal{(i) Post-Shock State: Following an exogenous shock $\varepsilon$, a bank $i$'s state is given by $E_i' = E_i - \Delta E_i$ and $A_i' = A_i - \Delta A_i = (1-\varepsilon) A_i$.}
	
	\textnormal{(ii) Bank Heterogeneity and Regulatory Status: Banks are heterogeneous, which implies that after the shock, some banks with relatively low leverage ratios satisfy $\frac{E_i'}{A_i'} < \bar{\theta}$, while banks with sufficient leverage satisfy $\frac{E_i'}{A_i'} \geq \bar{\theta}$. That is, $\theta^1 \geq \theta^2 \gg \cdots \gg \theta^j \geq \bar{\theta}$.}
	
	\textnormal{(iii) Liquidation Rule: To meet the requirement $\bar{\theta}$ while avoiding unnecessary asset devaluation, bank $i$ liquidates a notional value $s^i \in \mathbb{R}$, solving $\frac{E_i'}{A_i' - s^i} = \bar{\theta}$ (i.e., $s^i$ is solely a function of $\varepsilon$, $A_i$, $E_i$ and $\bar{\theta}$). Notably, if $\varepsilon = \theta^i$, then $s^i = A_i'$. Due to fire sale effect, the actual sales volume is $x^i$.}
\end{definition}
\subsection{Counterfactual Monopoly and Cluster}
\noindent Within our framework, homophily is modeled as incentive-compatible partitions. We aim to answer three questions: (i) the characteristics of an incentive-compatible partition; (ii) the internal structure of a partition—specifically, the attributes of a bank cluster that give rise to homophily; and (iii) comparative statics on homophily. 

We introduce two basic concepts which are helpful to address these questions: the \textit{Counterfactual Monopoly Space} and the \textit{Counterfactual Cluster Space}. They are designed to characterize and analyze the incentive-compatible partitions referenced from \citet{levy2015preferences}. A critical feature is that these incentive-compatible partitions are mutually independent. This property allows us to analyze the internal structure of each partition and conduct comparative statics in isolation, thereby significantly simplifying the problem by avoiding an intricate dissection of the entire system.
\begin{definition}\label{Counterfactual Monopoly Space}
	Counterfactual Monopoly Space
	
	\textnormal{Consider $n$ parallel spaces constructed in bijection with the rows of the original co-investment network, each satisfying the BCP. For the $i$-th space, the solution $\tilde{x}^{i}$ denotes the preimage $\Phi^{-1}(\tilde{s}^i)$, where the mapping $\Phi$ is given by:}
	\begin{equation}
		\Phi(\tilde{x}^i) = \tilde{x}^i + \sum\limits_{k}\left(1 - e^{-\beta \cdot \pi_{ik} \cdot \tilde{x}^i}\right) \cdot M_k \cdot p_{ki}
	\end{equation}
	
	\noindent\textnormal{The composite space formed by all parallel spaces is termed the \textit{Counterfactual Monopoly Space}, with its solution represented by the vector pairs $(\tilde{\mathbf{x}}, \tilde{\mathbf{s}})$. Furthermore, let $\varepsilon_{\tilde{s}^i=0}$ be the shock size such that $\tilde{s}^i(\varepsilon_{\tilde{s}^i=0}) = 0$ and thus $\tilde{x}^i=0$.}
\end{definition}

Analogous to the Counterfactual Monopoly Space (henceforth CMS), we define Counterfactual Cluster Space (henceforth CCS). Its solution, denoted by the vector $\tilde{\mathbf{\Phi}}\left(\mathbf{s}^{|\mathcal{P}|\times 1}\right)$, is implicitly defined as the solution to the corresponding system of equations. 
\begin{definition}\label{Counterfactual Cluster Space}
	Counterfactual Cluster Space
	
	\textnormal{The CCS is constructed by selecting subspaces from the CMS and combining them. The clearing process within $\mathcal{P}$ satisfies the BCP, and its equilibrium is given by the solution vector $\tilde{\mathbf{\Phi}}\left(\mathbf{s}^{|\mathcal{P}|\times 1}\right)$.}
\end{definition}

The left and right panels of Figure \ref{fig:Banking Modes} illustrate the CMS and CCS, respectively. We rewrite the solution system of the banks extracted from the CMS to form the CCS as $\mathbf{\Phi}^{-1}\left(\mathbf{s}^{|\mathcal{P}|\times 1}\right)$, where the individual solution is denoted by $\Phi^{-1}(\tilde{s}^i)$. An intuitive approach to compare the two spaces and analyze bank incentive compatibility is to take the difference between the two solution systems: $\tilde{\mathbf{\Phi}}\left(\mathbf{s}^{|\mathcal{P}|\times 1}\right) - \mathbf{\Phi}^{-1}\left(\mathbf{s}^{|\mathcal{P}|\times 1}\right)$. However, the viability of this metric hinges on the existence and uniqueness of the solutions, which are not always guaranteed. We explore this issue in detail in Online Appendix OA1.1, and provide necessary and sufficient conditions for the existence and uniqueness of solutions in Section \ref{Existence and Uniqueness of Equilibrium & Properties}.

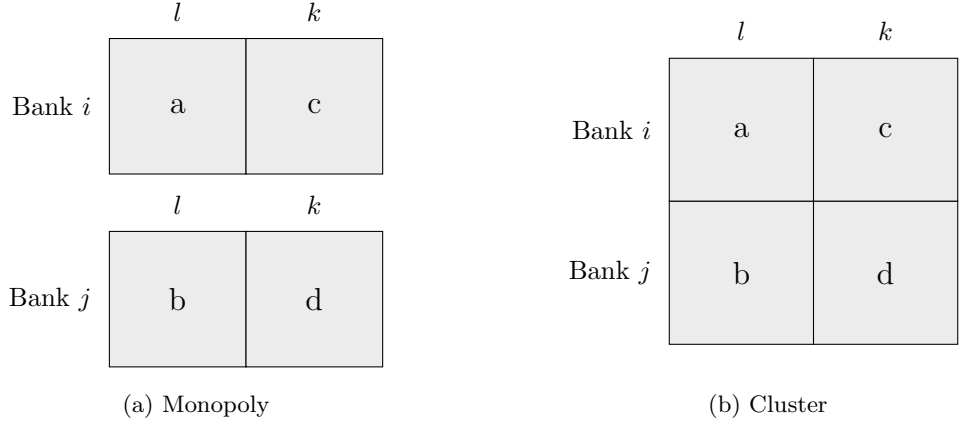
\begin{figure}[htbp]
	\centering
	
	\begin{subfigure}[b]{0.4\textwidth}
		\centering
		\begin{tikzpicture}[
			matrixcell/.style={draw, minimum width=1.8cm, minimum height=1.8cm,fill=gray!15, anchor=center, font=\large},
			]
			
			\matrix (mat1) [matrix of nodes,
			nodes={matrixcell},
			row sep=-\pgflinewidth,
			column sep=-\pgflinewidth,
			ampersand replacement=\&] {
				a \& c \\
			};
			
			\node[left=3pt of mat1-1-1] {Bank $i$};
			\node[above=3pt of mat1-1-1] {$l$};
			\node[above=3pt of mat1-1-2] {$k$};
			
			\matrix (mat2) [matrix of nodes,
			nodes={matrixcell},
			row sep=-\pgflinewidth,
			column sep=-\pgflinewidth,
			ampersand replacement=\&,
			below=0.5cm of mat1] {  
				b \& d \\
			};
			
			\node[left=3pt of mat2-1-1] {Bank $j$};
			\node[above=3pt of mat2-1-1] {$l$};
			\node[above=3pt of mat2-1-2] {$k$};
		\end{tikzpicture}
		\caption{Monopoly}
	\end{subfigure}
	\hspace{0.5cm}
	\begin{subfigure}[b]{0.4\textwidth}
		\centering
		\begin{tikzpicture}[
			matrixcell/.style={draw, minimum width=1.9cm, minimum height=1.9cm,fill=gray!15, anchor=center, font=\large},
			]
			
			\matrix (mat) [matrix of nodes,
			nodes={matrixcell},
			row sep=-\pgflinewidth,
			column sep=-\pgflinewidth,
			ampersand replacement=\&] {
				a \& c \\
				b  \& d \\
			};
			
			\node[left=3pt of mat-1-1] {Bank $i$};
			\node[left=3pt of mat-2-1] {Bank $j$};
			\node[above=3pt of mat-1-1] {$l$};
			\node[above=3pt of mat-1-2] {$k$};
		\end{tikzpicture}
		\vspace{0.3cm}
		\caption{Cluster}
	\end{subfigure}
	\caption{Banking Modes}
	\label{fig:Banking Modes}
\end{figure}

\section{Properties of Equilibrium}\label{Section: Properties of Equilibrium}
\noindent In this section, we establish the necessary and sufficient conditions for the existence and uniqueness of solutions under BCP, as well as their properties of these solutions. We then introduce the \textit{Crowding-Out Effect} to characterize $\tilde{\mathbf{\Phi}}\left(\mathbf{s}^{|\mathcal{P}|\times 1}\right) - \mathbf{\Phi}^{-1}\left(\mathbf{s}^{|\mathcal{P}|\times 1}\right)$. To highlight the economic intuition of incentive compatibility, we present the equilibrium properties of CMS and analyze a two-bank case. Section \ref{Characterization of Invariants, Perfection and Homophily} generalizes this idea, and establishes an economic analytical framework that guarantees equilibrium existence and uniqueness.

\subsection{Existence and Uniqueness of Equilibrium}\label{Existence and Uniqueness of Equilibrium & Properties}
\noindent The liquidation amount $s^i$ for bank $i$ is a function of the exogenous shock:
\begin{equation}\label{Shock-dependent liquidation}
	s^i=\frac{1}{\bar{\theta}}\left[\bar{\theta}A_i-E_i+(1-\bar{\theta})\varepsilon A_i  \right]
\end{equation}
where $s^i$ is independent of the devaluation parameter $\beta$. Under the definition of the BCP, the equilibrium solution system is characterized by equation system:
\begin{align}\label{Equilibrium solution under BCP}
	s^i&:=\underbrace{x^i}_{\text{Sales Volume}}+\underbrace{\sum\limits_k\left(1-e^{\sum\limits_i -\beta x^i\pi_{ik}} \right)M_k\cdot p_{ki}}_{\text{Devaluation}}  \\  &\quad\quad\quad\quad\quad\quad\quad\quad\quad \notag  \vdots
\end{align}
where the first term represents the bank's actual sales volume, and the second term corresponds to the asset devaluation resulting from fire sales. We can interpret $\beta$ as a measure of market liquidity, where a smaller $\beta$ corresponds to lower price devaluation and higher market liquidity. An immediate observation is that when $\beta = 0$, the equation system \ref{Equilibrium solution under BCP} admits a unique solution, since the contagion effect diminishes and each bank's actual asset sales depend only on its own assets. The following proposition establishes the necessary and sufficient conditions for the existence and uniqueness of solutions to the BCP, along with the properties of these solutions.
\begin{restatable}{thm}{SolutionsinFCP}\label{Solutions in FCP}
	Corresponding to BCP, we have:
	
	$(a)$ \textnormal{There exist a threshold $\bar{\beta}(\varepsilon)$, when $\beta\leq \bar{\beta}(\varepsilon)$, there exist an unique solution $\mathbf{x}^*$ in solution system \ref{Equilibrium solution under BCP}.}
	
	$(b)$ \textnormal{Given the condition of $a)$. There exist a threshold $\bar{\varepsilon}$, when $\varepsilon<\bar{\varepsilon}$, there exist at least one component $x^p$ of $\mathbf{x}^*$ is negative. The corresponding $s^p$ is negative as well.}
	
	$(c)$ \textnormal{Given the condition of $a)$. There exist a threshold $\tilde{\varepsilon}$, when $\varepsilon>\tilde{\varepsilon}$, there exist at least one component $x^p$ of $\mathbf{x}^*$ is greater than bank $p$'s asset $A_i'$. The corresponding $s^p$ is greater than $A_i'$ as well.}
\end{restatable}

Proposition $(a)$ formalizes the intuition that sufficiently strong market liquidity enables 
banks to match liquidation amounts exactly to regulatory requirements. We demonstrate the existence of a liquidity threshold contingent on the exogenous shock's magnitude. Proposition $(b)$ states that in conditions of abundant liquidity and moderate risk, the unrestricted credit policy applies to at least one bank—implying that while low-leverage banks conduct fire sales, others simultaneously engage in credit creation. Proposition $(c)$ describes a bank exit mechanism: when the financial crisis severity exceeds a critical level, a bank experiences complete equity erosion, and it cannot restore solvency even through a complete asset liquidation. Our assumption of no industrial expansion is made to facilitate the exposition of the above propositions. In practice, industrial expansion is possible under moderate risk environments, as examined in Online Appendix OA3.1. We are now ready to present the comparative statics results for the CMS, which are derived independently of both the strict no-expansion assumption and the liquidity threshold existence condition (i.e., $\beta \leq \bar{\beta}(\varepsilon)$).
\begin{restatable}{thm3}{CounterfactualMonopolySpaceLemma}\label{Corresponding to Counterfactual Monopoly Space}
	Corresponding to Counterfactual Monopoly Space, we have:
	
	$(a)$ \textnormal{$\mathbf{s}=\tilde{\mathbf{s}}$, and $\tilde{\mathbf{x}}$ exists regardless of $\beta$.}
	
	$(b)$ \textnormal{$\tilde{s}^i(\varepsilon) < 0$ and $\tilde{x}^i < 0$ if and only if $\varepsilon < \varepsilon_{\tilde{s}^i=0}$; conversely, $\tilde{s}^i(\varepsilon) > 0$ and $\tilde{x}^i > 0$ if and only if $\varepsilon > \varepsilon_{\tilde{s}^i=0}$.}
	
	$(c)$ \textnormal{$\Phi^{-1}\left[\tilde{s}^i\left(\varepsilon=0\right)\right]$ is non-decreasing in $A_i$ and $\beta$. Furthermore, $\sup\limits_{A_i,\beta}\Phi^{-1}\left[\tilde{s}^i\left(\varepsilon=0\right)\right]=0$.}
	
	$(d)$ \textnormal{$\Phi^{-1}\left[\tilde{s}^i\left(\varepsilon=0\right)\right]$ is non-increasing in $\theta^i$. Moreover, $\max\limits_{\theta^i}\Phi^{-1}\left[\tilde{s}^i\left(\varepsilon=0\right)\right]$ is non-increasing in $A_i$, or $\Phi^{-1}\left[\tilde{s}^i\left(\varepsilon=0\right)|\theta^i\right]$ is non-increasing in $A_i$.}
	
	$(e)$ \textnormal{$e^{-\beta\cdot\Phi^{-1}\left[\tilde{s}^i\left(\varepsilon=0\right)\right]}$ is non-decreasing in $\beta$. Furthermore,  $\sup\limits_{\Phi^{-1}\left[\tilde{s}^i\left(\varepsilon=0\right)\right]}\left[\sup\limits_{\beta}\sum\limits_k \left(e^{-\beta\cdot\pi_{ik}\cdot \Phi^{-1}\left[\tilde{s}^i\left(\varepsilon=0\right)\right]}\cdot \pi_{ik}\right)\right]=\frac{\theta^i}{\bar{\theta}}$, $\inf \limits_{\beta}e^{-\beta\cdot\Phi^{-1}\left[\tilde{s}^i\left(\varepsilon=0\right)\right]}=1$ and $\inf\limits_{\Phi^{-1}\left[\tilde{s}^i\left(\varepsilon=0\right)\right]}\left[\sup\limits_{\beta}e^{-\beta\cdot\Phi^{-1}\left[\tilde{s}^i\left(\varepsilon=0\right)\right]}\right]=\frac{\theta^i}{\bar{\theta}}$.}
	
	$(f)$ \textnormal{$e^{-\beta\cdot\Phi^{-1}\left[\tilde{s}^i\left(\varepsilon=0\right)|\theta^i\right]}$ is non-decreasing in $A_i$. Furthermore,  $\inf\limits_{\Phi^{-1}\left[\tilde{s}^i\left(\varepsilon=0\right)|\theta^i\right]}\left[\sup\limits_{A_i}e^{-\beta\cdot\Phi^{-1}\left[\tilde{s}^i\left(\varepsilon=0|\theta^i\right)\right]}\right]=\frac{\theta^i}{\bar{\theta}}$ and $\inf \limits_{A_i}e^{-\beta\cdot\Phi^{-1}\left[\tilde{s}^i\left(\varepsilon=0|\theta^i\right)\right]}=1$. $e^{-\beta\cdot\Phi^{-1}\left[\tilde{s}^i\left(\varepsilon=0\right)|\theta^i\right]}$ represents $e^{-\beta\cdot\Phi^{-1}\left[\tilde{s}^i\left(\varepsilon=0\right)\right]}$ conditional on fixed $\theta^i$.}
\end{restatable}

Result $(a)$ shows that the equilibrium of the CMS is guaranteed to exist, irrespective of the specific value of $\beta$, implying the existence and uniqueness of $\mathbf{\Phi}^{-1}\left(\mathbf{s}^{|\mathcal{P}|\times 1}\right)$. The analytical framework in Section \ref{Characterization of Invariants, Perfection and Homophily} ensures that $\tilde{\mathbf{\Phi}}\left(\mathbf{s}^{|\mathcal{P}|\times 1}\right)$ shares these properties, which we treat as established results to intuitively describe the crowding-out effect. Result $(b)$ characterizes liquidation transitions induced by exogenous shocks. We will utilize this fundamental result to investigate phase transition of the incentive compatibility measure $\tilde{\mathbf{\Phi}}\left(\mathbf{s}^{|\mathcal{P}|\times 1}\right) - \mathbf{\Phi}^{-1}\left(\mathbf{s}^{|\mathcal{P}|\times 1}\right)$ with respect to exogenous shocks.  Measure sign changes correspond to homophily degree shifts among individual banks, as their willingness to form clusters varies with crisis severity. The remaining results, termed \textit{ignition conditions}, characterize crowding-out effect outcomes without systemic risk ($\varepsilon=0$), providing the homophily transition baseline where $\varepsilon=0$ serves as the starting point.
\subsection{Crowding-Out Effect}\label{Section: Crowding-out effect}
\noindent We begin with two intuitive lemmas illustrating how the crowding-out effect shapes incentive-compatible partitions (homophily), then formally characterize its distinct forms. Consider a cluster $\mathcal{P}$ which can be precisely partitioned into $\mathcal{P} = \mathcal{P^+} \cup \mathcal{P^-}$, where $\mathcal{P^+}$ denotes the \textit{Maximal Bail-in Cluster} and $\mathcal{P^-}$ the \textit{Maximal Bail-out Cluster}. The following lemma holds:
\begin{restatable}{thm3}{LemmaCrowdingOutEffectOpposite}\label{Lemma: Crowding-Out Effect Opposite}
	Opposite Incentives
	
	\textnormal{Suppose the banks in the CCS can be decomposed into two disjoint subsets, $\mathcal{P} = \mathcal{P^+} \cup \mathcal{P^-}$, with $\mathcal{P^+}$ and $\mathcal{P^-}$ denoting respectively the sets of banks for which $\tilde{\mathbf{\Phi}}(\mathbf{s}^{|\mathcal{P}^+|\times 1}) \geq 0$ and $\tilde{\mathbf{\Phi}}(\mathbf{s}^{|\mathcal{P}^-|\times 1}) < 0$. Then the following results hold:}
	
	$(a)$ $\tilde{\mathbf{\Phi}}\left(\mathbf{s}^{|\mathcal{P}^+|\times 1}\right) - \left[\tilde{\mathbf{\Phi}}\left(\mathbf{s}^{|\mathcal{P}|\times 1}\right)\right]^{|\mathcal{P}^+|\times 1}<0$
	
	$(b)$ $\tilde{\mathbf{\Phi}}\left(\mathbf{s}^{|\mathcal{P}^-|\times 1}\right) - \left[\tilde{\mathbf{\Phi}}\left(\mathbf{s}^{|\mathcal{P}|\times 1}\right)\right]^{|\mathcal{P}^-|\times 1}>0$
	
	\noindent \textnormal{where the indices of banks and their equilibrium solutions in $\left[\tilde{\mathbf{\Phi}}\left(\mathbf{s}^{|\mathcal{P}|\times 1}\right)\right]^{|\mathcal{P}^+|\times 1}$ correspond exactly to those represented in $\tilde{\mathbf{\Phi}}\left(\mathbf{s}^{|\mathcal{P}^+|\times 1}\right)$.}
\end{restatable}

These two conditions demonstrate that cluster $\mathcal{P}$ is incentive-incompatible for sub-cluster $\mathcal{P^-}$, while the reverse holds for $\mathcal{P^+}$. Thus, $\mathcal{P^+}$ are disincentivized from clustering with $\mathcal{P^-}$, as such integration would compel them to increase asset sales. For $\mathcal{P^-}$, however, the combined cluster would expand their credit creation. This asymmetric impact—unilaterally harmful for $\mathcal{P^+}$ but beneficial for $\mathcal{P^-}$—defines the \textit{crowding-out effect} (henceforth COE).\footnote{A straightforward criterion applies: an inequality greater than zero favors the cluster denoted by the second term; conversely, negative value favors the first term's.} What, then, characterizes the incentive compatibility condition within a cluster? The following lemma formalizes this condition.
\begin{restatable}{thm3}{LemmaCrowdingOutEffectHomophily}\label{Lemma: Crowding-Out Effect Homophily}
	Partition and Homophily
	
	$(a)$ \textnormal{If $\mathcal{P}=\mathcal{P^+}$, then $\tilde{\mathbf{\Phi}}\left(\mathbf{s}^{|\mathcal{P}|\times 1}\right) - \mathbf{\Phi}^{-1}\left(\mathbf{s}^{|\mathcal{P}|\times 1}\right)<0$}
	
	$(b)$ \textnormal{If $\mathcal{P}=\mathcal{P^-}$, then $\tilde{\mathbf{\Phi}}\left(\mathbf{s}^{|\mathcal{P}|\times 1}\right) - \mathbf{\Phi}^{-1}\left(\mathbf{s}^{|\mathcal{P}|\times 1}\right)>0$}
\end{restatable}

The lemma indicates that banks in cluster $\mathcal{P^+}$ exhibit stronger incentives for coalition formation rather than operate as independent entities. We designate such incentive-compatible groupings as \textit{homophily}, and thus cluster $\mathcal{P^+}$ constitutes a type of homophily. In contrast, banks in cluster $\mathcal{P^-}$ prefer to remain dispersed rather than aggregated. The following lemma presents several COE variants, which will facilitate our categorization of different types of homophily in Section \ref{Characterization of Invariants, Perfection and Homophily}.
\begin{restatable}{thm3}{LemmaCrowdingOutEffect}\label{Lemma: Crowding-Out Effect}
	Crowding-Out Effect

	$(a)$ \textnormal{If a bank $i$ (cluster $\mathcal{P^I}$) is added to the $\mathcal{P^+}$ such that $\tilde{\mathbf{\Phi}}^{i}\left(\mathbf{s}^{|\mathcal{P}^+\cup\{i\}|\times 1}\right)<0$ ($\left[\tilde{\mathbf{\Phi}}\left(\mathbf{s}^{|\mathcal{P}^+\cup\mathcal{P^I}|\times 1}\right)\right]^{|\mathcal{P^I}|\times 1}<0$), then $\tilde{\mathbf{\Phi}}\left(\mathbf{s}^{|\mathcal{P}^+|\times 1}\right) - \left[\tilde{\mathbf{\Phi}}\left(\mathbf{s}^{|\mathcal{P}^+\cup\{i\}|\times 1}\right)\right]^{|\mathcal{P}^+|\times 1} < 0$ ($\tilde{\mathbf{\Phi}}\left(\mathbf{s}^{|\mathcal{P}^+|\times 1}\right) - \left[\tilde{\mathbf{\Phi}}\left(\mathbf{s}^{|\mathcal{P}^+\cup\mathcal{P^I}|\times 1}\right)\right]^{|\mathcal{P}^+|\times 1} < 0$) holds.}
	
	$(b)$ \textnormal{If a bank $j$ (cluster $\mathcal{P^J}$) is added to the $\mathcal{P^-}$ such that $\tilde{\mathbf{\Phi}}^{j}\left(\mathbf{s}^{|\mathcal{P}^-\cup\{j\}|\times 1}\right) > 0$ ($\left[\tilde{\mathbf{\Phi}}\left(\mathbf{s}^{|\mathcal{P}^-\cup\mathcal{P^J}|\times 1}\right)\right]^{|\mathcal{P^J}|\times 1} > 0$), then $\tilde{\mathbf{\Phi}}\left(\mathbf{s}^{|\mathcal{P}^-|\times 1}\right) - \left[\tilde{\mathbf{\Phi}}\left(\mathbf{s}^{|\mathcal{P}^-\cup\{j\}|\times 1}\right)\right]^{|\mathcal{P}^-|\times 1} > 0$ ($\tilde{\mathbf{\Phi}}\left(\mathbf{s}^{|\mathcal{P}^-|\times 1}\right) - \left[\tilde{\mathbf{\Phi}}\left(\mathbf{s}^{|\mathcal{P}^-\cup\mathcal{P^J}|\times 1}\right)\right]^{|\mathcal{P}^-|\times 1} > 0$) holds.}
	
	$(c)$ \textnormal{If a cluster $\mathcal{P^N}$ is added to the CCS $\mathcal{P}$ such that $\sum\limits_{b\in \mathcal{P}}\tilde{\mathbf{\Phi}}^{b}\left(\mathbf{s}^{|\mathcal{P}|\times 1}\right) < \sum\limits_{c\in \mathcal{P}\cup\mathcal{P^N}}\tilde{\mathbf{\Phi}}^{c}\left(\mathbf{s}^{|\mathcal{P}\cup\mathcal{P^N}|\times 1}\right)$, then $\tilde{\mathbf{\Phi}}^{i}\left(\mathbf{s}^{|\mathcal{P}\cup\mathcal{P^N}|\times 1}\right) < \tilde{\mathbf{\Phi}}^{i}\left(\mathbf{s}^{|\mathcal{P}|\times 1}\right)$ holds for any bank $i$ in $\mathcal{P}$ with $s^i < 0$.}
	
	$(d)$ \textnormal{If a cluster $\mathcal{P^N}$ is added to the CCS $\mathcal{P}$ such that $\sum\limits_{b\in \mathcal{P}}\tilde{\mathbf{\Phi}}^{b}\left(\mathbf{s}^{|\mathcal{P}|\times 1}\right) > \sum\limits_{c\in \mathcal{P}\cup\mathcal{P^N}}\tilde{\mathbf{\Phi}}^{c}\left(\mathbf{s}^{|\mathcal{P}\cup\mathcal{P^N}|\times 1}\right)$, then $\tilde{\mathbf{\Phi}}^{j}\left(\mathbf{s}^{|\mathcal{P}\cup\mathcal{P^N}|\times 1}\right) > \tilde{\mathbf{\Phi}}^{j}\left(\mathbf{s}^{|\mathcal{P}|\times 1}\right)$ holds for any bank $j$ in $\mathcal{P}$ with $s^j > 0$.}
\end{restatable}
\subsection{Hierarchy and Homophily Transition}
\noindent We employ a two-bank case to examine how exogenous shocks trigger phase transitions in homophily patterns. Moreover, we characterize the initial transition state's dependence on: individual bank attributes, their relative relationships, and market liquidity. Three homophily transition types are identified: Type-1 (bank $i$ dominates bank $j$ in both assets and leverage), Type-2 (near-identical attributes), and Type-3 (bank $i$ has higher leverage but not assets). Type-2 can be viewed as a limiting case of Type-1. The following proposition formalizes results for Type-1 and Type-2.
\begin{restatable}{thm}{CounterfactualMonopolySpace}\label{Propositions about Counterfactual Monopoly Space}
	Type-1 and Type-2 Homophily Transition:
	
	$(a)$ \textnormal{Type-1 Homophily Transition: Consider banks $i,j$ in industry $k$ with asymmetric leverage ratio $\theta^i \geq \theta^j\geq \bar{\theta}$ and investment dominance $A_i = V_{ik} \gg A_j = V_{jk}$. Under ignition condition $ e^{-\beta \cdot \Phi^{-1}\left[\tilde{s}^i\left(\varepsilon=0\right)\right]}\geq \frac{\theta^j}{\bar{\theta}}$ where $\beta>0$, the system exhibits:}
	
	\textnormal{1. \textit{Threshold Hierarchy}: $\varepsilon_{\tilde{s}^i=0} \geq \varepsilon_{\tilde{s}^j=0}$ with $(s^i-s^j)$ non-decreasing in $\varepsilon$.}
	
	\textnormal{2. \textit{Contagion}: Relative liquidation difference $\frac{s^i-x^i}{s^j-x^j} = \frac{A_i}{A_j}$ with cross-bank dependence $x^i = \frac{A_i(\theta^j-\theta^i)}{\bar{\theta}} + \frac{A_i}{A_j}\cdot x^j$ and $\frac{\partial x^i}{\partial x^j} = \frac{A_i}{A_j} > 0$.}
	
	\textnormal{3. \textit{Regime Transitions}:}
	\begin{itemize}
		\item $(i)$ \textnormal{For $\varepsilon \in [0, \varepsilon_{\tilde{s}^j=0}]$: $x^i \leq \tilde{x}^i < 0$, $\tilde{x}^j \leq 0 \leq x^j$}
		\item $(ii)$ \textnormal{For $\varepsilon \in [\varepsilon_{\tilde{s}^j=0}, \varepsilon_{\tilde{s}^i=0}]$: $x^i \leq \tilde{x}^i \leq 0$ , $0 \leq \tilde{x}^j \leq x^j$} 
		\item $(iii)$ \textnormal{For $\varepsilon \in [\varepsilon_{\tilde{s}^i=0}, \varepsilon_{x^i=0}]$: $x^i \leq  0\leq \tilde{x}^i$ , $0 \leq \tilde{x}^j \leq x^j$}
		\item $(iv)$ \textnormal{For $\varepsilon \geq \varepsilon_{x^i=0}$: $0 \leq x^i \leq \tilde{x}^i$ and $0 \leq x^j \leq \tilde{x}^j$}
	\end{itemize}
	
	$(b)$ \textnormal{Type-2 Homophily Transition: Consider banks $i,j$ in industry $k$ with symmetric leverage ratio $\theta^i=\theta^j$ and symmetric investment $A_i=V_{ik}= A_j=V_{jk}$ (i.e., $\frac{A_i}{A_j}\rightarrow 1$). The system exhibits:}
	
	\textnormal{1. \textit{Threshold Hierarchy}: $\varepsilon_{\tilde{s}^i=0} = \varepsilon_{\tilde{s}^j=0}$ with $s^i-s^j=0$.}
	
	\textnormal{2. \textit{Regime Transitions}:}
	\begin{itemize}
		\item $(i)$ \textnormal{For $\varepsilon \in [0, \varepsilon_{\tilde{s}^i=0}]$: $\tilde{x}^i\leq x^i < 0$ , $\tilde{x}^j\leq x^j < 0$}
		\item $(ii)$ \textnormal{For $\varepsilon \geq \varepsilon_{\tilde{s}^i=0}$: $0 \leq x^i \leq \tilde{x}^i$ and $0 \leq x^j \leq \tilde{x}^j$}
	\end{itemize}
\end{restatable}

The ignition condition functions to enforce a configuration where bank $i$ creates credit while bank $j$ liquidates assets in a risk-free setting. In Type-1, this condition is $\beta$-dependent. Lemma \ref{Corresponding to Counterfactual Monopoly Space} guarantees the irreversibility of this configuration, thereby ruling out the existence of role-swapping ignition conditions. Nevertheless, the system admits simultaneous credit creation at $\varepsilon=0$. We will examine this scenario in Type-3, characterizing the ignition condition through the lens of individual bank attributes.

The proposition illustrates three core properties of the clearing system. The \textit{threshold hierarchy} characterizes the relationship between bank $i$'s leverage ratio $\theta^i$ and its potential maximum tolerable risk level $\varepsilon_{\tilde{s}^i=0}$. In Type-1, bank $i$ can withstand greater risk exposure than bank $j$. \textit{Contagion} captures the one-to-one (bijective) correspondence in liquidation dependencies within a clustered banking pair. A key implication is that the total liquidation amount $x^i + x^j$ allows us to determine individual solutions through variable substitution. Furthermore, comparative statics shows simultaneous increases in their liquidation levels. \textit{Regime transition} characterizes the homophily transition induced by exogenous shocks. The proposition's inequality directly applies the COE from Lemma \ref{Lemma: Crowding-Out Effect}. Only in case $(iv)$ do both banks share the same homophily regime; in all other cases, bank $i$ crowds out bank $j$. Similarly, Type-2 exhibits the same homophily transition behavior.
\begin{restatable}{thm}{CounterfactualMonopolySpaceSnd}\label{Propositions about Counterfactual Monopoly Space2}
	Type-3 Homophily Transition:
	
	\textnormal{Consider banks $i,j$ in industry $k$ with asymmetric leverage ratio $\theta^i\gg \theta^j\geq \bar{\theta}$ and investment dominance $A_i = V_{ik} \ll A_j = V_{jk}$. Then the system exhibits:}
	
	\textnormal{1. There exists a threshold $\beta_{\left\{x^j\left(\varepsilon=0\right)=0\right\}}$ such that $x^j\left(\varepsilon=0\right)=0$ when $\beta= \beta_{\left\{x^j\left(\varepsilon=0\right)=0\right\}}$. Further more, $x^j\left(\varepsilon=0\right)<0$ when $\beta<\beta_{\left\{x^j\left(\varepsilon=0\right)=0\right\}}$ and $x^j\left(\varepsilon=0\right)>0$ when $\beta>\beta_{\left\{x^j\left(\varepsilon=0\right)=0\right\}}$.}
	
	\textnormal{2. $\beta_{\left\{x^j\left(\varepsilon=0\right)=0\right\}}$ is non-decreasing in $\theta^j$ and non-increasing in $A_j$ (i.e., bank $i$ is more likely to be crowded-out with larger $A_j$ conditional on $\theta^i$).}
	
	\textnormal{3. $\beta_{\left\{x^i\left(\varepsilon=0\right)=0\right\}}=+\infty$, that is $x^i\left(\varepsilon=0\right)<0$ is always guaranteed (i.e., regardless of $\beta$).}
	
	\textnormal{4. When $\beta>\beta_{\left\{x^j\left(\varepsilon=0\right)=0\right\}}$, the regime transitions are the same as we mentioned in Proposition \ref{Propositions about Counterfactual Monopoly Space} $a)$.}
	
	\textnormal{5. When $\beta<\beta_{\left\{x^j\left(\varepsilon=0\right)=0\right\}}$, we have $\varepsilon_{x^j=0}<\varepsilon_{s^j=0}<\varepsilon_{s^i=0}$, and the \textit{regime transitions} are:}
	\begin{itemize}
		\item $(i)$ \textnormal{For $\varepsilon \in [0, \varepsilon_{x^j=0}]$: $\tilde{x}^i\leq x^i \leq 0$ , $\tilde{x}^j\leq x^j \leq 0$}
		\item $(ii)$ \textnormal{For $\varepsilon \in \left[\varepsilon_{x^j=0},\varepsilon_{\tilde{s}^j=0}\right]$: $x^i \leq \tilde{x}^i \leq 0$ , $\tilde{x}^j \leq0 \leq  x^j$}
		\item $(iii)$ \textnormal{For $\varepsilon \in \left[\varepsilon_{\tilde{s}^j=0},\varepsilon_{\tilde{s}^i=0}\right]$: $x^i \leq \tilde{x}^i \leq 0$ and $0 \leq \tilde{x}^j \leq x^j$}
		\item $(iv)$ \textnormal{For $\varepsilon \in \left[\varepsilon_{\tilde{s}^i=0},\varepsilon_{x^i=0}\right]$: $x^i \leq 0 \leq \tilde{x}^i$ and $0 \leq \tilde{x}^j \leq x^j$}
		\item $(v)$ \textnormal{For $\varepsilon \geq \varepsilon_{x^i=0}$: $0 \leq x^i \leq \tilde{x}^i$ and $0 \leq x^j \leq \tilde{x}^j$}
	\end{itemize}
\end{restatable}

Proposition \ref{Propositions about Counterfactual Monopoly Space2} formalizes the $\beta$-dependent ignition condition in Type-3, establishing four results: $(i)$ the ignition condition is governed by the market liquidity threshold $\beta_{{x^j(\varepsilon=0)=0}}$; $(ii)$ $\beta_{{x^j(\varepsilon=0)=0}}$ admits characterization through individual attributes $\theta^i$ and $A_i$; $(iii)$ configuration irreversibility stems from the \textit{threshold hierarchy}; $(iv)$ with unchanged threshold hierarchy, homophily transition akin to Type-1 can be achieved by adjusting the $\beta_{{x^j(\varepsilon=0)=0}}$. Additionally, the proposition characterizes the emergence of homophily, exclusively during severe crises, under the condition of simultaneous credit creation by both banks at $\varepsilon=0$.
\begin{figure}[H]
	\centering
	\begin{tikzpicture}[scale=0.03]
		
		\begin{scope}
			\clip (-30,-80) rectangle (75,75); 
			\fill[green!20]
			plot[domain=-30:75] (\x,\x) --
			plot[domain=75:-30] (\x, {-75 + 2*\x}) --
			cycle;
		\end{scope}
		
		\draw[->] (-40,0) -- (100,0) node[right] {$x^j$};
		\draw[->] (0,-80) -- (0,100) node[above] {$x^i$};
		
		\draw[purple!50, thick, domain=-30:90,opacity=70] plot (\x, {-75 + 2*\x});
		\node[red!90, right] at (45,15) {$x^i=\frac{A_i\cdot\left(\theta^j-\theta^i\right)}{\bar{\theta}}+\frac{A_i}{A_j}\cdot x^j$};
		
		\draw[blue!70, thick] (20,-35) -- (30,-15)
		node[midway, below, sloped] {slope = $\frac{A_i}{A_j}$};
		
		\draw[gray, thick, dashed, domain=-30:90] plot (\x,\x);
		\node[black] at (110,85) {$y = x$};
		
		\filldraw[blue!60] (0,-75) circle (3pt) 
		node[left] {$\frac{A_i(\theta^j - \theta^i)}{\bar{\theta}}$};
		
		\draw[black, thick] plot coordinates {
			(-4.78, -26.09)
			(-3.49, -24.85)
			(-2.18, -23.58)
			(-0.86, -22.29)
			(1.85, -19.62)
			(3.23, -18.25)
			(4.63, -16.84)
			(6.05, -15.41)
			(8.94, -12.45)
			(10.42, -10.93)
			(11.92, -9.36)
			(13.43, -7.77)
			(14.97, -6.13)
			(16.53, -4.46)
			(18.11, -2.75)
			(19.70, -1.00)
			(21.32, 0.79)
			(22.96, 2.63)
			(24.63, 4.51)
			(26.31, 6.44)
			(28.01, 8.42)
			(29.74, 10.45)
			(31.48, 12.53)
			(33.25, 14.66)
			(35.04, 16.85)
			(36.85, 19.10)
			(38.69, 21.41)
			(40.54, 23.78)
			(42.42, 26.22)
			(44.31, 28.72)
			(46.23, 31.28)
			(48.17, 33.92)
			(50.13, 36.63)
			(52.11, 39.42)
			(54.11, 42.28)
			(56.14, 45.22)
			(58.18, 48.23)
			(60.24, 51.33)
			(62.32, 54.51)
			(64.43, 57.78)
			(66.55, 61.12)
			(68.69, 64.56)
			(70.85, 68.08)
			(73.02, 71.68)
			(75.22, 75.37)
		};
		
		\draw[red, thick] plot coordinates {
			(17.70, -39.61)
			(18.18, -38.64)
			(18.67, -37.66)
			(19.18, -36.64)
			(19.70, -35.60)
			(20.23, -34.54)
			(20.78, -33.45)
			(21.34, -32.32)
			(21.91, -31.17)
			(22.51, -29.99)
			(23.12, -28.77)
			(23.74, -27.51)
			(24.39, -26.22)
			(25.05, -24.89)
			(25.74, -23.52)
			(26.45, -22.11)
			(27.18, -20.65)
			(27.93, -19.14)
			(28.71, -17.58)
			(29.51, -15.97)
			(30.35, -14.30)
			(31.21, -12.58)
			(32.11, -10.79)
			(33.03, -8.94)
			(33.99, -7.01)
			(34.99, -5.01)
			(36.03, -2.94)
			(37.11, -0.78)
			(38.23, 1.46)
			(39.40, 3.79)
			(40.61, 6.22)
			(41.87, 8.74)
			(43.19, 11.38)
			(44.56, 14.11)
			(45.98, 16.97)
			(47.47, 19.93)
			(49.01, 23.02)
			(50.62, 26.24)
			(52.29, 29.57)
			(54.02, 33.04)
			(55.82, 36.64)
			(57.68, 40.36)
			(59.61, 44.21)
			(61.59, 48.19)
			(63.64, 52.29)
			(65.75, 56.51)
			(67.92, 60.84)
		};
		
		\node[black!90!black, left, font=\small,align=left] at (36,32) {Counterfactual\\Monopoly};
		
		\draw[->, line width=0.3, orange,>=stealth] (10.42, -10.93) -- (22.51, -29.99)
		node[midway, above, sloped] {};
		\draw[->, line width=0.8, orange,>=stealth] (-4.78, -26.09) -- (17.70, -39.61)
		node[midway, above, sloped] {};
		\draw[->, line width=0.3, orange,>=stealth] (18.11, -2.75) -- (25.74, -23.52)
		node[midway, above, sloped] {};
		\draw[->, line width=0.3, orange,>=stealth] (36.85, 19.10) -- (37.11, -0.78)
		node[midway, above, sloped] {};
		\draw[->, line width=0.3, orange,>=stealth] (24.63, 4.51) -- (30.35, -14.30)
		node[midway, above, sloped] {};
		\draw[->, line width=0.8, orange,>=stealth] (46.23, 31.28) -- (43.19, 11.38)
		node[midway, above, sloped] {};
		
		\node[black] at (-7,-43) {$x^i \leq x^j$};
		
	\end{tikzpicture}
	\caption{Crowding-Out Effect}
	\label{CrowdingEffect_2-1}
\end{figure}
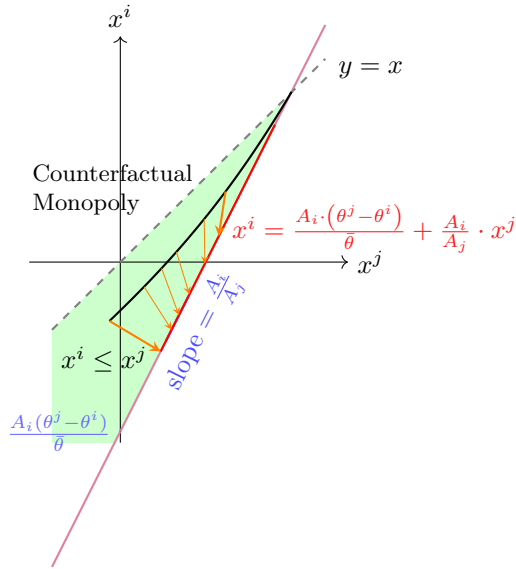
\subsection{Example}
\noindent We introduce an example to illustrate COE-induced homophily transition. Considering Type-1, bank $i$ ($A_i=100$, $E_i=8$, $\theta^i=8\%$) and bank $j$ ($A_j=50$, $E_j=2.5$, $\theta^j=5\%$) operate with market liquidity $\beta=0.02$. Figure \ref{CrowdingEffect_2-1} illustrates the COE state evolution as exogenous shocks escalate from moderate to extreme risk. The solid black line represents the banks' equilibrium solution pair in the CMS, while the red line corresponds to their CCS equilibrium. Arrows indicate CMS-to-CCS directions under fixed risk, enabling identification of their incentive compatibility conditions.

As established in Proposition \ref{Propositions about Counterfactual Monopoly Space}, bank $i$ crowds out bank $j$ at low risk, characterized by $\tilde{\mathbf{\Phi}}^i(\mathbf{s}^{|\mathcal{{i}\cup{j}}|\times 1}) - \Phi^{-1}(\tilde{s}^i) < 0$ and $\tilde{\mathbf{\Phi}}^j(\mathbf{s}^{|\mathcal{{i}\cup{j}}|\times 1}) - \Phi^{-1}(\tilde{s}^j) > 0$, visualized by bottom-right arrows in Figure \ref{CrowdingEffect_2-1}.  Under high-risk conditions, homophily emerges with $\tilde{\mathbf{\Phi}}^i\left(\mathbf{s}^{|\mathcal{\{i\}\cup\{j\}}|\times 1}\right) - \Phi^{-1}(\tilde{s}^i)<0$ and $\tilde{\mathbf{\Phi}}^j\left(\mathbf{s}^{|\mathcal{\{i\}\cup\{j\}}|\times 1}\right) - \Phi^{-1}(\tilde{s}^j)<0$, shown by bottom-left arrows.

\section{Characterization of Invariants, Perfection and Homophily}\label{Characterization of Invariants, Perfection and Homophily}
\noindent This section explores homophily emergence through bank heterogeneity, risk environment, and policy context. We develop a framework guaranteeing equilibrium existence in the CCM, thereby extending previous results. Exogenous-shock-driven homophily transition patterns is generalized as \textit{Compression Equivalence}, explaining banks' incentive compatibility condition changes via risk-oscillation effects on COE states. From a dual perspective, we investigate how partition operations and bank attributes (leverage/assets) trigger homophily transition—a phenomenon termed \textit{Partition-Induced Equilibrium Transition}. In addition to the COE, we introduce tools elucidating the intrinsic relationship between threshold hierarchy and homophily.
\subsection{Definitions of Invariants}
\noindent We first outline the analytical framework, then elaborate on its constituent components. The game is postulated to evolve sequentially from an initial state where all banks constitute a single homophily, or \textit{primitive generative space}. At each stage, banks decide whether to join a new homophilous groupings based on their individual circumstances (relative attributes and incentive compatibility) and environmental factors (peer decisions and mandatory policy requirements). The resulting configuration formed by existing spaces is designated the \textit{perfection generating space}(see Definition \ref{PerfectionGeneratingSpace}, henceforth PGS).\footnote{The subspaces of the PGS are mutually independent. This concept is similar to the notion of ``layers'' in \citet{delon2012local} and \citet{boerma2023composite}, but we construct it through the lens of incentive compatibility conditions.} Figure \ref{PGS_p=1} illustrates a simple PGS configuration.\footnote{We show another example in the Appendix, specifically, Figure \ref{PGS_p=1_Snd}.}
\begin{definition}\label{PerfectionGeneratingSpace}
	Perfection Generating Space
	
	\textnormal{The Primitive Generative Space $\mathcal{P}_{0\,|\,t=0}$ evolves under the joint operations ${p^s \oplus \mathcal{T}_p}$ to yield $\mathcal{P}_{0\,|\,t=1}$ and $\mathcal{P}_{1\,|\,t=1}$. The space $\mathcal{P}^{t=1}=\left[\mathcal{P}_{0\,|\,t=1},\mathcal{P}_{1\,|\,t=1}\right]$ is called the \textit{Perfection Generating Space at $t=1$}, whose \textit{dimension} $d\left(\mathcal{P}^{t=1}\right)=\big|\mathcal{P}^{t=1}\big|$ denotes the cardinality of $\mathcal{P}^{t=1}$. This notation applies recursively. Obviously, $\mathcal{P}^{t=0}=\left[\mathcal{P}_{0\,|\,t=0}\right]$ and $d\left(\mathcal{P}^{t=0}\right)=1$.}  
	
	\textnormal{Under the continuous application of composite operations, the space undergoes successive splitting until $\mathcal{P}^{t}=\mathcal{P}^{t+1}$, at which point we declare $\mathcal{P}^{t}$ to be \textit{stable or perfect}, and this process terminates at time $t$. Moreover, we stipulate that any subset of $\mathcal{P}^{t}$ is non-empty.}
\end{definition}

\begin{figure}[!htbp]
	\centering
	\includegraphics[width = 1\textwidth]{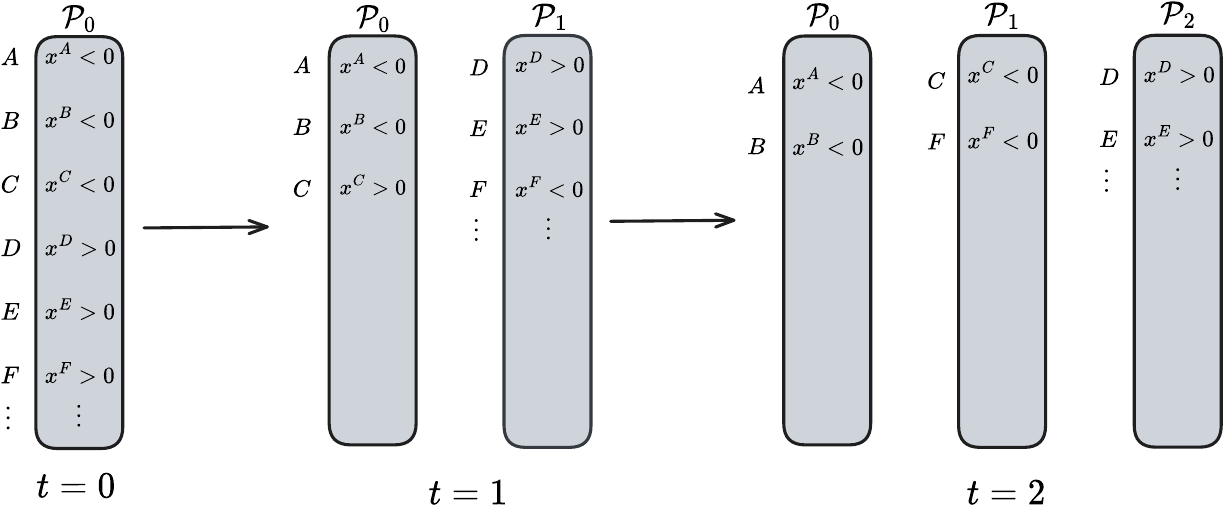}
	\caption{Perfection Generating Space: $p^s=1$}
	\label{PGS_p=1}
\end{figure}

We conceptualize the policy background as a \textit{swith} $p^s$ (see Definition \ref{Switch}), constraining banks' homophily behavior choices.\footnote{The research of \citet{chen2018nexus} and \citet{chen2025trade} actually implies a mandatory requirement of ``no hoarding'' for banks who are still sound under severe exogenous shocks, and such an imperative instruction miraculously enables the banking segmentation to endogenously fulfill the prerequisites of collective rationality of banks who are preparing to liquidate assets (i.e., restruct their balance sheets) to ensure regulatory compliance. We'll discuss this in Section \ref{Section: Homophily}.} Constraints apply selectively across banks: at $p^s=1$, the switch only binds credit-creating banks to maintain existing homophily, while asset-selling banks remain free. At $p^s=0$, all constraints are lifted. At $p^s\in (0,1)$, the switch is partially deactivated. Constraints can be temporary, contingent on relative attribute relationships among banks. As shown in Figure \ref{PGS_p=1}, bank $C$ is bound at $t=0$ but free at $t=1$.
\begin{definition}\label{Switch}
	Switch
	
	\textnormal{The \textit{Switch} $p^s$ is a Lebesgue measurable probabilistic variable and operation embedded within the Perfection Generating Space, whose priority supersedes the perfection operation. Specifically:}
	
	\textnormal{1. When $p^s=1$, banks with $x<0$ \textit{cannot transfer} (or be partitioned);}
	
	\textnormal{2. When $p^s=0$, the switch is \textit{disabled}, and only the perfection operation operates;}
	
	\textnormal{3. When $p^s\in\left(0,1\right)$, banks who in a bad state with $x<0$ transfer with probability $1-p^s$.}
\end{definition}

We observe in Figure \ref{PGS_p=1} that certain bank clusters persist in their subspaces across consecutive periods. For instance, banks $A$ and $B$ consistently remain in subspace $\mathcal{P}_0$ without transitioning. Meanwhile, some fixed clusters persistently exit their subspaces, as observed with banks $D$ and $E$, which continually depart from previous subspaces until PGS stabilization. We define a bank cluster with all members maintaining credit creation as a \textit{maximal bailout cluster}, and one with all members engaged in asset sales as a \textit{maximal bail-in cluster} (see Definition \ref{Maximal Bailout Cluster}).
\begin{definition}\label{Maximal Bailout Cluster}
	 Maximal Bailout (Bail-in) Cluster $\mathcal{B}(\varepsilon)$ $\left(\mathcal{B}^{\mathcal{I}}(\varepsilon)\right)$ 
	
	\textnormal{Under certain $\varepsilon$, Maximal Bailout (Bail-in) Cluster $\mathcal{B}(\varepsilon)$ $\left(\mathcal{B}^{\mathcal{I}}(\varepsilon)\right)$ refers to a clearing system where any bank $i\in \mathcal{B}(\varepsilon)$ $\left(i\in  \mathcal{B}^{\mathcal{I}}(\varepsilon)\right)$ with $s^i< 0$ $\left(s^i\geq 0\right)$ and $x^i< 0$ $\left(x^i\geq 0\right)$. For any bank $j\notin \mathcal{B}(\varepsilon)$ $\left(j\notin  \mathcal{B}^{\mathcal{I}}(\varepsilon)\right)$ satisfying $s^j\leq0$ $\left(s^j\geq 0\right)$ and $\tilde{x}^j\leq 0$ $\left(\tilde{x}^j\geq 0\right)$, the redesigned clearing system $\mathcal{B}'(\varepsilon)=\mathcal{B}(\varepsilon)\cup\{j\}$ $\left( \left(\mathcal{B}^{\mathcal{I}}(\varepsilon)\right)'=\mathcal{B}^{\mathcal{I}}(\varepsilon)\cup\{j\}\right)$ yields $x^j\geq0$ $\left(x^j< 0\right)$ while all other banks $i$ maintain $x^i< 0$ $\left(x^i\geq 0\right)$. If $\{j\}=\varnothing$, then $\mathcal{B}'(\varepsilon)=\mathcal{B}(\varepsilon)$ $\left( \left(\mathcal{B}^{\mathcal{I}}(\varepsilon)\right)'=\mathcal{B}^{\mathcal{I}}(\varepsilon)\right)$ is Maximal Bailout (Bail-in) Cluster.}
\end{definition}

How is the stable state of the PGS characterized? A canonical description is that after applying the joint operations ${p^s \oplus \mathcal{T}_p}$, all subspaces remain unchanged. Figure \ref{PGS_p=1} exhibits key features of stable subspaces: each contains distinct pure clusters, as formalized in Definition \ref{Maximal Bailout Cluster}. We formalize this idea in Definition \ref{Purification}.
\begin{definition}\label{Purification} 
	Space Purification Operation
	
	\textnormal{A subset $\mathcal{P}_{s\,|\,t}$ of $\mathcal{P}^{t}$ constitutes a \textit{Pure Bail-in Space} iff its clearing equilibrium satisfies $\min\limits_{\,i\in \mathcal{P}_{s\,|\,t}} x^i\geq 0$; it is termed a \textit{Pure Bailout Space} iff $\max\limits_{\,i\in \mathcal{P}_{s\,|\,t}} x^i< 0$ holds.}
	
	\textnormal{The \textit{Space Purification Operation} refers to the set of composite operations that enable a subset $\mathcal{P}_{s\,|\,t'_s}$ (abbreviated as $\mathcal{P}_{s}$ from a global perspective) of $\mathcal{P}^{t'_s}$ to achieve purification for the last time at $t=t'_s$. It is denoted as $\bigoplus_{\mathcal{P}_{s\,|\,t'_s}}=\left[ \left\{p^s \oplus \mathcal{T}_p\right\}^{t_s} , \cdots, \left\{p^s \oplus \mathcal{T}_p\right\}^{t'_s} \right]$ where its dimension $d\left(\bigoplus_{\mathcal{P}_{s}}\right) = \left| \bigoplus_{\mathcal{P}_{s\,|\,t'_s}} \right|$ represents the cardinality of $\bigoplus_{\mathcal{P}_{s\,|\,t'_s}}$ and $t_s$ denotes the initial occurrence time of $\mathcal{P}_{s}$. Under $\bigoplus_{\mathcal{P}_{s\,|\,t'_s}}$ and $\bigoplus_{\mathcal{P}_{s\,|\,t'_s+T}}$, we have $\mathcal{P}_{s\,|\,t'_s}=\mathcal{P}_{s\,|\,t'_s+T}$ where $T\in \mathbb{N}_+$.}
\end{definition}

\subsection{Incentive Compatible Partition}
\noindent We have frequently discussed the bank incentive compatibility conditions, we deploy them in the PGS to form an \textit{incentive compatible partition} analogous to \citet{levy2015preferences}. We interpret the operation of perfection from two dual perspectives. Individually: perfection operates when individual bank's departure from its current homophily yields dominant benefits relative to subspace retention, then such bank exits (subject to policy switch constraints). Collectively: perfection manifests when a bank coalition achieves superior benefits through collective departure versus current subspace retention, then such bank cluster exits conditional on switch compliance. In Figure \ref{PGS_p=1}, banks $C$, $D$, $E$, and $F$ constitute a departure-oriented cluster from $\mathcal{P}_0$, yet bank $C$'s exit is switch-prohibited. Consequently, the initial perfection yields a new homophily comprising $(D,E,F)$. The next stages will apply perfection operation again, and so on and so forth.
\begin{definition}\label{Perfection}
	States and Perfection
	
	\textnormal{1. \textit{Individual bank}}
	
	\textnormal{If $\tilde{x}^i - x^i \geq 0$, then bank $i$ is said to be in a \textit{good state} $\tau^i = 1$; otherwise, bank $i$ is in a \textit{bad state} $\tau^i = 0$. \textit{Perfection} $\mathcal{T}_p$ refers to the bipartition of a clearing system under certain shock $\varepsilon$: banks in the good state set $\tau^{\mathcal{G}ood}=\left\{i\big| \tau^i=1\right\}$ will remain, while banks in the bad state set $\tau^{\mathcal{B}ad}=\left\{j\big| \tau^j=0\right\}$ will exit. That is, $\tau^{\mathcal{G}ood} = \mathcal{I}_{bank} \setminus \tau^{\mathcal{B}ad}$. }
	
	\textnormal{2. \textit{Partitions}}
	
	\textnormal{If $\tilde{\mathbf{x}}_{_{\mathcal{P}_{s\,|\,t+1} } }-\,\,\mathbf{x}_{_{\mathcal{P}_{s\,|\,t}  }}\geq 0$ $\left(\tilde{\mathbf{x}}_{_{\mathcal{P}_{s+1\,|\,t+1} } }-\,\,\mathbf{x}_{_{\mathcal{P}_{s\,|\,t}  }}\geq 0\right)$, generative (generated) clearing system $\mathcal{P}_{s\,|\,t}\,\cap\,\mathcal{P}_{s\,|\,t+1}$ $\left(\mathcal{P}_{s\,|\,t}\,\cap\,\mathcal{P}_{s+1\,|\,t+1}\right)$ is said to be in a \textit{good state}; otherwise, it's in a \textit{bad state}. \textit{Perfection} $\mathcal{T}_p$ refers to the bipartition of a clearing system under certain shock $\varepsilon$: bank cluster in the bad state will exit.}
\end{definition}

Our analytical framework adopts the \textit{partition perfection} rather than scrutinizing individual bank incentive compatibility. Partition perfection actually ``aggregates'' the results of individual perfection. In Figure \ref{PGS_p=1}, if bank $D$ didn't transition with $(E,F)$, it would face maximal bail-out cluster $\mathcal{B}(\varepsilon)=\{A,B\}$, altering its incentive compatibility condition and causing migration to $\mathcal{P}_1$ via individual perfection at $t=1$.\footnote{This doesn't always hold for maximal bail-in cluster $\mathcal{B}^{\mathcal{I}}(\varepsilon)$, since the impairment of credit creation maybe too small to change the individual incentive of banks in $\mathcal{B}^{\mathcal{I}}(\varepsilon)$. Therefore the partition perfection is more of our model's setting. Remark \ref{The equivalent operations p^s=1} examines the relationships between perfection operation and switch. We prove it in Appendix \ref{AppendixSection: The equivalent operations}. In addition, partition perfection is well-defined for both $\mathcal{B}(\varepsilon)$ and $\mathcal{B}^{\mathcal{I}}(\varepsilon)$ under assumption \ref{finite-risk mitigation}, since their incentives are opposite (Lemma \ref{Lemma: Crowding-Out Effect Opposite}).} We therefore define \textit{perfection operation} $\mathcal{T}_p$ as partition perfection.

\begin{remark}\label{The equivalent operations p^s=1}
	The effect of $p^s=1$ with Partitions Perfection is equivalent to $p^s=1$ with Individual Perfection on bail-out banks and Partitions Perfection on bail-in banks.
\end{remark}
\subsection{Decomposition Theorems}
\noindent We first generalize previous results regarding exogenous-shock-induced homophily transitions by replacing bank $i$ with a maximal bail-out cluster $\mathcal{B}(\varepsilon)$ and combining it with bank $j \notin \mathcal{B}(\varepsilon)$. The \textit{compression equivalence theorem} states that as the external risk environment deteriorates, the minimal-leverage bank $f$ within $\mathcal{B}(\varepsilon)$ undergoes homophily transition like bank $j$ experienced. With initial cluster ${_n}\mathcal{B}(\varepsilon)$ with ${_{-1}}\mathcal{B}(\varepsilon) = \varnothing$, the remaining banks ${_{n+1}}\mathcal{B}(\varepsilon) = \mathcal{B}(\varepsilon) \setminus {f}$ recursively undergo this process under increasing $\varepsilon$ until ${_{\infty}}\mathcal{B}(\varepsilon) = \varnothing$. Section \ref{Section: Chain} discusses this theorem's role in linking threshold hierarchy and homophily.

\begin{restatable}{thm2}{DecompositionAndCompressionEquivalence}\label{Decomposition and Compression Equivalence}
	Weak Decomposition Theorem and Compression Equivalence Theorem
	
	\textnormal{1. \textit{Weak Decomposition}: }
	
	\textnormal{$\mathcal{B}(\varepsilon)$ and $\mathcal{B}^{\mathcal{I}}(\varepsilon)$ are guaranteed to exist in any Generative Space. Equivalently, any $\mathcal{P}_{s\,|\,t}$ can be decomposed into $\mathcal{B}(\varepsilon)\cup\mathcal{B}^{\mathcal{I}}(\varepsilon)\cup \bigg[\mathcal{P}_{s\,|\,t}\setminus\bigg(\mathcal{B}(\varepsilon)\cup \mathcal{B}^{\mathcal{I}}(\varepsilon)\bigg)\bigg]$. Consequently, $\mathcal{P}_{s\,|\,t}\setminus\bigg(\mathcal{B}(\varepsilon)\cup \mathcal{B}^{\mathcal{I}}(\varepsilon)\bigg)$ can be further decomposed until $\mathcal{B}(\varepsilon)$ and $\mathcal{B}^{\mathcal{I}}(\varepsilon)$ remain unchanged.}
	
	\textnormal{2. \textit{Compression Equivalence}:} 
	
	\textnormal{As for any bank $j\notin \mathcal{B}(\varepsilon)$ $\left(j\notin  \mathcal{B}^{\mathcal{I}}(\varepsilon)\right)$, composite entity $\left\{\mathcal{B}(\varepsilon)\right\}\cup\{j\}$ $\left(\left\{\mathcal{B}^{\mathcal{I}}(\varepsilon)\right\}\cup\{j\}\right)$ appears the same regime transitions as we mentioned in Proposition \ref{Propositions about Counterfactual Monopoly Space} and Proposition \ref{Propositions about Counterfactual Monopoly Space2}. That is, $\mathbf{\Phi}\left(\mathbf{s}^{| _{n+1}\mathcal{B}\left(\varepsilon\right) |\times 1}(\varepsilon)\right)  \leq \tilde{\mathbf{\Phi}}\left(\mathbf{s}^{| _{n+1}\mathcal{B}\left(\varepsilon\right) |\times 1}(\varepsilon)\right)\leq \mathbf{0}$ under $\varepsilon\in\left[ \varepsilon_{\max\limits_{i\in _{n}\mathcal{B}\left(\varepsilon\right)} x^i=0}, \varepsilon_{\max\limits_{f\in _{n+1}\mathcal{B}\left(\varepsilon\right)} x^f=0} \right]$ and ignition condition $\beta\geq\beta_{x^j(\varepsilon=0)=0}=\max \beta_{x^j(\varepsilon)=0}$, where $n\in \left\{ q\,|\, q\geq -1, q\in \mathbb{Z} \right\}$ and $\varepsilon_{\max\limits_{i\in _{-1}\mathcal{B}\left(\varepsilon\right)} x^i}=0$.}
\end{restatable}

Next, we delineate the distinctions between two decomposition theorems and their respective functions. The \textit{weak decomposition theorem} partitions subspace $\mathcal{P}_{s|t}$ into three components: the maximal bailout cluster $\mathcal{B}(\varepsilon)$, the maximal bail-in cluster $\mathcal{B}^{\mathcal{I}}(\varepsilon)$, and the residual set $\mathcal{P}_{s|t} \setminus \left( \mathcal{B}(\varepsilon) \cup \mathcal{B}^{\mathcal{I}}(\varepsilon) \right)$. In contrast, the \textit{strong decomposition theorem} utilizes properties of $\mathcal{B}(\varepsilon)$ and $\mathcal{B}^{\mathcal{I}}(\varepsilon)$ to further splits the remainder into $\left\{ i\,|\, s^i< 0,\, i\notin \mathcal{B}(\varepsilon) \right\}$ and $\left\{ j\,|\, s^j\geq 0,\, j\notin \mathcal{B}^{\mathcal{I}}(\varepsilon) \right\}$, then randomly selects a subset $\bigg\{ \left(\sigma\right) \diamond \left\{ i\,|\, s^i< 0,\, i\notin \mathcal{B}(\varepsilon) \right\} \bigg\}$ to incorporate with the $\mathcal{B}(\varepsilon)$-containing component.\footnote{The expression $(a) \diamond \{b\}$ represents randomly sampling $a \cdot |\{b\}|$ elements from set $b$ to form a new set and $a$ represents any fractional probability that divides the cardinality of the set. Specifically, we define $(a) \diamond \varnothing=\varnothing$. For notation simplicity, we denote $\{ b \}=\{i\,|\, i\in b\}$.} As emphasized in Remark \ref{Remark: The main functions of Strong Decomposition Theorem}, this theorem operationalizes the proof strategy for the \textit{partition-induced equilibrium transition theorem}. The mutual independence of components permits isolated analysis. The discrepancy between two decomposition theorems is illustrated in Figure \ref{fig: Weak and Strong Decomposition}.

\begin{figure}[H]
	\centering
	\begin{subfigure}[b]{0.45\textwidth}
		\includegraphics[width=\textwidth]{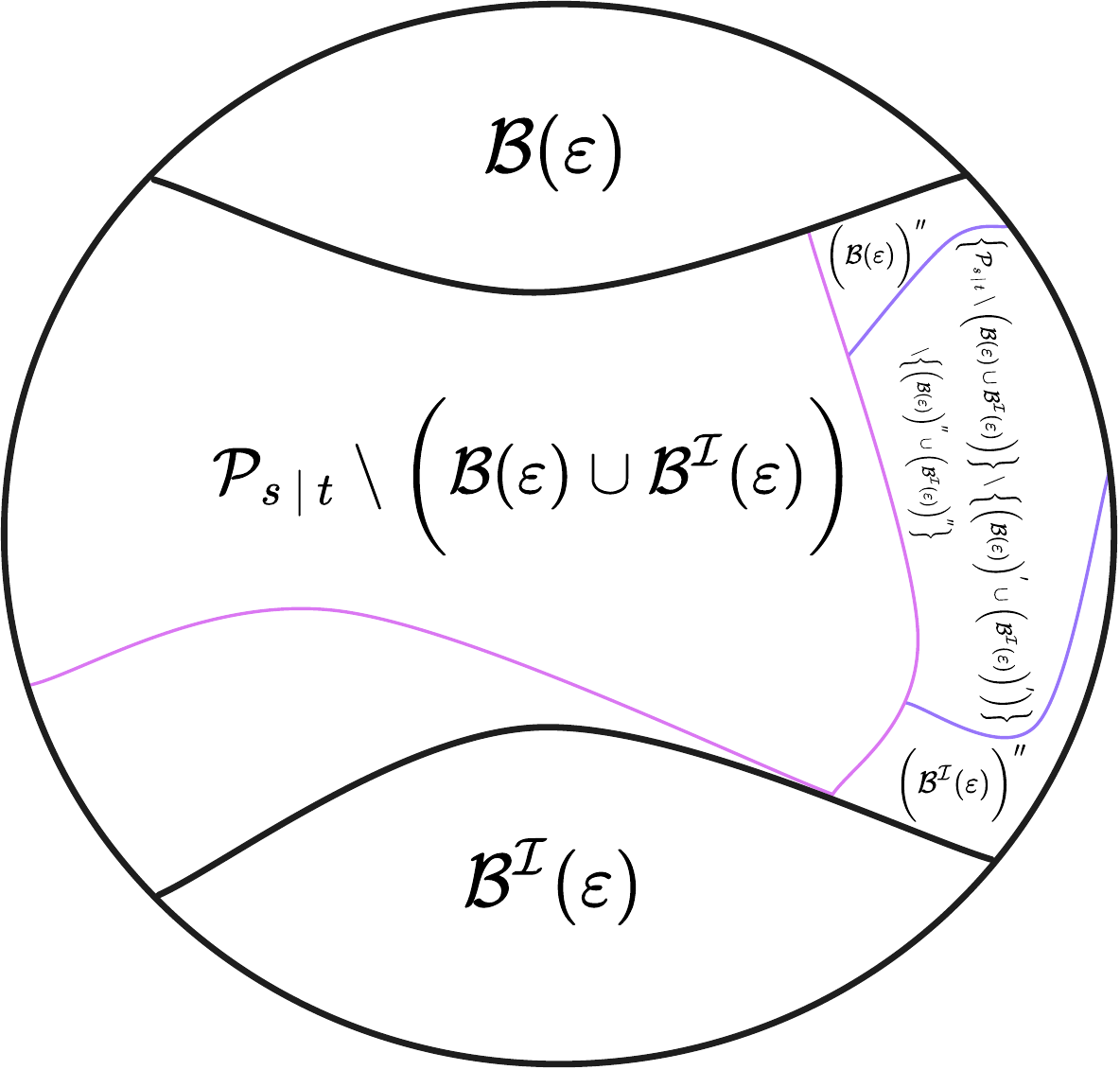}
		\caption{Weak Decomposition}
	\end{subfigure}
	\hspace{0.5cm}
	\begin{subfigure}[b]{0.44\textwidth}
		\includegraphics[width=\textwidth]{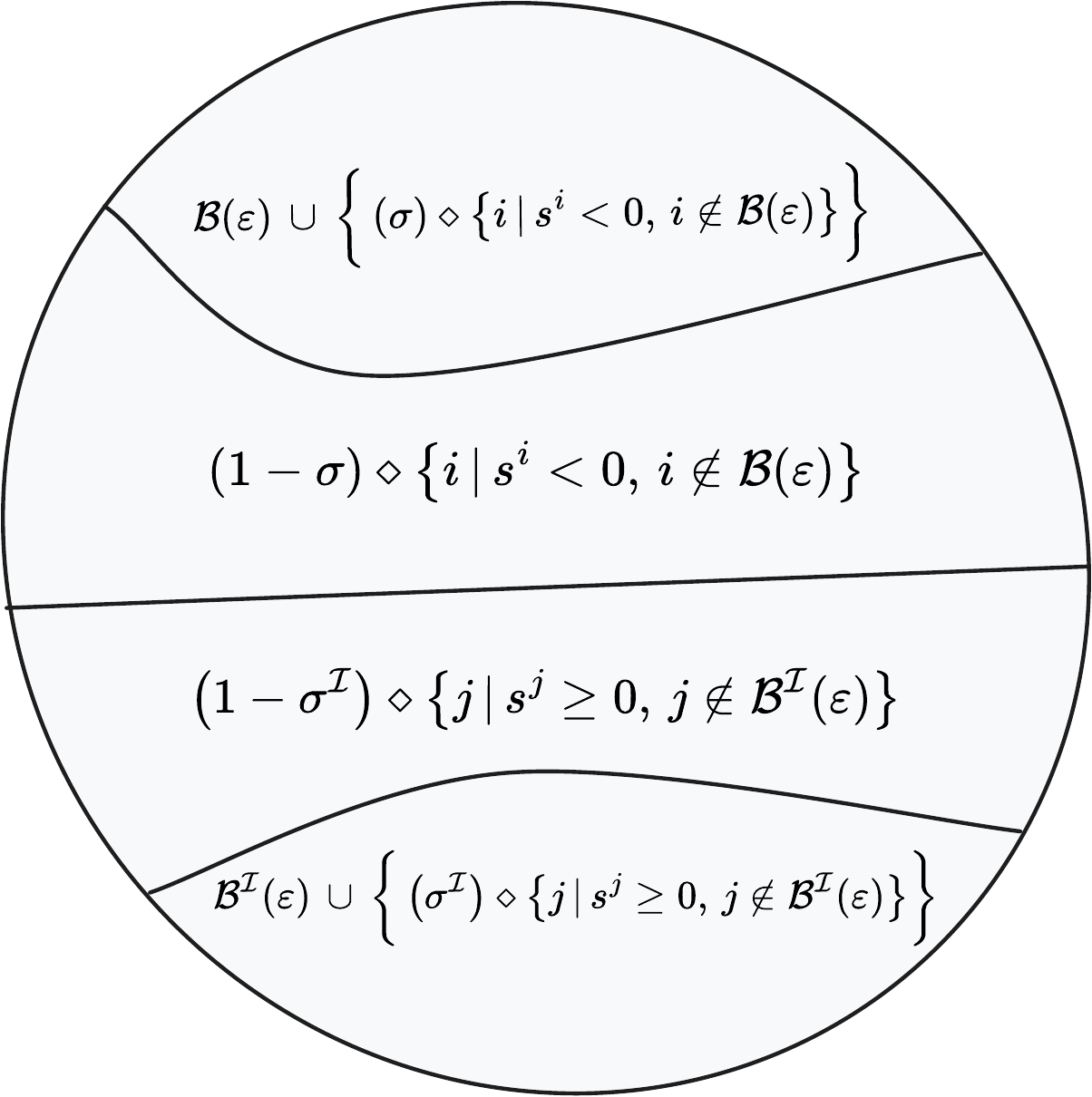}
		\caption{Strong Decomposition}
	\end{subfigure}
	\caption{Discrepancy between Weak and Strong Decomposition}
	\label{fig: Weak and Strong Decomposition}
\end{figure}

\begin{remark}\label{Remark: The main functions of Strong Decomposition Theorem}
	The main functions of Strong Decomposition Theorem
	
	\textnormal{We can randomly divide $\mathcal{P}_{s\,|\,t}$ into multiple different partitions, that is $^1\left(\mathcal{P}_{s\,|\,t}\right)\,\cup\,^2\left(\mathcal{P}_{s\,|\,t}\right)\,\cup\, \cdots\,=\,\mathcal{P}_{s\,|\,t}$. By Strong Decomposition Theorem, we can control inner states of each of these distinct partitions while remaining their independence.} 
	
	\textnormal{Furthermore, we substitute $\mathcal{B}(\varepsilon)$ with $\mathcal{B}(\varepsilon)\,\cup\,\bigg\{ \left(\sigma\right) \diamond \left\{ i\,|\, s^i< 0,\, i\notin \mathcal{B}(\varepsilon) \right\} \bigg\}$, preserving the essential feature $\mathcal{B}(\varepsilon)$. The fractional variable $\left(\sigma\right)$ helps control the number of elements in the first cluster $\mathcal{B}(\varepsilon)\,\cup\,\bigg\{ \left(\sigma\right) \diamond \left\{ i\,|\, s^i< 0,\, i\notin \mathcal{B}(\varepsilon) \right\} \bigg\}$. When $\sigma=0$, the decomposition pattern degenerates back to $\mathcal{B}(\varepsilon)$. However, as we adjust $\left(\sigma\right)$, we can continuously and incrementally (in steps of 1 unit) regulate the number of elements in the second cluster $\left(1-\sigma\right) \diamond \left\{ i\,|\, s^i< 0,\, i\notin \mathcal{B}(\varepsilon) \right\}$. At each level of cluster size, we can further control the second cluster's inner partial order relations. This is the origin of Partition-Induced Equilibrium Transition Theorem.}
\end{remark}

Partition-induced transition characterizes the phenomenon where strong decomposition of $\mathcal{P}_{s|t}$ or $\mathcal{P}_{s|t} \setminus \left( \mathcal{B}(\varepsilon) \cup \mathcal{B}^{\mathcal{I}}(\varepsilon) \right)$ permits behavioral regulation of components $\left(1-\sigma\right) \diamond \left\{ i\,|\, s^i< 0,\, i\notin \mathcal{B}(\varepsilon) \right\}$ and $\left(1-\sigma^{\mathcal{I}}\right) \diamond \left\{ j\,|\, s^j\geq 0,\, j\notin \mathcal{B}^{\mathcal{I}}(\varepsilon) \right\}$. Mirroring Proposition \ref{Propositions about Counterfactual Monopoly Space}'s ignition condition, we identify a corresponding threshold to trigger transitions while preserving the current $\varepsilon$. Figure \ref{Partition-Induced Equilibrium Transition}'s second row exemplifies this: both banks show $\tilde{x}^m < 0, \tilde{x}^n < 0$ in CMS, whereas in the CCS, bank $m$ persists credit creation incentives due to threshold hierarchy, while bank $n$'s solution depends on the ignition condition ($x^n \geq 0$ or $x^n \leq 0$). The third row reverses the second row's process, and requires significantly stronger risk mitigation since exogenous shocks now breach their leverage ratio.
\begin{figure}[H]
	\centering
	\includegraphics[width = 0.5\textwidth]{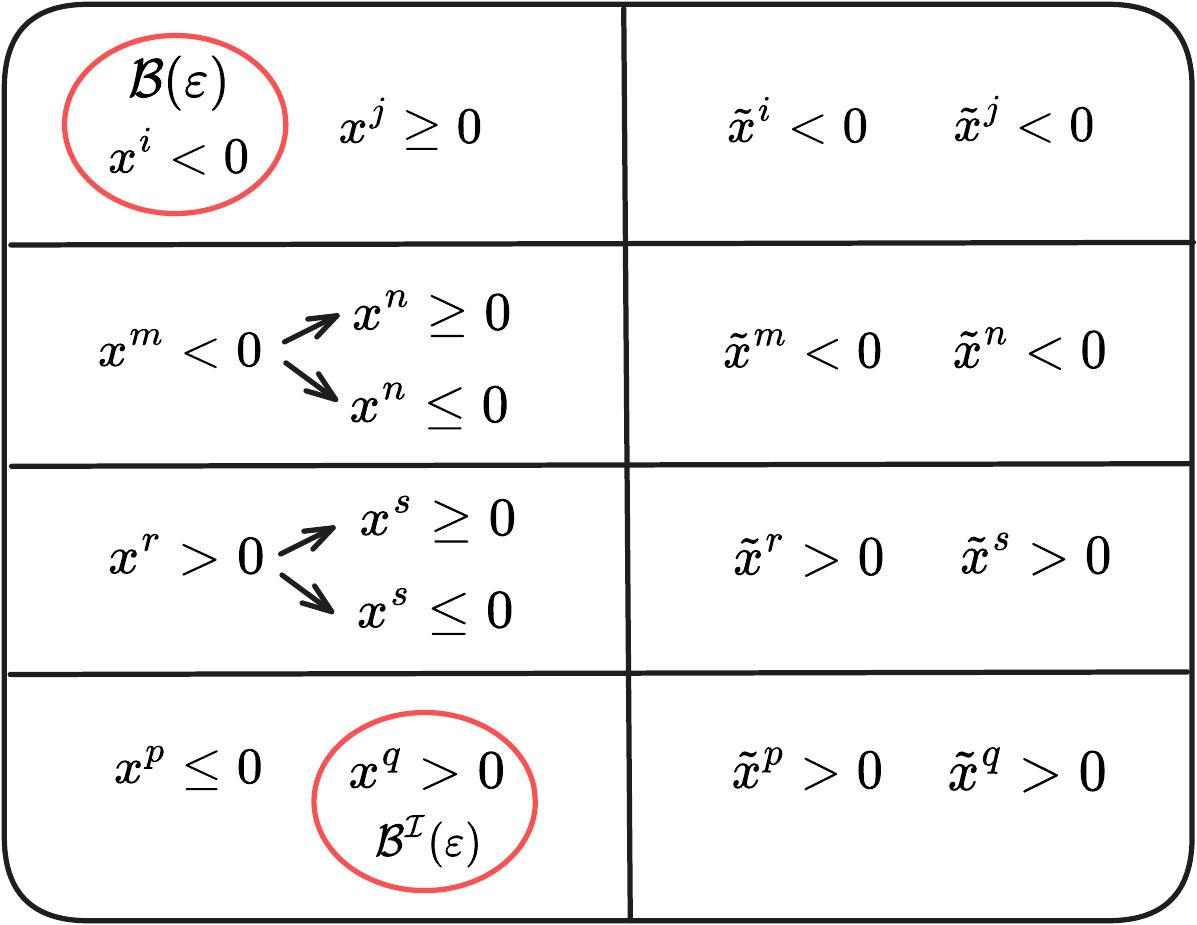}
	\caption{Partition-Induced Equilibrium Transition}
	\label{Partition-Induced Equilibrium Transition}
\end{figure}

The partition-induced equilibrium transition theorem enables elements within components to share homogeneous incentives. Figure \ref{PGS_p=1} demonstrates how the purified space aggregates banks with heterogeneous attributes into a single homophily. This configuration not only accommodates bank heterogeneity within the homophily but also implies fragmentation threshold that would trigger homophily dissolution. When a bank's attributes exceed this threshold, the cluster splits into two incentive-divergent clusters. This transition destabilizes the PGS, thereby necessitating iterative joint operations until the system gets perfect again. Section \ref{Section: Proof Outline of Theorems} details this attribute threshold and its relationship with exogenous shock $\varepsilon$.
\begin{restatable}{thm2}{PartitionInducedEquilibriumTransition}\label{Partition-induced transition}
	Partition-Induced Equilibrium Transition Theorem
	
	\textnormal{1. \textit{Strong Decomposition}:}
	
	 \textnormal{Any Generative Space can be decomposed into $\mathcal{B}(\varepsilon)\,\cup\,\bigg\{ \left(\sigma\right) \diamond \left\{ i\,|\, s^i< 0,\, i\notin \mathcal{B}(\varepsilon) \right\} \bigg\}$, $\left(1-\sigma\right) \diamond \left\{ i\,|\, s^i< 0,\, i\notin \mathcal{B}(\varepsilon) \right\}$, $\left(1-\sigma^{\mathcal{I}}\right) \diamond \left\{ j\,|\, s^j\geq 0,\, j\notin \mathcal{B}^{\mathcal{I}}(\varepsilon) \right\}$ and $\mathcal{B}^{\mathcal{I}}(\varepsilon)\,\cup\,\bigg\{ \left(\sigma^{\mathcal{I}}\right) \diamond \left\{ j\,|\, s^j\geq 0,\, j\notin \mathcal{B}^{\mathcal{I}}(\varepsilon) \right\} \bigg\}$.}
	
	\textnormal{2. \textit{Partition-Induced Transition}: }
	
	\textnormal{We can precisely control $\beta$, individual asset $A$ and individual leverage ratio $\theta$ to achieve exact Partition-Induced Transition while maintaining current $\varepsilon$.}
\end{restatable}

\subsection{Properties of Perfection Equilibrium}\label{Section: Properties of Perfection Equilibrium}
\noindent We explore the equilibrium of PGS under the following assumptions: $(i)$ shocks within the range where $\frac{\mathbb{E}\left[\theta^i\right]-\bar{\theta}}{1-\bar{\theta}}\leq \varepsilon \leq \max\limits_i\left\{\frac{\theta^i-\bar{\theta}}{1-\bar{\theta}}\right\}$; $(ii)$ the finite-risk mitigation condition \ref{finite-risk mitigation} for any $\mathcal{P}_{s\,|\,t}$ where $L^j=\frac{1}{\bar{\theta}}\cdot \left[ \bar{\theta}-\theta^j +\left(1-\bar{\theta}\right)\cdot \varepsilon \right]$. The equality of \ref{finite-risk mitigation} holds only when $\mathcal{P}_{s\,|\,t}$ gets prefect. We will interpret these assumptions in Section \ref{Section: Interpretation of Assumptions}.
\begin{align}\label{finite-risk mitigation}
	\frac{\sum\limits_{j\in\mathcal{P}_{s\,|\,t}} A_j\cdot L^j-\sum\limits_{j\in\mathcal{P}_{s\,|\,t}} x^j}{A^{\mathcal{P}_{s\,|\,t}}}\leq \sum\limits_{j\in\mathcal{P}_{s\,|\,t}}L^j
\end{align}

The first item of Lemma \ref{Lemma ForDecomposition and Compression Equivalence} demonstrates the existence and uniqueness of the equilibrium for arbitrary subspace $\mathcal{P}_{s|t}$. We present the key equilibrium properties in Appendix Lemma \ref{Perfection Switch=1}, dividing them into two aspects: fundamental characteristics of the subspace $\mathcal{P}_{s|t}$ and features of the equilibrium state. This section examines the economic implications of these properties.

We discuss the first part of Lemma \ref{Perfection Switch=1}. Results $a)$ and $b)$ generalizes the basic results of Proposition \ref{Propositions about Counterfactual Monopoly Space}. Result $c)$ is is an extension of $a)$, since we can represent the sum by utilizing a single bank's equilibrium solution. Result $d)$ demonstrates that no bank exhibits incentive to return to the earlier generative space in the new decision-making phase.\footnote{We'll link this property with knightian uncertainty in Section \ref{Section: Schelling Point and Knightian Point}} Result $e)$ constitutes a direct corollary of the partition-induced equilibrium transition theorem, indicating that the population size of banks in the primitive space can be adjusted through market liquidity parameter $\beta$ and $A_i$ modifications. Result $f)$ characterizes the joint operation's functionality and verifies that partitions conforming to this architectural scheme achieve stabilization within one operational step.

Next, we analyze Lemma \ref{Perfection Switch=1}'s second part. Results $a)$, $b)$ and $e)$ describe the properties of the first generative space and the last generated space: The primal subspace perpetually comprises the maximal bailout cluster $\mathcal{B}_0(\varepsilon)$ of the primitive generative space, with stable state $\mathcal{P}_0=\mathcal{B}_0(\varepsilon)$; last generated space invariably incorporates maximal bail-in cluster $\mathcal{B}^{\mathcal{I}}_0(\varepsilon)$ of the primitive generative space, with stable state $\mathcal{P}_{\max{s}}=\mathcal{B}^{\mathcal{I}}_0(\varepsilon)$. These results further imply that all generative subspaces achieve perfection, as each can be decomposed into three components mirroring the primitive generative space pattern via weak decomposition theorem's tripartite.

Result $c)$ delineates the feasible emergence times for full $\mathcal{B}^{\mathcal{I}}_0(\varepsilon)$ and specifies its latest possible appearance. This bounded characteristic is further generalized by result $d)$, which establishes upper and lower bounds for the PGS convergence time bounds. Furthermore, we have provided rigorous demonstration of the necessary and sufficient conditions for the equality in the proof. Results $f)$, $g)$ and $h)$ imply that obtaining a stable PGS essentially requires the application of space purification operations to all constituent subspaces.
\subsection{Homophily}\label{Section: Homophily}
\noindent We first consider a special case where the primitive generative space comprises three clusters satisfying Proposition \ref{Homogeneous and Homophily} conditions. The stable PGS manifests a stratified architecture through subspaces corresponding to these clusters. The result $d(\mathcal{P}^{\max t}) = 3$ indicates that homogeneity-hierarchy synthesis in Proposition \ref{Homogeneous and Homophily} reduces time complexity from $\mathcal{O}(n^2)$ (Lemma \ref{Perfection Switch=1}) to a constant value of 3. Furthermore, we can relax intra-cluster homogeneity conditions while maintaining the PGS's stable structure, since the Theorem \ref{Partition-induced transition} allows perturbations in the internal configurations (individual asset and leverage ratio) of clusters $\mathcal{C}_\mathbf{1}$, $\mathcal{C}_\mathbf{2}$, and $\mathcal{C}_\mathbf{3}$ while preserving the hierarchy and homophily: $d(\mathcal{P}^{\max t}) = 3$ with $\mathcal{P}_{1|\max t} = \mathcal{C}_\mathbf{1}$, $\mathcal{P}_{2|\max t} = \mathcal{C}_\mathbf{2}$, and $\mathcal{P}_{3|\max t} = \mathcal{C}_\mathbf{3}$.

When cluster $\mathcal{C}_\mathbf{1}$ conducts equity injections into $\mathcal{C}_\mathbf{2}$, and $\mathcal{C}_\mathbf{2}$ simultaneously provides equity to $\mathcal{C}_\mathbf{3}$ while sustaining internal cross-financing, Proposition \ref{Homogeneous and Homophily} becomes an optimal transport/matching problem, as explored in studies like \citet{boerma2023composite} and \citet{moldovanu2007contests}. We provide pertinent accounting standards in Online Appendix Section OA4. Additionally, it should be clarified that a stable PGS constitutes a transactional flow architecture, where short-term liquidity forms \textit{hierarchical structure} \citep{wang2025biggest} characterized by large banks capitalizing medium-sized institutions, which subsequently facilitate smaller bank funding via market-based channels such as wholesale lending \citep{chen2025trade}. Such a large-to-medium-to-small bank funding structure diverges from the canonical periphery-intermediary-periphery model. In Section \ref{Section: Core-Periphery Network}, we will interpret the prevalent \textit{core-periphery} financial structure \citep{jackson2021systemic} from a stock perspective. 
\begin{restatable}{thm}{HomogeneousAndHomophily}\label{Homogeneous and Homophily}
	Homogeneity, Hierarchy and Homophily
	
	\textnormal{There exist three distinct clusters, denoted as Cluster $\mathcal{C}_\mathbf{1}$, Cluster $\mathcal{C}_\mathbf{2}$, and Cluster $\mathcal{C}_\mathbf{3}$. Within each cluster, banks are \textit{homogeneous}, exhibiting identical leverage ratios and asset conditions. We denote the leverage ratio and assets of Cluster $\mathcal{C}_\mathbf{1}$ as $^{\mathcal{C}_\mathbf{1}}\theta$ and $^{\mathcal{C}_\mathbf{1}}A$ respectively, and apply analogous notation for the remaining clusters. These clusters satisfy $($\textit{hierarchy}$)$ under $p^s=1$:}
	
	1. $^{\mathcal{C}_\mathbf{1}}\theta\,\geq\, ^{\mathcal{C}_\mathbf{2}}\theta\,\gg\, ^{\mathcal{C}_\mathbf{3}}\theta$ and $^{\mathcal{C}_\mathbf{1}}A \,\geq       \,^{\mathcal{C}_\mathbf{2}}A$.
	
	2. $^{\mathcal{C}_\mathbf{1}}s(\varepsilon)<0$,  $^{\mathcal{C}_\mathbf{2}}s(\varepsilon)<0$ and $^{\mathcal{C}_\mathbf{3}}s(\varepsilon)>0$.
	
	3. $^{\mathcal{C}_\mathbf{1}}x(\varepsilon)<0$, $^{\mathcal{C}_\mathbf{2}}x(\varepsilon)<0$ and $^{\mathcal{C}_\mathbf{3}}x(\varepsilon)>0$ when three clusters coexist.
	
	4. $^{\mathcal{C}_\mathbf{1}}x(\varepsilon)<0$ and $^{\mathcal{C}_\mathbf{2}}x(\varepsilon)>0$ when  Cluster $\mathcal{C}_\mathbf{1}$ and Cluster $\mathcal{C}_\mathbf{2}$ coexist.
	
	\noindent \textnormal{Then it exhibits \textit{homophily}: $d\left(\mathcal{P}^{\max t}\right)=3$ with $\mathcal{P}_{1\,|\,\max t}=\mathcal{C}_\mathbf{1}$, $\mathcal{P}_{2\,|\,\max t}=\mathcal{C}_\mathbf{2}$ and $\mathcal{P}_{3\,|\,\max t}=\mathcal{C}_\mathbf{3}$.}
\end{restatable}

\subsection{Interpretation of Assumptions}\label{Section: Interpretation of Assumptions}
\noindent The condition $\frac{\mathbb{E}\left[\theta^i\right]-\bar{\theta}}{1-\bar{\theta}}\leq \varepsilon \leq \max\limits_i\left\{\frac{\theta^i-\bar{\theta}}{1-\bar{\theta}}\right\}$ implies that at least one bank but not all banks failed.

We reformulate finite-risk mitigation condition \ref{finite-risk mitigation} as \ref{finite-risk mitigation2} refer to Equation \ref{system P Sum}. The left side is the devaluation fraction. 
\begin{align}\label{finite-risk mitigation2}
	\underbrace{ 1-e^{-\beta \cdot \sum\limits_{j\in\mathcal{P}_{s\,|\,t}} x^j} }_{\text{Supply}}\leq \underbrace{ \sum\limits_{j\in\mathcal{P}_{s\,|\,t}}L^j }_{\text{Demand}}
\end{align}
\noindent We can reformulate $L^j$ as equation \ref{Meaning of Lj} which implies the risk mitigation requirement of bank $j$. 
\begin{align}\label{Meaning of Lj}
	\frac{\partial s^j}{\partial A_j}=\frac{\partial \left(A_j\cdot L^j \right)}{\partial A_j} =  L^j
\end{align}

\noindent Therefore, finite-risk mitigation condition \ref{finite-risk mitigation} implies that the current generative space cannot meet the requirements of all banks simultaneously.
\begin{restatable}{thm3}{LemmaForChangingTheSameDirec}\label{Lemma for changing in the same direction}
	A supplementary result to Lemma \ref{Perfection Switch=1} $1.c)$:
	
	$\sum\limits_{j\in \mathcal{P}_{s\,|\,t}}x^j_{s\,|\,t}$ and $\sum\limits_{j\in \mathcal{P}_{s\,|\,t}}A_j\cdot L^j$  exhibit co-directional movement.
\end{restatable}

The assumption of finite-risk mitigation will strengthen the result of Lemma \ref{Perfection Switch=1} $1.c)$, causing their changes to align. Lemma \ref{Lemma for changing in the same direction} demonstrates:
\begin{equation}\label{inequality: co-movement}
	\frac{\partial\sum\limits_{i\in\mathcal{P}_{s\,|\,t}} x^i\left(\mathbf{s}^{|\mathcal{P}_{s\,|\,t}|\times 1}\right)}{\partial \left(\sum\limits_{j\in \mathcal{P}_{s\,|\,t}}A_j\cdot L^j \right)}\geq0
\end{equation}
which implies Lemma \ref{Lemma For Decomposition Chains in any Generating Space}. This lemma shows that in any $\mathcal{P}_{s\,|\,t}$, the maximal bailout cluster $\mathcal{B}(\varepsilon)$ persists under joint operations while the maximal bail-in cluster $\mathcal{B}^{\mathcal{I}}(\varepsilon)$ departs. The finite cardinality of banks $|\mathcal{I}_{\text{Bank}}| < \infty$ ensures PGS to attain a stable state.

\begin{restatable}{thm3}{LemmaForDecompositionChainsPsEqualsOne}\label{Lemma For Decomposition Chains in any Generating Space}
	As for $p^s=1$, we have the following useful results:
	
	1. $A_i\cdot L^i<0,\,\forall i\in \mathcal{B}^{\bigstar\left({_\chi}\mathcal{P}_{s\,|\,t},\varepsilon\right)}$ and $A_j\cdot L^j\geq0,\,\forall j\in \mathcal{B}^{\bigstar\left({^\chi}\mathcal{P}_{s\,|\,t},\varepsilon\right)}$.
	
	2. Bank $j\in \mathcal{B}^{\bigstar\left({^\chi}\mathcal{P}_{s\,|\,t},\varepsilon\right)}$ will exit while bank $i\in \mathcal{B}^{\bigstar\left({_\chi}\mathcal{P}_{s\,|\,t},\varepsilon\right)}$ will stay after the execution of joint operation to $\mathcal{P}_{s\,|\,t}$ where ${^\chi}\mathcal{P}_{s\,|\,t}=\left\{ j\,|\, s^j\geq 0,\, j\in \mathcal{P}_{s\,|\,t} \right\}$ and ${_\chi}\mathcal{P}_{s\,|\,t}=\left\{ i\,|\, s^i< 0,\, i\in \mathcal{P}_{s\,|\,t} \right\}$.
\end{restatable}

\subsection{Chain}\label{Section: Chain}
\noindent Remark \ref{Remark: The main functions of Compression Equivalence Theorem} discusses the motivation behind our introduction of ``Chain'' (see Definition \ref{Chain}).
\begin{definition}\label{Chain}
	Chain $\bigstar\left(\mathcal{P}_{s\,|\,t},\varepsilon\right)$
	
	\textnormal{$\bigstar$ is an unique mapping result of correspondence $\lozenge:\mathcal{P}_{s\,|\,t}\rightarrow \bigstar\left(\mathcal{P}_{s\,|\,t},\varepsilon\right)$. It rearranges $\mathcal{P}_{s\,|\,t}$ in descending order of leverage ratios, with initial element $^0\bigstar\left(\mathcal{P}_{s\,|\,t},\varepsilon\right)$ and corresponding $x^{^0\bigstar\left(\mathcal{P}_{s\,|\,t},\varepsilon\right)}$.}
	
	\textnormal{$\bigstar\left(\mathcal{P}_{s\,|\,t},\varepsilon\right)$ is termed \textit{negative regular chain} $\left(\textit{positive regular chain}\right)$ if $x^{^w\bigstar\left(\mathcal{P}_{s\,|\,t},\varepsilon\right)}<0,\,\forall w\in \bigstar\left(\mathcal{P}_{s\,|\,t},\varepsilon\right)$ $\left( x^{^w\bigstar\left(\mathcal{P}_{s\,|\,t},\varepsilon\right)}\geq0,\,\forall w\in \bigstar\left(\mathcal{P}_{s\,|\,t},\varepsilon\right) \right)$ holds.}
	
	\textnormal{$\bigstar\left(\mathcal{P}_{s\,|\,t},\varepsilon\right)$ is termed \textit{negative pre-regular chain} $\left(\textit{positive pre-regular chain}\right)$ if $s^{^w\bigstar\left(\mathcal{P}_{s\,|\,t},\varepsilon\right)}<0,\,\forall w\in \bigstar\left(\mathcal{P}_{s\,|\,t},\varepsilon\right)$ $\left( s^{^w\bigstar\left(\mathcal{P}_{s\,|\,t},\varepsilon\right)}\geq0,\,\forall w\in \bigstar\left(\mathcal{P}_{s\,|\,t},\varepsilon\right) \right)$ holds.}
\end{definition}

\begin{remark}\label{Remark: The main functions of Compression Equivalence Theorem}
	The main functions of Compression Equivalence Theorem
	
	\textnormal{Compression Equivalence Theorem reveals the hierarchy among banks. That is, banks with lower leverage ratios will experience phase transitions earlier as $\varepsilon$ grows. Moreover, this observation motivates our introduction of the ``Chain'' concept, since the behavior is analogous to the chain's tail crossing the horizontal zero-axis. The decomposition of every pre-regular chain yields regular chains with exactly determined cardinalities (according to Lemma \ref{Lemma ForPartitionInduced Equilibrium Transition Snd} and Lemma \ref{Lemma ForPartitionInduced Equilibrium Transition}), and this phenomenon is topologically similar to banks' phase transition. Actually, we can develop a commutative diagram among $\varepsilon$, $A$ and $\theta$ by constructing appropriate correspondences. Figure \ref{Maps} illustrates this insight.}
\end{remark}

We denote $\hollowstar^{2}\left(\mathcal{P}_{s\,|\,t},\varepsilon\right)=\hollowstar\left(\hollowstar\left(\mathcal{P}_{s\,|\,t},\varepsilon\right),\varepsilon\right)$, and this notation applies recursively. Results $(a)$ and $(b)$ in Lemma \ref{Lemma for Chain} restate the strong decomposition theorem from a different perspective, and result $(c)$ shows that for any chain, the elements in the maximal bailout (bail-in) cluster remain unchanged after incorporating any element from the subspace $\mathcal{P}_{s\,|\,t}$.

\begin{restatable}{thm3}{LemmaForChain}\label{Lemma for Chain}
	Corresponding to chain $\bigstar\left(\mathcal{P}_{s\,|\,t},\varepsilon\right)$, we have:
	
	$(a)$ \textnormal{pre-regularity is the necessary condition for achieving regularity.}
	
	$(b)$ \textnormal{Negative (Positive) pre-regular chain $\bigstar\left(\mathcal{P}_{s\,|\,t},\varepsilon\right)$ can be decomposed into $\mathcal{B}^{\bigstar\left(\mathcal{P}_{s\,|\,t},\varepsilon\right)}\,\cup\,\hollowstar\left(\mathcal{P}_{s\,|\,t},\varepsilon\right)$, where $\mathcal{B}^{\bigstar\left(\mathcal{P}_{s\,|\,t},\varepsilon\right)}$ represents the maximal bailout (bail-in) cluster of $\bigstar\left(\mathcal{P}_{s\,|\,t},\varepsilon\right)$ and $\hollowstar\left(\mathcal{P}_{s\,|\,t},\varepsilon\right)$ is a negative (positive) pre-regular chain as well. That is, a pre-regular chain can be \textit{decomposed} into a regular chain and a different pre-regular chain.}
	
	$(c)$ \textnormal{$\mathcal{B}^{\bigstar\left(\mathcal{P}_{s\,|\,t},\varepsilon\right)}=\mathcal{B}^{\bigg\{\mathcal{B}^{\bigstar\left(\mathcal{P}_{s\,|\,t},\varepsilon\right)}\,\cup\, \left\{\left(\sigma\right) \diamond \hollowstar\left(\mathcal{P}_{s\,|\,t},\varepsilon\right)\right\}\bigg\}}$.}
\end{restatable}

\begin{remark}\label{Remark for Partial order and Hierarchy}
	\textnormal{Lemma \ref{Lemma ForDecomposition and Compression Equivalence} $6.ii)$ reveals the essence of our model as Figure \ref{Partial order and Hierarchy} illustrates. We prove that the configuration of Figure \ref{Partial order and Hierarchy} can be achieved in Appendix Section \ref{AppendixSection:Proof of Remark for Partial order and Hierarchy}}
\end{remark}

\begin{figure}[H]
	\centering
	\includegraphics[width = 0.5\textwidth]{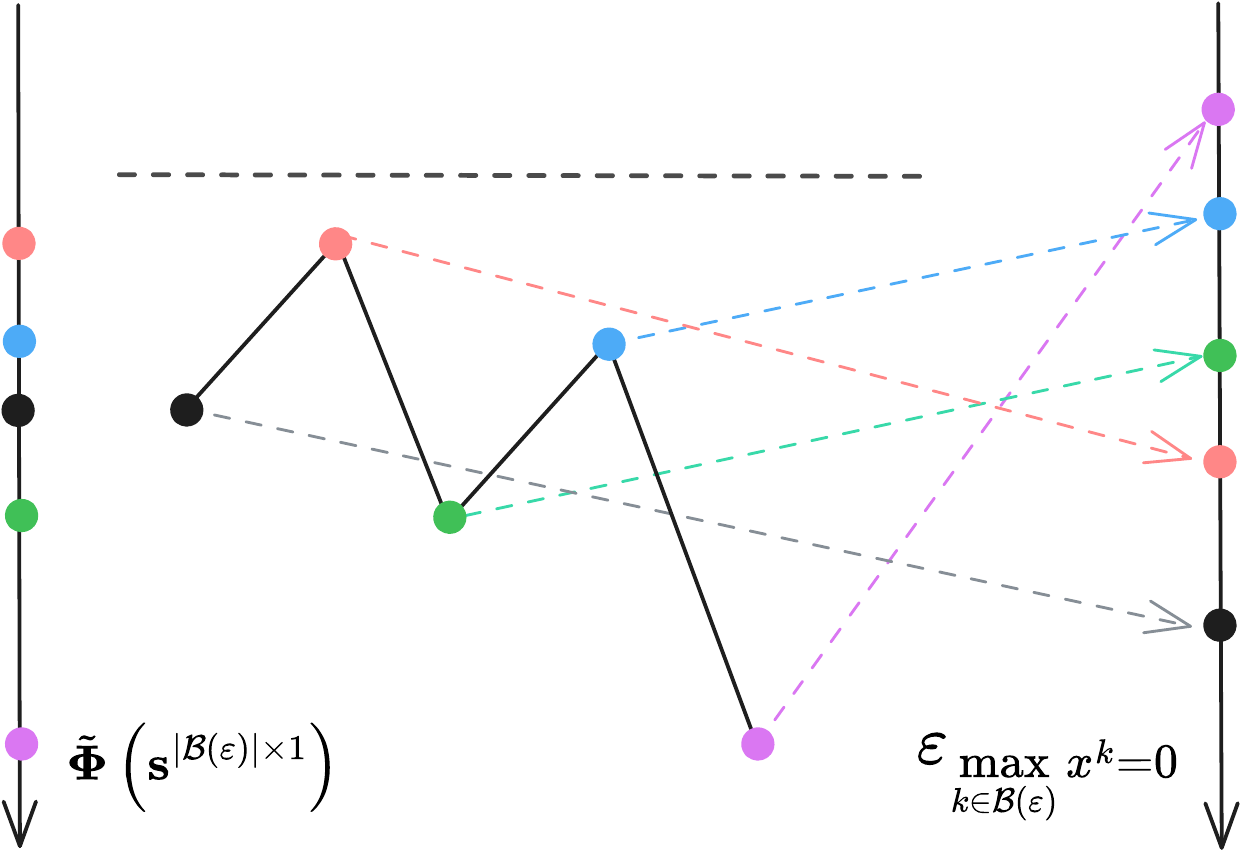}
	\caption{Partial Order and Hierarchy in Negative Pre-regular Chain}
	\label{Partial order and Hierarchy}
\end{figure}
\subsection{Proof Outline of Theorems}\label{Section: Proof Outline of Theorems}
\noindent Lemma \ref{Lemma ForDecomposition and Compression Equivalence} extends Lemma \ref{Corresponding to Counterfactual Monopoly Space}, serving as the foundational apparatus for proving the \textit{compression equivalence theorem}. This result encapsulates the dynamics whereby the minimal-leverage bank $f$ experiences homophily transition during exogenous shock amplification, with comparative statics for market liquidity, bank $f$'s assets/leverage, and peer banks' characteristics. The threshold hierarchy phenomenon receives particular emphasis—bank $f$ occupies the chain's terminal position, and notwithstanding potential asset magnitude or partial order elevation under moderate shocks (see Figure \ref{Partial order and Hierarchy}'s purple/green points), the hierarchical terminus proves most susceptible to severe environmental breaches. Figure \ref{Partial order and Hierarchy}'s left axis visualizes asset-determined credit creation partial orders under tranquility, whereas the right axis exhibits chain structure/hierarchical configuration. Additionally, we characterize bank $f$'s pre-transition trajectory, a process recursively applicable to the residual cluster ${_{n+1}}\mathcal{B}(\varepsilon) = \mathcal{B}(\varepsilon) \setminus {f}$. With escalating exogenous shocks, sequential bank breaches occur, manifesting as chain across the horizontal zero-axis.

\begin{restatable}{thm3}{LemmaForCommutativeDiagram}\label{Lemma For Commutative Diagram}
	Isomorphism
	
	\textnormal{1. $\mathcal{X}_{A\rightarrow \varepsilon}\circ \mathcal{X}_{\varepsilon\rightarrow A}\leftarrow \varepsilon=\varepsilon$ and $\mathcal{X}_{\varepsilon\rightarrow A}\circ \mathcal{X}_{A\rightarrow \varepsilon} \leftarrow A_{^w\bigstar\left(\mathcal{P}_{s\,|\,t},\varepsilon\right)}=A_{^w\bigstar\left(\mathcal{P}_{s\,|\,t},\varepsilon\right)}$.}
	
	\textnormal{2. $\mathcal{Y}_{\theta\rightarrow \varepsilon} \circ \mathcal{Y}_{\varepsilon\rightarrow  \theta}\leftarrow \varepsilon=\varepsilon$ and $\mathcal{Y}_{\varepsilon\rightarrow  \theta} \circ \mathcal{Y}_{\theta\rightarrow \varepsilon} \leftarrow  \theta^{^{w+1}\bigstar\left(\mathcal{P}_{s\,|\,t},\varepsilon\right)}=\theta^{^{w+1}\bigstar\left(\mathcal{P}_{s\,|\,t},\varepsilon\right)}$.}
	
	\textnormal{3. $\mathcal{X}_{\varepsilon\rightarrow A}\circ \mathcal{Y}_{\theta\rightarrow \varepsilon} \circ \mathcal{Y}_{\varepsilon\rightarrow  \theta} \circ \mathcal{X}_{A\rightarrow \varepsilon} \leftarrow A_{^w\bigstar\left(\mathcal{P}_{s\,|\,t},\varepsilon\right)}=A_{^w\bigstar\left(\mathcal{P}_{s\,|\,t},\varepsilon\right)}$.}
	
	\textnormal{4. $\mathcal{X}_{A\rightarrow \varepsilon}\circ \mathcal{X}_{\varepsilon\rightarrow A} \circ \mathcal{Y}_{\theta\rightarrow \varepsilon} \circ \mathcal{Y}_{\varepsilon\rightarrow  \theta}\leftarrow \varepsilon=\varepsilon$.}
\end{restatable}

We prove the \textit{partition-induced equilibrium transition theorem} via isomorphism construction. The core intuition is establishing bijective correspondence between individual bank assets $A_i$, leverage ratios $\theta^i$, and exogenous shock $\varepsilon$. This enables identifying equivalent $A_i$ and $\theta^i$ pairs that generate identical effects when exogenous shocks shift the chain upper-right in Figure \ref{Partial order and Hierarchy}—precisely the content of Lemma \ref{Lemma ForPartitionInduced Equilibrium Transition}. We denote the asset and leverage ratio of $^w\bigstar\left(\mathcal{P}_{s\,|\,t},\varepsilon\right)$ as $A_{^w\bigstar\left(\mathcal{P}_{s\,|\,t},\varepsilon\right)}$ and $\theta^{^{w}\bigstar\left(\mathcal{P}_{s\,|\,t},\varepsilon\right)}$, then utilize Lemma \ref{Lemma ForPartitionInduced Equilibrium Transition} to demonstrate variable isomorphism (Lemma \ref{Lemma For Commutative Diagram}), with Figure \ref{Maps} illustrating our proof scheme.
\begin{figure}[H]
	\centering
	\scalebox{0.9}{
		\begin{tikzpicture}[
			node distance=1.8cm and 1.8cm,  
			every node/.style={scale=1.2},
			>={Stealth[scale=1.2]},
			myloop/.style={looseness=8, min distance=4mm}  
			]
			
			\node (top) {$\bigg|  \mathcal{B}^{\bigstar\left(\mathcal{P}_{s\,|\,t},\varepsilon\right)} \bigg|$};
			\node (left) [below left=of top] {$\theta$};
			\node (right) [below right=of top] {$A$};
			\node (bottom) [below=3.0cm of $(left)!0.5!(right)$] {$\varepsilon$}; 
			
			\draw[->,black] (left) -- node[left] {$\mathcal{B}(\varepsilon)\circ\lozenge$} (top);
			\draw[->,black] (right) -- node[right] {$\mathcal{B}(\varepsilon)\circ\lozenge$} (top);
			\draw[->,black] (bottom) -- ++(0,0.5) -- node[right=0.1cm, yshift=0.6cm] {$\mathcal{B}(\varepsilon)\circ\lozenge$} (top); 
			
			\draw[->,black,myloop] (left) to[out=150,in=210] (left);
			\draw[->,black,myloop] (right) to[out=30,in=-30] (right);
			
			\draw[->,red] (right) to[out=-100,in=10] node[below right=-0.1cm] {$\mathcal{X}_{A\rightarrow \varepsilon}$} (bottom);
			\draw[->,blue] (bottom) to[out=80,in=190] node[above left=0.2cm and -0.4cm] {$\mathcal{X}_{\varepsilon\rightarrow A}$} (right); 
			
			\draw[->,purple!80!black] (left) to[out=-80,in=170] node[below left=-0.1cm] {$\mathcal{Y}_{\theta\rightarrow \varepsilon}$} (bottom);
			\draw[->,orange] (bottom) to[out=100,in=-10] node[above right=0.2cm and -0.4cm] {$\mathcal{Y}_{\varepsilon\rightarrow \theta}$} (left); 
			
		\end{tikzpicture}
	}
	\caption{Maps}
	\label{Maps}
\end{figure}

Building upon the strong decomposition theorem, we constrain the scope for proving the \textit{partition-induced equilibrium transition theorem} to the two central layers depicted in the right panel of Figure \ref{fig: Weak and Strong Decomposition}. Lemma \ref{Lemma ForPartitionInduced Equilibrium Transition Snd} demonstrates controllability over banking states across all layers. Thus, the lemma possesses general applicability to any partition of subspace $\mathcal{P}_{s\,|\,t}$, thereby completing the proof of Theorem \ref{Partition-induced transition}.
\begin{restatable}{thm3}{LemmaForPartitionInducedEquilibriumTransitionSnd}\label{Lemma ForPartitionInduced Equilibrium Transition Snd}
	\textnormal{We randomly select $m$ banks $\mathcal{B}^{\hollowstar}(m)=\mathcal{B}^{\hollowstar^i\left(\mathcal{P}_{s\,|\,t},\varepsilon\right)}\cup \cdots \cup\mathcal{B}^{\hollowstar^j\left(\mathcal{P}_{s\,|\,t},\varepsilon\right)}$ with $\bigg|\mathcal{B}^{\hollowstar}(m)\bigg|=m$ from Lemma \ref{Lemma ForPartitionInduced Equilibrium Transition} to form a new chain $\bigstar\left(\mathcal{B}^{\hollowstar}(m),\varepsilon\right)$. We can precisely control the number of either $\bigg|\mathcal{B}^{\bigstar\left(\mathcal{B}^{\hollowstar}(m),\varepsilon\right)}\bigg|$ or $\bigg|\hollowstar\left(\mathcal{B}^{\hollowstar}(m),\varepsilon\right)\bigg|$.}
\end{restatable}

We only rigorously demonstrate lemmas for negative pre-regular chain in our proof. However, according to Lemma \ref{Lemma for Chain} $(c)$, the proof is actually quite similar for the positive pre-regular chain. That is we only need to prove that $\mathcal{B}^{\mathcal{I}}$ possesses properties analogous to those in Lemma \ref{Lemma ForDecomposition and Compression Equivalence}, which is exactly the content of Lemma \ref{Lemma ForPartitionInduced Equilibrium Transition Third} and Lemma \ref{Lemma ForPartitionInduced Equilibrium Transition Fourth}.
\section{Micromotives, Mechanism and Macrobehavior}\label{Section: Micromotives, Mechanism and Macrobehavior}
\noindent This section interprets selected elements of \citet{schelling1978micromotives}, identifying two subspace $\mathcal{P}_{s|t}$ micromotives: the \textit{free-riding} mechanism manifests as macro-level complementarity within the maximal bail-in cluster $\mathcal{B}^{\mathcal{I}}(\varepsilon)$; while COE preserves substitution in $\mathcal{B}(\varepsilon)$. We also define the \textit{Schelling point} to characterize an agent's expectations of peer behavior and \textit{Knightian point} to capture personally optimal choices.
\subsection{Schelling Point and Knightian Point}\label{Section: Schelling Point and Knightian Point}
\noindent When making decisions, individuals expect others to act in their best interest, while simultaneously optimizing their own choices based others' choices. We characterize the former feasible action set as the \textit{Schelling point}; the latter strategy set constitutes the \textit{Knightian point} (see Definition \ref{Knightian Uncertainty}). As Schelling (1978, p. 23) describes:
\begin{quote}
	``[...] The free market may not do much, or anything, to distribute opportunities and resources among people the way you or I might like them distributed, and it may not lead people to like the activities we wish they liked or to want to consume the things we wish they wanted to consume ... Still,
	within those serious limitations, it does remarkably well in coordinating or harmonizing or integrating the efforts of myriads of self-serving individuals and organizations.''
\end{quote}
\begin{definition}\label{Knightian Uncertainty} 
	Bail-in Banks' Best Response $p^{\triangle}$ under Knightian Uncertainty
	
	\textnormal{We suppress the perfection operation $\mathcal{T}_p$. For bank $i$ who in a bad state with $x^i>0$, its best response under Knightian Uncertainty condition $p^s$ is $p^{\triangle}$, i.e., it exits the current generative space with probability $p^{\triangle}$.}
\end{definition}

\begin{figure}[H]
	\centering
	\includegraphics[width = 0.88\textwidth]{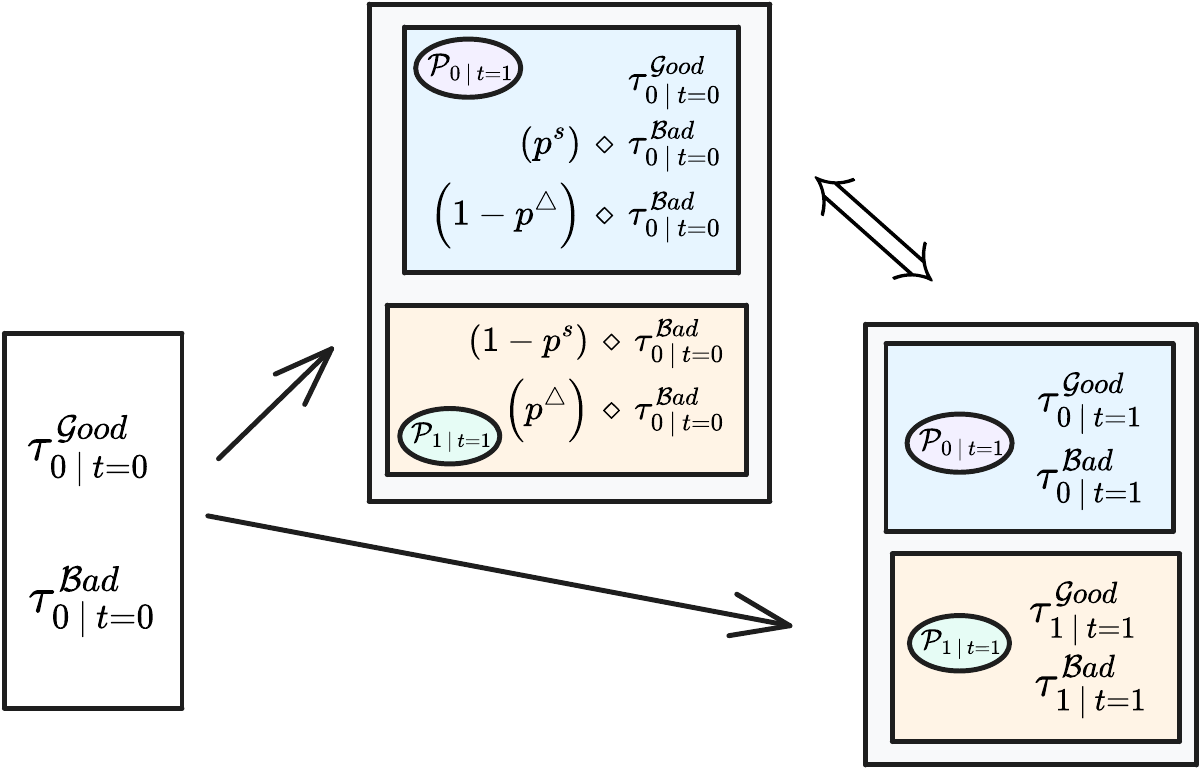}
	\caption{Perfection with $p^{\triangle}$ under Knightian Uncertainty}
	\label{Best Response under Knightian Uncertainty}
\end{figure}

In this section, we analyze the micromotives of the maximal bail-in cluster from dual perspectives. Figure \ref{Best Response under Knightian Uncertainty} shows the perfection mechanism under \textit{Knightian uncertainty}. Proposition \ref{Schelling Point and Knightian Point} formalizes our intuition. For $\mathcal{B}^{\mathcal{I}}(\varepsilon)$, the presence of the COE (Lemma \ref{Lemma: Crowding-Out Effect}) motivates avoidance of subspace sharing with $\mathcal{B}(\varepsilon)$ (Result $(a)$). Meanwhile, result $(b)$ establishes the dual perspective: when $\mathcal{B}(\varepsilon)$ operates without policy restrictions, $\mathcal{B}^{\mathcal{I}}(\varepsilon)$ necessarily abandons the current subspace to avoid contact with $\mathcal{B}(\varepsilon)$. This behavior aligns with the \textit{No-turning-back} property in Lemma \ref{Perfection Switch=1}, with both patterns caused by the COE.
\begin{restatable}{thm}{SchellingPointandKnightianPoint}\label{Schelling Point and Knightian Point} For maximal bail-in cluster $\mathcal{B}^{\mathcal{I}}(\varepsilon)$ and maximal bailout cluster $\mathcal{B}(\varepsilon)$, we have:
	
	$(a)$ \textnormal{The \textit{Schelling point} of $\mathcal{B}(\varepsilon)$ is $p^s=1$.}
	
	$(b)$ \textnormal{The \textit{Knightian point} of $\mathcal{B}^{\mathcal{I}}(\varepsilon)$ is $p^{\triangle}=1$.}
\end{restatable}
\subsection{Emergence of Complementarity}
\noindent We further explore the micromotives underlying maximal bail-in cluster $\mathcal{B}^{\mathcal{I}}(\varepsilon)$ formation. Let $[\mathcal{B}^{\mathcal{I}}(\varepsilon)]=\mathcal{B}^{\mathcal{I}}(\varepsilon)\cup \{j\}$ denote the augmented cluster through bank $j$ incorporation, which constitutes an extended maximal bail-in cluster.\footnote{$\mathbbm{1}_{\tilde{\mathbf{\Phi}}^j\left(\mathbf{s}^{|[\mathcal{B}^{\mathcal{I}}(\varepsilon)]|\times 1}\right)\geq 0}$ is an indicator function corresponding to whether bank $j$ is the new entrant of $\mathcal{B}^{\mathcal{I}}(\varepsilon)$ or not.} Proposition \ref{Emergence of Complementarity} shows the cluster's micromotivational structure via \textit{complementarity emergence}. This phenomenon indicates micro-level free-riding dynamics within $[\mathcal{B}^{\mathcal{I}}(\varepsilon)]$, manifesting as macro-homophily patterns. Each new entrant that maintains the cluster's intrinsic characteristics reduces liquidation burdens on incumbent members, thereby enabling greater asset retention.
\begin{restatable}{thm}{EmergenceOfComplementarity}\label{Emergence of Complementarity}
The system $[\mathcal{B}^{\mathcal{I}}(\varepsilon)]$ exhibits complementarity:
	
\begin{equation}\label{inequality: Emergence of Complementarity}
	\frac{\partial\left[\tilde{\mathbf{\Phi}}^i\left(\mathbf{s}^{|\mathcal{B}^{\mathcal{I}}(\varepsilon)|\times 1}\right)-\tilde{\mathbf{\Phi}}^i\left(\mathbf{s}^{|[\mathcal{B}^{\mathcal{I}}(\varepsilon)]|\times 1}\right)\right]}{\partial \left[\mathbbm{1}_{\tilde{\mathbf{\Phi}}^j\left(\mathbf{s}^{|[\mathcal{B}^{\mathcal{I}}(\varepsilon)]|\times 1}\right)\geq 0}\right]}\geq 0,\,\forall i\in\mathcal{B}^{\mathcal{I}}(\varepsilon)
\end{equation}
\textnormal{and the inequality \ref{inequality: Emergence of Complementarity2} holds when $\mathbbm{1}_{\tilde{\mathbf{\Phi}}^j\left(\mathbf{s}^{|[\mathcal{B}^{\mathcal{I}}(\varepsilon)]|\times 1}\right)\geq 0}>0$.}
\begin{equation}\label{inequality: Emergence of Complementarity2}
	\frac{\partial\left[\tilde{\mathbf{\Phi}}^i\left(\mathbf{s}^{|\mathcal{B}^{\mathcal{I}}(\varepsilon)|\times 1}\right)-\tilde{\mathbf{\Phi}}^i\left(\mathbf{s}^{|[\mathcal{B}^{\mathcal{I}}(\varepsilon)]|\times 1}\right)\right]}{\partial \left[\tilde{\mathbf{\Phi}}^y\left(\mathbf{s}^{|\mathcal{B}^{\mathcal{I}}(\varepsilon)|\times 1}\right)-\tilde{\mathbf{\Phi}}^y\left(\mathbf{s}^{|[\mathcal{B}^{\mathcal{I}}(\varepsilon)]|\times 1}\right)\right]}\geq 0,\,\forall i,y\in\mathcal{B}^{\mathcal{I}}(\varepsilon)
\end{equation}
\end{restatable}
\subsection{Persistence of Substitution}
\noindent In maximal bailout cluster $\mathcal{B}(\varepsilon)$, we can observe similar macroscopic properties. Let $[\mathcal{B}(\varepsilon)]=\mathcal{B}(\varepsilon)\cup \{j\}$ denote the augmented cluster through bank $j$ incorporation, which constitutes an extended maximal bail-out cluster. Proposition \ref{Persistence of Substitution} shows the micro-level \textit{substitution persistence}.

\begin{restatable}{thm}{PersistenceOfSubstitution}\label{Persistence of Substitution}
	The system $[\mathcal{B}(\varepsilon)]$ exhibits substitution:
	
\begin{equation}\label{inequality: Persistence of Substitution}
	\frac{\partial\left[\tilde{\mathbf{\Phi}}^i\left(\mathbf{s}^{|\mathcal{B}(\varepsilon)|\times 1}\right)-\tilde{\mathbf{\Phi}}^i\left(\mathbf{s}^{|[\mathcal{B}(\varepsilon)]|\times 1}\right)\right]}{\partial \left[\mathbbm{1}_{\tilde{\mathbf{\Phi}}^j\left(\mathbf{s}^{|[\mathcal{B}(\varepsilon)]|\times 1}\right)\leq 0}\right]}\leq 0,\,\forall i\in\mathcal{B}(\varepsilon)
\end{equation}
\textnormal{and the inequality \ref{inequality: Emergence of Complementarity2} also holds for any two elements in $\mathcal{B}(\varepsilon)$ when $\mathbbm{1}_{\tilde{\mathbf{\Phi}}^j\left(\mathbf{s}^{|[\mathcal{B}(\varepsilon)]|\times 1}\right)\leq 0}>0$.}
\end{restatable}

Proposition \ref{Emergence of Complementarity} and Proposition \ref{Persistence of Substitution} correspond to the content in Schelling (1978, p. 14). The difference is that, although our utility function is submodular, we can still identify a cluster exhibiting characteristics of complementarity. This cluster can independently form a homophily, which depends not only on the mechanism context but also on the micromotives of the banks within the cluster.
\begin{quote}
	``[...] People are responding to an environment that consists of other people responding to their environment, which consists of people responding to an environment of people's responses. Sometimes the dynamics are sequential: if your lights induce me to turn mine on, mine may induce somebody else but not you. Sometimes the dynamics are reciprocal: hearing your car horn, I honk mine, thus encouraging you to honk more insistently.''
\end{quote}
\subsection{Discussion of Utility Functions}
\noindent The phenomena characterized in Proposition \ref{Emergence of Complementarity} and Proposition \ref{Persistence of Substitution} manifest when the submodular utility function $u$ takes values in $\mathbb{R}$, which constitutes an extension of the research developed by \citet{topkis1978minimizing,topkis1979equilibrium}. For supermodular utility specifications, the inequality directions in \ref{inequality: Emergence of Complementarity} and \ref{inequality: Persistence of Substitution} would be reversed, as examined in Appendix \ref{AppendixSection: Discussion on inequality}.
\section{Discussion}\label{Section: Discussion}
\noindent In this section, we explain some examples.
\subsection{Schelling Model}
\noindent In the seat selection problem analyzed by Schelling (1978), individuals' utility functions exhibit strategic substitution—that is, individuals generally tend to avoid sitting next to others. However, introducing preference heterogeneity through divergent loss utilities (negative versus positive) significantly alters equilibrium configurations, especially when seating resources are limited. As illustrated in Figure \ref{PGS_p=1}, negative-utility types cluster together while positive-utility types approach them. When the former are unable to move freely, clustering naturally emerges. In addition to strategic interactions among students, external structures also matter: social pressure reduces front-row sitting, creating back-of-room clusters. Schelling (1978, p. 11) describes this as follows:
\begin{quote}
	``[...] There were eight hundred people in the hall, densely packed from the thirteenth row to the distant rear wall. Feeling a little as though I were addressing a crowd on the opposite bank of a river, I gave my lecture.''
\end{quote}

With negative loss utility, individuals choose homophily minimizing losses—that is, they seek to establish adjacency with others who are similar to themselves in certain attributes, sufficient to form a homophily connection (Theorem \ref{Partition-induced transition}). We analyze this via the perspective of loss minimization. Although the pursuit of utility maximization is equivalent in terms of the decision maker's objective, the latter is more intuitive for behavioral interpretation. In contrast, positive-loss-utility students choose neighbors more randomly, which may stem from distinctive attributes that make them more prominent, creating flexible seating patterns.
\subsection{Core-Periphery Network}\label{Section: Core-Periphery Network}
\noindent This section serves as a supplement to Section \ref{Section: Homophily}, examining the formation mechanisms of  \textit{core-periphery structure}, or why this structure can be indirectly identified by using the deposit structure of banks \citep{wang2023does}. \citet{jackson2021systemic} surveys related research.

Consider a bank run scenario where depositors need to withdraw their deposits from a specific bank, but with limited liquid assets and maturity mismatch (short deposits vs long loans) which create cash shortfalls. This implies that when a portion of depositors initiates large-scale withdrawals, the bank may be unable to meet all withdrawal demands, leading to a ``first-come, first-served'' allocation mechanism. In this setting, an individual depositor's withdrawal decision exhibits strategic substitution: the more one depositor withdraws, the less is available for others, thus withdrawal behavior displays typical submodular characteristics.

Within this framework, depositor utilities incorporate both personal withdrawal choices and total peer withdrawals. Expecting others to withdraw encourages early withdrawal, but coordination failure causes systemic risk. Some depositors (positive loss utility: less liquidity-sensitive, more trusting, or policy-constrained) may not withdraw immediately, alleviating run pressure. Nevertheless, when the number of banks is few and depositors cannot freely switch their deposits, this game structure promotes behavioral concentration in larger banks with better liquidity reserves to reduce risk (Assumption \ref{finite-risk mitigation2}).
\subsection{Supmodular Utility Function}
\noindent Consider worker and work scenarios. Appendix \ref{AppendixSection: Discussion on inequality} demonstrates that even when workers' actions show complementarity, positive utility functions can induce manifestations of strategic substitution. This motivational misalignment thereby expands the spectrum of equilibrium configurations.
\section{Conclusion}\label{Conlusion}
\noindent This paper makes three contributions. First, we reconstruct the theory of homophily formation. The proposed well-defined perfection dynamics provide a tractable analytical procedure that incorporates and quantifies exogenous environments—such as uncertain shocks and macro- and micro-prudential policies—as well as micromotives, such as agents' strategic interactions. The joint operator ensures this evolutionary process satisfying the incentive compatibility conditions of individual agents both between and within clusters. Moreover, Subspaces within the perfection generating space are mutually independent, and the dynamics exhibit finite upper and lower bounds on iteration time. At equilibrium, the stable subspaces are purified, implying that the weak and strong decomposition theorems, together with the isomorphic structure among variables, enable the perfection dynamics to be transformed—via the proposed invariants—into a finite static decomposition of the primitive generative space. This formulation facilitates comparative statics, and Partition-Induced Equilibrium Transition Theorem further permits an extension from a simple configuration to a more intricate structure encompassing heterogeneous bank settings.

Second, we extract some invariants and properties embedded in perfection dynamics to highlight the duality of Schelling point and Knightian point. We find that within $\mathcal{B}^{\mathcal{I}}(\varepsilon)$, macro-complementarity emerges from micro-substitution when local agents' utility function are submodular. Similar asymmetric behavior also emerges when the micromotives are supermodular. 

Third, we interpret the hierarchical structure as a form of short-term liquidity trading, which aligns with recent empirical observations. Its formation can be attributed to three mechanisms: $(i)$ micro-prudential policies, where monetary preferences channel liquidity toward large banks; $(ii)$ macro-prudential constraints, such as deposit rate ceilings; and $(iii)$ the market-based interbank wholesale funding. From a bank-run perspective, we analyze the classic core–periphery structure, where self-fulfilling depositor behaviors and deposit-rate caps enable the coexistence of the stock-based core–periphery and the flow-based hierarchical structure. The hierarchical structure can be further formulated as a matching or optimal transport problem, a topic that has been explored in several strands of the literature and is therefore not further pursued here.

\bibliographystyle{apalike}
\bibliography{reference}

\begin{thebibliography}{}

\bibitem[Acharya and Rajan, 2024]{acharya2024liquidity}
Acharya, V.~V. and Rajan, R. (2024).
\newblock Liquidity, liquidity everywhere, not a drop to use: Why flooding
  banks with central bank reserves may not expand liquidity.
\newblock {\em The Journal of Finance}, 79(5):2943--2991.

\bibitem[Anderson and Smith, 2024]{anderson2024comparative}
Anderson, A. and Smith, L. (2024).
\newblock The comparative statics of sorting.
\newblock {\em American Economic Review}, 114(3):709--751.

\bibitem[Bernard et~al., 2022]{bernard2022bail}
Bernard, B., Capponi, A., and Stiglitz, J.~E. (2022).
\newblock Bail-ins and bailouts: Incentives, connectivity, and systemic
  stability.
\newblock {\em Journal of Political Economy}, 130(7):1805--1859.

\bibitem[Boerma et~al., 2023]{boerma2023composite}
Boerma, J., Tsyvinski, A., Wang, R., and Zhang, Z. (2023).
\newblock Composite sorting.
\newblock Technical report, National Bureau of Economic Research.

\bibitem[Chen et~al., 2018]{chen2018nexus}
Chen, K., Ren, J., and Zha, T. (2018).
\newblock The nexus of monetary policy and shadow banking in china.
\newblock {\em American Economic Review}, 108(12):3891--3936.

\bibitem[Chen et~al., 2025]{chen2025trade}
Chen, K., Xiao, Y., and Zha, T. (2025).
\newblock A trade-off between monetary policy transmission and systemic risk in
  china.
\newblock Technical report, National Bureau of Economic Research.

\bibitem[Delon et~al., 2012]{delon2012local}
Delon, J., Salomon, J., and Sobolevski, A. (2012).
\newblock Local matching indicators for transport problems with concave costs.
\newblock {\em SIAM Journal on Discrete Mathematics}, 26(2):801--827.

\bibitem[Denbee et~al., 2021]{denbee2021network}
Denbee, E., Julliard, C., Li, Y., and Yuan, K. (2021).
\newblock Network risk and key players: A structural analysis of interbank
  liquidity.
\newblock {\em Journal of Financial Economics}, 141(3):831--859.

\bibitem[Eisenberg and Noe, 2001]{eisenberg2001systemic}
Eisenberg, L. and Noe, T.~H. (2001).
\newblock Systemic risk in financial systems.
\newblock {\em Management Science}, 47(2):236--249.

\bibitem[Eisenschmidt et~al., 2024]{eisenschmidt2024monetary}
Eisenschmidt, J., Ma, Y., and Zhang, A.~L. (2024).
\newblock Monetary policy transmission in segmented markets.
\newblock {\em Journal of Financial Economics}, 151:103738.

\bibitem[Elliott et~al., 2021]{elliott2021systemic}
Elliott, M., Georg, C.-P., and Hazell, J. (2021).
\newblock Systemic risk shifting in financial networks.
\newblock {\em Journal of Economic Theory}, 191:105157.

\bibitem[Hachem and Song, 2021]{hachem2021liquidity}
Hachem, K. and Song, Z. (2021).
\newblock Liquidity rules and credit booms.
\newblock {\em Journal of Political Economy}, 129(10):2721--2765.

\bibitem[Ham, 2021]{ham2021notionsanonymityfairnesssymmetry}
Ham, N. (2021).
\newblock Notions of anonymity, fairness and symmetry for finite strategic-form
  games.

\bibitem[Jackson and Pernoud, 2021]{jackson2021systemic}
Jackson, M.~O. and Pernoud, A. (2021).
\newblock Systemic risk in financial networks: A survey.
\newblock {\em Annual Review of Economics}, 13(1):171--202.

\bibitem[Kahn and Wagner, 2021]{kahn2021sources}
Kahn, C.~M. and Wagner, W. (2021).
\newblock Sources of liquidity and liquidity shortages.
\newblock {\em Journal of Financial Intermediation}, 46:100869.

\bibitem[Levy and Razin, 2015]{levy2015preferences}
Levy, G. and Razin, R. (2015).
\newblock Preferences over equality in the presence of costly income sorting.
\newblock {\em American Economic Journal: Microeconomics}, 7(2):308--337.

\bibitem[Lindenlaub, 2017]{lindenlaub2017sorting}
Lindenlaub, I. (2017).
\newblock Sorting multidimensional types: Theory and application.
\newblock {\em The Review of Economic Studies}, 84(2):718--789.

\bibitem[Moldovanu et~al., 2007]{moldovanu2007contests}
Moldovanu, B., Sela, A., and Shi, X. (2007).
\newblock Contests for status.
\newblock {\em Journal of political Economy}, 115(2):338--363.

\bibitem[Ortega et~al., 2021]{ortega2021schelling}
Ortega, D., Rodr{\'\i}guez-Laguna, J., and Korutcheva, E. (2021).
\newblock A schelling model with a variable threshold in a closed city
  segregation model. analysis of the universality classes.
\newblock {\em Physica A: Statistical Mechanics and its Applications},
  574:126010.

\bibitem[Schelling, 1978]{schelling1978micromotives}
Schelling, T. (1978).
\newblock Micromotives and macrobehavior.

\bibitem[Staab, 2024]{staab2024formation}
Staab, M. (2024).
\newblock The formation of social groups under status concern.
\newblock {\em Journal of Economic Theory}, 222:105924.

\bibitem[Topkis, 1978]{topkis1978minimizing}
Topkis, D.~M. (1978).
\newblock Minimizing a submodular function on a lattice.
\newblock {\em Operations research}, 26(2):305--321.

\bibitem[Topkis, 1979]{topkis1979equilibrium}
Topkis, D.~M. (1979).
\newblock Equilibrium points in nonzero-sum n-person submodular games.
\newblock {\em Siam Journal on control and optimization}, 17(6):773--787.

\bibitem[Wang et~al., 2025]{wang2025biggest}
Wang, T., Yang, H., Weng, J., and Guo, L. (2025).
\newblock The biggest bank may not be the most interconnected: A refined
  entropy-based approach for indicating the direction of interbank flow.
\newblock {\em International Review of Economics \& Finance}, page 104217.

\bibitem[Wang et~al., 2023]{wang2023does}
Wang, T., Zhao, S., Wang, W., and Yang, H. (2023).
\newblock How does exogenous shock change the structure of interbank network?:
  evidence from china under covid-19.
\newblock {\em Emerging Markets Finance and Trade}, 59(4):937--958.

\bibitem[Zakine et~al., 2024]{zakine2024socioeconomic}
Zakine, R., Garnier-Brun, J., Becharat, A.-C., and Benzaquen, M. (2024).
\newblock Socioeconomic agents as active matter in nonequilibrium
  sakoda-schelling models.
\newblock {\em Physical Review E}, 109(4):044310.

\end{thebibliography}

\renewcommand{\qedsymbol}{}
\appendix

\section{Existence and Uniqueness of Clearing Equilibrium}

\subsection{Proof of Lemma \ref{Drop in LR}}
\begin{proof}[\normalfont\bfseries Proof of lemma \ref{Drop in LR}]
	The devaluation of asset and equity is $\Delta A>0$. We have $\theta_i-\theta_i'=\frac{E_i}{A_i}-\frac{E_i-\Delta A}{A_i-\Delta A}=\frac{\left(A_i-E_i\right)\Delta A}{A_i(A_i-\Delta A)}$. Since $A_i-E_i>0$, we have $\theta_i-\theta_i'>0$.
\end{proof}

\subsection{Proof of Proposition \ref{Solutions in FCP}}
\begin{lemma}\label{Relationship of Monopoly and Specialization} \,
	
	$a)$ \textnormal{The determinants of matrix $\mathbf{\Pi}$ and matrix $\mathring{\mathbf{\Pi}}$ share the same sign. }
	
	$b)$ \textnormal{Under assumption of stochastic proportion of portfolios, $\det{\mathbf{\Pi}} > 0$.}
\end{lemma}
\begin{proof}[\normalfont\bfseries Proof of lemma \ref{Relationship of Monopoly and Specialization}]
	\,
	
	\textbf{Proof of $a)$}:
	Let $\mathcal{M}$ be the set of monopolistic banks and $\mathcal{N}$ the set of specialized banks. According to the rules of determinant operations, we obtain $\det{\mathbf{\Pi}}=\left(\prod\limits_{m} \pi_{mk}\right)\cdot\left(\prod\limits_n\pi_{nk}\right)\cdot\det{\mathring{\mathbf{\Pi}}}$ . Since $\prod\limits_{m} \pi_{mk}\,,\,\prod\limits_n\pi_{nk}>0$, then we have $\det{\mathbf{\Pi}}$ and $\det{\mathring{\mathbf{\Pi}}}$ share the same sign.
	
	\textbf{Proof of $b)$}:
	Under assumption, we have $\mathring{\mathbf{\Pi}}$ is a full rank matrix, so $\det{\mathring{\mathbf{\Pi}}}>0$, and by $a)$ we have $\det{\mathbf{\Pi}}>0$.
\end{proof}

\begin{proof}[\normalfont\bfseries Proof of Proposition~\ref{Solutions in FCP}]
	\,
	
	We firstly construct an iteration function of $\mathbf{x}$. The original function fails to conclude the conditions of existence and uniqueness of fixed point.
	
	Set $\mathbf{G}^{n\times 1}(\mathbf{x}^{n\times 1})=\mathbf{s}^{n\times 1}(\mathbf{x}^{n\times 1}) - \mathbf{s}^{n\times 1}$ and $\mathbf{s}=(s^i)^{n\times 1}$. We abbreviate it as $\mathbf{G}(\mathbf{x})=\mathbf{s}(\mathbf{x}) - \mathbf{s}$. We need to prove that the solution of $\mathbf{G}(\mathbf{x^*})=\mathbf{s}(\mathbf{x^*}) - \mathbf{s}$ exist and unique. We select one of the components of $\mathbf{G}(\mathbf{x})=0$ (i.e. $x^i+\sum\limits_k\left(1-e^{\sum\limits_i -\beta x^i\pi_{ik}} \right)M_k\cdot p_{ki}-s^i=0$) and rewrite the equation into formula \ref{iteration_x^i FCP}
	\begin{equation}\label{iteration_x^i FCP}
		x^i=s^i+\sum\limits_k\left(e^{\sum\limits_i -\beta x^i\pi_{ik}}-1 \right)M_k\cdot p_{ki}
	\end{equation}
	We define $Q_k = \sum_i \beta x^i \pi_{ik}$, observe that $Q_k$ is the same for all $i$ (since $k$ is a summation index), then Substitute $x^i$ into $Q_k$, we have
	\begin{equation*}
		Q_k = \sum_i \beta \left[ s^i - \sum_l \left( 1 - e^{-Q_l} \right) M_l p_{li} \right] \pi_{ik}
	\end{equation*}
	Expanding it, we get  
	\begin{equation*}\label{iteration_Q_k FCP}
		Q_k = \beta \sum_i s^i\pi_{ik} - \beta \sum_l \left( 1 - e^{-Q_l} \right) M_l \sum_i p_{li} \pi_{ik}
	\end{equation*}
	Let $\Pi_k = \sum_i s^i\pi_{ik}$ and $C_{kl} = \sum_i p_{li} \pi_{ik}$, then we have 
	\begin{equation*}\label{iteration_Q_kcomponent FCP}
		Q_k = \beta \Pi_k - \beta \sum_l \left( 1 - e^{-Q_l} \right) M_l C_{kl}
	\end{equation*}
	Defining the vector $\mathbf{Q} = (Q_1, \ldots, Q_k)^T$, $\mathbf{\Pi} = (\Pi_1, \ldots, \Pi_k)^T$, $\mathbf{C}$ represents the matrix $C_{kl}$, $\mathbf{M} = diag(M_1, \ldots, M_k)$, and $e^{-\mathbf{Q}}$ denotes element-wise exponential operation. The equation can be expressed as \ref{iteration_Q_Matrix FCP}:
	\begin{equation}\label{iteration_Q_Matrix FCP}
		\mathbf{Q} = \beta \mathbf{\Pi} - \beta \mathbf{C M}(1 - e^{-\mathbf{Q}})
	\end{equation}
	
	This constitutes a nonlinear equation in $\mathbf{Q}$, whose solution existence can be established through the fixed-point theorem and solution uniqueness can be established through contraction mapping (i.e., with Lipschitz constant less than 1). 
	
	\noindent\textbf{Proof of a)}:
	
	We define the mapping $\mathbf{F}(\mathbf{Q}) = \beta \mathbf{\Pi} - \beta \mathbf{C M}(1 - e^{-\mathbf{Q}})$ and redefine the domain of mapping $F$ as a closed ball: $\mathcal{Q} = B_R(0) = \{ \mathbf{Q} \in \mathbb{R}^k \mid \|\mathbf{Q}\| \leq R \}$. We need to prove that $\mathbf{F}$ is a continuous self-mapping, thus guaranteeing the existence of a fixed point via the Brouwer Fixed-Point Theorem.
	
	We need to prove $\|\mathbf{F}(\mathbf{Q})\| \leq R$ holds for all $\|\mathbf{Q}\| \leq R$, then we can meet the self-mapping condition $\mathbf{F}(B_R(0)) \subseteq B_R(0)$. Taking the $\ell^2$-norm, we have \ref{iteration_Q_Norm FCP}:
	\begin{equation}\label{iteration_Q_Norm FCP}
		\|\mathbf{F}(\mathbf{Q})\| \leq \beta  \|\mathbf{\Pi}\| + \beta \|\mathbf{C M}\| \cdot \|1 - e^{-\mathbf{Q}}\|
	\end{equation}
	Since $e^{-\mathbf{Q}}$ is an element-wise exponential and $\|\mathbf{Q}\| \leq R$, each component $Q_i \in [-R, R]$. Therefore $e^{-Q_i} \in [e^{-R}, e^{R}]$, hence $1 - e^{-Q_i} \in [1 - e^{R}, 1 - e^{-R}]$. Then we have $\|1 - e^{-\mathbf{Q}}\| \leq \sqrt{m} \max \{ |1 - e^{R}|, |1 - e^{-R}|\}$. Then we have (we deliberately drop $1-e^{-R}$, since it's bounded by $1$): 
	\begin{equation*}\label{iteration_Q_Supremum FCP}
		\|1 - e^{-\mathbf{Q}}\| \leq \sqrt{m} (e^{R}-1)
	\end{equation*}
	Since $0\leq\bar{\theta}\leq1$ and $\max{M_k}<+\infty$, we have $\mathbf{\Pi}$ is bounded by $S$ which subject to $\max\limits_i{|s^i|\pi_{ik}}<\max\limits_i{|s^i|}\leq \max\limits_i\{|s^i|,W\sqrt{m}\}<S<+\infty$ and $\mathbf{CM}$ is bounded by $W$ which subject to $\max{M_k}\cdot\sum_i p_{li}\pi_{ik}<\max{M_k}\cdot k <W<+\infty$ (i.e., $\|\mathbf{\Pi}\| \leq S$ and $\|\mathbf{CM}\| \leq W$). Then we have \ref{Norm_Supremum FCP}.
	\begin{equation}\label{Norm_Supremum FCP}
		\|F(Q)\| \leq \beta S + \beta W\sqrt{m}(e^{R}-1)
	\end{equation}
	
	To satisfy $F(B_R(0)) \subseteq B_R(0)$, we require: $\beta S + \beta W\sqrt{m}(e^{R}-1)\leq R$. i.e., $\beta S + \beta W\sqrt{m}(e^{R}-1)-R<0$ . Set $H(R)=\beta S + \beta W\sqrt{m}(e^{R}-1)-R$, we can easily obtain $H(0)-\beta S>0$ and $H(\infty)\approx \beta W\sqrt{m}e^{\infty}>0$. We hope to obtain $H(R)_{\min}<0$. Let $\frac{\partial H(R)}{\partial R}=\beta W\sqrt{m}e^{R_0}-1=0$, we obtain $R_0=-\ln{\left(\beta W\sqrt{m}\right)}$ (we require $\beta<\frac{1}{W\sqrt{m}}$). Since $\frac{\partial^2 H(R)}{\partial R^2}=\beta W\sqrt{m}e^{R_0}>0$, we have $H(R)_{\min}=H(R_0)<0$. Let $\mu(\beta)=H(R_0)=\beta S + 1-\beta W\sqrt{m}+\ln{\beta W\sqrt{m}}$. We can easily obtain $\mu(a)\rightarrow -\infty$ if $a\rightarrow 0^+$ and $\mu(b)\rightarrow \frac{1}{W\sqrt{m}}\cdot\left(S-W\sqrt{m}\right)+1>1$ if $b\rightarrow \frac{1}{W\sqrt{m}}$, and $\mu(\beta)$ is non-decreasing in $\beta$. Then by Intermediate Value Theorem, we can find a threshold $\bar{\beta}_1$ subject to $\mu(\bar{\beta}_1)=0$ and when $\beta\leq \bar{\beta}_1$, we meet the self-mapping condition $F(B_R(0)) \subseteq B_R(0)$. We also can rewrite the $\bar{\beta}_1$ as $\bar{\beta}_1(\varepsilon)$, since $S$ can be rewrite as $S(\varepsilon)$.
	
	Continuity of $F$ is guaranteed by continuity of linear term $\beta \mathbf{\Pi}$, continuity of matrix multiplication $\mathbf{C M}$, continuity of exponential operation $e^{-Q_l}$. Thus $F$ is a continuous mapping.
	
	Now we have $B_R(0)$ is a non-empty compact convex set in $\mathbb{R}^m$ and $F : B_R(0) \to B_R(0)$ is a continuous self-mapping. By the Brouwer Fixed-Point Theorem, there exists $Q^* \in B_R(0)$ such that $F(Q^*) = Q^*$.
	
	As for the uniqueness, we only need to prove that $F$ is a contraction mapping. The elements of the Jacobian matrix of $F$ are:
	\begin{equation*}\label{Jacobian Matrix}
		\frac{\partial F_k}{\partial Q_l} = J_{kl} =-\beta M_l C_{kl} e^{-Q_l}
	\end{equation*}
	By $Q_k = \sum_i \beta x^i \pi_{ik}>0$, we have $0<|\frac{\partial F_k}{\partial Q_l}|<\beta M_l C_{kl}$, we always can find such a $\bar{\beta}_2$ that $\bar{\beta}_2 \cdot\max\limits_{k,l}{M_l C_{kl}}<1$ is always true. Then the Lipschitz constant $L$ satisfies \ref{Lipschitz constant}:
	\begin{equation}\label{Lipschitz constant}
		L \leq \bar{\beta}_2 \max\limits_{k,l}{M_l C_{kl}}<1
	\end{equation}
	then $F$ is a contraction mapping, and the solution is unique.
	
	Consider $\mathbf{Q}= \mathbf{\Pi}^T \mathbf{x}$, in this equation, a unique correspondence between $\mathbf{Q}$ and $\mathbf{x}$ exists if and only if $\det{\mathbf{\Pi}^T}>0$. By Lemma \ref{Relationship of Monopoly and Specialization} $b)$, We have $\det{\mathbf{\Pi}>0}$. Since $\det{\mathbf{\Pi}}= \det{\mathbf{\Pi}^T}$, it follows that $\det{\mathbf{\Pi}^T}>0$.
	
	We take the threshold $\bar{\beta}(\varepsilon)=\min\{\bar{\beta}_1(\varepsilon),\bar{\beta}_2\}$. When $\beta<\bar{\beta}(\varepsilon)$, the existence and uniqueness of fixed point $\mathbf{x}^*$ are guaranteed.
	
	Q.E.D.
	
	\noindent\textbf{Proof of b)}:
	
	Recall that $s^i\left(\varepsilon\right)=\frac{1}{\bar{\theta}}\left[\bar{\theta}A_i-E_i+(1-\bar{\theta})\varepsilon A_i  \right]$, we have $\frac{\partial s^i\left(\varepsilon\right)}{\partial\varepsilon}=\frac{1-\bar{\theta}}{\theta}A_i>0$, so $s^i\left(\varepsilon\right)$ is non-decreasing in $\varepsilon$. When $\varepsilon\rightarrow 0$, we have $s^i\rightarrow A_i-\frac{E_i}{\bar{\theta}}=A_i(1-\frac{\theta^i}{\bar{\theta}})<0$. When $\varepsilon >\frac{E_i}{A_i}=\theta^i$, we have $s^i(\varepsilon)=\frac{1}{\bar{\theta}}\left[\bar{\theta}A_i-E_i+(1-\bar{\theta})\varepsilon A_i  \right]>\frac{1}{\bar{\theta}}\left[\bar{\theta}A_i-E_i+(1-\bar{\theta})E_i \right]=A_i-E_i>0$. By  Intermediate Value Theorem, there exists an $\varepsilon_z \in \left(0,\max\limits_i\left\{\frac{E_i}{A_i}\right\}\right)$ such that $s^i\left(\varepsilon_z\right)=0$. By monotoniticity of $s^i\left(\varepsilon\right)$, when $\varepsilon<\varepsilon_z$, we have $s^i\left(\varepsilon\right)<0$. Since $Q_k = \sum_i \beta x^i \pi_{ik}>0$, we have $\sum\limits_k\left(e^{\sum\limits_i -Q_k} -1\right)<0$. Formula \ref{Critical Point_credit} demonstrates that when $\varepsilon<\varepsilon_z=\bar{\varepsilon}$, we have at least one $s^p<0$ and $x^p<0$.
	\begin{equation}\label{Critical Point_credit}
		x^p\,=\,\left[\underbrace{s^p}_{<0}+\underbrace{\sum\limits_k\left(e^{\sum\limits_i -\beta x^i\pi_{ik}}-1 \right)M_k\cdot p_{ki}}_{<0}\right]\,<\,0
	\end{equation}
	
	\noindent\textbf{Proof of c)}:
	
	Set $d^i=\sum\limits_k\left(e^{\sum\limits_i -\beta x^i\pi_{ik}}-1 \right)M_k\cdot p_{ki}<0$. Substitute $s^i$ into the expression for $x^i$, we obtain $x^i=\frac{1}{\bar{\theta}}\left[\bar{\theta}A_i-E_i+(1-\bar{\theta})\varepsilon A_i  \right]+d^i$.  By solving the inequality $x^p>A_p$, we obtain $\varepsilon_p>\left(\frac{E_p}{\bar{\theta}-d^i}\right)\cdot\frac{\bar{\theta}}{1-\bar{\theta}}\cdot\frac{1}{A_p}>0$. Set $\tilde{\varepsilon}=\min\limits_p\left\{\left(\frac{E_p}{\bar{\theta}-d^i}\right)\cdot\frac{\bar{\theta}}{1-\bar{\theta}}\cdot\frac{1}{A_p}\right\}$, when $\varepsilon>\tilde{\varepsilon}$, there exist at least one $x^p>A_p$. Since $x^p=s^p+d^p<s^p$, we have $s^p>A_p$ as well.
\end{proof}
\section{Characterization of Crowding-Out Effect}\label{AppendixSection: Characterization of Crowding-Out Effect}
We prove Lemma \ref{Corresponding to Counterfactual Monopoly Space}, Proposition \ref{Propositions about Counterfactual Monopoly Space} and Proposition \ref{Propositions about Counterfactual Monopoly Space2} in Online Appendix Section OA2.
\subsection{Proof of Lemma \ref{Lemma: Crowding-Out Effect Opposite}}
\begin{proof}[\normalfont\bfseries Proof of Lemma \ref{Lemma: Crowding-Out Effect Opposite}]
	\,
	
	It's a direct corollary of Lemma \ref{Lemma: Crowding-Out Effect} $(a)$ and $(b)$.
\end{proof}
\subsection{Proof of Lemma \ref{Lemma: Crowding-Out Effect Homophily}}
\begin{restatable}{thm3}{LemmaCrowdingOutEffect}\label{SubLemma: Crowding-Out Effect}
	In addition to Lemma \ref{Lemma: Crowding-Out Effect}:
	
	$(a)$ \textnormal{If a bank $i$ (or a cluster $\mathcal{P^I}$) is added to the $\mathcal{P^+}$ such that $\tilde{\mathbf{\Phi}}^{i}\left(\mathbf{s}^{|\mathcal{P}^+\cup\{i\}|\times 1}\right)>0$ (or $\left[\tilde{\mathbf{\Phi}}\left(\mathbf{s}^{|\mathcal{P}^+\cup\mathcal{P^I}|\times 1}\right)\right]^{|\mathcal{P^I}|\times 1}>0$), then $\tilde{\mathbf{\Phi}}\left(\mathbf{s}^{|\mathcal{P}^+|\times 1}\right) - \left[\tilde{\mathbf{\Phi}}\left(\mathbf{s}^{|\mathcal{P}^+\cup\{i\}|\times 1}\right)\right]^{|\mathcal{P}^+|\times 1} > 0$ (or $\tilde{\mathbf{\Phi}}\left(\mathbf{s}^{|\mathcal{P}^+|\times 1}\right) - \left[\tilde{\mathbf{\Phi}}\left(\mathbf{s}^{|\mathcal{P}^+\cup\mathcal{P^I}|\times 1}\right)\right]^{|\mathcal{P}^+|\times 1} > 0$) holds.}
	
	$(b)$ \textnormal{If a bank $j$ (or a cluster $\mathcal{P^J}$) is added to the $\mathcal{P^-}$ such that $\tilde{\mathbf{\Phi}}^{j}\left(\mathbf{s}^{|\mathcal{P}^-\cup\{j\}|\times 1}\right) < 0$ (or $\left[\tilde{\mathbf{\Phi}}\left(\mathbf{s}^{|\mathcal{P}^-\cup\mathcal{P^J}|\times 1}\right)\right]^{|\mathcal{P^J}|\times 1} < 0$), then $\tilde{\mathbf{\Phi}}\left(\mathbf{s}^{|\mathcal{P}^-|\times 1}\right) - \left[\tilde{\mathbf{\Phi}}\left(\mathbf{s}^{|\mathcal{P}^-\cup\{j\}|\times 1}\right)\right]^{|\mathcal{P}^-|\times 1} < 0$ (or $\tilde{\mathbf{\Phi}}\left(\mathbf{s}^{|\mathcal{P}^-|\times 1}\right) - \left[\tilde{\mathbf{\Phi}}\left(\mathbf{s}^{|\mathcal{P}^-\cup\mathcal{P^J}|\times 1}\right)\right]^{|\mathcal{P}^-|\times 1} < 0$) holds.}
\end{restatable}

\begin{proof}[\normalfont\bfseries Proof of Lemma \ref{SubLemma: Crowding-Out Effect}]
	\,
	
	\noindent\textbf{Proof of $(a)$}:
	
	We prove it by contradiction. If the following inequality holds:
	\begin{align}\label{Equation system: SubLemma final contradiction inequality proof a}
		\tilde{\mathbf{\Phi}}\left(\mathbf{s}^{|\mathcal{P}^+|\times 1}\right) - \left[\tilde{\mathbf{\Phi}}\left(\mathbf{s}^{|\mathcal{P}^+\cup\mathcal{P^I}|\times 1}\right)\right]^{|\mathcal{P}^+|\times 1} < 0
	\end{align}
	Then the first row of system \ref{Equation system: COE new for proof a} will never be achieved, since $\tilde{\mathbf{\Phi}}^{i}\left(\mathbf{s}^{|\mathcal{P}^+\cup\{i\}|\times 1}\right)>0$ also.
	
	\noindent\textbf{Proof of $(b)$}:
	
	By the same argument, the following inequality can not hold.
	\begin{align}\label{Equation system: SubLemma final contradiction inequality proof b}
		\tilde{\mathbf{\Phi}}\left(\mathbf{s}^{|\mathcal{P}^-|\times 1}\right) - \left[\tilde{\mathbf{\Phi}}\left(\mathbf{s}^{|\mathcal{P}^-\cup\mathcal{P^J}|\times 1}\right)\right]^{|\mathcal{P}^-|\times 1} > 0
	\end{align}
	Therefore, we have:
	\begin{align}\label{Equation system: SubLemma Right inequality proof b}
		\tilde{\mathbf{\Phi}}\left(\mathbf{s}^{|\mathcal{P}^-|\times 1}\right) - \left[\tilde{\mathbf{\Phi}}\left(\mathbf{s}^{|\mathcal{P}^-\cup\mathcal{P^J}|\times 1}\right)\right]^{|\mathcal{P}^-|\times 1} < 0
	\end{align}
\end{proof}

\begin{proof}[\normalfont\bfseries Proof of Lemma \ref{Lemma: Crowding-Out Effect Homophily}]
	\,
	
	It's a direct corollary of Lemma \ref{SubLemma: Crowding-Out Effect} $(a)$ and $(b)$.
\end{proof}

\subsection{Proof of Lemma \ref{Lemma: Crowding-Out Effect}}
\begin{proof}[\normalfont\bfseries Proof of lemma \ref{Lemma: Crowding-Out Effect}]
	\,
	
	\noindent\textbf{Proof of $(a)$}:
	
	We write the original system \ref{Equation system: COE original for proof a}:
	\begin{align}\label{Equation system: COE original for proof a}
		s^i &= \tilde{\mathbf{\Phi}}^i\left(\mathbf{s}^{|\mathcal{P}^+|\times 1}\right) +  \left(1-e^{-\beta\cdot \sum\limits_{j\in \mathcal{P^+}}\tilde{\mathbf{\Phi}}^j\left(\mathbf{s}^{|\mathcal{P}^+|\times 1}\right)}\right)\cdot A_i  \\  &\quad\quad\quad\quad\quad\quad\quad\quad\quad \notag  \vdots
	\end{align}
	and the system $\mathcal{P}^+\cup\mathcal{P^I}$: 
	\begin{align}\label{Equation system: COE new for proof a}
		s^i &= \underbrace{\tilde{\mathbf{\Phi}}^i\left(\mathbf{s}^{|\mathcal{P}^+\cup\mathcal{P^I}|\times 1}\right)}_{\nearrow} +  \underbrace{\left(1-e^{-\beta\cdot \left[\sum\limits_{j\in \mathcal{P^+}}\tilde{\mathbf{\Phi}}^j\left(\mathbf{s}^{|\mathcal{P}^+\cup\mathcal{P^I}|\times 1}\right)   + \sum\limits_{c\in \mathcal{P^I}}\tilde{\mathbf{\Phi}}^c\left(\mathbf{s}^{|\mathcal{P}^+\cup\mathcal{P^I}|\times 1}\right)    \right]}\right)\cdot A_i}_{\searrow} \\&\quad\quad\quad\quad\quad\quad\quad\quad\quad \notag  \vdots \\  s^t &= \tilde{\mathbf{\Phi}}^t\left(\mathbf{s}^{|\mathcal{P}^+\cup\mathcal{P^I}|\times 1}\right) +  \left(1-e^{-\beta\cdot \left[\sum\limits_{j\in \mathcal{P^+}}\tilde{\mathbf{\Phi}}^j\left(\mathbf{s}^{|\mathcal{P}^+\cup\mathcal{P^I}|\times 1}\right)   + \sum\limits_{c\in \mathcal{P^I}}\tilde{\mathbf{\Phi}}^c\left(\mathbf{s}^{|\mathcal{P}^+\cup\mathcal{P^I}|\times 1}\right)    \right]}\right)\cdot A_t  \notag \\&\quad\quad\quad\quad\quad\quad\quad\quad\quad \notag  \vdots
	\end{align}
	We observe that if bank $i$ persists the action $\tilde{\mathbf{\Phi}}^i\left(\mathbf{s}^{|\mathcal{P}^+|\times 1}\right)$, then the second term in new system \ref{Equation system: COE new for proof a} will decrease due to $\left[\tilde{\mathbf{\Phi}}\left(\mathbf{s}^{|\mathcal{P}^+\cup\mathcal{P^I}|\times 1}\right)\right]^{|\mathcal{P^I}|\times 1}<0$. In order to meet the requirement of fixed $s^i$, bank $i$ has to increase its sales volume, therefore $\tilde{\mathbf{\Phi}}\left(\mathbf{s}^{|\mathcal{P}^+|\times 1}\right) < \left[\tilde{\mathbf{\Phi}}\left(\mathbf{s}^{|\mathcal{P}^+\cup\mathcal{P^I}|\times 1}\right)\right]^{|\mathcal{P}^+|\times 1}$. The property of \textit{cross-bank dependence} in Lemma \ref{Perfection Switch=1} and the first item of Lemma \ref{Lemma ForDecomposition and Compression Equivalence} ensure the existence and uniqueness of the new equilibrium. Therefore, we have:
	\begin{align}\label{Equation system: COE final inequality proof a}
		\tilde{\mathbf{\Phi}}\left(\mathbf{s}^{|\mathcal{P}^+|\times 1}\right) - \left[\tilde{\mathbf{\Phi}}\left(\mathbf{s}^{|\mathcal{P}^+\cup\mathcal{P^I}|\times 1}\right)\right]^{|\mathcal{P}^+|\times 1} < 0
	\end{align}
	where the individual bank $i$ is the special case of cluster $\mathcal{P^I}$.
	
	Another viewpoint for this proof is by contradiction, if the following inequality holds:
	\begin{align}\label{Equation system: COE final contradiction inequality proof a}
		\tilde{\mathbf{\Phi}}\left(\mathbf{s}^{|\mathcal{P}^+|\times 1}\right) - \left[\tilde{\mathbf{\Phi}}\left(\mathbf{s}^{|\mathcal{P}^+\cup\mathcal{P^I}|\times 1}\right)\right]^{|\mathcal{P}^+|\times 1} > 0
	\end{align}
	Then the first row of system \ref{Equation system: COE new for proof a} will never be achieved.
	
	\noindent\textbf{Proof of $(b)$}:
	
	By the same argument, $\left[\tilde{\mathbf{\Phi}}\left(\mathbf{s}^{|\mathcal{P}^-\cup\mathcal{P^J}|\times 1}\right)\right]^{|\mathcal{P^J}|\times 1} > 0$ leads to the increase in the second term of system \ref{Equation system: COE new for proof a}. Then bank $i$ has to decrease its sales volume. Therefore: 
	\begin{align}\label{Equation system: COE final inequality proof b}
		\tilde{\mathbf{\Phi}}\left(\mathbf{s}^{|\mathcal{P}^-|\times 1}\right) - \left[\tilde{\mathbf{\Phi}}\left(\mathbf{s}^{|\mathcal{P}^-\cup\mathcal{P^J}|\times 1}\right)\right]^{|\mathcal{P}^-|\times 1} > 0
	\end{align}
	where the individual bank $j$ is the special case of cluster $\mathcal{P^J}$.
	
	\noindent\textbf{Proof of $(c)$}:
	
	It's a more relaxed condition than Lemma \ref{Lemma: Crowding-Out Effect} $(a)$ and $(b)$. The second term of system \ref{Equation system: COE new for proof a} must increase under the condition $\sum\limits_{b\in \mathcal{P}}\tilde{\mathbf{\Phi}}^{b}\left(\mathbf{s}^{|\mathcal{P}|\times 1}\right) < \sum\limits_{c\in \mathcal{P}\cup\mathcal{P^N}}\tilde{\mathbf{\Phi}}^{c}\left(\mathbf{s}^{|\mathcal{P}\cup\mathcal{P^N}|\times 1}\right)$, then the first term has to decrease since $s^i<0$, that is:
	\begin{align}\label{Equation system: COE final inequality proof c}
		\tilde{\mathbf{\Phi}}^{i}\left(\mathbf{s}^{|\mathcal{P}\cup\mathcal{P^N}|\times 1}\right) < \tilde{\mathbf{\Phi}}^{i}\left(\mathbf{s}^{|\mathcal{P}|\times 1}\right)
	\end{align}
	
	\noindent\textbf{Proof of $(d)$}:
	
	By the same argument, The second term of system \ref{Equation system: COE new for proof a} must decrease under the condition $\sum\limits_{b\in \mathcal{P}}\tilde{\mathbf{\Phi}}^{b}\left(\mathbf{s}^{|\mathcal{P}|\times 1}\right) > \sum\limits_{c\in \mathcal{P}\cup\mathcal{P^N}}\tilde{\mathbf{\Phi}}^{c}\left(\mathbf{s}^{|\mathcal{P}\cup\mathcal{P^N}|\times 1}\right)$, then the first term has to decrease since $s^j<0$, that is:
	\begin{align}\label{Equation system: COE final inequality proof d}
		\tilde{\mathbf{\Phi}}^{j}\left(\mathbf{s}^{|\mathcal{P}\cup\mathcal{P^N}|\times 1}\right) > \tilde{\mathbf{\Phi}}^{j}\left(\mathbf{s}^{|\mathcal{P}|\times 1}\right)
	\end{align}
\end{proof}
\section{Proof of Theorems}\label{Proof of Theorems}
\subsection{Proof of Lemma \ref{Lemma For Commutative Diagram} and \ref{Lemma ForPartitionInduced Equilibrium Transition Snd}: preliminary 1}\label{preliminaries:part1}
For any $\mathcal{P}_{s\,|\,t}$ with $\theta^j\ll \theta^i$ where $i\in \mathcal{B}(\varepsilon)$ (note that it can be always achieved according to Lemma \ref{Perfection Switch=1} 1. $b)$). We can find an $\varepsilon$ such that $s^j\leq 0$ and $s^i<0$. Suppose all banks in $\mathcal{B}(\varepsilon)$ or $\mathcal{P}_{s\,|\,t}$ are strictly heterogeneous (note that this assumption is established in the proof of Theorem \ref{Decomposition and Compression Equivalence}). 

Similar to the concept of Counterfactual Monopoly Space, we define Counterfactual Cluster Solution $\tilde{\mathbf{\Phi}}\left(\mathbf{s}^{|\mathcal{B}\left(\varepsilon\right)|\times 1}\right)$, which refers to the system constructed by $\mathcal{B}(\varepsilon)=\mathcal{P}_{s\,|\,t}\setminus \{j\}$.

\begin{restatable}{thm3}{LemmaForDecompositionAndCompressionEquivalence}\label{Lemma ForDecomposition and Compression Equivalence}
	Under $\varepsilon$, the system $\mathcal{P}_{s\,|\,t}$ with $\theta^j\ll \theta^i$, $s^j\leq 0$ and $s^i<0$  where $i\in \mathcal{B}(\varepsilon)$ exhibits:
	
	1. $\tilde{\mathbf{\Phi}}\left(\mathbf{s}^{|\mathcal{B}\left(\varepsilon\right)|\times 1}\right)$ is guaranteed to exist and unique regardless of $\varepsilon$ and $\beta<+\infty$.
	
	2. $e^{-\beta\cdot\sum\limits_{i\in\mathcal{B}\left(\varepsilon\right)} \tilde{\Phi}^i\left(\mathbf{s}^{|\mathcal{B}\left(\varepsilon\right)|\times 1}\right)}$ is non-decreasing in $\beta$, $A_i$ and $\sum\limits_i A_i$.
	
	3. Denote $f$ satisfies $\theta^f=\min\limits_{k\in\mathcal{B}\left(\varepsilon\right)} \theta^k$ and all $\theta^k$ remain fixed. Then $\mathcal{B}(\varepsilon)=\mathcal{P}_{s\,|\,t}\setminus \{j\}$ requires:
	
	\hspace{1em}$i)$ Ignition condition: $e^{-\beta\cdot\sum\limits_{i\in\mathcal{B}\left(\varepsilon\right)} \tilde{\Phi}^i\left(\mathbf{s}^{|\mathcal{B}\left(\varepsilon\right)|\times 1}\right)}\geq \frac{\theta^j}{\bar{\theta}}+ \underbrace{\left(1-\frac{1}{\bar{\theta}}\right)\cdot \varepsilon}_{<0}$.
	
	\hspace{1em}$ii)$ $\sup\limits_{\beta,\, A_{m\neq f},\,\sum\limits_{m\neq f} A_m} \left[e^{-\beta\cdot\sum\limits_{i\in\mathcal{B}\left(\varepsilon\right)} \tilde{\Phi}^i\left(\mathbf{s}^{|\mathcal{B}\left(\varepsilon\right)|\times 1}\right)}\right]=\min\limits_{k\in\mathcal{B}\left(\varepsilon\right)}\frac{\theta^k}{\bar{\theta}}+\left(1-\frac{1}{\bar{\theta}}\right)\cdot \varepsilon$.
	
	\hspace{1em}$iii)$ $\beta$ satisfies $\beta_{x^j(\varepsilon)=0}<\beta<\beta_{x^f(\varepsilon)<0}$.
	
	\hspace{1em}$iv)$ The tuple $\mathbf{A}_{m\neq f}^{|\mathcal{B}\left(\varepsilon\right)-1|\times 1}$ or $\sum\limits_{m\neq f} A_m$ is bounded.
	
	4. Increasing only $A_f$ into $+\infty$ doesn't reverse the conclusion of $\mathcal{B}(\varepsilon)=\mathcal{P}_{s\,|\,t}\setminus \{j\}$ in 3.
	
	5. Threshold $\beta_{x^j(\varepsilon)=0}$ is non-increasing in $\varepsilon$, $\max \beta_{x^j(\varepsilon)=0}=\beta_{x^j(\varepsilon=0)=0}$. $\beta_{x^j(\varepsilon)=0}$ is non-decreasing in $\theta^j$.
	
	6. Under $\max \beta_{x^j(\varepsilon)=0}<\beta<\beta_{x^f(\varepsilon)<0}$, we have:
	
	\hspace{1em}$i)$ \hspace{0.5em}Back to $\varepsilon=0$, the signs of $\tilde{\Phi}^i\left(\mathbf{s}^{|\mathcal{B}\left(\varepsilon\right)|\times 1}\right)$ and $x^j$ remain unchanged.
	
	\hspace{1em}$ii)$ \hspace{0.2em}As $\varepsilon$ grows, partial order among $\tilde{\Phi}^i\left(\mathbf{s}^{|\mathcal{B}\left(\varepsilon\right)|\times 1}\right)$ doesn't preserve, but Hierarchy remains.
	
	\hspace{1em}$iii)$ \hspace{-0.2em}As $\varepsilon$ grows, There exist a threshold $\varepsilon_{x^f(\varepsilon)=0}=\varepsilon_{\max\limits_{k\in\mathcal{B}(\varepsilon)} x^k =0 }$ such that $\max\limits_{k\in\mathcal{B}(\varepsilon)} x^k =0$ or $x^f(\varepsilon_{\max\limits_{k\in\mathcal{B}(\varepsilon)} x^k =0 })$=0 where $f$ satisfies $\theta^f=\min\limits_k \theta^k$.
	
	\hspace{1em}$iv)$ \hspace{0.2em}$\mathcal{B}(\varepsilon)\setminus \{f\}$ is a Maximal Bailout Cluster when $\varepsilon=\varepsilon_{\max\limits_{k\in\mathcal{B}(\varepsilon)} x^k =0 }$.
\end{restatable}

\begin{proof}[\normalfont\bfseries Proof of Lemma \ref{Lemma ForDecomposition and Compression Equivalence}]
	\,
	
	\noindent\textbf{Proof of $1$}: 
	
	Recall that $s^i=\frac{1}{\bar{\theta}}\left(\bar{\theta}A_i-E_i+(1-\bar{\theta})\cdot \varepsilon\cdot A_i  \right)$ and $s^i=x^i+\left(1-e^{-\beta\cdot \sum\limits_{j\in \mathcal{P}_{s\,|\,t} }x^j}\right)\cdot A_i$. Then the system $\mathcal{P}_{s\,|\,t}$ exhibits \ref{system P} where $L^i=\frac{1}{\bar{\theta}}\cdot \left[ \bar{\theta}-\theta^i +\left(1-\bar{\theta}\right)\cdot \varepsilon \right]$: 
	\begin{align}\label{system P}
		A_i\cdot L^i=x^i+&\left(1-e^{-\beta\cdot \sum\limits_{j\in \mathcal{P}_{s\,|\,t}}x^j}\right)\cdot A_i \notag  \\  &\vdots \quad\quad\quad\quad\quad\quad\quad    \\   A_m\cdot L^m=x^m+&\left(1-e^{-\beta\cdot \sum\limits_{j\in \mathcal{P}_{s\,|\,t}}x^j}\right)\cdot A_m \notag
	\end{align}
	\noindent By adding the $\big|\mathcal{P}_{s\,|\,t}\big|$ equations together, we obtain equation \ref{system P Sum} where $A^{\mathcal{P}_{s\,|\,t}}=\sum\limits_{j\in \mathcal{P}_{s\,|\,t}}A_j$. Refer to Lemma OA1.1 $c)$ and Intermediate Value Theorem, there exists an unique solution $\sum\limits_{j\in \mathcal{P}_{s\,|\,t}} x^j$ of equation \ref{system P Sum}.
	\begin{align}\label{system P Sum}
		\sum\limits_{j\in \mathcal{P}_{s\,|\,t}}A_j\cdot L^j=\sum\limits_{j\in \mathcal{P}_{s\,|\,t}} x^j+\left(1-e^{-\beta\cdot \sum\limits_{j\in \mathcal{P}_{s\,|\,t} }x^j}\right)\cdot A^{\mathcal{P}_{s\,|\,t}}
	\end{align}
	\noindent Refer to \ref{Perfection Switch=1} $1.a)$, we can represent $\sum\limits_{j\in \mathcal{P}_{s\,|\,t}} x^j$ by any $x^j$ where $j\in \mathcal{P}_{s\,|\,t}$. That is, any $x^j$ where $j\in \mathcal{P}_{s\,|\,t}$ is unique. 
	
	The proof process does not involve specific $\varepsilon$ and $\beta<+\infty$. Moreover $\mathcal{B}(\varepsilon)=\mathcal{P}_{s\,|\,t}\setminus \{j\}$ and $\mathcal{P}_{s\,|\,t}$ share the same system structure (i.e., only one column). Therefore, $\tilde{\mathbf{\Phi}}\left(\mathbf{s}^{|\mathcal{B}\left(\varepsilon\right)|\times 1}\right)$ is guaranteed to exist and unique regardless of $\varepsilon$ and $\beta<+\infty$.
	
	\noindent\textbf{Proof of $2$}: 
	
	We firstly prove that $e^{-\beta\cdot\sum\limits_{i\in\mathcal{B}\left(\varepsilon\right)} \tilde{\Phi}^i\left(\mathbf{s}^{|\mathcal{B}\left(\varepsilon\right)|\times 1}\right)}$ is non-decreasing in $A_i$. Taking the derivative with respect to $A_i$ on both sides of \ref{system P Sum} yields \ref{Decom_Exp_Ai}.
	\begin{align}\label{Decom_Exp_Ai}
		L^i = \frac{\partial \sum\limits_{i\in\mathcal{B}\left(\varepsilon\right)} \tilde{\Phi}^i\left(\mathbf{s}^{|\mathcal{B}\left(\varepsilon\right)|\times 1}\right)}{\partial A_i} +1 - e^{-\beta\cdot\sum\limits_{i\in\mathcal{B}\left(\varepsilon\right)} \tilde{\Phi}^i\left(\mathbf{s}^{|\mathcal{B}\left(\varepsilon\right)|\times 1}\right)} + A^{\mathcal{B}\left(\varepsilon\right)}\cdot\beta\cdot e^{-\beta\cdot\sum\limits_{i\in\mathcal{B}\left(\varepsilon\right)} \tilde{\Phi}^i\left(\mathbf{s}^{|\mathcal{B}\left(\varepsilon\right)|\times 1}\right)} \cdot \frac{\partial\sum\limits_{i\in\mathcal{B}\left(\varepsilon\right)} \tilde{\Phi}^i\left(\mathbf{s}^{|\mathcal{B}\left(\varepsilon\right)|\times 1}\right)}{\partial A_i} 
	\end{align}
	\noindent By solving it, we obtain \ref{Decom_Exp_Ai Snd}.
	\begin{align}\label{Decom_Exp_Ai Snd}
		\frac{\partial\sum\limits_{i\in\mathcal{B}\left(\varepsilon\right)} \tilde{\Phi}^i\left(\mathbf{s}^{|\mathcal{B}\left(\varepsilon\right)|\times 1}\right)}{\partial A_i} &=\frac{L^i-1+e^{-\beta\cdot\sum\limits_{i\in\mathcal{B}\left(\varepsilon\right)} \tilde{\Phi}^i\left(\mathbf{s}^{|\mathcal{B}\left(\varepsilon\right)|\times 1}\right)}}{1+A^{\mathcal{B}\left(\varepsilon\right)}\cdot\beta\cdot e^{-\beta\cdot\sum\limits_{i\in\mathcal{B}\left(\varepsilon\right)} \tilde{\Phi}^i\left(\mathbf{s}^{|\mathcal{B}\left(\varepsilon\right)|\times 1}\right)}} \notag  \\   \Longleftrightarrow  \frac{\partial\sum\limits_{i\in\mathcal{B}\left(\varepsilon\right)} \tilde{\Phi}^i\left(\mathbf{s}^{|\mathcal{B}\left(\varepsilon\right)|\times 1}\right)}{\partial A_i} &=\frac{\frac{\tilde{\Phi}^i\left(\mathbf{s}^{|\mathcal{B}\left(\varepsilon\right)|\times 1}\right)}{A_i}+\left(1-e^{-\beta\cdot\sum\limits_{i\in\mathcal{B}\left(\varepsilon\right)} \tilde{\Phi}^i\left(\mathbf{s}^{|\mathcal{B}\left(\varepsilon\right)|\times 1}\right)}\right)-1+e^{-\beta\cdot\sum\limits_{i\in\mathcal{B}\left(\varepsilon\right)} \tilde{\Phi}^i\left(\mathbf{s}^{|\mathcal{B}\left(\varepsilon\right)|\times 1}\right)}}{1+A^{\mathcal{B}\left(\varepsilon\right)}\cdot\beta\cdot e^{-\beta\cdot\sum\limits_{i\in\mathcal{B}\left(\varepsilon\right)} \notag  \tilde{\Phi}^i\left(\mathbf{s}^{|\mathcal{B}\left(\varepsilon\right)|\times 1}\right)}} \\   \Longleftrightarrow  \frac{\partial\sum\limits_{i\in\mathcal{B}\left(\varepsilon\right)} \tilde{\Phi}^i\left(\mathbf{s}^{|\mathcal{B}\left(\varepsilon\right)|\times 1}\right)}{\partial A_i}  &=\frac{\tilde{\Phi}^i\left(\mathbf{s}^{|\mathcal{B}\left(\varepsilon\right)|\times 1}\right)}{A_i\cdot \left[1+A^{\mathcal{B}\left(\varepsilon\right)}\cdot\beta\cdot e^{-\beta\cdot\sum\limits_{i\in\mathcal{B}\left(\varepsilon\right)} \tilde{\Phi}^i\left(\mathbf{s}^{|\mathcal{B}\left(\varepsilon\right)|\times 1}\right)} \right]}<0 
	\end{align}
	\noindent Taking the derivative with respect to $A_i$ on $e^{-\beta\cdot\sum\limits_{i\in\mathcal{B}\left(\varepsilon\right)} \tilde{\Phi}^i\left(\mathbf{s}^{|\mathcal{B}\left(\varepsilon\right)|\times 1}\right)}$ yields \ref{Decom_Exp_Ai Third}:
	\begin{align}\label{Decom_Exp_Ai Third}
		\frac{\partial e^{-\beta\cdot\sum\limits_{i\in\mathcal{B}\left(\varepsilon\right)} \tilde{\Phi}^i\left(\mathbf{s}^{|\mathcal{B}\left(\varepsilon\right)|\times 1}\right)}}{\partial A_i}=-\beta\cdot e^{-\beta\cdot\sum\limits_{i\in\mathcal{B}\left(\varepsilon\right)} \tilde{\Phi}^i\left(\mathbf{s}^{|\mathcal{B}\left(\varepsilon\right)|\times 1}\right)}\cdot \frac{\partial \sum\limits_{i\in\mathcal{B}\left(\varepsilon\right)} \tilde{\Phi}^i\left(\mathbf{s}^{|\mathcal{B}\left(\varepsilon\right)|\times 1}\right)}{\partial A_i}>0
	\end{align}
	\noindent Therefore, $e^{-\beta\cdot\sum\limits_{i\in\mathcal{B}\left(\varepsilon\right)} \tilde{\Phi}^i\left(\mathbf{s}^{|\mathcal{B}\left(\varepsilon\right)|\times 1}\right)}$ is non-decreasing in $A_i$.
	
	Then we prove that $e^{-\beta\cdot\sum\limits_{i\in\mathcal{B}\left(\varepsilon\right)} \tilde{\Phi}^i\left(\mathbf{s}^{|\mathcal{B}\left(\varepsilon\right)|\times 1}\right)}$ is non-decreasing in $A^{\mathcal{B}\left(\varepsilon\right)}=\sum\limits_{i\in\mathcal{B}\left(\varepsilon\right)} A_i$. Taking the derivative with respect to $A$ on both sides of \ref{system P Sum} yields \ref{Decom_Exp_A} (Note that $\frac{\partial A_i}{\partial \sum\limits_{i\in\mathcal{B}\left(\varepsilon\right)} A_i}=\frac{1}{\partial\frac{\sum\limits_{i\in\mathcal{B}\left(\varepsilon\right)} A_i}{\partial A_i}}=1$).
	\begin{align}\label{Decom_Exp_A}
		\sum\limits_{i\in\mathcal{B}\left(\varepsilon\right)}L^i &= \frac{\partial\sum\limits_{i\in\mathcal{B}\left(\varepsilon\right)} \tilde{\Phi}^i\left(\mathbf{s}^{|\mathcal{B}\left(\varepsilon\right)|\times 1}\right)}{\partial A^{\mathcal{B}\left(\varepsilon\right)}} +1 - e^{-\beta\cdot\sum\limits_{i\in\mathcal{B}\left(\varepsilon\right)} \tilde{\Phi}^i\left(\mathbf{s}^{|\mathcal{B}\left(\varepsilon\right)|\times 1}\right)} \notag \\ &\quad + A^{\mathcal{B}\left(\varepsilon\right)}\cdot\beta\cdot e^{-\beta\cdot\sum\limits_{i\in\mathcal{B}\left(\varepsilon\right)} \tilde{\Phi}^i\left(\mathbf{s}^{|\mathcal{B}\left(\varepsilon\right)|\times 1}\right)} \cdot \frac{\partial\sum\limits_{i\in\mathcal{B}\left(\varepsilon\right)} \tilde{\Phi}^i\left(\mathbf{s}^{|\mathcal{B}\left(\varepsilon\right)|\times 1}\right)}{\partial A^{\mathcal{B}\left(\varepsilon\right)}} 
	\end{align}
	\noindent By solving it, we obtain \ref{Decom_Exp_A Snd}.
	\begin{align}\label{Decom_Exp_A Snd}
		\frac{\partial\sum\limits_{i\in\mathcal{B}\left(\varepsilon\right)} \tilde{\Phi}^i\left(\mathbf{s}^{|\mathcal{B}\left(\varepsilon\right)|\times 1}\right)}{\partial A^{\mathcal{B}\left(\varepsilon\right)}} &=\frac{\sum\limits_{i\in\mathcal{B}\left(\varepsilon\right)}L^i-1+e^{-\beta\cdot\sum\limits_{i\in\mathcal{B}\left(\varepsilon\right)} \tilde{\Phi}^i\left(\mathbf{s}^{|\mathcal{B}\left(\varepsilon\right)|\times 1}\right)}}{1+A^{\mathcal{B}\left(\varepsilon\right)}\cdot\beta\cdot e^{-\beta\cdot\sum\limits_{i\in\mathcal{B}\left(\varepsilon\right)} \tilde{\Phi}^i\left(\mathbf{s}^{|\mathcal{B}\left(\varepsilon\right)|\times 1}\right)}} \notag  \\   \Longleftrightarrow  \frac{\partial\sum\limits_{i\in\mathcal{B}\left(\varepsilon\right)} \tilde{\Phi}^i\left(\mathbf{s}^{|\mathcal{B}\left(\varepsilon\right)|\times 1}\right)}{\partial A^{\mathcal{B}\left(\varepsilon\right)}} &=\frac{\sum\limits_{i\in\mathcal{B}\left(\varepsilon\right)}\frac{\tilde{\Phi}^i\left(\mathbf{s}^{|\mathcal{B}\left(\varepsilon\right)|\times 1}\right)}{A_i}+\sum\limits_{i\in\mathcal{B}\left(\varepsilon\right)}\left(1-e^{-\beta\cdot\sum\limits_{i\in\mathcal{B}\left(\varepsilon\right)} \tilde{\Phi}^i\left(\mathbf{s}^{|\mathcal{B}\left(\varepsilon\right)|\times 1}\right)}\right)-1+e^{-\beta\cdot\sum\limits_{i\in\mathcal{B}\left(\varepsilon\right)} \tilde{\Phi}^i\left(\mathbf{s}^{|\mathcal{B}\left(\varepsilon\right)|\times 1}\right)}}{1+A^{\mathcal{B}\left(\varepsilon\right)}\cdot\beta\cdot e^{-\beta\cdot\sum\limits_{i\in\mathcal{B}\left(\varepsilon\right)} \tilde{\Phi}^i\left(\mathbf{s}^{|\mathcal{B}\left(\varepsilon\right)|\times 1}\right)}} \notag \\   \Longleftrightarrow  \frac{\partial\sum\limits_{i\in\mathcal{B}\left(\varepsilon\right)} \tilde{\Phi}^i\left(\mathbf{s}^{|\mathcal{B}\left(\varepsilon\right)|\times 1}\right)}{\partial A^{\mathcal{B}\left(\varepsilon\right)}} &=\frac{\sum\limits_{i\in\mathcal{B}\left(\varepsilon\right)}\frac{\tilde{\Phi}^i\left(\mathbf{s}^{|\mathcal{B}\left(\varepsilon\right)|\times 1}\right)}{A_i} + \left(\big|\mathcal{B}\left(\varepsilon\right)\big|-1\right)\cdot \left(1 - e^{-\beta\cdot\sum\limits_{i\in\mathcal{B}\left(\varepsilon\right)} \tilde{\Phi}^i\left(\mathbf{s}^{|\mathcal{B}\left(\varepsilon\right)|\times 1}\right)}\right)}{1+A^{\mathcal{B}\left(\varepsilon\right)}\cdot\beta\cdot e^{-\beta\cdot\sum\limits_{i\in\mathcal{B}\left(\varepsilon\right)} \tilde{\Phi}^i\left(\mathbf{s}^{|\mathcal{B}\left(\varepsilon\right)|\times 1}\right)}}<0
	\end{align}
	\noindent Taking the derivative with respect to $A^{\mathcal{B}\left(\varepsilon\right)}$ on $e^{-\beta\cdot\sum\limits_{i\in\mathcal{B}\left(\varepsilon\right)} \tilde{\Phi}^i\left(\mathbf{s}^{|\mathcal{B}\left(\varepsilon\right)|\times 1}\right)}$ yields \ref{Decom_Exp_A Third}:
	\begin{align}\label{Decom_Exp_A Third}
		\frac{\partial e^{-\beta\cdot\sum\limits_{i\in\mathcal{B}\left(\varepsilon\right)} \tilde{\Phi}^i\left(\mathbf{s}^{|\mathcal{B}\left(\varepsilon\right)|\times 1}\right)}}{\partial A^{\mathcal{B}\left(\varepsilon\right)}}=-\beta\cdot e^{-\beta\cdot\sum\limits_{i\in\mathcal{B}\left(\varepsilon\right)} \tilde{\Phi}^i\left(\mathbf{s}^{|\mathcal{B}\left(\varepsilon\right)|\times 1}\right)}\cdot \frac{\partial\sum\limits_{i\in\mathcal{B}\left(\varepsilon\right)} \tilde{\Phi}^i\left(\mathbf{s}^{|\mathcal{B}\left(\varepsilon\right)|\times 1}\right)}{\partial A^{\mathcal{B}\left(\varepsilon\right)}}>0
	\end{align}
	\noindent Therefore, $e^{-\beta\cdot\sum\limits_{i\in\mathcal{B}\left(\varepsilon\right)} \tilde{\Phi}^i\left(\mathbf{s}^{|\mathcal{B}\left(\varepsilon\right)|\times 1}\right)}$ is non-decreasing in $A^{\mathcal{B}\left(\varepsilon\right)}$.
	
	Then we prove that $e^{-\beta\cdot\sum\limits_{i\in\mathcal{B}\left(\varepsilon\right)} \tilde{\Phi}^i\left(\mathbf{s}^{|\mathcal{B}\left(\varepsilon\right)|\times 1}\right)}$ is non-decreasing in $\beta$. Taking the derivative with respect to $\beta$ on both sides of \ref{system P Sum} yields \ref{Decom_Exp_Beta}.
	\begin{align}\label{Decom_Exp_Beta}
		0 = \frac{\partial\sum\limits_{i\in\mathcal{B}\left(\varepsilon\right)} \tilde{\Phi}^i\left(\mathbf{s}^{|\mathcal{B}\left(\varepsilon\right)|\times 1}\right)}{\partial \beta}  +  A^{\mathcal{B}\left(\varepsilon\right)}\cdot e^{-\beta\cdot\sum\limits_{i\in\mathcal{B}\left(\varepsilon\right)} \tilde{\Phi}^i\left(\mathbf{s}^{|\mathcal{B}\left(\varepsilon\right)|\times 1}\right)} \cdot \left(  \sum\limits_{i\in\mathcal{B}\left(\varepsilon\right)} \tilde{\Phi}^i\left(\mathbf{s}^{|\mathcal{B}\left(\varepsilon\right)|\times 1}\right) +\beta\cdot \frac{\partial\sum\limits_{i\in\mathcal{B}\left(\varepsilon\right)} \tilde{\Phi}^i\left(\mathbf{s}^{|\mathcal{B}\left(\varepsilon\right)|\times 1}\right)}{\partial \beta}  \right)
	\end{align}
	\noindent By solving it, we obtain \ref{Decom_Exp_Beta Snd}
	\begin{align}\label{Decom_Exp_Beta Snd}
		\frac{\partial\sum\limits_{i\in\mathcal{B}\left(\varepsilon\right)} \tilde{\Phi}^i\left(\mathbf{s}^{|\mathcal{B}\left(\varepsilon\right)|\times 1}\right)}{\partial \beta}=-\frac{A^{\mathcal{B}\left(\varepsilon\right)}\cdot e^{-\beta\cdot\sum\limits_{i\in\mathcal{B}\left(\varepsilon\right)} \tilde{\Phi}^i\left(\mathbf{s}^{|\mathcal{B}\left(\varepsilon\right)|\times 1}\right)} \cdot \sum\limits_{i\in\mathcal{B}\left(\varepsilon\right)} \tilde{\Phi}^i\left(\mathbf{s}^{|\mathcal{B}\left(\varepsilon\right)|\times 1}\right)}{1+A^{\mathcal{B}\left(\varepsilon\right)}\cdot \beta \cdot e^{-\beta\cdot\sum\limits_{i\in\mathcal{B}\left(\varepsilon\right)} \tilde{\Phi}^i\left(\mathbf{s}^{|\mathcal{B}\left(\varepsilon\right)|\times 1}\right)}}>0
	\end{align}
	\noindent and \ref{Decom_Exp_Beta Third}.
	\begin{align}\label{Decom_Exp_Beta Third}
		&\quad -e^{-\beta\cdot\sum\limits_{i\in\mathcal{B}\left(\varepsilon\right)} \tilde{\Phi}^i\left(\mathbf{s}^{|\mathcal{B}\left(\varepsilon\right)|\times 1}\right)} \cdot \left(  \sum\limits_{i\in\mathcal{B}\left(\varepsilon\right)} \tilde{\Phi}^i\left(\mathbf{s}^{|\mathcal{B}\left(\varepsilon\right)|\times 1}\right) +\beta\cdot \frac{\partial\sum\limits_{i\in\mathcal{B}\left(\varepsilon\right)} \tilde{\Phi}^i\left(\mathbf{s}^{|\mathcal{B}\left(\varepsilon\right)|\times 1}\right)}{\partial \beta}  \right)    \\ &= \frac{1}{A^{\mathcal{B}\left(\varepsilon\right)}} \cdot \frac{\partial\sum\limits_{i\in\mathcal{B}\left(\varepsilon\right)} \tilde{\Phi}^i\left(\mathbf{s}^{|\mathcal{B}\left(\varepsilon\right)|\times 1}\right)}{\partial \beta}>0 \notag
	\end{align}
	Taking the derivative with respect to $\beta$ on $e^{-\beta\cdot\sum\limits_{i\in\mathcal{B}\left(\varepsilon\right)} \tilde{\Phi}^i\left(\mathbf{s}^{|\mathcal{B}\left(\varepsilon\right)|\times 1}\right)}$ yields \ref{Decom_Exp_Beta Fourth}:
	\begin{align}\label{Decom_Exp_Beta Fourth}
		&\quad\frac{\partial e^{-\beta\cdot\sum\limits_{i\in\mathcal{B}\left(\varepsilon\right)} \tilde{\Phi}^i\left(\mathbf{s}^{|\mathcal{B}\left(\varepsilon\right)|\times 1}\right)}}{\partial \beta} \notag \\  &=-e^{-\beta\cdot\sum\limits_{i\in\mathcal{B}\left(\varepsilon\right)} \tilde{\Phi}^i\left(\mathbf{s}^{|\mathcal{B}\left(\varepsilon\right)|\times 1}\right)} \cdot \left(  \sum\limits_{i\in\mathcal{B}\left(\varepsilon\right)} \tilde{\Phi}^i\left(\mathbf{s}^{|\mathcal{B}\left(\varepsilon\right)|\times 1}\right) +\beta\cdot \frac{\partial\sum\limits_{i\in\mathcal{B}\left(\varepsilon\right)} \tilde{\Phi}^i\left(\mathbf{s}^{|\mathcal{B}\left(\varepsilon\right)|\times 1}\right)}{\partial \beta}  \right)>0 
	\end{align}
	\noindent Therefore, $e^{-\beta\cdot\sum\limits_{i\in\mathcal{B}\left(\varepsilon\right)} \tilde{\Phi}^i\left(\mathbf{s}^{|\mathcal{B}\left(\varepsilon\right)|\times 1}\right)}$ is non-decreasing in $\beta$.
	
	\noindent\textbf{Proof of $3$}: 
	
	\textbf{Proof of $i)$}: $\mathcal{B}(\varepsilon)=\mathcal{P}_{s\,|\,t}\setminus \{j\}$ implies that $x^j\geq 0$ in system $\mathcal{P}_{s\,|\,t}$. By the same argument in Proposition \ref{Propositions about Counterfactual Monopoly Space2}, we obtain that $e^{-\beta\cdot\sum\limits_{i\in\mathcal{B}\left(\varepsilon\right)} \tilde{\Phi}^i\left(\mathbf{s}^{|\mathcal{B}\left(\varepsilon\right)|\times 1}\right)}=\frac{\theta^j}{\bar{\theta}}+ \left(1-\frac{1}{\bar{\theta}}\right)\cdot \varepsilon$ is a feasible equation when $x^j=0$ in system $\mathcal{P}_{s\,|\,t}$. As $\beta$ or $A^{\mathcal{B}\left(\varepsilon\right)}$ increase, $e^{-\beta\cdot\sum\limits_{i\in\mathcal{B}\left(\varepsilon\right)} \tilde{\Phi}^i\left(\mathbf{s}^{|\mathcal{B}\left(\varepsilon\right)|\times 1}\right)}$ increases. Suppose $x^j\leq 0$, we have $\hat{s}^j<s^j$ which contradicts to the equality condition, that is, $\hat{x}^j>0$ as $\beta$ or $A^{\mathcal{B}\left(\varepsilon\right)}$ increase.
	
	\textbf{Proof of $ii)$}: Refer to Proposition \ref{Propositions about Counterfactual Monopoly Space} $a).1$ Threshold Hierarchy, bank $f$ will be the most likely one to leave $\mathcal{B}(\varepsilon)$. Within $\mathcal{B}(\varepsilon)=\mathcal{P}_{s\,|\,t}\setminus \{j\}$, we need to ensure that $f$ does not change sign when $\beta$ or $A^{\mathcal{B}\left(\varepsilon\right)}$ increase, which requires $e^{-\beta\cdot\sum\limits_{i\in\mathcal{B}\left(\varepsilon\right)} \tilde{\Phi}^i\left(\mathbf{s}^{|\mathcal{B}\left(\varepsilon\right)|\times 1}\right)} < \frac{\theta^f}{\bar{\theta}}+\left(1-\frac{1}{\bar{\theta}}\right)\cdot \varepsilon$. That is $\sup\limits_{\beta,\, A_{m\neq f},\,\sum\limits_{m\neq f} A_m} \left[e^{-\beta\cdot\sum\limits_{i\in\mathcal{B}\left(\varepsilon\right)} \tilde{\Phi}^i\left(\mathbf{s}^{|\mathcal{B}\left(\varepsilon\right)|\times 1}\right)}\right]=\min\limits_{k\in\mathcal{B}\left(\varepsilon\right)}\frac{\theta^k}{\bar{\theta}}+\left(1-\frac{1}{\bar{\theta}}\right)\cdot \varepsilon=\frac{\theta^f}{\bar{\theta}}+\left(1-\frac{1}{\bar{\theta}}\right)\cdot \varepsilon$.
	
	\textbf{Proof of $iii)$}: We further transfer the equation system \ref{system P} into \ref{system Exp}.
	\begin{align}\label{system Exp}
		e^{-\beta\cdot\sum\limits_{i\in\mathcal{B}\left(\varepsilon\right)} \tilde{\Phi}^i\left(\mathbf{s}^{|\mathcal{B}\left(\varepsilon\right)|\times 1}\right)}=\frac{\theta^m}{\bar{\theta}}&+\left(1-\frac{1}{\bar{\theta}}\right)\cdot \varepsilon+\frac{\tilde{\Phi}^j\left(\mathbf{s}^{|\mathcal{B}\left(\varepsilon\right)|\times 1}\right)}{A_m}   \notag  \\  &\vdots \quad\quad\quad\quad\quad\quad\quad    \\   e^{-\beta\cdot\sum\limits_{i\in\mathcal{B}\left(\varepsilon\right)} \tilde{\Phi}^i\left(\mathbf{s}^{|\mathcal{B}\left(\varepsilon\right)|\times 1}\right)}=\frac{\theta^f}{\bar{\theta}}&+\left(1-\frac{1}{\bar{\theta}}\right)\cdot \varepsilon+\frac{\tilde{\Phi}^f\left(\mathbf{s}^{|\mathcal{B}\left(\varepsilon\right)|\times 1}\right)}{A_f} \notag 
	\end{align}
	\noindent By the same argument in the proof of \ref{Corresponding to Counterfactual Monopoly Space} $e)$, we obtain that $\sup\limits_{\beta} e^{-\beta\cdot\sum\limits_{i\in\mathcal{B}\left(\varepsilon\right)} \tilde{\Phi}^i\left(\mathbf{s}^{|\mathcal{B}\left(\varepsilon\right)|\times 1}\right)} = \max\limits_{i\in\mathcal{B}\left(\varepsilon\right)} \frac{\theta^i}{\bar{\theta}}+\left(1-\frac{1}{\bar{\theta}}\right)\cdot \varepsilon$. That is $\sup\limits_{\beta} e^{-\beta\cdot\sum\limits_{i\in\mathcal{B}\left(\varepsilon\right)} \tilde{\Phi}^i\left(\mathbf{s}^{|\mathcal{B}\left(\varepsilon\right)|\times 1}\right)}>\frac{\theta^f}{\bar{\theta}}+\left(1-\frac{1}{\bar{\theta}}\right)\cdot \varepsilon$. By Intermediate Value Theorem, there exists a threshold $\beta_{x^f(\varepsilon)<0}$ such that $\tilde{\Phi}^f\left(\mathbf{s}^{|\mathcal{B}\left(\varepsilon\right)|\times 1}\right)<0$ when $\beta<\beta_{x^f(\varepsilon)<0}$. By the same argument in the proof of Proposition \ref{Propositions about Counterfactual Monopoly Space2} $a)$, there exists another threshold $\beta_{x^j(\varepsilon)=0}$ such that $x^j(\varepsilon)>0$ when $\beta>\beta_{x^j(\varepsilon)=0}$.
	
	\textbf{Proof of $iv)$}: We set $A_m\rightarrow +\infty$ in the first equation of \ref{system Exp} and obtain that $\sup\limits_{A_m} e^{-\beta\cdot\sum\limits_{i\in\mathcal{B}\left(\varepsilon\right)} \tilde{\Phi}^i\left(\mathbf{s}^{|\mathcal{B}\left(\varepsilon\right)|\times 1}\right)} = \frac{\theta^m}{\bar{\theta}}+\left(1-\frac{1}{\bar{\theta}}\right)\cdot \varepsilon$. That is $\sup\limits_{A_m} e^{-\beta\cdot\sum\limits_{i\in\mathcal{B}\left(\varepsilon\right)} \tilde{\Phi}^i\left(\mathbf{s}^{|\mathcal{B}\left(\varepsilon\right)|\times 1}\right)}>\frac{\theta^f}{\bar{\theta}}+\left(1-\frac{1}{\bar{\theta}}\right)\cdot \varepsilon$ which contradicts $ii)$. The same conclusion holds for any bank $m\neq j$. Therefore the tuple $\mathbf{A}_{m\neq f}^{|\mathcal{B}\left(\varepsilon\right)-1|\times 1}$ or $\sum\limits_{m\neq f} A_m$ is bounded.
	
	\noindent\textbf{Proof of $4$}: 
	
	We set $A_f\rightarrow +\infty$ in the first equation of \ref{system Exp} and obtain that $\sup\limits_{A_f} e^{-\beta\cdot\sum\limits_{i\in\mathcal{B}\left(\varepsilon\right)} \tilde{\Phi}^i\left(\mathbf{s}^{|\mathcal{B}\left(\varepsilon\right)|\times 1}\right)} = \frac{\theta^f}{\bar{\theta}}+\left(1-\frac{1}{\bar{\theta}}\right)\cdot \varepsilon$. That is $e^{-\beta\cdot\sum\limits_{i\in\mathcal{B}\left(\varepsilon\right)} \tilde{\Phi}^i\left(\mathbf{s}^{|\mathcal{B}\left(\varepsilon\right)|\times 1}\right)}<\frac{\theta^f}{\bar{\theta}}+\left(1-\frac{1}{\bar{\theta}}\right)\cdot \varepsilon$ which satisfies the ignition condition $3.i)$. Meanwhile, $e^{-\beta\cdot\sum\limits_{i\in\mathcal{B}\left(\varepsilon\right)} \tilde{\Phi}^i\left(\mathbf{s}^{|\mathcal{B}\left(\varepsilon\right)|\times 1}\right)}<\frac{\theta^m}{\bar{\theta}}+\left(1-\frac{1}{\bar{\theta}}\right)\cdot \varepsilon,\,\forall m\neq f$. Therefore $\mathcal{B}(\varepsilon)=\mathcal{P}_{s\,|\,t}\setminus \{j\}$ holds under $A_f\rightarrow +\infty$.
	
	\noindent\textbf{Proof of $5$}: 
	
	By multiplying $A_m$ on both sides of $m$th equation and adding the $\big|\mathcal{B}(\varepsilon)\big|$ equations together, we obtain equation \ref{system EXP Sum} where $A^{\mathcal{B}(\varepsilon)}=\sum\limits_{m\in \mathcal{B}(\varepsilon)}A_m$ and $E^{\mathcal{B}(\varepsilon)}=\sum\limits_{m\in \mathcal{B}(\varepsilon)}E_m$. 
	\begin{align}\label{system EXP Sum}
		A^{\mathcal{B}(\varepsilon)}\cdot e^{-\beta\cdot\sum\limits_{i\in\mathcal{B}\left(\varepsilon\right)} \tilde{\Phi}^i\left(\mathbf{s}^{|\mathcal{B}\left(\varepsilon\right)|\times 1}\right)} -\sum\limits_{i\in\mathcal{B}\left(\varepsilon\right)} \tilde{\Phi}^i\left(\mathbf{s}^{|\mathcal{B}\left(\varepsilon\right)|\times 1}\right) = \frac{E^{\mathcal{B}(\varepsilon)}}{\bar{\theta}}  + \left(1-\frac{1}{\bar{\theta}}\right)\cdot \varepsilon \cdot A^{\mathcal{B}(\varepsilon)}
	\end{align}
	\noindent The ignition condition $3.i)$ and threshold $\beta_{x^j(\varepsilon)=0}$ implies \ref{Equation threshold Beta}.
	\begin{align}\label{Equation threshold Beta}
		e^{-\beta_{x^j(\varepsilon)=0}\cdot\sum\limits_{i\in\mathcal{B}\left(\varepsilon\right)} \tilde{\Phi}^i\left(\mathbf{s}^{|\mathcal{B}\left(\varepsilon\right)|\times 1}\right)} = \frac{\theta^j}{\bar{\theta}}+\left(1-\frac{1}{\bar{\theta}}\right)\cdot \varepsilon
	\end{align}
	\noindent We Substitute \ref{Equation threshold Beta} into \ref{system EXP Sum}, we obtain \ref{Equation threshold Beta2}.
	\begin{align}\label{Equation threshold Beta2}
		A^{\mathcal{B}(\varepsilon)}\cdot \left[\frac{\theta^j}{\bar{\theta}}+\left(1-\frac{1}{\bar{\theta}}\right)\cdot \varepsilon \right] -\sum\limits_{i\in\mathcal{B}\left(\varepsilon\right)} \tilde{\Phi}^i\left(\mathbf{s}^{|\mathcal{B}\left(\varepsilon\right)|\times 1}\right) = \frac{E^{\mathcal{B}(\varepsilon)}}{\bar{\theta}}  + \left(1-\frac{1}{\bar{\theta}}\right)\cdot \varepsilon \cdot A^{\mathcal{B}(\varepsilon)}
	\end{align}
	\noindent Taking the derivative with respect to $\varepsilon$ on both sides of \ref{Equation threshold Beta2} yields \ref{Decom_Sum X_ThresholdBeta}.
	\begin{align}\label{Decom_Sum X_ThresholdBeta}
		&\left(1-\frac{1}{\bar{\theta}}\right) \cdot A^{\mathcal{B}(\varepsilon)} - \frac{\partial \sum\limits_{i\in\mathcal{B}\left(\varepsilon\right)} \tilde{\Phi}^i\left(\mathbf{s}^{|\mathcal{B}\left(\varepsilon\right)|\times 1}\right)}{ \partial \varepsilon} = \left(1-\frac{1}{\bar{\theta}}\right) \cdot A^{\mathcal{B}(\varepsilon)}  \notag \\  &\quad\quad\Longleftrightarrow    \frac{\partial \sum\limits_{i\in\mathcal{B}\left(\varepsilon\right)} \tilde{\Phi}^i\left(\mathbf{s}^{|\mathcal{B}\left(\varepsilon\right)|\times 1}\right)}{ \partial \varepsilon}=0 
	\end{align}
	\noindent Taking the derivative with respect to $\varepsilon$ on both sides of \ref{system EXP Sum} yields \ref{Decom_Beta_vs_Epsilon}.
	\begin{align}\label{Decom_Beta_vs_Epsilon}
		&\left(1-\frac{1}{\bar{\theta}}\right) \cdot A^{\mathcal{B}(\varepsilon)} = -\frac{\partial \sum\limits_{i\in\mathcal{B}\left(\varepsilon\right)} \tilde{\Phi}^i\left(\mathbf{s}^{|\mathcal{B}\left(\varepsilon\right)|\times 1}\right)}{ \partial \varepsilon} \\  &- A^{\mathcal{B}(\varepsilon)} \cdot e^{-\beta_{x^j(\varepsilon)=0} \cdot \sum\limits_{i\in\mathcal{B}\left(\varepsilon\right)} \tilde{\Phi}^i\left(\mathbf{s}^{|\mathcal{B}\left(\varepsilon\right)|\times 1}\right)}   \cdot \left( \frac{\partial \beta_{x^j(\varepsilon)=0}}{ \partial \varepsilon} \cdot \sum\limits_{i\in\mathcal{B}\left(\varepsilon\right)} \tilde{\Phi}^i\left(\mathbf{s}^{|\mathcal{B}\left(\varepsilon\right)|\times 1}\right)  + \frac{\partial \sum\limits_{i\in\mathcal{B}\left(\varepsilon\right)} \tilde{\Phi}^i\left(\mathbf{s}^{|\mathcal{B}\left(\varepsilon\right)|\times 1}\right)}{ \partial \varepsilon}\cdot \beta_{x^j(\varepsilon)=0}  \right)  \notag
	\end{align}
	\noindent We Substitute \ref{Decom_Sum X_ThresholdBeta} into \ref{Decom_Beta_vs_Epsilon}, we obtain \ref{Decom_Beta_vs_Epsilon2}.
	\begin{align}\label{Decom_Beta_vs_Epsilon2}
		\frac{\partial \beta_{x^j(\varepsilon)=0}}{ \partial \varepsilon} = \frac{\left(1-\bar{\theta}\right)\cdot \sum\limits_{i\in\mathcal{B}\left(\varepsilon\right)} \tilde{\Phi}^i\left(\mathbf{s}^{|\mathcal{B}\left(\varepsilon\right)|\times 1}\right)}{\bar{\theta}\cdot e^{-\beta_{x^j(\varepsilon)=0}\cdot\sum\limits_{i\in\mathcal{B}\left(\varepsilon\right)} \tilde{\Phi}^i\left(\mathbf{s}^{|\mathcal{B}\left(\varepsilon\right)|\times 1}\right)}}<0
	\end{align}
	\noindent Therefore, threshold $\beta_{x^j(\varepsilon)=0}$ is non-increasing in $\varepsilon$ and $\max \beta_{x^j(\varepsilon)=0}=\beta_{x^j(\varepsilon=0)=0}$.
	
	Since $e^{-\beta_{x^j(\varepsilon)=0}\cdot\sum\limits_{i\in\mathcal{B}\left(\varepsilon\right)} \tilde{\Phi}^i\left(\mathbf{s}^{|\mathcal{B}\left(\varepsilon\right)|\times 1}\right)}$ is non-decreasing in $\beta_{x^j(\varepsilon)=0}$ according to $2$. Once $\theta^j$ increases, we need to balance the equation \ref{Equation threshold Beta}, which implies that we need to increase $\beta_{x^j(\varepsilon)=0}$. Therefore $\beta_{x^j(\varepsilon)=0}$ is non-decreasing in $\theta^j$.
	
	\noindent\textbf{Proof of $6$}: 
	
	\textbf{Proof of $i)$}: $\max \beta_{x^j(\varepsilon)=0}<\beta$ implies that $\beta_{x^j(\varepsilon)=0}<\beta_{x^j(\varepsilon=0)=0}<\beta$. That is $x^j(\varepsilon)>0$. Refer to equation OA3.7, we obtain that $\mathbf{x}$ is non-decreasing in $\varepsilon$. That is $\tilde{\Phi}^i\left(\mathbf{s}^{|\mathcal{B}\left(\varepsilon\right)|\times 1}\right)<0,\forall i \in \mathcal{B}\left(\varepsilon\right)$ holds when $\varepsilon=0$. Moreover, $\beta_{x^j(\varepsilon=0)=0}<\beta$ implies that $x^j(\varepsilon=0)>0$. Therefore, the signs of $\tilde{\Phi}^i\left(\mathbf{s}^{|\mathcal{B}\left(\varepsilon\right)|\times 1}\right)$ and $x^j$ remain unchanged if $\varepsilon$ turns back to $\varepsilon=0$.
	
	\textbf{Proof of $ii)$}: We set $A_f\rightarrow +\infty$ in the last equation of \ref{system Exp} and obtain that $\tilde{\Phi}^f\left(\mathbf{s}^{|\mathcal{B}\left(\varepsilon\right)|\times 1}\right)$ can be arbitrarily larger than any other $\tilde{\Phi}^m\left(\mathbf{s}^{|\mathcal{B}\left(\varepsilon\right)|\times 1}\right),\,\forall m\in \mathcal{B}\left(\varepsilon\right),\,m\neq f$ while maintaining Hierarchy according to $4$ (i.e., $\sup\limits_{A_f}e^{-\beta\cdot\sum\limits_{i\in\mathcal{B}\left(\varepsilon\right)} \tilde{\Phi}^i\left(\mathbf{s}^{|\mathcal{B}\left(\varepsilon\right)|\times 1}\right)}=\frac{\theta^f}{\bar{\theta}}+\left(1-\frac{1}{\bar{\theta}}\right)\cdot \varepsilon<\frac{\theta^m}{\bar{\theta}}+\left(1-\frac{1}{\bar{\theta}}\right)\cdot \varepsilon,\,\forall m\neq f$).
	
	\textbf{Proof of $iii)$}: We firstly prove that such a threshold $\varepsilon_{\max\limits_{k\in\mathcal{B}(\varepsilon)} x^k =0 }$ always exists as $\varepsilon$ grows. Suppose not, choose $\varepsilon \geq \max\limits_{i\in \mathcal{B}(\varepsilon)} \frac{\theta^i-\bar{\theta}}{1-\bar{\theta}}$, that is, $s^i\geq0,\,\forall i\in \mathcal{B}(\varepsilon)$. if $\tilde{\Phi}^i\left(\mathbf{s}^{|\mathcal{B}\left(\varepsilon\right)|\times 1}\right)<0,\,\forall i\in \mathcal{B}(\varepsilon)$ still holds, we have $\hat{s}^i<0,\,\forall i\in \mathcal{B}(\varepsilon)$ which contradicts to $s^i\geq0,\,\forall i\in \mathcal{B}(\varepsilon)$. Therefore, at least one $x^i\geq 0,\, i\in \mathcal{B}(\varepsilon)$. By Intermediate Value Theorem, threshold $\varepsilon_{\max\limits_{k\in\mathcal{B}(\varepsilon)} x^k =0 }$ always exists as $\varepsilon$ grows. Then we prove that $\varepsilon_{x^f(\varepsilon)=0}=\varepsilon_{\max\limits_{k\in\mathcal{B}(\varepsilon)} x^k =0 }$. According to Lemma \ref{Perfection Switch=1} $1.b)$, the state $x^m\geq0,\, m\in \mathcal{B}(\varepsilon),\,m\neq f$ while $x^f\leq0$ cannot hold. According to Lemma \ref{Perfection Switch=1} $1.a)$, $x^m<0,\, m\in \mathcal{B}(\varepsilon),\,m\neq f$ while $x^f\leq0$. Therefore bank $f$ will be the first bank to leave $\mathcal{B}(\varepsilon)=\mathcal{P}_{s\,|\,t}\setminus \{j\}$, that is $\varepsilon_{x^f(\varepsilon)=0}=\varepsilon_{\max\limits_{k\in\mathcal{B}(\varepsilon)} x^k =0 }$.
	
	\textbf{Proof of $iv)$}: According to \ref{Perfection Switch=1} $1.a)$, $x^m<0,\, m\in \mathcal{B}(\varepsilon),\,m\neq f$ while $x^f=0$. We obtain that $x^f=0$ when combine the system $\mathcal{B}(\varepsilon)\setminus \{f\}$ with $\{f\}$. We obtain that $x^j>0$ when combine the system $\mathcal{B}(\varepsilon)\setminus \{f\}$ with $\{j\}$ (i.e.,  $e^{-\beta\cdot\sum\limits_{i\in\mathcal{B}(\varepsilon)\setminus \{f\}} \tilde{\Phi}^i\left(\mathbf{s}^{|\mathcal{B}(\varepsilon)\setminus \{f\}|\times 1}\right)}>\frac{\theta^f}{\bar{\theta}}+\left(1-\frac{1}{\bar{\theta}}\right)\cdot \varepsilon>\frac{\theta^j}{\bar{\theta}}+\left(1-\frac{1}{\bar{\theta}}\right)\cdot \varepsilon$). Combining $\mathcal{B}(\varepsilon)\setminus \{f\}$ and $\{f\}\cup \{j\}$ leads to $x^j>0$ and $x^f>0$. The first one $x^j>0$ can be demonstrated by the same argument in $6.i)$. The second one $x^f>0$ can be demonstrated by the same argument for Crowding-out Effect. Therefore, the system $\mathcal{B}(\varepsilon)\setminus \{f\}$ meets the requirements of Maximal Bail-out Cluster, that is, $\mathcal{B}(\varepsilon)\setminus \{f\}$ is a Maximal Bailout Cluster when $\varepsilon=\varepsilon_{\max\limits_{k\in\mathcal{B}(\varepsilon)} x^k =0 }$.
	
	A more intuitive proof is that $e^{-\beta\cdot\sum\limits_{i\in\mathcal{B}\left(\varepsilon\right)} \tilde{\Phi}^i\left(\mathbf{s}^{|\mathcal{B}\left(\varepsilon\right)|\times 1}\right)}=\frac{\theta^f}{\bar{\theta}}+\left(1-\frac{1}{\bar{\theta}}\right)\cdot \varepsilon$ when $\varepsilon=\varepsilon_{\max\limits_{k\in\mathcal{B}(\varepsilon)} x^k =0 }$. That is $e^{-\beta\cdot\sum\limits_{i\in\mathcal{B}\left(\varepsilon\right)} \tilde{\Phi}^i\left(\mathbf{s}^{|\mathcal{B}\left(\varepsilon\right)|\times 1}\right)}<\frac{\theta^m}{\bar{\theta}}+\left(1-\frac{1}{\bar{\theta}}\right)\cdot \varepsilon,\,\forall m\neq f$, and the $\tilde{\Phi}^m\left(\mathbf{s}^{|\mathcal{B}\left(\varepsilon\right)|\times 1}\right)<0,\,m\in\mathcal{B}(\varepsilon),\,\forall m\neq f$ even though we drop $\{f\}$ whose $\tilde{\Phi}^f\left(\mathbf{s}^{|\mathcal{B}\left(\varepsilon\right)|\times 1}\right)=0$. 
\end{proof}

\subsection{Proof of Lemma \ref{Lemma For Commutative Diagram} and \ref{Lemma ForPartitionInduced Equilibrium Transition Snd}: preliminary 2}\label{preliminaries:part2}
\begin{restatable}{thm3}{LemmaForPartitionInducedEquilibriumTransition}\label{Lemma ForPartitionInduced Equilibrium Transition}
	Under $\varepsilon< \min\limits_{i\in \mathcal{P}_{s\,|\,t}} \frac{\theta^i-\bar{\theta}}{1-\bar{\theta}}$ with $\beta$ satisfies $e^{-\beta\cdot \tilde{x}^{^w\bigstar\left(\mathcal{P}_{s\,|\,t},\varepsilon\right)}}< \frac{\theta^{^{w+1}\bigstar\left(\mathcal{P}_{s\,|\,t},\varepsilon\right)}}{\bar{\theta}}+\left(1-\frac{1}{\bar{\theta}}\right)\cdot \varepsilon,\,\forall w\in \bigstar\left(\mathcal{P}_{s\,|\,t},\varepsilon\right)$ the system $\mathcal{P}_{s\,|\,t}$ exhibits:
	
	1. We can control each of the sets $\mathcal{B}^{\bigstar\left(\mathcal{P}_{s\,|\,t},\varepsilon\right)},\,\mathcal{B}^{\hollowstar\left(\mathcal{P}_{s\,|\,t},\varepsilon\right)},\, \mathcal{B}^{\hollowstar^{2}\left(\mathcal{P}_{s\,|\,t},\varepsilon\right)},\cdots,\,\mathcal{B}^{\hollowstar^{\big|\bigstar\left(\mathcal{P}_{s\,|\,t},\varepsilon\right)\big|-1}\left(\mathcal{P}_{s\,|\,t},\varepsilon\right)}$ contains only one bank by regulating the individual $A$ and $\theta$. 
	
	2. The bijective correspondence $\mathcal{X}_{A\rightarrow \varepsilon}: A_{^w\bigstar\left(\mathcal{P}_{s\,|\,t},\varepsilon\right)} \rightarrow \varepsilon_{x^{^{w+1}\bigstar\left(\mathcal{P}_{s\,|\,t},\varepsilon\right)}=0}$ $\left(R\rightarrow R\right)$ holds, where $\varepsilon_{x^{^{w+1}\bigstar\left(\mathcal{P}_{s\,|\,t},\varepsilon\right)}=0}$ satisfies $e^{-\beta\cdot \tilde{x}^{^w\bigstar\left(\mathcal{P}_{s\,|\,t},\,\varepsilon_{x^{^{w+1}\bigstar\left(\mathcal{P}_{s\,|\,t},\varepsilon\right)}=0}\right)}}= \frac{\theta^{^{w+1}\bigstar\left(\mathcal{P}_{s\,|\,t},\,\varepsilon_{x^{^{w+1}\bigstar\left(\mathcal{P}_{s\,|\,t},\varepsilon\right)}=0}\right)}}{\bar{\theta}}+\left(1-\frac{1}{\bar{\theta}}\right)\cdot \varepsilon,\,\forall w\in \bigstar\left(\mathcal{P}_{s\,|\,t},\varepsilon\right)$.
	
	3. The bijective correspondence $\mathcal{X}_{\varepsilon\rightarrow A}:\varepsilon \rightarrow A^{^w\bigstar\left(\mathcal{P}_{s\,|\,t},\varepsilon\right)}$ $\left(R\rightarrow R\right)$ is guaranteed to hold, where $A^{^w\bigstar\left(\mathcal{P}_{s\,|\,t},\varepsilon\right)}$ satisfies $e^{-\beta\cdot \tilde{x}^{^w\bigstar\left(\mathcal{P}_{s\,|\,t},\varepsilon\right)}}= \frac{\theta^{^{w+1}\bigstar\left(\mathcal{P}_{s\,|\,t},\varepsilon\right)}}{\bar{\theta}}+\left(1-\frac{1}{\bar{\theta}}\right)\cdot \varepsilon,\,\forall w\in \bigstar\left(\mathcal{P}_{s\,|\,t},\varepsilon\right)$.
	
	4. The bijective correspondence $\mathcal{Y}_{\theta\rightarrow \varepsilon}:\theta^{^{w+1}\bigstar\left(\mathcal{P}_{s\,|\,t},\varepsilon\right)}\rightarrow \varepsilon_{x^{^{w+1}\bigstar\left(\mathcal{P}_{s\,|\,t},\varepsilon\right)}=0}$ $\left(R\rightarrow R\right)$ holds, where $\varepsilon_{x^{^{w+1}\bigstar\left(\mathcal{P}_{s\,|\,t},\varepsilon\right)}=0}$ satisfies $e^{-\beta\cdot \tilde{x}^{^w\bigstar\left(\mathcal{P}_{s\,|\,t},\,\varepsilon_{x^{^{w+1}\bigstar\left(\mathcal{P}_{s\,|\,t},\varepsilon\right)}=0}\right)}}= \frac{\theta^{^{w+1}\bigstar\left(\mathcal{P}_{s\,|\,t},\,\varepsilon_{x^{^{w+1}\bigstar\left(\mathcal{P}_{s\,|\,t},\varepsilon\right)}=0}\right)}}{\bar{\theta}}+\left(1-\frac{1}{\bar{\theta}}\right)\cdot \varepsilon,\,\forall w\in \bigstar\left(\mathcal{P}_{s\,|\,t},\varepsilon\right)$.
	
	5. The bijective correspondence $\mathcal{Y}_{\varepsilon\rightarrow \theta}:\varepsilon \rightarrow \theta_{^{w+1}\bigstar\left(\mathcal{P}_{s\,|\,t},\varepsilon\right)}$ $\left(R\rightarrow R\right)$ is guaranteed to hold, where $\theta_{^{w+1}\bigstar\left(\mathcal{P}_{s\,|\,t},\varepsilon\right)}$ satisfies $e^{-\beta\cdot \tilde{x}^{^w\bigstar\left(\mathcal{P}_{s\,|\,t},\varepsilon\right)}}= \frac{\theta_{^{w+1}\bigstar\left(\mathcal{P}_{s\,|\,t},\varepsilon\right)}}{\bar{\theta}}+\left(1-\frac{1}{\bar{\theta}}\right)\cdot \varepsilon,\,\forall w\in \bigstar\left(\mathcal{P}_{s\,|\,t},\varepsilon\right)$.
\end{restatable}

\begin{proof}[\normalfont\bfseries Proof of Lemma \ref{Lemma ForPartitionInduced Equilibrium Transition}]
	\,
	
	\noindent\textbf{Proof of $1$}: It's equivalent to prove that $\mathcal{B}^{\bigstar\left(\mathcal{P}_{s\,|\,t},\varepsilon\right)}={^0}\bigstar\left(\mathcal{P}_{s\,|\,t},\varepsilon\right),\,\mathcal{B}^{\hollowstar\left(\mathcal{P}_{s\,|\,t},\varepsilon\right)}={^1}\bigstar\left(\mathcal{P}_{s\,|\,t},\varepsilon\right),\,\cdots$ can be controlled by setting appropriate individual $A$ and $\theta$. Since all $\theta$ haved been controlled by heterogeneity, we will focus solely on $A$ in the subsequent analysis.
	
	In the system $\bigstar\left(\mathcal{P}_{s\,|\,t},\varepsilon\right)$, we have \ref{Chain-0 MaximalBailout} refer to Lemma \ref{Lemma ForDecomposition and Compression Equivalence} $3.i)$ and $3.ii)$.
	\begin{align}\label{Chain-0 MaximalBailout}
		\sup\limits_{A_{{^0}\bigstar\left(\mathcal{P}_{s\,|\,t},\varepsilon\right)}} \left[e^{-\beta\cdot \tilde{\Phi}^{{^0}\bigstar\left(\mathcal{P}_{s\,|\,t},\varepsilon\right)}}\right]=\frac{\theta^{{^0}\bigstar\left(\mathcal{P}_{s\,|\,t},\varepsilon\right)}}{\bar{\theta}}+\left(1-\frac{1}{\bar{\theta}}\right)\cdot \varepsilon\gg\max\limits_{k\in \hollowstar\left(\mathcal{P}_{s\,|\,t},\varepsilon\right)}\frac{\theta^k}{\bar{\theta}}+\left(1-\frac{1}{\bar{\theta}}\right)\cdot \varepsilon
	\end{align}
	\noindent Furthermore, we can obtain similar expression \ref{Chain-1 MaximalBailout} in the system $\hollowstar\left(\mathcal{P}_{s\,|\,t},\varepsilon\right)$.
	\begin{align}\label{Chain-1 MaximalBailout}
		\sup\limits_{A_{{^0}\hollowstar\left(\mathcal{P}_{s\,|\,t},\varepsilon\right)}} \left[e^{-\beta\cdot \tilde{\Phi}^{{^0}\hollowstar\left(\mathcal{P}_{s\,|\,t},\varepsilon\right)}}\right]=\frac{\theta^{{^0}\hollowstar\left(\mathcal{P}_{s\,|\,t},\varepsilon\right)}}{\bar{\theta}}+\left(1-\frac{1}{\bar{\theta}}\right)\cdot \varepsilon\gg\max\limits_{k\in \hollowstar^{1}\left(\mathcal{P}_{s\,|\,t},\varepsilon\right)}\frac{\theta^k}{\bar{\theta}}+\left(1-\frac{1}{\bar{\theta}}\right)\cdot \varepsilon
	\end{align}
	\noindent These expressions can be recursively formulated until the last system $\hollowstar^{\big|\bigstar\left(\mathcal{P}_{s\,|\,t},\varepsilon\right)\big|-1}\left(\mathcal{P}_{s\,|\,t},\varepsilon\right)$ which contains only ${^{\big|\bigstar\left(\mathcal{P}_{s\,|\,t},\varepsilon\right)\big|-1}}\bigstar\left(\mathcal{P}_{s\,|\,t},\varepsilon\right)$.
	
	Since each inequality we mentioned above is only determined by the individual asset and leverage ratio of the first bank of the chain. Therefore we can ensure that inequalities \ref{Chain-0 MaximalBailout}, \ref{Chain-1 MaximalBailout}, $\cdots$ all hold by precisely controlling $A_{{^0}\bigstar\left(\mathcal{P}_{s\,|\,t},\varepsilon\right)},\,A_{{^1}\bigstar\left(\mathcal{P}_{s\,|\,t},\varepsilon\right)},\,\cdots,\,A_{{^{\big|\bigstar\left(\mathcal{P}_{s\,|\,t},\varepsilon\right)\big|-1}}\bigstar\left(\mathcal{P}_{s\,|\,t},\varepsilon\right)}$. That is, each of the sets $\mathcal{B}^{\bigstar\left(\mathcal{P}_{s\,|\,t},\varepsilon\right)},\,\mathcal{B}^{\hollowstar\left(\mathcal{P}_{s\,|\,t},\varepsilon\right)},\, \mathcal{B}^{\hollowstar^{2}\left(\mathcal{P}_{s\,|\,t},\varepsilon\right)},\cdots,\,\mathcal{B}^{\hollowstar^{\big|\bigstar\left(\mathcal{P}_{s\,|\,t},\varepsilon\right)\big|-1}\left(\mathcal{P}_{s\,|\,t},\varepsilon\right)}$ contains only one bank. 
	
	\noindent\textbf{Proof of $2$}: Refer to Compression Equivalence Theorem \ref{Decomposition and Compression Equivalence}, we obtain that $\varepsilon_{x^{^{w+1}\bigstar\left(\mathcal{P}_{s\,|\,t},\varepsilon\right)}=0}$ is equivalent to $\varepsilon_{\max\limits_{f\in _{n+1}\mathcal{B}\left(\varepsilon\right)} x^f=0}$. Therefore, $\varepsilon_{x^{^{w+1}\bigstar\left(\mathcal{P}_{s\,|\,t},\varepsilon\right)}=0}$ exists and is unique for the corresponding $A_{^w\bigstar\left(\mathcal{P}_{s\,|\,t},\varepsilon\right)}$.
	
	\noindent\textbf{Proof of $3$}: Refer to Lemma \ref{Lemma ForPartitionInduced Equilibrium Transition} $1$, there exist the unique minimum elements in \ref{Unique minimum A}
	\begin{align}\label{Unique minimum A}
		\bigg\{A_{{^0}\bigstar\left(\mathcal{P}_{s\,|\,t},\varepsilon\right)},\,A_{{^1}\bigstar\left(\mathcal{P}_{s\,|\,t},\varepsilon\right)},\,\cdots,\,A_{{^{\big|\bigstar\left(\mathcal{P}_{s\,|\,t},\varepsilon\right)\big|-1}}\bigstar\left(\mathcal{P}_{s\,|\,t},\varepsilon\right)}\bigg\}
	\end{align}
	which precisely satisfy equations \ref{Correspondence_Varespilon to A}.
	\begin{align}\label{Correspondence_Varespilon to A}
		e^{-\beta\cdot \tilde{\Phi}^{{^0}\bigstar\left(\mathcal{P}_{s\,|\,t},\varepsilon\right)}}&=\max\limits_{k\in \hollowstar\left(\mathcal{P}_{s\,|\,t},\varepsilon\right)}\frac{\theta^k}{\bar{\theta}}+\left(1-\frac{1}{\bar{\theta}}\right)\cdot \varepsilon \notag  \\ e^{-\beta\cdot \tilde{\Phi}^{{^0}\hollowstar\left(\mathcal{P}_{s\,|\,t},\varepsilon\right)}}&=\max\limits_{k\in \hollowstar^{1}\left(\mathcal{P}_{s\,|\,t},\varepsilon\right)}\frac{\theta^k}{\bar{\theta}}+\left(1-\frac{1}{\bar{\theta}}\right)\cdot \varepsilon  \\   &\vdots \notag
	\end{align}
	\noindent Therefore, the bijective correspondence $\mathcal{X}_{\varepsilon\rightarrow A}:\varepsilon \rightarrow A^{^w\bigstar\left(\mathcal{P}_{s\,|\,t},\varepsilon\right)}$ $\left(R\rightarrow R\right)$ is guaranteed to hold.
	
	\noindent\textbf{Proof of $4$}: $\varepsilon_{x^{^{w+1}\bigstar\left(\mathcal{P}_{s\,|\,t},\varepsilon\right)}=0}$ is the same in Lemma \ref{Lemma ForPartitionInduced Equilibrium Transition} $2$, since the system in Lemma \ref{Lemma ForPartitionInduced Equilibrium Transition} $2$ share the same parameters in Lemma \ref{Lemma ForPartitionInduced Equilibrium Transition} $4$(i.e., they have the same $A_{^w\bigstar\left(\mathcal{P}_{s\,|\,t},\varepsilon\right)}$ and $\theta^{^{w}\bigstar\left(\mathcal{P}_{s\,|\,t},\varepsilon\right)},\,\forall w\in \bigstar\left(\mathcal{P}_{s\,|\,t},\varepsilon\right)$).
	
	The bijective correspondence $\mathcal{X}_{A\rightarrow \varepsilon}: A_{^w\bigstar\left(\mathcal{P}_{s\,|\,t},\varepsilon\right)} \rightarrow \varepsilon_{x^{^{w+1}\bigstar\left(\mathcal{P}_{s\,|\,t},\varepsilon\right)}=0}$ $\left(R\rightarrow R\right)$ is actually implies the bijective set to set correspondence $\mathcal{X}_{A\rightarrow \varepsilon}: \left(A_{^w\bigstar\left(\mathcal{P}_{s\,|\,t},\varepsilon\right)}, \theta^{^{w+1}\bigstar\left(\mathcal{P}_{s\,|\,t},\varepsilon\right)} \right) \rightarrow \varepsilon_{x^{^{w+1}\bigstar\left(\mathcal{P}_{s\,|\,t},\varepsilon\right)}=0}$$\left(\mathbb{R}^2\rightarrow R\right)$. Therefore, the bijective correspondence $\mathcal{Y}_{\theta\rightarrow \varepsilon}:\theta^{^{w+1}\bigstar\left(\mathcal{P}_{s\,|\,t},\varepsilon\right)}\rightarrow \varepsilon_{x^{^{w+1}\bigstar\left(\mathcal{P}_{s\,|\,t},\varepsilon\right)}=0}$ $\left(R\rightarrow R\right)$ is the dual expression of $\mathcal{X}_{A\rightarrow \varepsilon}$.
	
	\noindent\textbf{Proof of $5$}: We can gradually decrease $\theta^{^{w+1}\bigstar\left(\mathcal{P}_{s\,|\,t},\varepsilon\right)}$ in $e^{-\beta\cdot \tilde{x}^{^w\bigstar\left(\mathcal{P}_{s\,|\,t},\varepsilon\right)}}< \frac{\theta^{^{w+1}\bigstar\left(\mathcal{P}_{s\,|\,t},\varepsilon\right)}}{\bar{\theta}}+\left(1-\frac{1}{\bar{\theta}}\right)\cdot \varepsilon,\,\forall w\in \bigstar\left(\mathcal{P}_{s\,|\,t},\varepsilon\right)$ to satisfy \ref{Correspondence_Varespilon to Theta1}.
	\begin{align}\label{Correspondence_Varespilon to Theta1}
		e^{-\beta\cdot \tilde{x}^{^w\bigstar\left(\mathcal{P}_{s\,|\,t},\varepsilon\right)}}= \frac{\theta^{^{w+1}\bigstar\left(\mathcal{P}_{s\,|\,t},\varepsilon\right)}}{\bar{\theta}}+\left(1-\frac{1}{\bar{\theta}}\right)\cdot \varepsilon,\,\forall w\in \bigstar\left(\mathcal{P}_{s\,|\,t},\varepsilon\right)
	\end{align}
	Definitely it can be achieved. We firstly transfer \ref{Correspondence_Varespilon to Theta1} into \ref{Correspondence_Varespilon to Theta2}.
	\begin{align}\label{Correspondence_Varespilon to Theta2}
		\theta^{^{w+1}\bigstar\left(\mathcal{P}_{s\,|\,t},\varepsilon\right)}= \bigg[e^{-\beta\cdot \tilde{x}^{^w\bigstar\left(\mathcal{P}_{s\,|\,t},\varepsilon\right)}}-\left(1-\frac{1}{\bar{\theta}}\right)\cdot \varepsilon\bigg]\cdot \bar{\theta},\,\forall w\in \bigstar\left(\mathcal{P}_{s\,|\,t},\varepsilon\right)
	\end{align}
	Then we only need to check whether $\theta^{^{w+1}\bigstar\left(\mathcal{P}_{s\,|\,t},\varepsilon\right)}$ satisfies $\varepsilon< \min\limits_{i\in \mathcal{P}_{s\,|\,t}} \frac{\theta^i-\bar{\theta}}{1-\bar{\theta}}$ or not. That is \ref{Correspondence_Varespilon to Theta3}. 
	\begin{align}\label{Correspondence_Varespilon to Theta3}
		\frac{\theta^{^{w+1}\bigstar\left(\mathcal{P}_{s\,|\,t},\varepsilon\right)}-\bar{\theta}}{1-\bar{\theta}}&=\frac{\bigg[e^{-\beta\cdot \tilde{x}^{^w\bigstar\left(\mathcal{P}_{s\,|\,t},\varepsilon\right)}}-\left(1-\frac{1}{\bar{\theta}}\right)\cdot \varepsilon\bigg]\cdot \bar{\theta}-\bar{\theta}}{1-\bar{\theta}}\notag \\ &> \frac{\bigg[e^{-\beta\cdot \tilde{x}^{^w\bigstar\left(\mathcal{P}_{s\,|\,t},\varepsilon\right)}}-1\bigg]\cdot \bar{\theta}}{{1-\bar{\theta}}} +\varepsilon \\ &> 0+\varepsilon =  \varepsilon  \notag
	\end{align}
	
	Therefore, there exist such $\theta^{^{w+1}\bigstar\left(\mathcal{P}_{s\,|\,t},\varepsilon\right)},\,\forall w\in \bigstar\left(\mathcal{P}_{s\,|\,t},\varepsilon\right)$ to satisfy the condition we mentioned above. That is, the bijective correspondence $\mathcal{Y}_{\varepsilon\rightarrow \theta}:\varepsilon \rightarrow \theta_{^{w+1}\bigstar\left(\mathcal{P}_{s\,|\,t},\varepsilon\right)}$ $\left(R\rightarrow R\right)$ is guaranteed to hold.
\end{proof}

\subsection{Proof of Lemma \ref{Lemma for Chain}}
\begin{proof}[\normalfont\bfseries Proof of Lemma \ref{Lemma for Chain}]
	\,
	
	\noindent\textbf{Proof of $a)$}: 
	
	For any negative regular chain, it will make $\hat{s}^{^w\bigstar\left(\mathcal{P}_{s\,|\,t},\varepsilon\right)}<0$ hold. For any positive regular chain, it will make $\hat{s}^{^w\bigstar\left(\mathcal{P}_{s\,|\,t},\varepsilon\right)}>0$ hold. Therefore, pre-regularity is the necessary condition for achieving regularity.
	
	\noindent\textbf{Proof of $b)$}: 
	
	Refer to Theorem \ref{Decomposition and Compression Equivalence}, we obtain that $\mathcal{B}^{\mathcal{I}}(\varepsilon)=\varnothing$ $\bigg(\mathcal{B}(\varepsilon)=\varnothing\bigg)$ in any negative (positive) pre-regular chain $\bigstar\left(\mathcal{P}_{s\,|\,t},\varepsilon\right)$. Again by Theorem \ref{Decomposition and Compression Equivalence}, We can decompose negative (positive) pre-regular chain $\bigstar\left(\mathcal{P}_{s\,|\,t},\varepsilon\right)$ into a Maximal Bailout (Bail-in) Cluster $\mathcal{B}^{\bigstar\left(\mathcal{P}_{s\,|\,t},\varepsilon\right)}$ and a subset of $\bigstar\left(\mathcal{P}_{s\,|\,t},\varepsilon\right)$ (i.e., $\bigstar\left(\mathcal{P}_{s\,|\,t},\varepsilon\right)\setminus \mathcal{B}^{\bigstar\left(\mathcal{P}_{s\,|\,t},\varepsilon\right)}$) . Obviously $\hollowstar\left(\mathcal{P}_{s\,|\,t},\varepsilon\right)=\bigstar\left(\mathcal{P}_{s\,|\,t},\varepsilon\right)\setminus \mathcal{B}^{\bigstar\left(\mathcal{P}_{s\,|\,t},\varepsilon\right)}$ is also a negative (positive) pre-regular chain. Therefore, negative (positive) pre-regular chain $\bigstar\left(\mathcal{P}_{s\,|\,t},\varepsilon\right)$ can be decomposed into a regular chain $\mathcal{B}^{\bigstar\left(\mathcal{P}_{s\,|\,t},\varepsilon\right)}$ and a different pre-regular chain $\hollowstar\left(\mathcal{P}_{s\,|\,t},\varepsilon\right)$. 
	
	\noindent\textbf{Proof of $c)$}: 
	
	By definition of $\mathcal{B}^{\bigstar\left(\mathcal{P}_{s\,|\,t},\varepsilon\right)}$, for banks in system $\bigg\{\mathcal{B}^{\bigstar\left(\mathcal{P}_{s\,|\,t},\varepsilon\right)}\,\cup\, \left\{\left(\sigma\right) \diamond \hollowstar\left(\mathcal{P}_{s\,|\,t},\varepsilon\right)\right\}\bigg\}$, we have $x^{^m\mathcal{B}^{\bigstar\left(\mathcal{P}_{s\,|\,t},\varepsilon\right)}}<0,\,\forall m\in \mathcal{B}^{\bigstar\left(\mathcal{P}_{s\,|\,t},\varepsilon\right)}$ and $x^{^n\mathcal{B}^{\bigstar\left(\mathcal{P}_{s\,|\,t},\varepsilon\right)}}>0,\,\forall n\in \left\{\left(\sigma\right) \diamond \hollowstar\left(\mathcal{P}_{s\,|\,t},\varepsilon\right)\right\}$. Therefore, $\mathcal{B}^{\bigstar\left(\mathcal{P}_{s\,|\,t},\varepsilon\right)}=\mathcal{B}^{\bigg\{\mathcal{B}^{\bigstar\left(\mathcal{P}_{s\,|\,t},\varepsilon\right)}\,\cup\, \left\{\left(\sigma\right) \diamond \hollowstar\left(\mathcal{P}_{s\,|\,t},\varepsilon\right)\right\}\bigg\}}$.
\end{proof}

\subsection{Proof of Lemma \ref{Lemma For Commutative Diagram}}
\begin{proof}[\normalfont\bfseries Proof of Lemma \ref{Lemma For Commutative Diagram}]
	\,
	
	\noindent\textbf{Proof of $1$}: We input $\varepsilon$ into $\mathcal{X}_{\varepsilon\rightarrow A}$ and obtain $A^{^w\bigstar\left(\mathcal{P}_{s\,|\,t},\varepsilon\right)}$ which satisfies equation \ref{X_input A}. Observe that $\left(A^{^w\bigstar\left(\mathcal{P}_{s\,|\,t},\varepsilon\right)},\varepsilon \right)$ is a feasible pair satisfies equation \ref{X_input A}. Therefore, inputting $A^{^w\bigstar\left(\mathcal{P}_{s\,|\,t},\varepsilon\right)}$ into $\mathcal{X}_{A\rightarrow \varepsilon}$ leads to $\varepsilon$. That is $\mathcal{X}_{A\rightarrow \varepsilon}\circ \mathcal{X}_{\varepsilon\rightarrow A}\leftarrow \varepsilon=\varepsilon$.
	\begin{align}\label{X_input A}
		e^{-\beta\cdot \tilde{x}^{^w\bigstar\left(\mathcal{P}_{s\,|\,t},\varepsilon\right)}}= \frac{\theta^{^{w+1}\bigstar\left(\mathcal{P}_{s\,|\,t},\varepsilon\right)}}{\bar{\theta}}+\left(1-\frac{1}{\bar{\theta}}\right)\cdot \varepsilon,\,\forall w\in \bigstar\left(\mathcal{P}_{s\,|\,t},\varepsilon\right)
	\end{align}
	We input $A_{^w\bigstar\left(\mathcal{P}_{s\,|\,t},\varepsilon\right)}$ into $\mathcal{X}_{A\rightarrow \varepsilon}$ and obtain $\varepsilon_{x^{^{w+1}\bigstar\left(\mathcal{P}_{s\,|\,t},\varepsilon\right)}=0}$ which satisfies equation \ref{X_input Epsilon}. Observe that $\left(A_{^w\bigstar\left(\mathcal{P}_{s\,|\,t},\varepsilon\right)},\varepsilon_{x^{^{w+1}\bigstar\left(\mathcal{P}_{s\,|\,t},\varepsilon\right)}=0} \right)$ is a feasible pair satisfies equation \ref{X_input Epsilon}. Therefore, inputting $\varepsilon_{x^{^{w+1}\bigstar\left(\mathcal{P}_{s\,|\,t},\varepsilon\right)}=0}$ into $\mathcal{X}_{\varepsilon\rightarrow A}$ leads to $A_{^w\bigstar\left(\mathcal{P}_{s\,|\,t},\varepsilon\right)}$. That is $\mathcal{X}_{\varepsilon\rightarrow A}\circ \mathcal{X}_{A\rightarrow \varepsilon} \leftarrow A_{^w\bigstar\left(\mathcal{P}_{s\,|\,t},\varepsilon\right)}=A_{^w\bigstar\left(\mathcal{P}_{s\,|\,t},\varepsilon\right)}$.
	\begin{align}\label{X_input Epsilon}
		e^{-\beta\cdot \tilde{x}^{^w\bigstar\left(\mathcal{P}_{s\,|\,t},\,\varepsilon_{x^{^{w+1}\bigstar\left(\mathcal{P}_{s\,|\,t},\varepsilon\right)}=0}\right)}}= \frac{\theta^{^{w+1}\bigstar\left(\mathcal{P}_{s\,|\,t},\,\varepsilon_{x^{^{w+1}\bigstar\left(\mathcal{P}_{s\,|\,t},\varepsilon\right)}=0}\right)}}{\bar{\theta}}+\left(1-\frac{1}{\bar{\theta}}\right)\cdot \varepsilon,\,\forall w\in \bigstar\left(\mathcal{P}_{s\,|\,t},\varepsilon\right)
	\end{align}
	
	\noindent\textbf{Proof of $2$}:We input $\varepsilon$ into $\mathcal{Y}_{\varepsilon\rightarrow  \theta}$ and obtain $\theta_{^{w+1}\bigstar\left(\mathcal{P}_{s\,|\,t},\varepsilon\right)}$ which satisfies equation \ref{Y_input Theta}. Observe that $\left(\theta_{^{w+1}\bigstar\left(\mathcal{P}_{s\,|\,t},\varepsilon\right)},\varepsilon \right)$ is a feasible pair satisfies equation \ref{Y_input Theta}. Therefore, inputting $\theta_{^{w+1}\bigstar\left(\mathcal{P}_{s\,|\,t},\varepsilon\right)}$ into $\mathcal{Y}_{\theta\rightarrow \varepsilon}$ leads to $\varepsilon$. That is $\mathcal{Y}_{\theta\rightarrow \varepsilon} \circ \mathcal{Y}_{\varepsilon\rightarrow  \theta}\leftarrow \varepsilon=\varepsilon$.
	\begin{align}\label{Y_input Theta}
		e^{-\beta\cdot \tilde{x}^{^w\bigstar\left(\mathcal{P}_{s\,|\,t},\varepsilon\right)}}= \frac{\theta_{^{w+1}\bigstar\left(\mathcal{P}_{s\,|\,t},\varepsilon\right)}}{\bar{\theta}}+\left(1-\frac{1}{\bar{\theta}}\right)\cdot \varepsilon,\,\forall w\in \bigstar\left(\mathcal{P}_{s\,|\,t},\varepsilon\right)
	\end{align}
	We input $\theta^{^{w+1}\bigstar\left(\mathcal{P}_{s\,|\,t},\varepsilon\right)}$ into $\mathcal{Y}_{\theta\rightarrow \varepsilon}$ and obtain $\varepsilon_{x^{^{w+1}\bigstar\left(\mathcal{P}_{s\,|\,t},\varepsilon\right)}=0}$ which satisfies equation \ref{Y_input Epsilon}. Observe that $\left(\theta^{^{w+1}\bigstar\left(\mathcal{P}_{s\,|\,t},\varepsilon\right)},\varepsilon_{x^{^{w+1}\bigstar\left(\mathcal{P}_{s\,|\,t},\varepsilon\right)}=0} \right)$ is a feasible pair satisfies equation \ref{Y_input Epsilon}. Therefore, inputting $\varepsilon_{x^{^{w+1}\bigstar\left(\mathcal{P}_{s\,|\,t},\varepsilon\right)}=0}$ into $\mathcal{Y}_{\varepsilon\rightarrow  \theta}$ leads to $\theta^{^{w+1}\bigstar\left(\mathcal{P}_{s\,|\,t},\varepsilon\right)}$. That is $\mathcal{Y}_{\varepsilon\rightarrow  \theta} \circ \mathcal{Y}_{\theta\rightarrow \varepsilon} \leftarrow  \theta^{^{w}\bigstar\left(\mathcal{P}_{s\,|\,t},\varepsilon\right)}=\theta^{^{w}\bigstar\left(\mathcal{P}_{s\,|\,t},\varepsilon\right)}$.
	\begin{align}\label{Y_input Epsilon}
		e^{-\beta\cdot \tilde{x}^{^w\bigstar\left(\mathcal{P}_{s\,|\,t},\,\varepsilon_{x^{^{w+1}\bigstar\left(\mathcal{P}_{s\,|\,t},\varepsilon\right)}=0}\right)}}= \frac{\theta^{^{w+1}\bigstar\left(\mathcal{P}_{s\,|\,t},\,\varepsilon_{x^{^{w+1}\bigstar\left(\mathcal{P}_{s\,|\,t},\varepsilon\right)}=0}\right)}}{\bar{\theta}}+\left(1-\frac{1}{\bar{\theta}}\right)\cdot \varepsilon,\,\forall w\in \bigstar\left(\mathcal{P}_{s\,|\,t},\varepsilon\right)
	\end{align}
	
	\noindent\textbf{Proof of $3$}: Refer to \ref{Circle_input A}, we obtain that $\mathcal{X}_{\varepsilon\rightarrow A}\circ \mathcal{Y}_{\theta\rightarrow \varepsilon} \circ \mathcal{Y}_{\varepsilon\rightarrow  \theta} \circ \mathcal{X}_{A\rightarrow \varepsilon} \leftarrow A_{^w\bigstar\left(\mathcal{P}_{s\,|\,t},\varepsilon\right)}=A_{^w\bigstar\left(\mathcal{P}_{s\,|\,t},\varepsilon\right)}$.
	\begin{align}\label{Circle_input A}
		&\quad\,\,\quad\mathcal{X}_{\varepsilon\rightarrow A}\circ \mathcal{Y}_{\theta\rightarrow \varepsilon} \circ \mathcal{Y}_{\varepsilon\rightarrow  \theta} \circ \mathcal{X}_{A\rightarrow \varepsilon} \leftarrow A_{^w\bigstar\left(\mathcal{P}_{s\,|\,t},\varepsilon\right)} \notag  \\ &\Longleftrightarrow  \mathcal{X}_{\varepsilon\rightarrow A}\circ \mathcal{Y}_{\theta\rightarrow \varepsilon} \circ \mathcal{Y}_{\varepsilon\rightarrow  \theta}  \leftarrow \varepsilon_{x^{^{w+1}\bigstar\left(\mathcal{P}_{s\,|\,t},\varepsilon\right)}=0}   \notag \\ &\Longleftrightarrow  \mathcal{X}_{\varepsilon\rightarrow A}\circ \mathcal{Y}_{\theta\rightarrow \varepsilon}   \leftarrow \theta^{^{w+1}\bigstar\left(\mathcal{P}_{s\,|\,t},\varepsilon\right)}   \\ &\Longleftrightarrow  \mathcal{X}_{\varepsilon\rightarrow A}  \leftarrow \varepsilon_{x^{^{w+1}\bigstar\left(\mathcal{P}_{s\,|\,t},\varepsilon\right)}=0}\notag  \\ &\Longleftrightarrow  A_{^w\bigstar\left(\mathcal{P}_{s\,|\,t},\varepsilon\right)}  \notag
	\end{align}
	
	\noindent\textbf{Proof of $4$}: Refer to \ref{Circle_input Epsilon}, we obtain that $\mathcal{X}_{A\rightarrow \varepsilon}\circ \mathcal{X}_{\varepsilon\rightarrow A} \circ \mathcal{Y}_{\theta\rightarrow \varepsilon} \circ \mathcal{Y}_{\varepsilon\rightarrow  \theta}\leftarrow \varepsilon=\varepsilon$.
	\begin{align}\label{Circle_input Epsilon}
		&\quad\,\,\quad   \mathcal{X}_{A\rightarrow \varepsilon}\circ \mathcal{X}_{\varepsilon\rightarrow A} \circ \mathcal{Y}_{\theta\rightarrow \varepsilon} \circ \mathcal{Y}_{\varepsilon\rightarrow  \theta}\leftarrow \varepsilon  \notag  \\ &\Longleftrightarrow  \mathcal{X}_{A\rightarrow \varepsilon}\circ \mathcal{X}_{\varepsilon\rightarrow A} \circ \mathcal{Y}_{\theta\rightarrow \varepsilon} \leftarrow \theta_{^{w+1}\bigstar\left(\mathcal{P}_{s\,|\,t},\varepsilon\right)}     \notag \\ &\Longleftrightarrow  \mathcal{X}_{A\rightarrow \varepsilon}\circ \mathcal{X}_{\varepsilon\rightarrow A}  \leftarrow \varepsilon    \\ &\Longleftrightarrow \mathcal{X}_{A\rightarrow \varepsilon}  \leftarrow A^{^w\bigstar\left(\mathcal{P}_{s\,|\,t},\varepsilon\right)}    \notag  \\ &\Longleftrightarrow  \varepsilon   \notag 
	\end{align}
\end{proof}

\subsection{Proof of Lemma \ref{Lemma ForPartitionInduced Equilibrium Transition Snd}}
\begin{proof}[\normalfont\bfseries Proof of Lemma \ref{Lemma ForPartitionInduced Equilibrium Transition Snd}]
	\,
	
	Refer to Strong Decomposition Theorem \ref{Partition-induced transition}, we can further divide $\bigstar\left(\mathcal{B}^{\hollowstar}(m),\varepsilon\right)$ into two parts. We denote $\bigstar_k=\bigg\{ {^0}\bigstar\left(\mathcal{B}^{\hollowstar}(m),\varepsilon\right),\,\cdots,\, {^{k-1}}\bigstar\left(\mathcal{B}^{\hollowstar}(m),\varepsilon\right)  \bigg\}$ which represents the first $k$ elements of set $\bigstar\left(\mathcal{B}^{\hollowstar}(m),\varepsilon\right)$ and $\bigstar_{m-k}$ represents the left part.
	
	For the system $\bigstar_k$, We firstly prove that it can be shaped to satisfy the requirements in \ref{Lemma ForDecomposition and Compression Equivalence} $3$ where $\theta^j=\theta^{{^0}\bigstar_{m-k}}$ by setting appropriate individual $A$ and $\theta$. We apply Lemma \ref{Lemma ForPartitionInduced Equilibrium Transition} to prove it. 
	
	We can find such a tuple $\bigg( A^{^{0}\bigstar_k},\,\cdots,\, A^{^{k-2}\bigstar_k} \bigg)$ which satisfies $e^{-\beta\cdot \tilde{x}^{^w\bigstar_k}}= \frac{\theta^{{^{w+1}}\bigstar_k} }{\bar{\theta}}+\left(1-\frac{1}{\bar{\theta}}\right)\cdot \varepsilon,\,\forall w\in \bigstar_k\setminus {^{k-1}}\bigstar_k$ refer to Lemma \ref{Lemma ForPartitionInduced Equilibrium Transition} $3$. We can further find a tuple $\bigg( A^{^{0}\bigstar_k}+\delta,\,\cdots,\, A^{^{k-2}\bigstar_k} +\delta\bigg)$ which satisfies \ref{Threshold A plus delta}. Then we ensure that $\bigstar_k\subset\mathcal{B}^{\bigstar\left(\mathcal{B}^{\hollowstar}(m),\varepsilon\right)}$.
	\begin{align}\label{Threshold A plus delta}
		\bigg( A^{^{0}\bigstar_k},\,\cdots,\, A^{^{k-2}\bigstar_k} \bigg) \ll \bigg( A^{^{0}\bigstar_k}+\delta,\,\cdots,\, A^{^{k-2}\bigstar_k} +\delta\bigg)
	\end{align}
	\noindent We apply Lemma \ref{Lemma ForPartitionInduced Equilibrium Transition} $3$ again to find such a $A^{^{k-1}\bigstar_k}$ which satisfies $e^{-\beta\cdot \tilde{x}^{^{k-1}\bigstar_k}}= \frac{ \theta^{{^0}\bigstar_{m-k}} }{\bar{\theta}}+\left(1-\frac{1}{\bar{\theta}}\right)\cdot \varepsilon$. Then we set $A_{^{k-1}\bigstar_k}=A^{^{k-1}\bigstar_k}+\delta$ which satisfies \ref{Threshold biggest A plus delta}.
	\begin{align}\label{Threshold biggest A plus delta}
		e^{-\beta\cdot \tilde{x}^{^{k-1}\bigstar_k}}&\gg \frac{ \theta^{{^w}\bigstar_{m-k}} }{\bar{\theta}}+\left(1-\frac{1}{\bar{\theta}}\right)\cdot \varepsilon,\,\forall w\in \bigstar_{m-k} \\ \Longleftrightarrow  e^{-\beta\cdot \tilde{x}^{^{k-1}\bigstar_k}}&\gg \frac{ \theta^{{^0}\bigstar_{m-k}} }{\bar{\theta}}+\left(1-\frac{1}{\bar{\theta}}\right)\cdot \varepsilon \notag
	\end{align}
	\noindent We conclude that $\bigstar_{m-k}\nsubseteq \mathcal{B}^{\bigstar\left(\mathcal{B}^{\hollowstar}(m),\varepsilon\right)}$ by inequalities \ref{Threshold A plus delta} and \ref{Threshold biggest A plus delta}. That is $\bigstar_k=\mathcal{B}^{\bigstar\left(\mathcal{B}^{\hollowstar}(m),\varepsilon\right)}$.
	
	Therefore we precisely control the number of $\bigg|\mathcal{B}^{\bigstar\left(\mathcal{B}^{\hollowstar}(m),\varepsilon\right)}\bigg|$, specifically $k$, and $m-k$ for $\bigg|\hollowstar\left(\mathcal{B}^{\hollowstar}(m),\varepsilon\right)\bigg|$.
\end{proof}

\subsection{Supplementary Results of Section \ref{preliminaries:part1} and \ref{preliminaries:part2}}
\begin{restatable}{thm3}{LemmaForPartitionInducedEquilibriumTransitionThird}\label{Lemma ForPartitionInduced Equilibrium Transition Third}
	Under $\varepsilon$, the system $\mathcal{P}_{s\,|\,t}$ with $\theta^j\gg \theta^i$, $s^j\geq 0$ and $s^i>0$  where $i\in \mathcal{B}^{\mathcal{I}}(\varepsilon)$ exhibits:
	
	1. $e^{-\beta\cdot\sum\limits_{i\in\mathcal{B}^{\mathcal{I}}(\varepsilon)} \tilde{\Phi}^i\left(\mathbf{s}^{|\mathcal{B}^{\mathcal{I}}(\varepsilon)|\times 1}\right)}$ is non-inreasing in $\beta$, $A_i$ and $\sum\limits_i A_i$.
	
	2. Denote $f$ satisfies $\theta^f=\max\limits_{k\in\mathcal{B}^{\mathcal{I}}(\varepsilon)} \theta^k$ and all $\theta^k$ remain fixed. Then $\mathcal{B}^{\mathcal{I}}(\varepsilon)=\mathcal{P}_{s\,|\,t}\setminus \{j\}$ requires:
	
	\hspace{1em}$i)$ Ignition condition: $e^{-\beta\cdot\sum\limits_{i\in\mathcal{B}^{\mathcal{I}}(\varepsilon)} \tilde{\Phi}^i\left(\mathbf{s}^{|\mathcal{B}^{\mathcal{I}}(\varepsilon)|\times 1}\right)} \leq \frac{\theta^j}{\bar{\theta}}+ \underbrace{\left(1-\frac{1}{\bar{\theta}}\right)\cdot \varepsilon}_{<0}$.
	
	\hspace{1em}$ii)$ $\inf\limits_{\beta,\, A_{m\neq f},\,\sum\limits_{m\neq f} A_m} \left[e^{-\beta\cdot\sum\limits_{i\in\mathcal{B}^{\mathcal{I}}(\varepsilon)} \tilde{\Phi}^i\left(\mathbf{s}^{|\mathcal{B}^{\mathcal{I}}(\varepsilon)|\times 1}\right)}\right]=\max\limits_{k\in\mathcal{B}^{\mathcal{I}}(\varepsilon)}\frac{\theta^k}{\bar{\theta}}+\left(1-\frac{1}{\bar{\theta}}\right)\cdot \varepsilon$.
	
	\hspace{1em}$iii)$ $\beta$ satisfies $\beta_{x^j(\varepsilon)=0}<\beta<\beta_{x^f(\varepsilon)<0}$.
	
	\hspace{1em}$iv)$ The tuple $\mathbf{A}_{m\neq f}^{|\mathcal{B}^{\mathcal{I}}(\varepsilon)-1|\times 1}$ or $\sum\limits_{m\neq f} A_m$ is bounded.
	
	3. Increasing only $A_f$ into $+\infty$ doesn't reverse the conclusion of $\mathcal{B}^{\mathcal{I}}(\varepsilon)=\mathcal{P}_{s\,|\,t}\setminus \{j\}$ in 3.
	
	4. Threshold $\beta_{x^j(\varepsilon)=0}$ is non-decreasing in $\varepsilon$, $\min \beta_{x^j(\varepsilon)=0}=\beta_{x^j(\varepsilon=\varepsilon_0)=0}$. $\beta_{x^j(\varepsilon)=0}$ is non-increasing in $\theta^j$.
	
	5. Under $\min \beta_{x^j(\varepsilon)=0}<\beta<\beta_{x^f(\varepsilon)<0}$, we have:
	
	$i)$ as $\varepsilon$ grows, the signs of $\tilde{\Phi}^i\left(\mathbf{s}^{|\mathcal{B}^{\mathcal{I}}(\varepsilon)|\times 1}\right)$ and $x^j$ remain unchanged.
	
	$ii)$ As $\varepsilon$ decreases, partial order among $\tilde{\Phi}^i\left(\mathbf{s}^{|\mathcal{B}^{\mathcal{I}}(\varepsilon)|\times 1}\right)$ doesn't preserve, but Hierarchy remains.
	
	$iii)$ As $\varepsilon$ decreases, There exist a threshold $\varepsilon_{x^f(\varepsilon)=0}=\varepsilon_{\min\limits_{k\in\mathcal{B}^{\mathcal{I}}(\varepsilon)} x^k =0 }$ such that $\min\limits_{k\in\mathcal{B}^{\mathcal{I}}(\varepsilon)} x^k =0$ or $x^f(\varepsilon_{\min\limits_{k\in\mathcal{B}^{\mathcal{I}}(\varepsilon)} x^k =0 })$=0 where $f$ satisfies $\theta^f=\max\limits_k \theta^k$.
	
	$iv)$ $\mathcal{B}(\varepsilon)\setminus \{f\}$ is a Maximal Bail-in Cluster when $\varepsilon=\varepsilon_{\min\limits_{k\in\mathcal{B}^{\mathcal{I}}(\varepsilon)} x^k =0 }$.
\end{restatable}

\begin{proof}[\normalfont\bfseries Proof of Lemma \ref{Lemma ForPartitionInduced Equilibrium Transition Third}]
	\,
	
	We only need to reverse the demonstrations in the proof of Lemma \ref{Lemma ForDecomposition and Compression Equivalence}, i.e., $\tilde{\Phi}^i\left(\mathbf{s}^{|\mathcal{B}^{\mathcal{I}}(\varepsilon)|\times 1}\right)>0,\,\forall i\in \mathcal{B}^{\mathcal{I}}(\varepsilon)$.
\end{proof}

\begin{restatable}{thm3}{LemmaForPartitionInducedEquilibriumTransitionThird}\label{Lemma ForPartitionInduced Equilibrium Transition Fourth}
	Under $\varepsilon> \max\limits_{i\in \mathcal{P}_{s\,|\,t}} \frac{\theta^i-\bar{\theta}}{1-\bar{\theta}}$ with $\beta$ satisfies $e^{-\beta\cdot \tilde{x}^{^{w+1}\bigstar\left(\mathcal{P}_{s\,|\,t},\varepsilon\right)}}> \frac{\theta^{^{w}\bigstar\left(\mathcal{P}_{s\,|\,t},\varepsilon\right)}}{\bar{\theta}}+\left(1-\frac{1}{\bar{\theta}}\right)\cdot \varepsilon,\,\forall w\in \bigstar\left(\mathcal{P}_{s\,|\,t},\varepsilon\right)$ the system $\mathcal{P}_{s\,|\,t}$ exhibits:
	
	1. We can control each of the sets $\mathcal{B}^{\hollowstar^{q}\left(\mathcal{P}_{s\,|\,t},\varepsilon\right)}$ contains $\big| \hollowstar^{q}\left(\mathcal{P}_{s\,|\,t},\varepsilon\right) \big|-1$  banks by regulating the individual $A$ and $\theta$ where $q\in \bigg[ 0, \bigg|\bigstar\left(\mathcal{P}_{s\,|\,t},\varepsilon\right)\big|-1 \bigg],\,q\in \mathbb{N}$ and $\hollowstar^{0}\left(\mathcal{P}_{s\,|\,t},\varepsilon\right)=\bigstar\left(\mathcal{P}_{s\,|\,t},\varepsilon\right)$.
	
	2. The bijective correspondence $\mathcal{X}_{A\rightarrow \varepsilon}: A_{^{w+1}\bigstar\left(\mathcal{P}_{s\,|\,t},\varepsilon\right)} \rightarrow \varepsilon_{x^{^{w}\bigstar\left(\mathcal{P}_{s\,|\,t},\varepsilon\right)}=0}$ $\left(R\rightarrow R\right)$ holds, where $\varepsilon_{x^{^{w}\bigstar\left(\mathcal{P}_{s\,|\,t},\varepsilon\right)}=0}$ satisfies $e^{-\beta\cdot \tilde{x}^{^{w+1}\bigstar\left(\mathcal{P}_{s\,|\,t},\,\varepsilon_{x^{^{w}\bigstar\left(\mathcal{P}_{s\,|\,t},\varepsilon\right)}=0}\right)}}= \frac{\theta^{^{w}\bigstar\left(\mathcal{P}_{s\,|\,t},\,\varepsilon_{x^{^{w}\bigstar\left(\mathcal{P}_{s\,|\,t},\varepsilon\right)}=0}\right)}}{\bar{\theta}}+\left(1-\frac{1}{\bar{\theta}}\right)\cdot \varepsilon,\,\forall w\in \bigstar\left(\mathcal{P}_{s\,|\,t},\varepsilon\right)$.
	
	3. The bijective correspondence $\mathcal{X}_{\varepsilon\rightarrow A}:\varepsilon \rightarrow A^{^{w+1}\bigstar\left(\mathcal{P}_{s\,|\,t},\varepsilon\right)}$ $\left(R\rightarrow R\right)$ is guaranteed to hold, where $A^{^{w+1}\bigstar\left(\mathcal{P}_{s\,|\,t},\varepsilon\right)}$ satisfies $e^{-\beta\cdot \tilde{x}^{^{w+1}\bigstar\left(\mathcal{P}_{s\,|\,t},\varepsilon\right)}}= \frac{\theta^{^{w}\bigstar\left(\mathcal{P}_{s\,|\,t},\varepsilon\right)}}{\bar{\theta}}+\left(1-\frac{1}{\bar{\theta}}\right)\cdot \varepsilon,\,\forall w\in \bigstar\left(\mathcal{P}_{s\,|\,t},\varepsilon\right)$.
	
	4. The bijective correspondence $\mathcal{Y}_{\theta\rightarrow \varepsilon}:\theta^{^{w}\bigstar\left(\mathcal{P}_{s\,|\,t},\varepsilon\right)}\rightarrow \varepsilon_{x^{^{w}\bigstar\left(\mathcal{P}_{s\,|\,t},\varepsilon\right)}=0}$ $\left(R\rightarrow R\right)$ holds, where $\varepsilon_{x^{^{w}\bigstar\left(\mathcal{P}_{s\,|\,t},\varepsilon\right)}=0}$ satisfies $e^{-\beta\cdot \tilde{x}^{^{w+1}\bigstar\left(\mathcal{P}_{s\,|\,t},\,\varepsilon_{x^{^{w+1}\bigstar\left(\mathcal{P}_{s\,|\,t},\varepsilon\right)}=0}\right)}}= \frac{\theta^{^{w}\bigstar\left(\mathcal{P}_{s\,|\,t},\,\varepsilon_{x^{^{w}\bigstar\left(\mathcal{P}_{s\,|\,t},\varepsilon\right)}=0}\right)}}{\bar{\theta}}+\left(1-\frac{1}{\bar{\theta}}\right)\cdot \varepsilon,\,\forall w\in \bigstar\left(\mathcal{P}_{s\,|\,t},\varepsilon\right)$.
	
	5. The bijective correspondence $\mathcal{Y}_{\varepsilon\rightarrow \theta}:\varepsilon \rightarrow \theta_{^{w}\bigstar\left(\mathcal{P}_{s\,|\,t},\varepsilon\right)}$ $\left(R\rightarrow R\right)$ is guaranteed to hold, where $\theta_{^{w}\bigstar\left(\mathcal{P}_{s\,|\,t},\varepsilon\right)}$ satisfies $e^{-\beta\cdot \tilde{x}^{^{w+1}\bigstar\left(\mathcal{P}_{s\,|\,t},\varepsilon\right)}}= \frac{\theta_{^{w}\bigstar\left(\mathcal{P}_{s\,|\,t},\varepsilon\right)}}{\bar{\theta}}+\left(1-\frac{1}{\bar{\theta}}\right)\cdot \varepsilon,\,\forall w\in \bigstar\left(\mathcal{P}_{s\,|\,t},\varepsilon\right)$.
	
\end{restatable}

\begin{proof}[\normalfont\bfseries Proof of Lemma \ref{Lemma ForPartitionInduced Equilibrium Transition Fourth}]
	\,
	
	Just follow the same argument in the proof of Lemma \ref{Lemma ForPartitionInduced Equilibrium Transition}.
\end{proof}

\subsection{Proof of Weak Decomposition and Compression Equivalence Theorem}
\begin{proof}[\normalfont\bfseries Proof of Theorem \ref{Decomposition and Compression Equivalence}]
	\,
	
	\textbf{Proof of Weak Decomposition Theorem}
	
	We prove that it holds for any $\varepsilon$ by classification. Recall that $s^i=\frac{1}{\bar{\theta}}\left(\bar{\theta}A_i-E_i+(1-\bar{\theta})\cdot \varepsilon\cdot A_i  \right)$,  we obtain that $s^i\leq0$ if $\varepsilon\leq \frac{\theta^i-\bar{\theta}}{1-\bar{\theta}}$. We randomly choose a $\mathcal{P}_{s\,|\,t}\neq \varnothing$ in the following analysis. all banks in $\mathcal{P}_{s\,|\,t}$ are strictly heterogeneous, that is, any two banks cannot have the same assets and leverage ratio. Such an assumption is reasonable because, according to Lemma \ref{Perfection Switch=1} 1. $a)$,  homogeneous banks would yield homogeneous solutions, leading to redundant analysis. 
	
	1. $\varepsilon< \min\limits_{i\in \mathcal{P}_{s\,|\,t}} \frac{\theta^i-\bar{\theta}}{1-\bar{\theta}}$
	
	That is $s^i<0$, $\forall i\in \mathcal{P}_{s\,|\,t}$. By definition of $\mathcal{B}^{\mathcal{I}}(\varepsilon)$, we obtain that $\mathcal{B}^{\mathcal{I}}(\varepsilon)=\varnothing$. 
	
	Then we prove that $\mathcal{B}(\varepsilon)\neq \varnothing$. Suppose $\mathcal{B}(\varepsilon)= \varnothing$, choose any $s^i<0$, we can obtain $\varnothing \cup \{i\}$ by combining them, that is, $x^i<0$ which contradicts to the definition of $\mathcal{B}(\varepsilon)$. Therefore, $\varnothing\neq \mathcal{B}(\varepsilon)$ or $\mathcal{B}(\varepsilon) \neq \varnothing$. Then $\mathcal{P}_{s\,|\,t}$ can be decomposed into $\mathcal{B}(\varepsilon)\cup \left(\mathcal{P}_{s\,|\,t}\setminus \mathcal{B}(\varepsilon)\right) \cup \varnothing$.
	
	2. $\varepsilon= \min\limits_{i\in \mathcal{P}_{s\,|\,t}} \frac{\theta^i-\bar{\theta}}{1-\bar{\theta}}$
	
	That is, at least one (note that it's can be written as ``only one'', since banks are totally heterogeneous) bank $k$ satisfies that $s^k=0$. By the same argument in 1, we obtain that $\mathcal{B}(\varepsilon)\neq \varnothing$. 
	
	We argue that $\mathcal{B}^{\mathcal{I}}(\varepsilon)=\{k\}$. We choose any element with $s^j>0$ from $\mathcal{P}_{s\,|\,t}\setminus \bigg\{ \mathcal{B}(\varepsilon) \cup \{k\}\bigg\}$ and combine it with $\{k\}$. Obviously new system $\{k\}\cup \varnothing$ meets the requirement of $\mathcal{B}^{\mathcal{I}}(\varepsilon)$. Therefore, $\mathcal{B}^{\mathcal{I}}(\varepsilon)=\{k\}$. Then $\mathcal{P}_{s\,|\,t}$ can be decomposed into $\mathcal{B}(\varepsilon)\cup  \bigg\{\mathcal{P}_{s\,|\,t}\setminus \left\{ \mathcal{B}(\varepsilon) \cup \{k\}\right\}\bigg\} \cup \{k\}$.
	
	3. $\varepsilon \in \left(\min\limits_{i\in \mathcal{P}_{s\,|\,t}} \frac{\theta^i-\bar{\theta}}{1-\bar{\theta}}, \max\limits_{i\in \mathcal{P}_{s\,|\,t}} \frac{\theta^i-\bar{\theta}}{1-\bar{\theta}}\right)$
	
	That is at least one $s^j>0$ and at least one $s^i<0$.
	
	By the same argument in 1, we obtain that $\mathcal{B}(\varepsilon)\neq \varnothing$. We argue that $\mathcal{B}^{\mathcal{I}}(\varepsilon)\neq \varnothing$. Suppose $\mathcal{B}^{\mathcal{I}}(\varepsilon)= \varnothing$, we combine $\mathcal{B}^{\mathcal{I}}(\varepsilon)$ and $s^j$ into $\mathcal{B}^{\mathcal{I}}(\varepsilon) \cup \{j\}$, then we obtain that $x^j>0$ which contradicts to the definition of $\mathcal{B}^{\mathcal{I}}(\varepsilon)$. Therefore, $\varnothing\neq \mathcal{B}^{\mathcal{I}}(\varepsilon)$ or $\mathcal{B}^{\mathcal{I}}(\varepsilon) \neq \varnothing$. Then $\mathcal{P}_{s\,|\,t}$ can be decomposed into $\mathcal{B}(\varepsilon)\cup\mathcal{B}^{\mathcal{I}}(\varepsilon)\cup \bigg[\mathcal{P}_{s\,|\,t}\setminus\bigg(\mathcal{B}(\varepsilon)\cup \mathcal{B}^{\mathcal{I}}(\varepsilon)\bigg)\bigg]$.
	
	4. $\varepsilon \geq \max\limits_{i\in \mathcal{P}_{s\,|\,t}} \frac{\theta^i-\bar{\theta}}{1-\bar{\theta}}$
	
	That is all $s^j\geq0$. By definition of $\mathcal{B}(\varepsilon)$, we obtain that $\mathcal{B}(\varepsilon)= \varnothing$. By the same argument in 3, we obtain that $\mathcal{B}^{\mathcal{I}}(\varepsilon)\neq \varnothing$. Then $\mathcal{P}_{s\,|\,t}$ can be decomposed into $\varnothing\cup \left(\mathcal{P}_{s\,|\,t}\setminus \mathcal{B}^{\mathcal{I}}(\varepsilon)\right)\cup \mathcal{B}^{\mathcal{I}}(\varepsilon)$.
	
	Through the above demonstration, we have completed the proof of Weak Decomposition Theorem. Furthermore, for set $\mathcal{P}_{s\,|\,t}\setminus\bigg(\mathcal{B}(\varepsilon)\cup \mathcal{B}^{\mathcal{I}}(\varepsilon)\bigg)$, we can apply Weak Decomposition Theorem again. We prove that this process will converge. In the proof of Weak Decomposition Theorem, we observe that each decomposition requires $\mathcal{B}(\varepsilon)\subset \mathcal{P}_{s\,|\,t}$, $\mathcal{B}^{\mathcal{I}}(\varepsilon)\subset \mathcal{P}_{s\,|\,t}$ and the cardinality $\bigg| \mathcal{B}(\varepsilon) \cup \mathcal{B}^{\mathcal{I}}(\varepsilon) \bigg|\neq 0$. That is the composition process will proceed at most $\bigg|\mathcal{P}_{s\,|\,t} \bigg|<+\infty$ times. Suppose this process not converges, then we obtain that $\bigg|\mathcal{P}_{s\,|\,t} \bigg|\rightarrow +\infty$ which leads to a contradiction. Therefore, $\mathcal{P}_{s\,|\,t}\setminus\bigg(\mathcal{B}(\varepsilon)\cup \mathcal{B}^{\mathcal{I}}(\varepsilon)\bigg)$ can be further decomposed until $\mathcal{B}(\varepsilon)$ and $\mathcal{B}^{\mathcal{I}}(\varepsilon)$ remain unchanged or this process will converge.
	
	\textbf{Proof of Compression Equivalence Theorem}
	
	We firstly prove it holds in the first interval $\varepsilon\in\left[ \varepsilon_{\max\limits_{i\in _{-1}\mathcal{B}\left(\varepsilon\right)} x^i=0}, \varepsilon_{\max\limits_{f\in _{0}\mathcal{B}\left(\varepsilon\right)} x^f=0} \right]$. According to Lemma \ref{Lemma ForDecomposition and Compression Equivalence} $6.i)$, we obtain that $x^j>0$ and $\tilde{\Phi}^i\left(\mathbf{s}^{|{_0}\mathcal{B}\left(\varepsilon\right)|\times 1}\right)<0,\forall i \,\in{_0}\mathcal{B}\left(\varepsilon\right)$ under $\varepsilon\in\left[ \varepsilon_{\max\limits_{i\in _{-1}\mathcal{B}\left(\varepsilon\right)} x^i=0}, \varepsilon_{\max\limits_{f\in _{0}\mathcal{B}\left(\varepsilon\right)} x^f=0} \right)$ and $\beta\geq\beta_{x^j(\varepsilon=0)=0}=\max \beta_{x^j(\varepsilon)=0}$. When $\varepsilon=\varepsilon_{\max\limits_{f\in _{0}\mathcal{B}\left(\varepsilon\right)} x^f=0}$, we obtain that $\tilde{\Phi}^f\left(\mathbf{s}^{|{_0}\mathcal{B}\left(\varepsilon\right)|\times 1}\right)=0$. Refer to Lemma \ref{Lemma ForDecomposition and Compression Equivalence} $6.iv)$, we obtain that  $_1\mathcal{B}\left(\varepsilon\right) = {_0}\mathcal{B}\left(\varepsilon\right)\setminus \{f\}$ is Maximal Bailout Cluster. Refer to Lemma \ref{Lemma ForDecomposition and Compression Equivalence} $6.iii)$, we obtain that there exist a threshold $\varepsilon_{x^f(\varepsilon)=0}=\varepsilon_{\max\limits_{k\in{_1}\mathcal{B}(\varepsilon)} x^k =0 }$ such that $\max\limits_{k\in{_1}\mathcal{B}(\varepsilon)} x^k =0$ or $x^f(\varepsilon_{\max\limits_{k\in_1\mathcal{B}(\varepsilon)} x^k =0 })$=0 where $f$ satisfies $\theta^f=\min\limits_{k\in{_1}\mathcal{B}(\varepsilon)} \theta^k$ as $\varepsilon$ grows. In the second interval, we obtain that $\tilde{\Phi}^i\left(\mathbf{s}^{|{_1}\mathcal{B}\left(\varepsilon\right)|\times 1}\right)<0,\forall i \,\in{_1}\mathcal{B}\left(\varepsilon\right)$ under $\varepsilon\in\left[ \varepsilon_{\max\limits_{i\in _{0}\mathcal{B}\left(\varepsilon\right)} x^i=0}, \varepsilon_{\max\limits_{f\in _{1}\mathcal{B}\left(\varepsilon\right)} x^f=0} \right)$ by Lemma \ref{Lemma ForDecomposition and Compression Equivalence} $6.i)$. Then we can recursively apply Lemma \ref{Lemma ForDecomposition and Compression Equivalence} $6.iv)$, $6.iii)$ and $6.i)$ until $_{n}\mathcal{B}\left(\varepsilon\right)=_{n+1}\mathcal{B}\left(\varepsilon\right)$. Moreover, by the same argument for Crowding-out Effect, we have $\mathbf{\Phi}\left(\mathbf{s}^{| _{n+1}\mathcal{B}\left(\varepsilon\right) |\times 1}(\varepsilon)\right)  \leq \tilde{\mathbf{\Phi}}\left(\mathbf{s}^{| _{n+1}\mathcal{B}\left(\varepsilon\right) |\times 1}(\varepsilon)\right)$. Therefore, we have $\mathbf{\Phi}\left(\mathbf{s}^{| _{n+1}\mathcal{B}\left(\varepsilon\right) |\times 1}(\varepsilon)\right)  \leq \tilde{\mathbf{\Phi}}\left(\mathbf{s}^{| _{n+1}\mathcal{B}\left(\varepsilon\right) |\times 1}(\varepsilon)\right)\leq \mathbf{0}$ under $\varepsilon\in\left[ \varepsilon_{\max\limits_{i\in _{n}\mathcal{B}\left(\varepsilon\right)} x^i=0}, \varepsilon_{\max\limits_{f\in _{n+1}\mathcal{B}\left(\varepsilon\right)} x^f=0} \right]$ and ignition condition $\beta\geq\beta_{x^j(\varepsilon=0)=0}=\max \beta_{x^j(\varepsilon)=0}$, where $n\in \left\{ q\,|\, q\geq -1, q\in \mathbb{Z} \right\}$ and $\varepsilon_{\max\limits_{i\in _{-1}\mathcal{B}\left(\varepsilon\right)} x^i}=0$.
\end{proof}

\subsection{Proof of Strong Decomposition and Partition-Induced Equilibrium Transition Theorem}
\begin{proof}[\normalfont\bfseries Proof of Theorem \ref{Partition-induced transition}]
	\,
	
	\textbf{Proof of Strong Decomposition Theorem}
	
	According to Theorem \ref{Decomposition and Compression Equivalence}, system $\mathcal{P}_{s\,|\,t}\setminus\bigg(\mathcal{B}(\varepsilon)\cup \mathcal{B}^{\mathcal{I}}(\varepsilon)\bigg)$ can be further decomposed. Unlike the previously mentioned trichotomy, we adopt a bipartition here. That is, $\mathcal{P}_{s\,|\,t}\setminus\bigg(\mathcal{B}(\varepsilon)\cup \mathcal{B}^{\mathcal{I}}(\varepsilon)\bigg)=\bigg\{ j\,|\, s^j\geq 0, j\notin \mathcal{B}^{\mathcal{I}}(\varepsilon),j\in \mathcal{P}_{s\,|\,t} \bigg\}\,\cup\,\bigg\{ i\,|\, s^i< 0, i\notin \mathcal{B}(\varepsilon),i\in \mathcal{P}_{s\,|\,t} \bigg\}$, we firstly prove that such a bipartition can be achieved. 
	
	We first implement a basic bipartition $\mathcal{P}_{s\,|\,t}=\bigg\{ j\,|\, s^j\geq 0, j\in \mathcal{P}_{s\,|\,t} \bigg\}\,\cup\,\bigg\{ i\,|\, s^i< 0, i\in \mathcal{P}_{s\,|\,t} \bigg\}$. Next, we apply Theorem \ref{Decomposition and Compression Equivalence} to ${_{s\geq0}}\mathcal{P}_{s\,|\,t}=\bigg\{ j\,|\, s^j\geq 0, j\in \mathcal{P}_{s\,|\,t} \bigg\}$. Obviously $\mathcal{B}({_{s\geq0}}\mathcal{P}_{s\,|\,t},\varepsilon)=\varnothing$ in ${_{s\geq0}}\mathcal{P}_{s\,|\,t}$. Therefore, we can apply bipartition for ${_{s\geq0}}\mathcal{P}_{s\,|\,t}$, that is, $\bigg\{ j\,|\, s^j\geq 0, j\in \mathcal{P}_{s\,|\,t} \bigg\}=\mathcal{B}^{\mathcal{I}}(\varepsilon)\cup \bigg\{ j\,|\, s^j\geq 0, j\notin \mathcal{B}^{\mathcal{I}}(\varepsilon),j\in \mathcal{P}_{s\,|\,t} \bigg\}$. The same argument holds for ${_{s<0}}\mathcal{P}_{s\,|\,t}=\bigg\{ i\,|\, s^i< 0, i\in \mathcal{P}_{s\,|\,t} \bigg\}$, that is, $\bigg\{ i\,|\, s^i< 0, i\in \mathcal{P}_{s\,|\,t} \bigg\}=\mathcal{B}(\varepsilon) \cup \bigg\{ i\,|\, s^i< 0, i\notin \mathcal{B}(\varepsilon),i\in \mathcal{P}_{s\,|\,t} \bigg\}$. 
	
	Therefore, we can divide $\mathcal{P}_{s\,|\,t}$ into four distinct parts in just one step, that is \ref{Basic Four Parts of system P}.
	\begin{align}\label{Basic Four Parts of system P}
		\mathcal{P}_{s\,|\,t}=\mathcal{B}(\varepsilon)\cup \mathcal{B}^{\mathcal{I}}(\varepsilon)\cup \bigg\{ j\,|\, s^j\geq 0, j\notin \mathcal{B}^{\mathcal{I}}(\varepsilon),j\in \mathcal{P}_{s\,|\,t} \bigg\}\,\cup\,\bigg\{ i\,|\, s^i< 0, i\notin \mathcal{B}(\varepsilon),i\in \mathcal{P}_{s\,|\,t} \bigg\}
	\end{align}
	\noindent Furthermore, we randomly transfer some elements from $\bigg\{ i\,|\, s^i< 0, i\notin \mathcal{B}(\varepsilon),i\in \mathcal{P}_{s\,|\,t} \bigg\}$ to $\mathcal{B}(\varepsilon)$, forming $\mathcal{B}(\varepsilon)\,\cup\,\bigg\{ \left(\sigma\right) \diamond \left\{ i\,|\, s^i< 0,\, i\notin \mathcal{B}(\varepsilon) \right\} \bigg\}$, while the remaining part constitutes $\left(1-\sigma\right) \diamond \left\{ i\,|\, s^i< 0,\, i\notin \mathcal{B}(\varepsilon) \right\}$ (Note that we omit $i\in \mathcal{P}_{s\,|\,t}$ here to simplify algebra). The same operation is applied for $\mathcal{B}^{\mathcal{I}}(\varepsilon)$ and $\bigg\{ j\,|\, s^j\geq 0, j\notin \mathcal{B}^{\mathcal{I}}(\varepsilon),j\in \mathcal{P}_{s\,|\,t} \bigg\}$. Therefore, we can reformulate \ref{Basic Four Parts of system P} as \ref{Strong Four Parts of system P}.
	\begin{align}\label{Strong Four Parts of system P}
		\mathcal{P}_{s\,|\,t}&= \mathcal{B}(\varepsilon)\,\cup\,\bigg\{ \left(\sigma\right) \diamond \left\{ i\,|\, s^i< 0,\, i\notin \mathcal{B}(\varepsilon) \right\} \bigg\} \notag \\    &\quad\, \cup \left(1-\sigma\right) \diamond \left\{ i\,|\, s^i< 0,\, i\notin \mathcal{B}(\varepsilon) \right\}  \notag  \\ &\quad\, \cup \left(1-\sigma^{\mathcal{I}}\right) \diamond \left\{ j\,|\, s^j\geq 0,\, j\notin \mathcal{B}^{\mathcal{I}}(\varepsilon) \right\}      \\ &\quad\, \cup \,\,\mathcal{B}^{\mathcal{I}}(\varepsilon)\,\cup\,\bigg\{ \left(\sigma^{\mathcal{I}}\right) \diamond \left\{ j\,|\, s^j\geq 0,\, j\notin \mathcal{B}^{\mathcal{I}}(\varepsilon) \right\} \bigg\} \notag
	\end{align}
	
	\textbf{Proof of Partition-Induced Transition Theorem}
	
	For any $\mathcal{P}_{s\,|\,t}$, we can roughly divide it into different parts ${^0}\mathcal{P}_{s\,|\,t}\cup\cdots \cup {^\chi}\mathcal{P}_{s\,|\,t}$. Then by Strong Decomposition Theorem, we can further divide any ${^\chi}\mathcal{P}_{s\,|\,t}$ into four parts. We don't care about the set $\mathcal{B}(\varepsilon)\,\cup\,\bigg\{ \left(\sigma\right) \diamond \left\{ i\,|\, s^i< 0,\, i\notin \mathcal{B}(\varepsilon) \right\} \bigg\}$ and set $\mathcal{B}^{\mathcal{I}}(\varepsilon)\,\cup\,\bigg\{ \left(\sigma^{\mathcal{I}}\right) \diamond \left\{ j\,|\, s^j\geq 0,\, j\notin \mathcal{B}^{\mathcal{I}}(\varepsilon) \right\} \bigg\}$, since the elements in one of them are forming all either positive or negative solutions by the definition of $\mathcal{B}(\varepsilon)$ and $\mathcal{B}^{\mathcal{I}}$ (or refer to Lemma \ref{Lemma for Chain} $(c)$). The left two parts are all pre-regular chains.We denote $\left(1-\sigma\right) \diamond \left\{ i\,|\, s^i< 0,\, i\notin \mathcal{B}(\varepsilon),\,i\in {^\chi}\mathcal{P}_{s\,|\,t} \right\}$ as $\bigstar\left({^\chi}\mathcal{P}_{s\,|\,t},\varepsilon\right)$. That is, we now turn to study the case $\bigstar\left({^0}\mathcal{P}_{s\,|\,t},\varepsilon\right)\cup\cdots \cup  \bigstar\left({^\chi}\mathcal{P}_{s\,|\,t},\varepsilon\right)$ which is a negative pre-regular chain. Then we can apply Lemma \ref{Lemma ForPartitionInduced Equilibrium Transition} and Lemma \ref{Lemma ForPartitionInduced Equilibrium Transition Snd} to precisely control $\bigg|\mathcal{B}^{\bigstar\left({^0}\mathcal{P}_{s\,|\,t},\varepsilon\right)}\bigg|,\,\cdots,\, \bigg|\mathcal{B}^{\bigstar\left({^\chi}\mathcal{P}_{s\,|\,t},\varepsilon\right)}\bigg|$. This process requires the identical $\varepsilon$. The similar analysis holds for the set $\left(1-\sigma^{\mathcal{I}}\right) \diamond \left\{ j\,|\, s^j\geq 0,\, j\notin \mathcal{B}^{\mathcal{I}}(\varepsilon) \right\}$ which is a positive pre-regular chain.
	
	Therefore, We can precisely achieve exact Partition-Induced Transition while maintaining current $\varepsilon$ by controlling $\beta$, individual asset $A$ and individual leverage ratio $\theta$.
\end{proof}

\section{Perfection Equilibrium Properties}\label{AppendixSection: Perfection Equilibrium Properties}
\subsection{Example of Perfection Generating Space}
See Figure \ref{PGS_p=1_Snd}. Bank $F$ may not exit $\mathcal{P}_0$ at $t=1$ if it doesn't meet the requirement of partition perfection, and this can be achieved according to Theorem \ref{Partition-induced transition}. But it will eventually leave $\mathcal{P}_0$, since it will become the new maximal bail-in cluster $\mathcal{B}^{\mathcal{I}}(\varepsilon)$ or sooner or later (Lemma \ref{Lemma For Decomposition Chains in any Generating Space}).

\begin{figure}[!htbp]
	\centering
	\includegraphics[width = 1\textwidth]{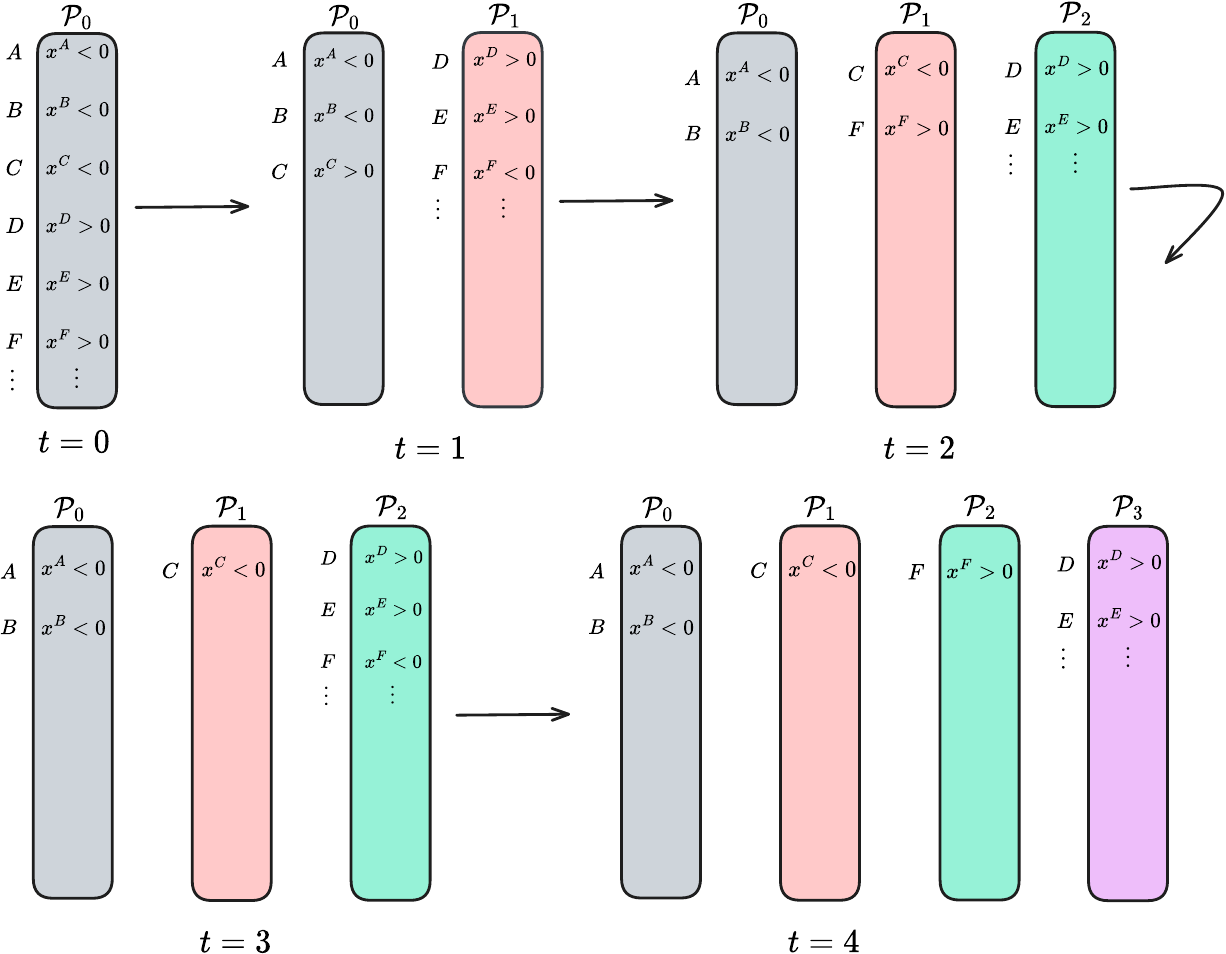}
	\caption{Perfection Generating Space: $p^s=1$, Second Version}
	\label{PGS_p=1_Snd}
\end{figure}

\subsection{Proof of Remark \ref{The equivalent operations p^s=1}}\label{AppendixSection: The equivalent operations}
\begin{proof}[\normalfont\bfseries Proof of Remark \ref{The equivalent operations p^s=1}]
	\,
	
	When $\sum\limits_{j\in \mathcal{P}_{s\,|\,t}}x^j_{s\,|\,t}\geq0$, bank $i\in {_\chi}\mathcal{P}_{s\,|\,t}=\left\{ i\,|\, s^i< 0,\, i\in \mathcal{P}_{s\,|\,t} \right\}$ will meet $x^i<0$ and be constraint to stay in $\mathcal{P}_{s\,|\,t}$ by operation $p^s=1$. Bank $j\in {^\chi}\mathcal{P}_{s\,|\,t}=\left\{ j\,|\, s^j\geq 0,\, j\in \mathcal{P}_{s\,|\,t} \right\}$ will behave under operation Partitions Perfection. 
	
	When $\sum\limits_{j\in \mathcal{P}_{s\,|\,t}}x^j_{s\,|\,t}<0$, bank $j\in {^\chi}\mathcal{P}_{s\,|\,t}=\left\{ j\,|\, s^j\geq 0,\, j\in \mathcal{P}_{s\,|\,t} \right\}$ will behave under Partitions Perfection operation, specifically, all of them will leave $\mathcal{P}_{s\,|\,t}$.  Bank $i\in {_\chi}\mathcal{P}_{s\,|\,t}=\left\{ i\,|\, s^i< 0,\, i\in \mathcal{P}_{s\,|\,t} \right\}$ whose $x^i<0$ will stay regardless of Partitions Perfection. While those banks with $x^i>0$ will leave $\mathcal{P}_{s\,|\,t}$ under Partitions Perfection, or Individual Perfection, since bank $i$ with $x^i>0$ is now on a personal bad state. Therefore, bank $i\in {_\chi}\mathcal{P}_{s\,|\,t}=\left\{ i\,|\, s^i< 0,\, i\in \mathcal{P}_{s\,|\,t} \right\}$ experience operation $p^s=1$ with Individual Perfection, while bank $j\in {^\chi}\mathcal{P}_{s\,|\,t}=\left\{ j\,|\, s^j\geq 0,\, j\in \mathcal{P}_{s\,|\,t} \right\}$ experience Partitions Perfection.
\end{proof}

\subsection{Proof of Lemma \ref{Lemma for changing in the same direction}}
\begin{proof}[\normalfont\bfseries Proof of Lemma \ref{Lemma For Decomposition Chains in any Generating Space}]
	\,
	
	Taking the derivative with respect to $\sum\limits_{j\in \mathcal{P}_{s\,|\,t}}A_j\cdot L^j$ on both sides of \ref{system P Sum} yields \ref{Decom_Exp_AcdotL}.
	
	\begin{align}\label{Decom_Exp_AcdotL}
		1 &= \frac{\partial\sum\limits_{i\in\mathcal{P}_{s\,|\,t}} x^i\left(\mathbf{s}^{|\mathcal{P}_{s\,|\,t}|\times 1}\right)}{\partial \left(\sum\limits_{j\in \mathcal{P}_{s\,|\,t}}A_j\cdot L^j \right)} + \left(1 - e^{-\beta\cdot\sum\limits_{i\in\mathcal{P}_{s\,|\,t}} x^i\left(\mathbf{s}^{|\mathcal{P}_{s\,|\,t}|\times 1}\right)}\right) \cdot \frac{1}{\sum\limits_{j\in \mathcal{P}_{s\,|\,t}}L^j} \notag \\ &\quad + A^{\mathcal{P}_{s\,|\,t}}\cdot\beta\cdot e^{-\beta\cdot\sum\limits_{i\in\mathcal{P}_{s\,|\,t}} x^i\left(\mathbf{s}^{|\mathcal{P}_{s\,|\,t}|\times 1}\right)} \cdot \frac{\partial\sum\limits_{i\in\mathcal{P}_{s\,|\,t}} x^i\left(\mathbf{s}^{|\mathcal{P}_{s\,|\,t}|\times 1}\right)}{\partial \left(\sum\limits_{j\in \mathcal{P}_{s\,|\,t}}A_j\cdot L^j \right)} 
	\end{align}
	\noindent By solving it, we obtain equation \ref{Decom_Exp_AcdotL Snd} which holds under the finite-risk mitigation condition \ref{finite-risk mitigation} or \ref{finite-risk mitigation2}.
	\begin{align}\label{Decom_Exp_AcdotL Snd}
		\frac{\partial\sum\limits_{i\in\mathcal{P}_{s\,|\,t}} x^i\left(\mathbf{s}^{|\mathcal{P}_{s\,|\,t}|\times 1}\right)}{\partial \left(\sum\limits_{j\in \mathcal{P}_{s\,|\,t}}A_j\cdot L^j \right)} &=\frac{1- \left(1 - e^{-\beta\cdot\sum\limits_{i\in\mathcal{P}_{s\,|\,t}} x^i\left(\mathbf{s}^{|\mathcal{P}_{s\,|\,t}|\times 1}\right)}\right) \cdot \frac{1}{\sum\limits_{j\in \mathcal{P}_{s\,|\,t}}L^j}}{1+A^{\mathcal{P}_{s\,|\,t}}\cdot\beta\cdot e^{-\beta\cdot\sum\limits_{i\in\mathcal{P}_{s\,|\,t}} x^i\left(\mathbf{s}^{|\mathcal{P}_{s\,|\,t}|\times 1}\right)}} \geq0
	\end{align}
\end{proof}

\subsection{Proof of Lemma \ref{Lemma For Decomposition Chains in any Generating Space}}
\begin{proof}[\normalfont\bfseries Proof of Lemma \ref{Lemma For Decomposition Chains in any Generating Space}]
	\,
	
	\noindent\textbf{Proof of $1$}: When $s^i=A_i\cdot L^i< 0,\,\forall i\in \mathcal{B}^{\bigstar\left({_\chi}\mathcal{P}_{s\,|\,t},\varepsilon\right)}$, we obtain that $L^i<0,\,\forall i\in \mathcal{B}^{\bigstar\left({_\chi}\mathcal{P}_{s\,|\,t},\varepsilon\right)}$. When $s^j=A_j\cdot L^j\geq 0,\,\forall j\in \mathcal{B}^{\bigstar\left({^\chi}\mathcal{P}_{s\,|\,t},\varepsilon\right)}$, we obtain that $L^j\geq0,\,\forall j\in \mathcal{B}^{\bigstar\left({^\chi}\mathcal{P}_{s\,|\,t},\varepsilon\right)}$. 
	
	\noindent\textbf{Proof of $2$}: 
	We prove it by Strong Decomposition Theorem \ref{Partition-induced transition}. We divide $\mathcal{P}_{s\,|\,t}$ into $\mathcal{B}^{\bigstar\left({_\chi}\mathcal{P}_{s\,|\,t},\varepsilon\right)}\cup \hollowstar\left({_\chi}\mathcal{P}_{s\,|\,t},\varepsilon\right)\cup {^\chi}\mathcal{P}_{s\,|\,t}$ to demonstrate that bank $i\in \mathcal{B}^{\bigstar\left({_\chi}\mathcal{P}_{s\,|\,t},\varepsilon\right)}$ will stay after the application of joint operation.
	
	We argue that combining $\mathcal{B}^{\bigstar\left({_\chi}\mathcal{P}_{s\,|\,t},\varepsilon\right)}$ and $\hollowstar\left({_\chi}\mathcal{P}_{s\,|\,t},\varepsilon\right)$ leads to a solution which is strictly less than Counterfactual Cluster Solution of equation system \ref{Comparison between MaxBailoutCluster and normal}. The definition of $\mathcal{B}^{\bigstar\left({_\chi}\mathcal{P}_{s\,|\,t},\varepsilon\right)}$ ensures that the newcomers in $\hollowstar\left({_\chi}\mathcal{P}_{s\,|\,t},\varepsilon\right)$ possess a positive solution. Then, according to the Crowding-out Effect, it follows that $x^i\ll \tilde{\Phi}^i\left(\mathbf{s}^{|\mathcal{B}^{\bigstar\left({_\chi}\mathcal{P}_{s\,|\,t},\varepsilon\right)}|\times 1}\right)<0,\,\forall i\in \mathcal{B}^{\bigstar\left({_\chi}\mathcal{P}_{s\,|\,t},\varepsilon\right)}$, i.e., the solution of the new system $\mathcal{B}^{\bigstar\left({_\chi}\mathcal{P}_{s\,|\,t},\varepsilon\right)}\cup \hollowstar\left({_\chi}\mathcal{P}_{s\,|\,t},\varepsilon\right)$ is strictly less than that of system $\mathcal{B}^{\bigstar\left({_\chi}\mathcal{P}_{s\,|\,t},\varepsilon\right)}$.
	\begin{align}\label{Comparison between MaxBailoutCluster and normal}
		A_i\cdot L^i=\tilde{\Phi}^i\left(\mathbf{s}^{|\mathcal{B}^{\bigstar\left({_\chi}\mathcal{P}_{s\,|\,t},\varepsilon\right)}|\times 1}\right)+&\left(1-e^{-\beta\cdot \sum\limits_{p\in \mathcal{B}^{\bigstar\left({_\chi}\mathcal{P}_{s\,|\,t},\varepsilon\right)}}\tilde{\Phi}^p\left(\mathbf{s}^{|\mathcal{B}^{\bigstar\left({_\chi}\mathcal{P}_{s\,|\,t},\varepsilon\right)}|\times 1}\right)}\right)\cdot A_i \notag  \\  &\vdots \quad\quad\quad\quad\quad\quad\quad    \\   A_m\cdot L^m=\tilde{\Phi}^m\left(\mathbf{s}^{|\mathcal{B}^{\bigstar\left({_\chi}\mathcal{P}_{s\,|\,t},\varepsilon\right)}|\times 1}\right)+&\left(1-e^{-\beta\cdot \sum\limits_{p\in \mathcal{B}^{\bigstar\left({_\chi}\mathcal{P}_{s\,|\,t},\varepsilon\right)}}\tilde{\Phi}^p\left(\mathbf{s}^{|\mathcal{B}^{\bigstar\left({_\chi}\mathcal{P}_{s\,|\,t},\varepsilon\right)}|\times 1}\right)}\right)\cdot A_m \notag 
	\end{align}
	\noindent We further introduce ${^\chi}\mathcal{P}_{s\,|\,t}$ into the system which will lead to an increase in $\sum\limits_{k}A_k\cdot L^k$ (or it's unchanged if $A_j\cdot L^j=0,\,\forall j\in {^\chi}\mathcal{P}_{s\,|\,t}$). According to Lemma \ref{Perfection Switch=1} $1.c)$ and Lemma \ref{Lemma for changing in the same direction}, we have $\sum\limits_{n\in \mathcal{B}^{\bigstar\left({_\chi}\mathcal{P}_{s\,|\,t},\varepsilon\right)}\cup \hollowstar\left({_\chi}\mathcal{P}_{s\,|\,t},\varepsilon\right)}x^n\leq \sum\limits_{q\in \mathcal{P}_{s\,|\,t}}x^q_{s\,|\,t}$. Again by applying Crowding-out Effect, we conclude that $\left(x^i\right)'\leq x^i<0,\,\forall i\in \mathcal{B}^{\bigstar\left({_\chi}\mathcal{P}_{s\,|\,t},\varepsilon\right)}$, i.e., the solution of the new system $\mathcal{P}_{s\,|\,t}$ is less than that of system $\mathcal{B}^{\bigstar\left({_\chi}\mathcal{P}_{s\,|\,t},\varepsilon\right)}\cup \hollowstar\left({_\chi}\mathcal{P}_{s\,|\,t},\varepsilon\right)$.
	
	Therefore, all banks in $\mathcal{B}^{\bigstar\left({_\chi}\mathcal{P}_{s\,|\,t},\varepsilon\right)}$ possess a negative solution which is constrained to remain.
	
	Then we divide $\mathcal{P}_{s\,|\,t}$ into $\mathcal{B}^{\bigstar\left({^\chi}\mathcal{P}_{s\,|\,t},\varepsilon\right)}\cup \hollowstar\left({^\chi}\mathcal{P}_{s\,|\,t},\varepsilon\right)\cup {_\chi}\mathcal{P}_{s\,|\,t}$ to demonstrate that bank $j\in \mathcal{B}^{\bigstar\left({^\chi}\mathcal{P}_{s\,|\,t},\varepsilon\right)}$ will leave after the application of joint operation.
	
	We argue that combining $\mathcal{B}^{\bigstar\left({^\chi}\mathcal{P}_{s\,|\,t},\varepsilon\right)}$ and $\hollowstar\left({^\chi}\mathcal{P}_{s\,|\,t},\varepsilon\right)$ leads to a solution which is strictly greater than Counterfactual Cluster Solution of equation system \ref{Comparison between MaxBailoutCluster and normal}. The definition of $\mathcal{B}^{\bigstar\left({^\chi}\mathcal{P}_{s\,|\,t},\varepsilon\right)}$ ensures that the newcomers in $\hollowstar\left({^\chi}\mathcal{P}_{s\,|\,t},\varepsilon\right)$ possess a negative solution. Then, according to the Crowding-out Effect, it follows that $0< \tilde{\Phi}^j\left(\mathbf{s}^{|\mathcal{B}^{\bigstar\left({_\chi}\mathcal{P}_{s\,|\,t},\varepsilon\right)}|\times 1}\right)\ll x^j,\,\forall j\in \mathcal{B}^{\bigstar\left({^\chi}\mathcal{P}_{s\,|\,t},\varepsilon\right)}$, i.e., the solution of the new system $\mathcal{B}^{\bigstar\left({^\chi}\mathcal{P}_{s\,|\,t},\varepsilon\right)}\cup \hollowstar\left({^\chi}\mathcal{P}_{s\,|\,t},\varepsilon\right)$ is strictly greater than that of system $\mathcal{B}^{\bigstar\left({^\chi}\mathcal{P}_{s\,|\,t},\varepsilon\right)}$. 
	
	We further introduce ${_\chi}\mathcal{P}_{s\,|\,t}$ into the system which will lead to an decrease in $\sum\limits_{k}A_k\cdot L^k$. According to Lemma \ref{Perfection Switch=1} $1.c)$ and Lemma \ref{Lemma for changing in the same direction}, we have $\sum\limits_{q\in \mathcal{P}_{s\,|\,t}}x^q_{s\,|\,t}\ll \sum\limits_{n\in \mathcal{B}^{\bigstar\left({^\chi}\mathcal{P}_{s\,|\,t},\varepsilon\right)}\cup \hollowstar\left({^\chi}\mathcal{P}_{s\,|\,t},\varepsilon\right)}x^n$. Again by applying Crowding-out Effect, we conclude that $0<x^j\ll \left(x^j\right)',\,\forall i\in \mathcal{B}^{\bigstar\left({_\chi}\mathcal{P}_{s\,|\,t},\varepsilon\right)}$, i.e., the solution of the new system $\mathcal{P}_{s\,|\,t}$ is strictly greater than that of system $\mathcal{B}^{\bigstar\left({^\chi}\mathcal{P}_{s\,|\,t},\varepsilon\right)}\cup \hollowstar\left({^\chi}\mathcal{P}_{s\,|\,t},\varepsilon\right)$. That is,each bank $j$ in $\mathcal{B}^{\bigstar\left({^\chi}\mathcal{P}_{s\,|\,t},\varepsilon\right)}$ possesses a solution of system $\mathcal{P}_{s\,|\,t}$ which is greater than Counterfactual Cluster Solution.
	
	Therefore, bank $j\in \mathcal{B}^{\bigstar\left({^\chi}\mathcal{P}_{s\,|\,t},\varepsilon\right)}$ will exit after the execution of joint operation.
\end{proof}

\subsection{Proof of Remark \ref{Remark for Partial order and Hierarchy}}\label{AppendixSection:Proof of Remark for Partial order and Hierarchy}
\begin{proof}[\normalfont\bfseries Proof of Remark \ref{Remark for Partial order and Hierarchy}]
	\,
	
	We prove that configurations such as Figure \ref{Partial order and Hierarchy} illustrated can be achieved.
	
	By the same argument in the proof of Lemma \ref{Corresponding to Counterfactual Monopoly Space} $f)$, As $A_f\rightarrow +\infty$, $x^f$ is bouned by a negative value $\bar{x}^f$, that is $\frac{x^f}{A_f}\rightarrow 0$.
	
	Refer to Lemma \ref{Perfection Switch=1} $1.a)$, any other bank $i\in \mathcal{B}\left(\varepsilon\right)$ meets equation \ref{Relations for Partial order and Hierarchy}. Therefore, all banks other than bank $f$ can freely choose their partial orders.
	\begin{align}\label{Relations for Partial order and Hierarchy}
		&x^i = \frac{A_i\cdot \left(\theta^f-\theta^i\right)}{\bar{\theta}} +\frac{A_i}{A_f} \cdot x^f ,\,\forall i\in \mathcal{B}\left(\varepsilon\right) \notag \\  \Longleftrightarrow \quad&x^i = \frac{A_i\cdot \left(\theta^f-\theta^i\right)}{\bar{\theta}} +0^{-},\,\forall i\in \mathcal{B}\left(\varepsilon\right)   \\  \Longleftrightarrow \quad&x^i \rightarrow \frac{A_i\cdot \left(\theta^f-\theta^i\right)}{\bar{\theta}},\,\forall i\in \mathcal{B}\left(\varepsilon\right) \notag
	\end{align}
	\noindent Moreover, we can solve $x^f$ by recursively apply Lemma \ref{Perfection Switch=1} $1.a)$, that is exactly what equation \ref{Relations for Partial order and Hierarchy2} depicts.
	\begin{align}\label{Relations for Partial order and Hierarchy2}
		&\quad\quad\quad\quad\quad\quad\quad\,\, e^{-\beta\cdot\sum\limits_{i\in\mathcal{P}_{s\,|\,t}\left(\varepsilon\right)} x^i\left(\mathbf{s}^{|\mathcal{P}_{s\,|\,t}|\times 1}\right)} = \frac{\theta^f}{\bar{\theta}} +\left(1-\frac{1}{\bar{\theta}}\right)\cdot \varepsilon  +\frac{x^f}{A_f}    \notag \\  \Longleftrightarrow \quad&  -\beta\cdot \left[ x^f+\sum\limits_{i\in\mathcal{P}_{s\,|\,t},\,i\neq f} x^i\left(\mathbf{s}^{|\mathcal{P}_{s\,|\,t}|\times 1}\right) \right] =  \frac{\theta^f}{\bar{\theta}} +\left(1-\frac{1}{\bar{\theta}}\right)\cdot \varepsilon +0^- \notag \\  \Longleftrightarrow \quad& -\beta\cdot \left[ \left(1+  \frac{A^{\mathcal{P}_{s\,|\,t}\setminus \{f\}}}{A_f}  \right)\cdot x^f +    \sum\limits_{i\in\mathcal{P}_{s\,|\,t},\,i\neq f} \frac{A_i\cdot \left(\theta^f-\theta^i\right)}{\bar{\theta}} \right] \rightarrow \frac{\theta^f}{\bar{\theta}} +\left(1-\frac{1}{\bar{\theta}}\right)\cdot \varepsilon  \\&     \notag \\  \Longleftrightarrow \quad&\quad\quad\quad\quad\quad\quad\quad\,\,  x^f \rightarrow  -\frac{\frac{\frac{\theta^f}{\bar{\theta}} +\left(1-\frac{1}{\bar{\theta}}\right)\cdot \varepsilon}{\beta}+ \sum\limits_{i\in\mathcal{P}_{s\,|\,t},\,i\neq f} \frac{A_i\cdot \left(\theta^f-\theta^i\right)}{\bar{\theta}}}{1+  \frac{A^{\mathcal{P}_{s\,|\,t}\setminus \{f\}}}{A_f}} \notag
	\end{align}
	\noindent We set $A_i\rightarrow 0^+,\, \forall i\in\mathcal{P}_{s\,|\,t},\,i\neq f$, we obtain that $x^i\rightarrow 0^-,\, \forall i\in\mathcal{P}_{s\,|\,t},\,i\neq f$ according to \ref{Relations for Partial order and Hierarchy}. While $x^f\rightarrow -\frac{\frac{\theta^f}{\bar{\theta}} +\left(1-\frac{1}{\bar{\theta}}\right)\cdot \varepsilon}{\beta}\ll 0$.
	
	Therefore, Figure \ref{Partial order and Hierarchy} can be achieved.
\end{proof}

\subsection{Main Properties and proofs}
We assume $\sum\limits_{i\in \mathcal{P}_{0\,|\,t=0}} x^i<0$. The inequality is used for simplifying algebra, it's orthogonal to the main results of properties, we examine it in lemma \ref{Perfection Switch=1} 2. $e)$.
\begin{restatable}{thm3}{LemmaPerfectionSwitchEqsOne}\label{Perfection Switch=1}
	As for $p^s=1$ with Partitions Perfection, we have:
	
	1. \textbf{Universal Properties}
	
	$a)$ Cross-bank dependence $x^i = \frac{A_i(\theta^j-\theta^i)}{\bar{\theta}} + \frac{A_i}{A_j}\cdot x^j$ holds in any $\mathcal{P}_{s\,|\,t}$.
	
	$b)$ \textbf{Hierarchy}: For $\theta^i\gg \theta^j$, then the state $x^i\geq0$ while $x^j\leq 0$ cannot hold in any Perfection Generating Space, regardless of their asset size relationship.
	
	$c)$ $\sum\limits_{j\in \mathcal{P}_{s\,|\,t}}x^j_{s\,|\,t}$ and $\sum\limits_{j\in \mathcal{P}_{s\,|\,t}}A_j\cdot L^j$ share the same sign in any $\mathcal{P}_{s\,|\,t}$, where   $L^j=\frac{1}{\bar{\theta}}\cdot \left[ \bar{\theta}-\theta^j +\left(1-\bar{\theta}\right)\cdot \varepsilon \right]$.
	
	$d)$ \textbf{No-turning-back}: for any bank $i\in \mathcal{P}_{s\,|\,t}$ with $s\geq 1$, the inclusion $i\in \mathcal{P}_{s-1\,|\, t+1}$ cannot hold. Equivalently, $d\left(\mathcal{P}^t\right)$ is non-decreasing over time.  
	
	$e)$ $\big|\mathcal{B}_0(\varepsilon)\big|$ is non-increasing in $\beta$ and $A_i$ where $i\in \mathcal{B}_0(\varepsilon)$ and $\mathcal{B}_0(\varepsilon)$ represents the Maximal Bailout Cluster of $\mathcal{P}_{0\,|\,t=0}$.
	
	$f)$ Splitting $\bigg\{ i \,\bigg|\, i\in \mathcal{P}_{0\,|\,t=0}, s^i\leq 0 \bigg\}$ into $\mathcal{B}_0(\varepsilon)$ and $\bigg\{ j \,\bigg|\, j\in \mathcal{P}_{0\,|\,t=0}\setminus \mathcal{B}_0(\varepsilon), s^j\leq 0 \bigg\}$ requires only one step. Analogously, the same holds for the partition of $\bigg\{ i \,\bigg|\, i\in \mathcal{P}_{0\,|\,t=0}, s^i\geq 0 \bigg\}$.
	
	2. \textbf{Properties of Equilibrium} 
	
	$a)$ $1\leq d\left(\bigoplus_{\mathcal{P}_{0}}\right)\leq 2$, and $\mathcal{P}_0=\mathcal{B}_0(\varepsilon)$ under $\bigoplus_{\mathcal{P}_{0}}$ or $\bigoplus_{\mathcal{P}_{0}} \,\cup\, \left\{p^s \oplus \mathcal{T}_p\right\}^{\mathring{t}}$ where $t'_0\leq \mathring{t} \leq \max\limits_s{t'_s}$.
	
	$b)$ $\max t= \max\limits_s{t_s} = t'_{\{\max s\}}$, $d\left(\bigoplus_{ \mathcal{P}_{\max\{s\}}  }\right)=1$ and $\mathcal{P}_{\max\{s\}}=\mathcal{B}^{\mathcal{I}}_0(\varepsilon)=\cdots=\mathcal{B}^{\mathcal{I}}_{\max\{s\}\,|\,t}(\varepsilon)$ where $\mathcal{B}^{\mathcal{I}}_0(\varepsilon)$ $\left(\mathcal{B}^{\mathcal{I}}_{\max\{s\}\,|\,t}(\varepsilon)\right)$ represents the Maximal Bail-in Cluster of $\mathcal{P}_{0\,|\,t=0}$ $\left(\mathcal{P}_{\max\{s\}\,|\,t}\right)$. Equivalently, newly generated space (or the last subspace of any $\mathcal{P}^{t}$) is guaranteed to include $\mathcal{B}^{\mathcal{I}}_0(\varepsilon)$.
	
	$c)$ Let $t_{\mathcal{B}^{\mathcal{I}}_0(\varepsilon)}$ denote the first time point at which $\mathcal{B}^{\mathcal{I}}_0(\varepsilon)$ appears solely, we have $1\leq t_{\mathcal{B}^{\mathcal{I}}_0(\varepsilon)}\leq \bigg| \bigg\{i\, \bigg|\, i\in \mathcal{P}_{0\,|\,t=0}\setminus \bigg(\mathcal{B}_0(\varepsilon)\cup \mathcal{B}^{\mathcal{I}}_0(\varepsilon)\bigg),s^i\leq 0 \bigg\} \bigg|+2$.
	
	$d)$ $0 \leq 2\cdot \sum\limits_{s\notin\{0,\max\{s\}\}}d\left(\bigoplus_{ \mathcal{P}_{s} }\right)\leq 8\cdot m -2   \,+\,  \left( n+1 \right)\cdot n$, where $m=\bigg|\mathcal{P}_{0\,|\,t=1}\setminus\mathcal{B}_0(\varepsilon)\bigg|$ and $n=\bigg|\mathcal{P}_{1\,|\,t=1}\setminus\mathcal{B}^{\mathcal{I}}_0(\varepsilon)\bigg|$.The necessary condition for the final equality to hold is that $\mathcal{P}_{0\,|\,t=1}=\bigg\{ i\,\bigg|\,i\in \mathcal{P}_{0\,|\,t=0},s^i\leq 0,x^i_{t=0}<0 \bigg\}$ and  $\mathcal{P}_{1\,|\,t=1}=\bigg\{ j\,\bigg|\,j\in \mathcal{P}_{0\,|\,t=0},s^j\geq 0,x^j_{t=0}>0 \bigg\}$.
	
	$e)$ $\mathcal{P}_{0}$ will eventually get stable with $\mathcal{P}_0=\mathcal{B}_0(\varepsilon)=\cdots=\mathcal{B}_{0\,|\,t}(\varepsilon)$ without the assumption $\sum\limits_{i\in \mathcal{P}_{0\,|\,t=0}} x^i<0$. Generally, any $\mathcal{P}_{s\,|\,t}$ will get perfect finally.
	
	$f)$ $\mathcal{P}_0$ is a Pure Bailout Space while $\mathcal{P}_{\max\{s\}}$ is a Pure Bail-in Space. Generally, $\mathcal{P}_s$ is either a Pure Bailout Space or a Pure Bail-in Space when the last Perfection Generating Space $\mathcal{P}^t$ gets perfect.
	
	$g)$ All bail-in banks ($s^i>0$) with $x^i\geq0$ and all bailout banks $s^i\leq0$ with $x^i<0$ when the system gets perfect.
	
	$h)$ we have $f)\Longleftrightarrow g)\Longleftrightarrow \mathcal{P}^{t}=\mathcal{P}^{t+1}$.
\end{restatable}

\begin{proof}[\normalfont\bfseries Proof of Lemma \ref{Perfection Switch=1}]
	\,
	
	\noindent\textbf{Proof of \textbf{Universal Properties}}: 
	
	\textbf{Proof of $a)$}: We reformulate equation system \ref{system P} as equation system \ref{system P_Switch=1} where $L^i=\frac{1}{\bar{\theta}}\cdot \left[ \bar{\theta}-\theta^i +\left(1-\bar{\theta}\right)\cdot \varepsilon \right]$.
	\begin{align}\label{system P_Switch=1}
		L^i=\frac{x^i}{A_i}+&\left(1-e^{-\beta\cdot \sum\limits_{j\in \mathcal{P}_{s\,|\,t}}x^j}\right) \quad\quad\quad\quad  \notag  \\  &\vdots \quad\quad\quad\quad\quad\quad\quad    \\   L^m=\frac{x^m}{A_m}+&\left(1-e^{-\beta\cdot \sum\limits_{j\in \mathcal{P}_{s\,|\,t}}x^j}\right) \quad\quad\quad\quad \notag
	\end{align}
	\noindent The difference between the first and last equations yields \ref{Difference in system P}.
	\begin{align}\label{Difference in system P}
		&\frac{\theta^m-\theta^i}{\bar{\theta}}=\frac{x^i}{A_i}-\frac{x^m}{A_m} \notag\\ \Longleftrightarrow &\,x^i = \frac{A_i\cdot \left(\theta^m-\theta^i\right)}{\bar{\theta}} +\frac{A_i}{A_m} \cdot x^m ,\,\forall m\in \mathcal{P}_{s\,|\,t}
	\end{align} 
	
	\textbf{Proof of $b)$}: When $\theta^i\gg \theta^j$, $x^i\geq0$ and $x^j\leq 0$, we obtain \ref{Hierarchy_impossible} which constitutes a contradiction. Therefore the state $x^i\geq0$ while $x^j\leq 0$ cannot hold when $\theta^i\gg \theta^j$.
	\begin{align}\label{Hierarchy_impossible}
		0\leq x^i = \frac{A_i\cdot \left(\theta^j-\theta^i\right)}{\bar{\theta}} +\frac{A_i}{A_j} \cdot x^j <0 ,\,\forall m 
	\end{align} 
	
	\textbf{Proof of $c)$}: Refer to Lemma \ref{Lemma ForDecomposition and Compression Equivalence}, Equation {system P Sum} has an unique solution, and $\sum\limits_{j\in \mathcal{P}_{s\,|\,t}}x^j_{s\,|\,t}=0$ when $\sum\limits_{j\in \mathcal{P}_{s\,|\,t}}A_j\cdot L^j=0$.  Refer to Lemma OA1.1 $d)$, we obtain that $\sum\limits_{j\in \mathcal{P}_{s\,|\,t}}x^j_{s\,|\,t}<0$ when $\sum\limits_{j\in \mathcal{P}_{s\,|\,t}}A_j\cdot L^j<0$, and  $\sum\limits_{j\in \mathcal{P}_{s\,|\,t}}x^j_{s\,|\,t}>0$ when $\sum\limits_{j\in \mathcal{P}_{s\,|\,t}}A_j\cdot L^j>0$. 
	
	Therefore, $\sum\limits_{j\in \mathcal{P}_{s\,|\,t}}x^j_{s\,|\,t}$ and $\sum\limits_{j\in \mathcal{P}_{s\,|\,t}}A_j\cdot L^j$ have the same sign.
	
	\textbf{Proof of $d)$}: It is equivalent to proving that $\sum\limits_{j\in \mathcal{P}_{s\,|\,t}}x^j_{s\,|\,t}$ or $\sum\limits_{j\in \mathcal{P}_{s\,|\,t}}A_j\cdot L^j$ is non-increasing over $t$ for all $1\leq s\leq \max{s}-1$. Since an increase in $\sum\limits_{j\in \mathcal{P}_{s\,|\,t}}x^j_{s\,|\,t}$ leads to less $x^f$ according to equation \ref{Do not move backward} and Crowding-out Effect (i.e., bank $f$ would fall into a bad (or worse) state if it moves backward). 
	\begin{align}\label{Do not move backward}
		s^f = \underbrace{x^f}_{\searrow} +  \underbrace{\left(1-e^{-\beta\cdot \sum\limits_{j\in \mathcal{P}_{s\,|\,t}}x^j}\right)\cdot A_f}_{\nearrow}
	\end{align} 
	\noindent Refer to Lemma \ref{Lemma For Decomposition Chains in any Generating Space} $2$,  we obtain that the Maximal Bailout Cluster of $\mathcal{P}_{s\,|\,t}$ will stay, while the Maximal Bail-in Cluster of $\mathcal{P}_{s\,|\,t}$ will leave. That is $\sum\limits_{q\in \mathcal{P}_{s\,|\,t+1}}x^q_{s\,|\,t+1}\leq \sum\limits_{j\in \mathcal{P}_{s\,|\,t}}x^j_{s\,|\,t}$. Therefore, no bank has incentive to move backward. Moreover, this implies that banks will either proliferate new space or maintain the current position in the next period. That is, $d\left(\mathcal{P}^t\right)$ is non-decreasing along the time domain.
	
	\textbf{Proof of $e)$}: Refer to Lemma \ref{Lemma ForDecomposition and Compression Equivalence} $2$ and $3$, and Lemma \ref{Lemma ForPartitionInduced Equilibrium Transition Snd}, we obtain that $\big|\mathcal{B}_0(\varepsilon)\big|$ is non-decreasing in $\beta$ and $A_i$. Since it will make more banks in set $\bigg\{ i \,\bigg|\, i\in \mathcal{P}_{0\,|\,t=0}, s^i\leq 0 \bigg\}$ change their states from $x<0$ into $x\geq 0$.
	
	\textbf{Proof of $f)$}: It's a special case of Strong Decomposition Theorem \ref{Partition-induced transition} where $\sigma=0$ and $\sigma^{\mathcal{I}}=0$.
	
	\noindent\textbf{Proof of \textbf{Properties of Equilibrium}}:
	
	\textbf{Proof of $a)$}: All bank $j\in{^\chi}\mathcal{P}_{s\,|\,t}=\left\{ j\,|\, s^j\geq 0,\, j\in \mathcal{P}_{s\,|\,t} \right\}$ will leave $\mathcal{P}_0$ since $\sum\limits_{i\in \mathcal{P}_{0\,|\,t=0}} x^i<0$. Refer to Lemma \ref{Lemma For Decomposition Chains in any Generating Space}, $\mathcal{B}_0(\varepsilon)$ will stay in $\mathcal{P}_{0\,|\,1}$. Bank $i\in\left\{ i\,|\, s^i< 0,\, i\in \mathcal{P}_{s\,|\,t},\,i\notin \mathcal{B}_0(\varepsilon) \right\}$ will leave $\mathcal{P}_0$ if $x^i\geq0$ at $t=0$, while the left part with $x^i<0$ will stay in $\mathcal{P}_0$. 
	
	If all bank $i\in\left\{ i\,|\, s^i< 0,\, i\in \mathcal{P}_{s\,|\,t},\,i\notin \mathcal{B}_0(\varepsilon) \right\}$ with $x^i\geq0$, then $\mathcal{P}_0=\mathcal{B}_0(\varepsilon)$ after the first joint operation, i.e., $d\left(\bigoplus_{\mathcal{P}_{0}}\right)=1$. If not, the left bank $i\in\left\{ i\,|\, s^i< 0,\, i\in \mathcal{P}_{s\,|\,t},\,i\notin \mathcal{B}_0(\varepsilon) ,\,  i\notin \mathcal{P}_{1\,|\,1} \right\}$ will possess $x^i_{0\,|\,1}>0$ by the definition of $\mathcal{B}_0(\varepsilon)$. In the next period, banks with $x^i_{0\,|\,1}>0$ will leave $\mathcal{P}_{0\,|\,1}$, rendering only $\mathcal{B}_0(\varepsilon)$ in $\mathcal{P}_0$ and $d\left(\bigoplus_{\mathcal{P}_{0}}\right)=2$. Refer to Lemma \ref{Lemma For Decomposition Chains in any Generating Space} again, the Maximal Bailout Cluster $\mathcal{B}_0(\varepsilon)$ will permanently stay in $\mathcal{P}_0$ in the following application of joint operations. That is $\mathcal{P}_0=\mathcal{B}_0(\varepsilon)$ under $\bigoplus_{\mathcal{P}_{0}} \,\cup\, \left\{p^s \oplus \mathcal{T}_p\right\}^{\mathring{t}}$.
	
	\textbf{Proof of $b)$}: By definition of $\mathcal{P}_{\max\{s\}}$, the subspace $\mathcal{P}_{\max\{s\}}$ is the final component to achieve stability in the system. Therefore, $d\left(\bigoplus_{ \mathcal{P}_{\max\{s\}}  }\right)=1$. According to Lemma \ref{Lemma For Decomposition Chains in any Generating Space}, $\mathcal{B}^{\mathcal{I}}_0(\varepsilon)$ will depart from its current subspace and enter newly generated space in each epoch, until the system gets perfect or the last generated subspace has no further entrants. That is $\mathcal{P}_{\max\{s\}}=\mathcal{B}^{\mathcal{I}}_0(\varepsilon)=\cdots=\mathcal{B}^{\mathcal{I}}_{\max\{s\}\,|\,t}(\varepsilon)$.
	
	\textbf{Proof of $c)$}: $t_{\mathcal{B}^{\mathcal{I}}_0(\varepsilon)}=1$ happens when $\mathcal{P}_{0\,|\,t=0}=\mathcal{B}_0(\varepsilon)\,\cup\,\mathcal{B}^{\mathcal{I}}_0(\varepsilon)$. Refer to Lemma \ref{Lemma For Decomposition Chains in any Generating Space}, the system will be stable after the first application of joint operation. That is $t_{\mathcal{B}^{\mathcal{I}}_0(\varepsilon)}=d\left(\bigoplus_{\mathcal{P}_{0}}\right)=1$.
	
	The next inequality requires $\mathcal{P}_{1\,|\,t=1}=\mathcal{P}_{0\,|\,t=0}\setminus \mathcal{B}_0(\varepsilon)$. Refer to Lemma \ref{Lemma For Decomposition Chains in any Generating Space} again, each joint operation will retain at least $\mathcal{B}^{\bigstar\left({_\chi}\mathcal{P}_{s\,|\,t},\varepsilon\right)}$, implying that removing exactly one element from $\mathcal{P}_{1\,|\,t=1}$ per epoch will maximize the hitting time for the first occurrence of $\mathcal{B}^{\mathcal{I}}_0(\varepsilon)$. 
	
	Accurately completing this segmentation process requires meeting two conditions. One of them is $\sum\limits_{i\in \mathcal{P}_{\max\{s\}\,|\,t}} x^i<0$, which gives cluster $j\in{^\chi}\mathcal{P}_{s\,|\,t}=\left\{ j\,|\, s^j\geq 0,\, j\in \mathcal{P}_{s\,|\,t} \right\}$ the motivation to leave. The other is that the internal structure of cluster $^{-}_{\chi}\mathcal{P}_{s\,|\,t}=\bigg\{i\, \bigg|\, i\in \mathcal{P}_{0\,|\,t=0}\setminus \bigg(\mathcal{B}_0(\varepsilon)\cup \mathcal{B}^{\mathcal{I}}_0(\varepsilon)\bigg),s^i\leq 0 \bigg\}$ needs to satisfy the condition that each $\mathcal{B}^{\bigstar\left({^-}{_\chi}\mathcal{P}_{s\,|\,t},\varepsilon\right)},\,\mathcal{B}^{\hollowstar\left({^-}{_\chi}\mathcal{P}_{s\,|\,t},\varepsilon\right)},\, \mathcal{B}^{\hollowstar^{2}\left({^-}{_\chi}\mathcal{P}_{s\,|\,t},\varepsilon\right)},\cdots,\,\mathcal{B}^{\hollowstar^{\big|\bigstar\left({^-}{_\chi}\mathcal{P}_{s\,|\,t},\varepsilon\right)\big|-1}\left({^-}{_\chi}\mathcal{P}_{s\,|\,t},\varepsilon\right)}$ contains only one element, meaning that in each period, we only keep one bank in the previous last generative subspace, because apart from this bank, the solutions for all other banks are greater than $0$. These two conditions can be achieved by Lemma \ref{Lemma ForPartitionInduced Equilibrium Transition}.
	
	When the last generated space only contains bank $j\in{^\chi}\mathcal{P}_{s\,|\,t}=\left\{ j\,|\, s^j\geq 0,\, j\in \mathcal{P}_{s\,|\,t} \right\}$, i.e., $t=\bigg| \bigg\{i\, \bigg|\, i\in \mathcal{P}_{0\,|\,t=0}\setminus \bigg(\mathcal{B}_0(\varepsilon)\cup \mathcal{B}^{\mathcal{I}}_0(\varepsilon)\bigg),s^i\leq 0 \bigg\} \bigg|+1$. Refer to Lemma \ref{Perfection Switch=1} $1.f)$, we only need one step to separate them into $\mathcal{B}^{\mathcal{I}}_0(\varepsilon)$ and ${^\chi}\mathcal{P}_{s\,|\,t}\setminus \mathcal{B}^{\mathcal{I}}_0(\varepsilon)$.
	
	Therefore $1\leq t_{\mathcal{B}^{\mathcal{I}}_0(\varepsilon)}\leq \bigg| \bigg\{i\, \bigg|\, i\in \mathcal{P}_{0\,|\,t=0}\setminus \bigg(\mathcal{B}_0(\varepsilon)\cup \mathcal{B}^{\mathcal{I}}_0(\varepsilon)\bigg),s^i\leq 0 \bigg\} \bigg|+2$.
	
	\textbf{Proof of $d)$}: We firstly prove that there is no partition satisfying the configuration illustrated in Figure \ref{Impossible Partition in Negative Pre-regular Chain} for any ${_\chi}\mathcal{P}_{s\,|\,t}=\left\{ i\,|\, s^i< 0,\, i\in \mathcal{P}_{s\,|\,t} \right\}$.Figure \ref{Impossible Partition in Negative Pre-regular Chain}'s leftmost diagram implies $\mathcal{B}\left(\varepsilon\right)=\left\{ x^1,\,\cdots,\, x^{n-1} \right\}$, while the central diagram implies $\mathcal{B}'\left(\varepsilon\right)=\left\{ x^1,\,\cdots,\, x^{n-k} \right\}$. This constitutes a contradiction because Lemma \ref{Lemma For Decomposition Chains in any Generating Space} establishes $|\mathcal{B}\left(\varepsilon\right)|\equiv|\mathcal{B}'\left(\varepsilon\right)|$, yet $\mathcal{B}'\left(\varepsilon\right)$ contains strictly fewer elements than $\mathcal{B}\left(\varepsilon\right)$ in the graphical representation. This implies that every partition of negative pre-regular chain $\bigstar\left({_\chi}\mathcal{P}_{s\,|\,t},\varepsilon\right)$ must retain at least one entity. We denote the departing component as $^{-}_{\chi}\mathcal{P}_{s\,|\,t}^1=\bigg\{i\, \bigg|\, i\in \bigg\{\mathcal{P}_{0\,|\,t=0}\setminus \bigg(\mathcal{B}_0(\varepsilon)\cup \mathcal{B}^{\mathcal{I}}_0(\varepsilon)\bigg)\bigg\} \setminus \mathcal{B}^{\bigstar\left({_\chi}\mathcal{P}_{s\,|\,t},\varepsilon\right)},s^i\leq 0 \bigg\}$, this notation is applied recursively, where $^{-}_{\chi}\mathcal{P}_{s\,|\,t}^0={^{-}_{\chi}}\mathcal{P}_{s\,|\,t}$ and $^{-}_{\chi}\mathcal{P}_{s\,|\,t}^2=\bigg\{i\, \bigg|\, i\in {^{-}_{\chi}}\mathcal{P}_{s\,|\,t}^1 \setminus  \mathcal{B}^{\bigstar\left({^{-}_{\chi}}\mathcal{P}_{s\,|\,t}^1,\varepsilon\right)}  \bigg\}$.
	\begin{figure}[!htbp]
		\centering
		\includegraphics[width = 0.6\textwidth]{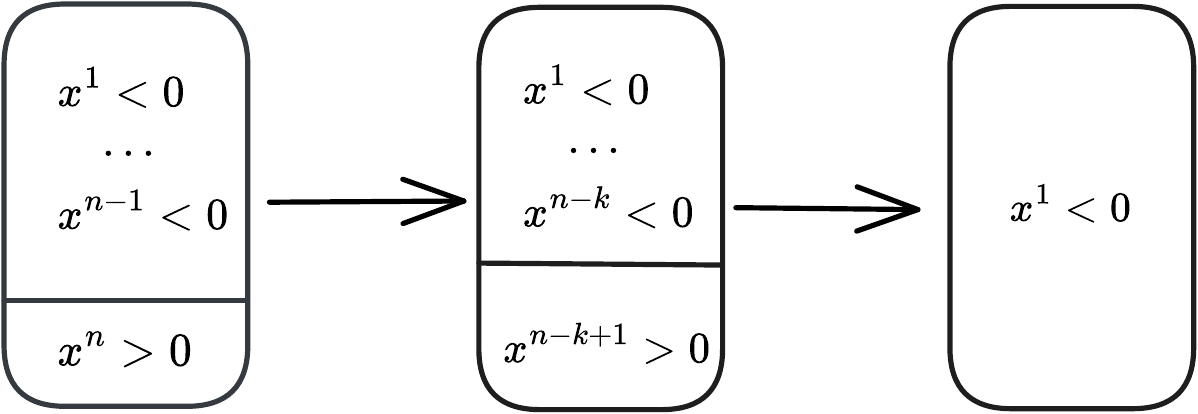}
		\caption{Impossible Partition in Negative Pre-regular Chain}
		\label{Impossible Partition in Negative Pre-regular Chain}
	\end{figure}
	
	Then we prove that a positive pre-regular chain $\bigstar\left({^\chi}\mathcal{P}_{s\,|\,t},\varepsilon\right)$ permits precisely one element removal per partition operation as Figure \ref{Possible Partition in Positive Pre-regular Chain} illustrates. In each space, the bank with the maximal index independently forms $\mathcal{B}^{\bigstar\left({^\chi}\mathcal{P}_{s\,|\,t},\varepsilon\right)}$. According to Lemma \ref{Lemma For Decomposition Chains in any Generating Space}, $\mathcal{B}^{\bigstar\left({^\chi}\mathcal{P}_{s\,|\,t},\varepsilon\right)}$ will depart from its original generative space. We denote the left part as $^{\chi}_{+}\mathcal{P}_{s\,|\,t}^1=\bigg\{i\, \bigg|\, i\in \bigg\{\mathcal{P}_{0\,|\,t=0}\setminus \bigg(\mathcal{B}_0(\varepsilon)\cup \mathcal{B}^{\mathcal{I}}_0(\varepsilon)\bigg)\bigg\} \setminus \mathcal{B}^{\bigstar\left({^\chi}\mathcal{P}_{s\,|\,t},\varepsilon\right)},s^i\geq 0 \bigg\}$, this notation is applied recursively, where ${^{\chi}_{+}}\mathcal{P}_{s\,|\,t}^2=\bigg\{i\, \bigg|\, i\in {^{\chi}_{+}}\mathcal{P}_{s\,|\,t}^1 \setminus  \mathcal{B}^{\bigstar\left({^{\chi}_{+}}\mathcal{P}_{s\,|\,t}^1,\varepsilon\right)}  \bigg\}$.
	\begin{figure}[!htbp]
		\centering
		\includegraphics[width = 0.6\textwidth]{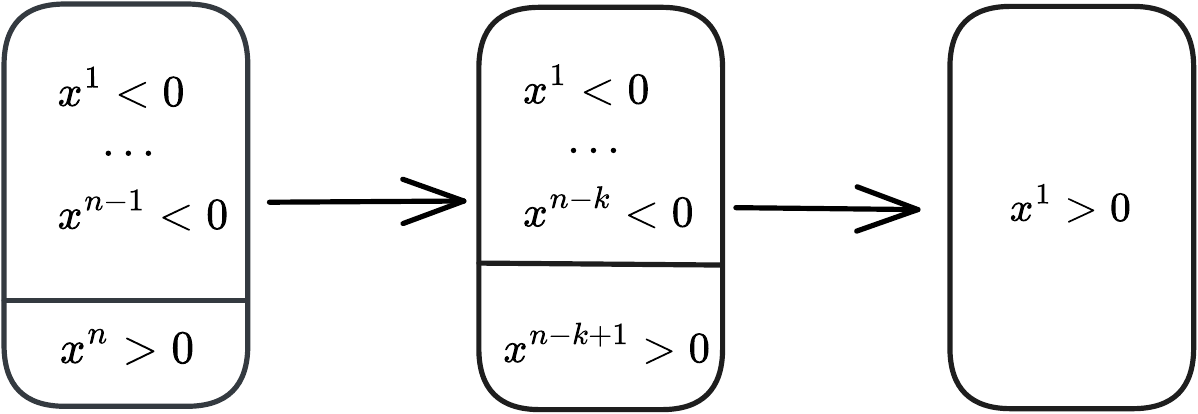}
		\caption{Possible Partition in Positive Pre-regular Chain}
		\label{Possible Partition in Positive Pre-regular Chain}
	\end{figure}
	
	Each subspace $k+1$ will experience at most four steps to get perfect during negative pre-regular chain rounds as Figure {Former Recursion} illustrates, where $k\leq \bigg|{^{-}_{\chi}}\mathcal{P}_{s\,|\,t}^0\bigg|-1,\,k\in \mathbb{N}$. In the last subspace $\bigg|{^{-}_{\chi}}\mathcal{P}_{s\,|\,t}^0\bigg|$, it only experience three steps because no further partition in a singleton negative regular chain refer to Lemma \ref{Lemma For Decomposition Chains in any Generating Space}. We argue that the unit-decrement partitioning strategy yields the largest partition count. Each subspace in negative pre-regular chain rounds requires at most four steps to get perfect, thus increasing the number of subspaces in such rounds directly extends the iteration duration. Therefore, the total step is $4\cdot\bigg| {^{-}_{\chi}}\mathcal{P}_{s\,|\,t}\bigg| -1$.
	\begin{figure}[!htbp]
		\centering
		\includegraphics[width = 0.85\textwidth]{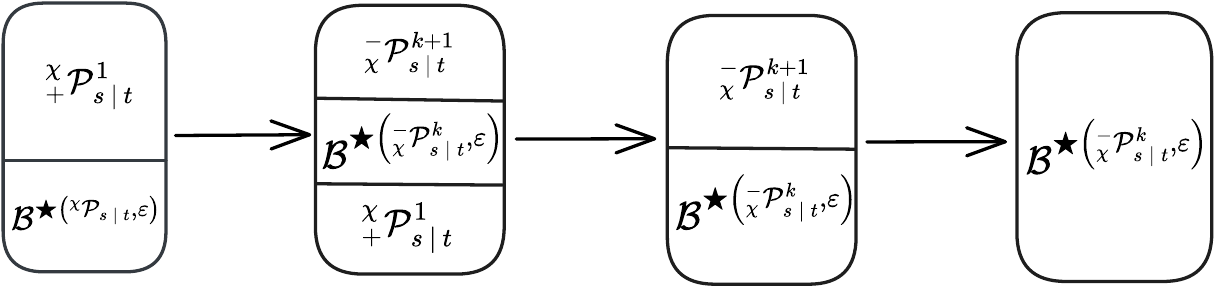}
		\caption{Iterations during negative pre-regular chain rounds}
		\label{Former Recursion}
	\end{figure}
	
	Each space $k+1+m$ requires at most $\bigg| {^{\chi}_{+}}\mathcal{P}_{s\,|\,t} \bigg|+1-m$ steps to get perfect during positive pre-regular chain rounds where $m\leq \bigg| {^{\chi}_{+}}\mathcal{P}_{s\,|\,t} \bigg|-1,\,m\in \mathbb{N_+}$. That is because we can iteratively remove single elements from space $k+1+m$ (we just proved it before, that is, ``a positive pre-regular chain $\bigstar\left({^\chi}\mathcal{P}_{s\,|\,t},\varepsilon\right)$ permits precisely one element removal per partition operation''), with each detached element to crowd out element in the next adjacent generated space which contains only one element. Figure \ref{Later Recursion} demonstrates the scenario when $m=1$. Therefore, the total step is $(1+\bigg| {^{\chi}_{+}}\mathcal{P}_{s\,|\,t} \bigg|)\cdot \bigg| {^{\chi}_{+}}\mathcal{P}_{s\,|\,t} \bigg|/2$.
	\begin{figure}[!htbp]
		\centering
		\includegraphics[width = 0.85\textwidth]{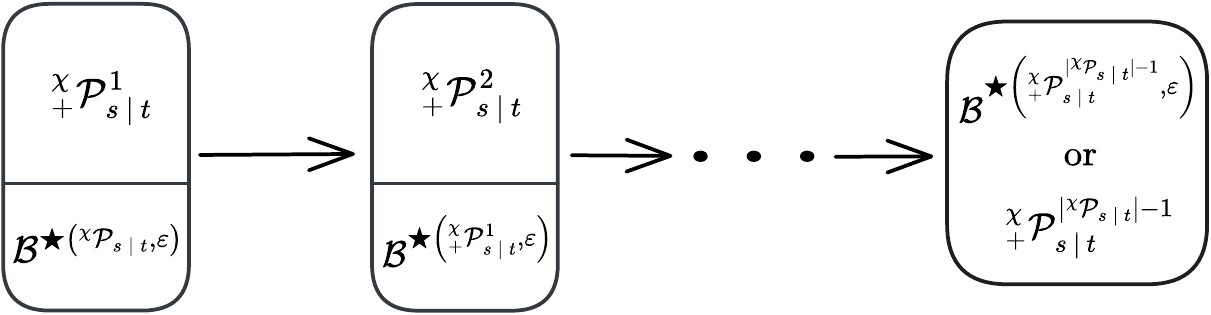}
		\caption{Iterations during positive pre-regular chain rounds}
		\label{Later Recursion}
	\end{figure}
	
	Moreover, the necessary condition for the scenario we mentioned above is ${_\chi}\mathcal{P}_{s\,|\,t}=\mathcal{P}_{0\,|\,t=1}=\bigg\{ i\,\bigg|\,i\in \mathcal{P}_{0\,|\,t=0},s^i\leq 0,x^i_{t=0}<0 \bigg\}$ and ${^\chi}\mathcal{P}_{s\,|\,t}=\mathcal{P}_{1\,|\,t=1}=\bigg\{ j\,\bigg|\,j\in \mathcal{P}_{0\,|\,t=0},s^j\geq 0,x^j_{t=0}>0 \bigg\}$. That is ${^{-}_{\chi}}\mathcal{P}_{s\,|\,t}={_\chi}\mathcal{P}_{s\,|\,t}\setminus \mathcal{B}_0(\varepsilon)=\mathcal{P}_{0\,|\,t=1}\setminus\mathcal{B}_0(\varepsilon)$ and ${^{\chi}_{+}}\mathcal{P}_{s\,|\,t}={^\chi}\mathcal{P}_{s\,|\,t}\setminus \mathcal{B}^{\mathcal{I}}_0(\varepsilon)=\mathcal{P}_{1\,|\,t=1}\setminus\mathcal{B}^{\mathcal{I}}_0(\varepsilon)$. Then we perform simple algebra to obtain $2\cdot \max\bigg[\sum\limits_{s\notin\{0,\max\{s\}\}}d\left(\bigoplus_{ \mathcal{P}_{s} }\right)\bigg]=8\cdot m -2   \,+\,  \left( n+1 \right)\cdot n$.
	
	\textbf{Proof of $e)$}: We only need to prove that $\mathcal{P}_{0}$ will eventually get stable when $\sum\limits_{i\in \mathcal{P}_{0\,|\,t=0}} x^i\geq0$. 
	Refer to Lemma \ref{Lemma for Chain} $b)$, we can decompose ${^\chi}\mathcal{P}_{s\,|\,t}=\left\{ j\,|\, s^j\geq 0,\, j\in \mathcal{P}_{s\,|\,t} \right\}$ into a positive regular chain $\bigstar\left({^\chi}\mathcal{P}_{s\,|\,t},\varepsilon\right)$ and a positive pre-regular chain ${^\chi}\mathcal{P}_{s\,|\,t}\setminus \bigstar\left({^\chi}\mathcal{P}_{s\,|\,t},\varepsilon\right)$. After the first partition, $\bigstar\left({^\chi}\mathcal{P}_{s\,|\,t},\varepsilon\right)$ will exit $\mathcal{P}_{0}$ according to Lemma \ref{Lemma For Decomposition Chains in any Generating Space}. We can further decompose ${^\chi}\mathcal{P}_{s\,|\,t}\setminus \bigstar\left({^\chi}\mathcal{P}_{s\,|\,t},\varepsilon\right)$ into a positive regular chain and a positive pre-regular chain by Lemma \ref{Lemma for Chain} $b)$ again, obviously, such a positive regular chain will exit $\mathcal{P}_{0}$. We recursively apply Lemma \ref{Lemma For Decomposition Chains in any Generating Space} and Lemma \ref{Lemma for Chain} $b)$ until a positive pre-regular chain can only be decomposed into a positive regular chain (i.e., itself, no further decomposition). Then we apply Lemma \ref{Lemma For Decomposition Chains in any Generating Space} again, such a (last negative regular) chain will leave $\mathcal{P}_{0}$, that is, $\mathcal{P}_{0}$ now is a negative pre-regular chain. Refer to Lemma \ref{Perfection Switch=1} $1.f)$,  $\mathcal{P}_{0}$ will be stable with $\mathcal{P}_{0}=\mathcal{B}_0\left(\varepsilon\right)=\mathcal{B}^{\bigstar\left({_\chi}\mathcal{P}_{s\,|\,t},\varepsilon\right)}$ in the next application of joint operation. Moreover, $\mathcal{P}_0=\mathcal{B}_0(\varepsilon)=\cdots=\mathcal{B}_{0\,|\,t}(\varepsilon)$ is guaranteed to hold refer to Lemma \ref{Lemma For Decomposition Chains in any Generating Space}.
	
	For any other any $\mathcal{P}_{s\,|\,t}$, we can divide then into ${_\chi}\mathcal{P}_{s\,|\,t}=\left\{ i\,|\, s^i< 0,\, i\in \mathcal{P}_{s\,|\,t} \right\}$ and ${^\chi}\mathcal{P}_{s\,|\,t}=\left\{ j\,|\, s^j\geq 0,\, j\in \mathcal{P}_{s\,|\,t} \right\}$ by Strong Decomposition Theorem \ref{Partition-induced transition} where $\sigma=\sigma^{\mathcal{I}}=1$. Then we can apply the same line of reasoning which used for the stability of $\mathcal{P}_{0}$ to demonstrate that any $\mathcal{P}_{s\,|\,t}$ will eventually be perfect.
	
	\textbf{Proof of $f)$}: According to Lemma \ref{Perfection Switch=1} $2.d)$, the system's stabilization time is finite, and the terminal subspace of the stable Perfection Generating Space contains only $\bigstar\left({^\chi}\mathcal{P}_{s\,|\,t},\varepsilon\right)$ which is a positive pre-regular chain, obviously the space is a Pure Bail-in Space. According to Lemma \ref{Perfection Switch=1} $2.e)$, $\mathcal{P}_{0}$ will get stable with $\mathcal{P}_{0}=\mathcal{B}_0\left(\varepsilon\right)=\bigstar\left({_\chi}\mathcal{P}_{s\,|\,t},\varepsilon\right)$, obviously $\mathcal{P}_{0}$ is a Pure Bailout Space. Lemma \ref{Perfection Switch=1} $2.e)$ implies that the system will get perfect finally. We can firstly apply Strong Decomposition Theorem \ref{Partition-induced transition} to divide any $\mathcal{P}_{s\,|\,t}$ into a positive pre-regular chain and a negative pre-regular chain. Then we can recursively apply Lemma \ref{Lemma For Decomposition Chains in any Generating Space}, Lemma \ref{Lemma for Chain} $b)$ to the decomposed pre-regular chains to get each space perfect with a regular chain, that is, any $\mathcal{P}_{s\,|\,t}$ will be a Pure Space eventually. The last negative regular chain ${_\chi}\tilde{\mathcal{P}}_{s\,|\,t}$ is the border of Pure Bail-in Space and Pure Bail-in Space according to Lemma \ref{Perfection Switch=1} $2.e)$ and Lemma \ref{Perfection Switch=1} $2.d)$. 
	
	\textbf{Proof of $g)$}: If not, suppose bank $m\in \left\{j \,|\, x^j<0,\,s^j>0,\,j\in \mathcal{P}_{s\,|\,t} \right\}$ and bank $n\in \left\{j \,|\, x^j\geq0,\,s^j>0,\,j\in \mathcal{P}_{s\,|\,t} \right\}$. Refer to the definition of Partition Perfection and Lemma \ref{Lemma For Decomposition Chains in any Generating Space}, bank $n$ who is an element of $\mathcal{B}^{\bigstar\left(\mathcal{P}_{s\,|\,t},\varepsilon\right)}$ will exit $\mathcal{P}_{s\,|\,t}$. That is the system is not perfect. Suppose bank $m\in \left\{j \,|\, x^j>0,\,s^j\leq0,\,j\in \mathcal{P}_{s\,|\,t} \right\}$ and bank $n\in \left\{j \,|\, x^j<0,\,s^j\leq0,\,j\in \mathcal{P}_{s\,|\,t} \right\}$. Obviously bank $m$ is a bad state and it implies that the system is not stable since bank $m$ will exit $\mathcal{P}_{s\,|\,t}$.
	
	\textbf{Proof of $h)$}: $f)\Longrightarrow g)$: Since $\mathcal{P}_s$ is either a Pure Bailout Space or a Pure Bail-in Space, we have all bail-in banks ($s^i>0$) with $x^i\geq0$ and all bailout banks $s^i\leq0$ with $x^i<0$. $g)\Longrightarrow \mathcal{P}^{t}=\mathcal{P}^{t+1}$: All bail-in banks are in a good state. All bailout banks cannot move. Therefore the system is stable, that is $\mathcal{P}^{t}=\mathcal{P}^{t+1}$. $\mathcal{P}^{t}=\mathcal{P}^{t+1}\Longrightarrow f)$: By definition of $\mathcal{P}^{t}=\mathcal{P}^{t+1}$, all bailout banks $s^i\leq0$ with $x^i<0$ and all bail-in banks ($s^i>0$) with $x^i\geq0$, that is $\mathcal{P}_s$ is either a Pure Bailout Space or a Pure Bail-in Space when the last Perfection Generating Space is perfect.
\end{proof}

\section{Proof of Propositions}
\subsection{Proof of Proposition \ref{Homogeneous and Homophily}}
As Figure \ref{Homgeneity and hierarchy} illustrates.
\begin{figure}[!htbp]
	\centering
	\includegraphics[width = 0.9\textwidth]{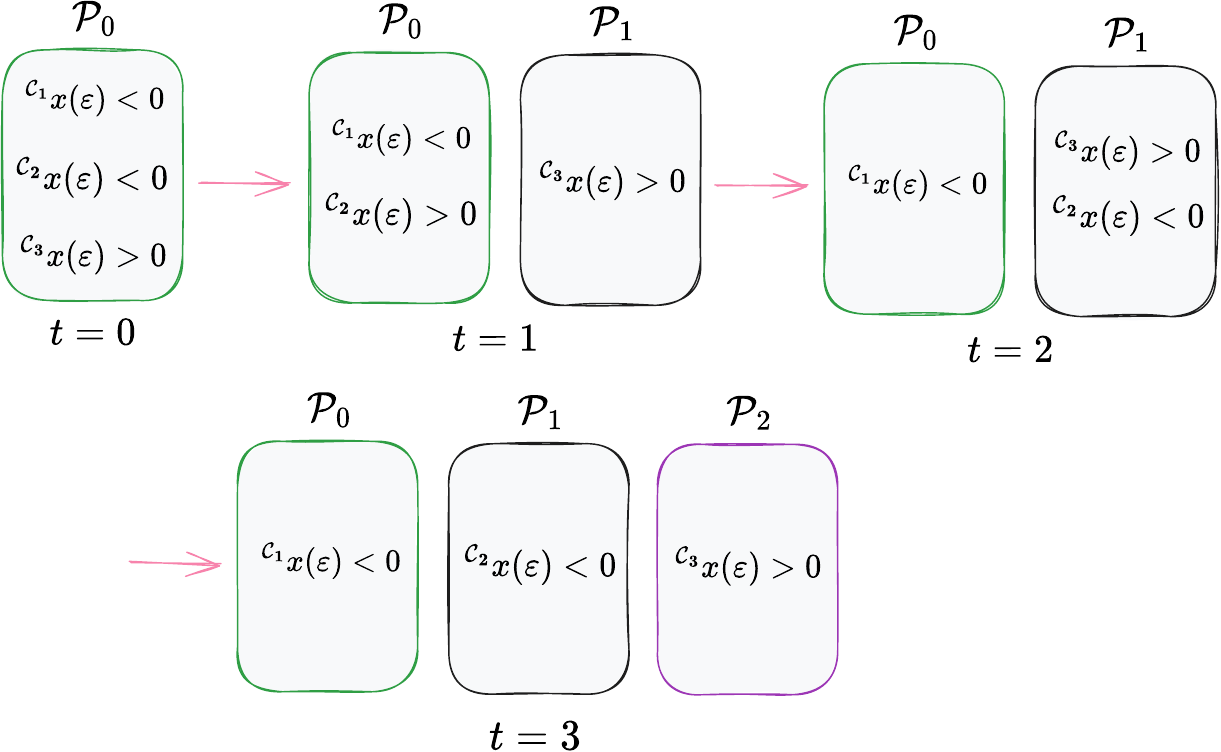}
	\caption{Perfection dynamics}
	\label{Homgeneity and hierarchy}
\end{figure}
\subsection{Proof of Proposition \ref{Schelling Point and Knightian Point}}
Refer to Lemma \ref{Lemma: Crowding-Out Effect} $(a)$, we have $p^s=1$ and $p^{\triangle}=1$.
\subsection{Proof of Proposition \ref{Emergence of Complementarity}}
Inequality \ref{inequality: Emergence of Complementarity} is a direct corollary of Lemma \ref{SubLemma: Crowding-Out Effect} $(a)$. The second inequality \ref{inequality: Emergence of Complementarity2} is a direct corollary of Lemma \ref{Lemma: Crowding-Out Effect Homophily}.
\subsection{Proof of Proposition \ref{Persistence of Substitution}}
Inequality \ref{inequality: Persistence of Substitution} is a direct corollary of Lemma \ref{SubLemma: Crowding-Out Effect} $(b)$. The second inequality is a direct corollary of Lemma \ref{Lemma: Crowding-Out Effect Homophily}.
\section{Discussion on inequality \ref{inequality: Emergence of Complementarity} and \ref{inequality: Persistence of Substitution}}\label{AppendixSection: Discussion on inequality}
We set the liquidation demand function as $s^i=x^i+\sum\limits_k\left(e^{\sum\limits_i -\beta x^i\pi_{ik}}-1 \right)M_k\cdot p_{ki}$. Obviously, $s^i$ is supmodular. We firstly prove that both $(i)$ $s^i<0$ with $x^i<0$ $\forall i$ and $(ii)$ $s^i>0$ with $x^i>0$ $\forall i$ are all possible in a subspace $\mathcal{P}_{s\,|\,t}$ (Lemma \ref{Supmodular: Properties}).
\begin{restatable}{thm3}{PropertiesofSupmodularFunction}\label{Supmodular: Properties}
	Properties of supmodular function $s^i$: 
	
	$S_1\cap S_2 \neq \varnothing$ where $S_1:=\left\{ \left(\beta,M_k\right) \,|\,  \exists (x^i,s^i)<0 \right\}$ and $S_2:=\left\{ \left(\beta,M_k\right) \,|\,  \exists (x^i,s^i)>0 \right\}$.
\end{restatable}
\begin{proof}[\normalfont\bfseries Proof of Lemma \ref{Supmodular: Properties}]
	\,
	
	We first prove that $S_1 \neq \varnothing$. Considering a set of symmetric solutions, we can write the system in a form similar to \ref{Equilibrium solution under BCP}, then aggregate to obtain a single equation of the form $s(x) = x + \left( e^{-\beta x} - 1 \right) M$. Differentiating with respect to $x$, we have:
	\begin{align}\label{Supmodular properties negative: Derivatives}
		\frac{ds}{dx} = 1 - \beta M e^{-\beta x} = 0
	\end{align}
	The extreme point is given by equation \ref{Supmodular properties negative: Extrem point},
	\begin{align}\label{Supmodular properties negative: Extrem point}
		e^{-\beta x} = \frac{1}{\beta M}
	\end{align}
	which implies $x = \frac{\ln(\beta M)}{\beta}$. When $x < 0$, we require $\frac{\ln(\beta M)}{\beta} < 0$. Since $\beta > 0$, this condition is equivalent to inequality \ref{Supmodular properties negative: Extrem point equivalent}
	\begin{align}\label{Supmodular properties negative: Extrem point equivalent}
		\beta \cdot M < 1
	\end{align}
	At the extreme point, $s(M) = \frac{\ln(\beta M)}{\beta} + \frac{1}{\beta} - M$, which is a monotonically increasing function of $M$, and $s(\frac{1}{\beta}) = 0$. Furthermore, due to inequality \ref{Supmodular properties negative: Extrem point equivalent}, we have $s(M) < 0$ for all $x < 0$. Therefore, $S_1 \neq \varnothing$.
	
	Furthermore, we prove that $S_2 \neq \varnothing$, $\forall \beta, M$. As $x \rightarrow \infty$, specifically taking $x = 2 \cdot M > 0$, we have:
	\begin{align}\label{Supmodular properties positive: infinity}
		\lim_{x \rightarrow \infty} s(x) &= x + \left( e^{-\beta x} - 1 \right) M \notag \\
		\Rightarrow \lim_{x \rightarrow \infty} s(x) &> x - M = M > 0
	\end{align}
	Thus, $S_2 \neq \varnothing$, and $S_1 \subsetneq S_2$. Therefore, $S_1 \cap S_2 \neq \varnothing$.
\end{proof}

Lemma \ref{Supmodular: Crowding-Out Effect} demonstrates the COE of supmodular function $s^i$.
\begin{restatable}{thm3}{LemmaCrowdingOutEffect}\label{Supmodular: Crowding-Out Effect}
	Crowding-Out Effect of supmodular function:
	
	$(a)$ \textnormal{If a bank $i$ (or a cluster $\mathcal{P^I}$) is added to the $\mathcal{P^+}$ such that $\tilde{\mathbf{\Phi}}^{i}\left(\mathbf{s}^{|\mathcal{P}^+\cup\{i\}|\times 1}\right)>0$ (or $\left[\tilde{\mathbf{\Phi}}\left(\mathbf{s}^{|\mathcal{P}^+\cup\mathcal{P^I}|\times 1}\right)\right]^{|\mathcal{P^I}|\times 1}>0$), then $\tilde{\mathbf{\Phi}}\left(\mathbf{s}^{|\mathcal{P}^+|\times 1}\right) - \left[\tilde{\mathbf{\Phi}}\left(\mathbf{s}^{|\mathcal{P}^+\cup\{i\}|\times 1}\right)\right]^{|\mathcal{P}^+|\times 1} < 0$ (or $\tilde{\mathbf{\Phi}}\left(\mathbf{s}^{|\mathcal{P}^+|\times 1}\right) - \left[\tilde{\mathbf{\Phi}}\left(\mathbf{s}^{|\mathcal{P}^+\cup\mathcal{P^I}|\times 1}\right)\right]^{|\mathcal{P}^+|\times 1} < 0$) holds.}
	
	$(b)$ \textnormal{If a bank $j$ (or a cluster $\mathcal{P^J}$) is added to the $\mathcal{P^-}$ such that $\tilde{\mathbf{\Phi}}^{j}\left(\mathbf{s}^{|\mathcal{P}^-\cup\{j\}|\times 1}\right) < 0$ (or $\left[\tilde{\mathbf{\Phi}}\left(\mathbf{s}^{|\mathcal{P}^-\cup\mathcal{P^J}|\times 1}\right)\right]^{|\mathcal{P^J}|\times 1} < 0$), then $\tilde{\mathbf{\Phi}}\left(\mathbf{s}^{|\mathcal{P}^-|\times 1}\right) - \left[\tilde{\mathbf{\Phi}}\left(\mathbf{s}^{|\mathcal{P}^-\cup\{j\}|\times 1}\right)\right]^{|\mathcal{P}^-|\times 1} > 0$ (or $\tilde{\mathbf{\Phi}}\left(\mathbf{s}^{|\mathcal{P}^-|\times 1}\right) - \left[\tilde{\mathbf{\Phi}}\left(\mathbf{s}^{|\mathcal{P}^-\cup\mathcal{P^J}|\times 1}\right)\right]^{|\mathcal{P}^-|\times 1} > 0$) holds.}
\end{restatable}
\begin{proof}[\normalfont\bfseries Proof of Lemma \ref{Supmodular: Crowding-Out Effect}]
	\,
	
	By the same argument of Lemma \ref{SubLemma: Crowding-Out Effect}, inequality \ref{Equation system: supmodular inequality proof a} cannot hold in scenario $(a)$ and \ref{Equation system: supmodular inequality proof b} cannot hold in scenario $(b)$.
	\begin{align}
		\tilde{\mathbf{\Phi}}\left(\mathbf{s}^{|\mathcal{P}^+|\times 1}\right) - \left[\tilde{\mathbf{\Phi}}\left(\mathbf{s}^{|\mathcal{P}^+\cup\mathcal{P^I}|\times 1}\right)\right]^{|\mathcal{P}^+|\times 1} > 0 \label{Equation system: supmodular inequality proof a} \\ \tilde{\mathbf{\Phi}}\left(\mathbf{s}^{|\mathcal{P}^-|\times 1}\right) - \left[\tilde{\mathbf{\Phi}}\left(\mathbf{s}^{|\mathcal{P}^-\cup\mathcal{P^J}|\times 1}\right)\right]^{|\mathcal{P}^-|\times 1} < 0  \label{Equation system: supmodular inequality proof b}
	\end{align}
\end{proof}

\begin{restatable}{thm}{EmergenceOfComplementarity}\label{Emergence of Substitution}
	The system $[\mathcal{B}^{\mathcal{I}}(\varepsilon)]$ exhibits substitution:

\begin{equation}\label{inequality: Emergence of Substitution}
	\frac{\partial\left[\tilde{\mathbf{\Phi}}^i\left(\mathbf{s}^{|\mathcal{B}^{\mathcal{I}}(\varepsilon)|\times 1}\right)-\tilde{\mathbf{\Phi}}^i\left(\mathbf{s}^{|[\mathcal{B}^{\mathcal{I}}(\varepsilon)]|\times 1}\right)\right]}{\partial \left[\mathbbm{1}_{\tilde{\mathbf{\Phi}}^j\left(\mathbf{s}^{|[\mathcal{B}^{\mathcal{I}}(\varepsilon)]|\times 1}\right)\geq 0}\right]}\leq 0,\,\forall i\in\mathcal{B}^{\mathcal{I}}(\varepsilon)
\end{equation}
\textnormal{and the inequality \ref{inequality: Emergence of Substitution2} holds when $\mathbbm{1}_{\tilde{\mathbf{\Phi}}^j\left(\mathbf{s}^{|[\mathcal{B}^{\mathcal{I}}(\varepsilon)]|\times 1}\right)\geq 0}>0$.}
\begin{equation}\label{inequality: Emergence of Substitution2}
	\frac{\partial\left[\tilde{\mathbf{\Phi}}^i\left(\mathbf{s}^{|\mathcal{B}^{\mathcal{I}}(\varepsilon)|\times 1}\right)-\tilde{\mathbf{\Phi}}^i\left(\mathbf{s}^{|[\mathcal{B}^{\mathcal{I}}(\varepsilon)]|\times 1}\right)\right]}{\partial \left[\tilde{\mathbf{\Phi}}^y\left(\mathbf{s}^{|\mathcal{B}^{\mathcal{I}}(\varepsilon)|\times 1}\right)-\tilde{\mathbf{\Phi}}^y\left(\mathbf{s}^{|[\mathcal{B}^{\mathcal{I}}(\varepsilon)]|\times 1}\right)\right]}\geq 0,\,\forall i,y\in\mathcal{B}^{\mathcal{I}}(\varepsilon)
\end{equation}
\end{restatable}
\begin{proof}[\normalfont\bfseries Proof of Proposition \ref{Emergence of Substitution}]
	\,
	
	It's a direct corollary of Lemma \ref{Supmodular: Crowding-Out Effect} $(a)$.
\end{proof}

\begin{restatable}{thm}{PersistenceOfComplementarity}\label{Persistence of Complementarity}
	The system $[\mathcal{B}(\varepsilon)]$ exhibits complementarity:
	
	\begin{equation}\label{inequality: Persistence of Complementarity}
		\frac{\partial\left[\tilde{\mathbf{\Phi}}^i\left(\mathbf{s}^{|\mathcal{B}(\varepsilon)|\times 1}\right)-\tilde{\mathbf{\Phi}}^i\left(\mathbf{s}^{|[\mathcal{B}(\varepsilon)]|\times 1}\right)\right]}{\partial \left[\mathbbm{1}_{\tilde{\mathbf{\Phi}}^j\left(\mathbf{s}^{|[\mathcal{B}(\varepsilon)]|\times 1}\right)\leq 0}\right]}\geq 0,\,\forall i\in\mathcal{B}(\varepsilon)
	\end{equation}
	\textnormal{and the inequality \ref{inequality: Emergence of Substitution2} also holds for any two elements in $\mathcal{B}(\varepsilon)$ when $\mathbbm{1}_{\tilde{\mathbf{\Phi}}^j\left(\mathbf{s}^{|[\mathcal{B}(\varepsilon)]|\times 1}\right)\leq 0}>0$.}
\end{restatable}
\begin{proof}[\normalfont\bfseries Proof of Proposition \ref{Persistence of Complementarity}]
	\,
	
	It's a direct corollary of Lemma \ref{Supmodular: Crowding-Out Effect} $(b)$.
\end{proof}


\end{document}